\documentclass[11pt]{article}
\usepackage[margin=0.75in]{geometry}
\usepackage{amsmath,amsthm,amssymb,amsfonts,microtype}
\usepackage{graphicx,float,mathrsfs}
\usepackage{multirow}
\usepackage{hyperref}

\usepackage{tikz}

\newtheorem{theorem}{Theorem}[section]
\newtheorem{lemma}[theorem]{Lemma}
\newtheorem{prop}[theorem]{Proposition}
\newtheorem{cor}[theorem]{Corollary}

\theoremstyle{definition}
\newtheorem{remark}{Remark}[section]
\newtheorem{example}{Example}[section]

\newcommand{\R}{\mathbb{R}}
\newcommand{\C}{\mathbb{C}}
\newcommand{\N}{\mathbb{N}}

\newcommand{\M}{\mathcal{M}}

\newcommand{\X}{\mathcal{X}}
\newcommand{\Y}{\mathcal{Y}}
\newcommand{\PP}{\mathcal{P}}
\newcommand{\PPP}{\tilde{\PP}}

\newcommand{\scri}{\mathscr{I}}
\newcommand{\g}{\mathbf{g}}
\newcommand{\h}{\mathbf{h}}

\newcommand{\U}{\mathcal{U}}
\newcommand{\V}{\mathcal{V}}
\newcommand{\Vb}{\V_{\mathrm{b}}}

\newcommand{\tm}{{}^{\mathrm{de,sc}}T\M}
\newcommand{\ttm}{{}^{\mathrm{de,sc}}T^*\M}
\newcommand{\sym}[1]{{S}^{\mathsf{#1}}(\PP)}

\newcommand{\symm}[1]{{S}^{\mathsf{#1}}(M)}
\newcommand{\symmcl}[1]{{S}^{\mathsf{#1}}_{\mathrm{cl}}(M)}

\newcommand{\soe}{S^{\mathsf{0}}_{\epsilon}}

\newcommand{\db}{\mathrm{Diff}_{\mathrm{b}}}
\newcommand{\ps}[1][]{\Psi_{\mathrm{de,sc}}^{\mathsf{#1}}(\M)}

\newcommand{\Hp}{\tilde{H}_p}

\newcommand{\wf}[1]{\mathrm{WF}_{\mathrm{de,sc}}^{\mathsf{#1}}}
\newcommand{\wfs}{\mathrm{WF}_{\mathrm{de,sc}}'}

\newcommand{\LL}{L^2_{\mathrm{de,sc}}(\M)}
\newcommand{\sob}[1]{H^{\mathsf{#1}}_{\mathrm{de,sc}}(\M)}

\newcommand{\sch}{\mathcal{S}}

\newcommand{\os}{\otimes_{\mathrm{sym}}}

\usepackage[backend=biber,sorting=nyt,giveninits=true,maxbibnames=10]{biblatex}
\renewbibmacro{in:}{\ifentrytype{article}{}{\printtext{\bibstring{in}\intitlepunct}}}
\AtEveryBibitem{\clearfield{number}\clearfield{issue}\clearfield{month}}
\DeclareFieldFormat{pages}{#1}
\begin{document}

\title{The Feynman propagator for massive Klein-Gordon fields on radiative asymptotically flat spacetimes}
\author{
    {Mikhail Molodyk\thanks{Stanford University, Department of Physics and Leinweber Institute for Theoretical Physics at Stanford, Stanford, CA 94305, USA. E-mail: mam765@stanford.edu} }
    and 
    {Andr\'as Vasy\thanks{Stanford University, Department of Mathematics, Stanford, CA 94305, USA. E-mail: andras@math.stanford.edu}}    
    }
\date{}
\maketitle

\begin{abstract}
On a large class of asymptotically flat spacetimes which includes radiative perturbations of Minkowski space, we define a distinguished global Feynman propagator for massive Klein-Gordon fields by means of the microlocal approach to non-elliptic Fredholm theory, working in the de,sc-pseudodifferential algebra due to Sussman. We extend the limiting absorption principle (the "$i\varepsilon$ prescription" for the Feynman propagator) to this setting. Motivated by the complicated Hamilton flow structure arising in this problem, we also prove a new localized radial point estimate in the spirit of Haber--Vasy which, under appropriate nondegeneracy assumptions, allows one to propagate microlocal regularity into a single radial point belonging to a larger radial set which can be a source, sink, or saddle for the Hamilton flow.
\end{abstract}

\section{Introduction}

In this paper, we analyze the Klein-Gordon operator $P=\Box_{\g}+m^2$, where $\g$ is the Lorentzian (signature $-,+,\ldots,+$) metric of a spacetime $\M^{\circ}$, $\Box_{\g}=d^*d$ its wave operator (adjoints always being with respect to $L^2(\M^{\circ},\g)$), and $m>0$, on a large class of asymptotically flat spacetimes which includes radiative perturbations of Minkowski space. Our main result is a canonical definition of distinguished retarded, advanced, Feynman, and anti-Feynman propagators as inverses of $P$ acting between weighted Sobolev spaces whose orders encode the direction of propagation of positive- and negative-frequency singularities. While retarded and advanced propagators are uniquely defined on any globally hyperbolic spacetime by their support properties, the problem of defining distinguished Feynman and anti-Feynman propagators is more subtle and related to the problem of defining distinguished states on the algebra of free quantum fields in a curved spacetime; see Section~\ref{sec:propagators} for a review of work on this subject. A Lorentzian version of the limiting absorption principle, which states that the Feynman propagator can be understood as a limit $\lim_{\varepsilon\to 0^+} (P-i\varepsilon)^{-1}$ and which was shown for less general asymptotically Minkowski spacetimes by Vasy \cite{Vasy-SA} and Taira \cite{Taira}, remains valid in this setting. We use the framework of double-edge--scattering (de,sc-) pseudodifferential operators developed by Sussman \cite{Sussman}, which we extend to be applicable to a more general class of metrics similar to those considered by Hintz and Vasy \cite{HV-eb} in their related analysis of the massless wave equation. This is necessary to include physical examples with gravitational radiation escaping through null infinity, as arising from solving Einstein's equation with perturbed Minkowski initial data \cite{HV-stability}.

The simplest version of our results, in the Feynman case specifically, can be stated slightly informally as follows.
\begin{theorem}
Let $(\R^{3+1},\g)$ be a small perturbation of Minkowski space in the sense of solutions to the vacuum Einstein initial value problem, with symbolic initial data given by small-mass asymptotically Schwarzschild data plus small terms decaying slightly faster than Schwarzschild-like data as in Theorem~1.1 of \cite{HV-stability}. Then for any $f\in\sch(\R^{3+1})$, there exists a unique solution $u\in\sch'(\R^{3+1})$ to $Pu=f$ such that the de,sc-wavefront set of $u$ is contained in the union of the positive mass shell over future timelike infinity and the negative mass shell over past timelike infinity. It satisfies $u=\lim_{\varepsilon\to 0^+} (P-i\varepsilon)^{-1}f$ in $\sch'(\R^{3+1})$.
\label{thm:intro}
\end{theorem}
Here $\sch(\R^{3+1})$ and $\sch'(\R^{3+1})$ are the spaces of Schwartz functions and tempered distributions on spacetime, respectively. In the last statement, $(P-i\varepsilon)^{-1}$ is understood as the inverse of $(P-i\varepsilon):\sch(\R^{3+1})\to\sch(\R^{3+1})$, whose invertibility for metrics of this class is shown in Theorem~\ref{thm:invertibility-complex}, cf.\ \cite{JMS}. The result remains valid for much more general spacetimes which, in a precise sense, have the same asymptotic structure and are non-trapping. Moreover, as already mentioned, the existence-and-uniqueness statement is a rough version of a more precise characterization of $P$ as an invertible operator acting between certain Hilbert spaces based on \textit{weighted de,sc-Sobolev spaces}; see Theorems~\ref{thm:invertibility} and \ref{thm:intro-precise}. The characterization as a limit in the topology of $\sch'(\R^{3+1})$ is a weaker version of Theorem~\ref{thm:lim-abs}, which shows convergence in a weighted Sobolev space. The wavefront set property gives a concrete meaning to the notion of solutions to the Klein-Gordon equation which are asymptotically positive- or negative-frequency in the far future or far past.

Sobolev spaces and microlocalization in geometric microlocal analysis need not be mysterious; they correspond to a certain scaling of the momentum variables as one approaches infinity, i.e.\ the boundary when one compactifies the underlying physical space, as we do. Concretely, the sc-Sobolev spaces are just standard weighted Sobolev spaces on $\R^{3+1}$, based on regularity with respect to the translation invariant vector fields $\partial_{z_j}$, $z=(t,x)$. Away from null infinity these are the spaces used, with microlocalization with respect to the dual momenta $\zeta_j$, so for instance for bounded $\zeta$, this means (in the region $|x|<C|t|$ for some $C>0$, for simplicity) localization near a certain $\zeta$ and $x/t$ for large $|t|$, i.e. asymptotically conically in position space and in the standard sense in momentum space. (Thus, the roles of position and momentum are reversed relative to standard microlocal analysis in this region.) On the other hand, at null infinity a different, the double edge (de), scaling is used because this desingularizes the metric and its Hamilton flow. The compactification of the position space as a manifold with corners combines all asymptotic regimes, in a way described in Section~\ref{sec:geom}, and is done in a way which assures that the metric is conormal, or symbolic, i.e.\ regular with respect to vector fields tangent to all boundary faces. We remark here that it might be tempting to use a different microlocalization at null infinity, based on $3$-body scattering \cite{Vasy:Propagation-2}, which in a sense would be the `obvious' choice as it desingularizes the metric much as in $N$-body scattering and which would {\em not} change the basic Sobolev spaces; however this choice still leaves a degenerate Hamilton flow: it is the Lorentzian (rather than Riemannian) nature of the metric that allows making a different, double edge, choice, that also makes the flow well-behaved.

Our general approach is an instance of the microlocal approach to Fredholm theory for non-elliptic operators, based on combining (microlocal) elliptic estimates, propagation-of-singularities estimates, and radial point estimates to prove that the operator in question is Fredholm when acting between suitable function spaces. This method, originating in \cite{Vasy-AH-KdS}, has since been successfully applied to wave/Klein-Gordon operators on various classes of spacetimes \cite{HV-semilinear,GR-H-V,Vasy-positivity,B-D-GR-causal,B-D-GR-Feynman} as well as Helmholtz \cite{Grenoble-notes,GR-H-Sh-Zh} and time-dependent Schr\"{o}dinger \cite{GR-G-H} operators. Going a step further to prove invertibility of the operator between those spaces usually requires additional input, for wave equations typically (as in our case) in the form of some additional causal structure assumptions which allow one to use energy estimates.

In particular, Gell-Redman, Haber, and Vasy \cite{GR-H-V} used this method to define the Feynman propagator for the massless wave equation on a more restrictive class of spacetimes well-behaved on a compactification modeled on the radial compactification of Minkowski space, in which null infinity is a codimension-one submanifold of the boundary and which is poorly equipped to describe settings with radiation. In the same class of spacetimes, the case of the massive Klein-Gordon equation, which is analytically simpler, is briefly discussed in ~\cite{Grenoble-notes}. Our work is an extension of the construction for massive fields to more general spacetimes, which are well-behaved on a compactification with separate boundary hypersurfaces for each of timelike, spacelike, and null infinity, thus combining features of the radial compactification and the Penrose compactification.

Versions of this setting have been used by Baskin, Vasy, and Wunsch \cite{BVW,BVW-long-range} to study asymptotics of solutions to the wave equation on a more restrictive class of spacetimes and by Hintz and Vasy \cite{HV-stability} to study stability of Minkowski space. We remark here that the latter has a long history in mathematics, starting with the groundbreaking works of Christodoulou and Klainerman \cite{Christodoulou-Klainerman:Global}, later simplified by Lindblad and Rodnianski \cite{Lindblad-Rodnianski:Global-stability}, and with the asymptotic structure we need obtained in the aforementioned \cite{HV-stability}. In the physics literature, where one usually considers the Penrose diagram, rather than the radial compactification, as the starting point, similar frameworks for considering all of these asymptotic regimes simultaneously have been developed by Comp\`{e}re, Gralla, and Wei \cite{C-G-W} and Borthwick, Chantreau, and Herfray \cite{BCH}. Sussman \cite{Sussman} and Hintz and Vasy \cite{HV-eb} introduced the tools appropriate for microlocal analysis on this compactification of massive and massless waves respectively, based on the \textit{double-edge--scattering} (de,sc-) and \textit{edge-b} (e,b-) pseudodifferential calculi.

The Klein-Gordon operator with $m>0$ has nondegenerate principal symbol in the de,sc-calculus, which is fully symbolic. This means that the Fredholm setup only requires understanding propagation of singularities {\em along the Hamilton (or bicharacteristic) flow} associated to the symbol (including at radial points) -- there are no ``normal operator" considerations, which arise in the massless case. Similarly to the radial-compactification setting, we use variable-order spaces for the Feynman and anti-Feynman propagators since singularities are propagated in opposite directions (from a physical-space perspective) in the positive- and negative-frequency components of the characteristic set. However, since we track the decay at all five of past/future timelike infinity, past/future null infinity, and spacelike infinity separately, the variable orders can be taken constant on each connected component of the characteristic set, so for propagation of singularities their variability is irrelevant.

The Hamilton flow structure on the compactification with a blown-up null infinity is significantly more complicated than on the radial-type compactification; this was already observed in \cite{HV-eb} and \cite{Sussman}, with the latter, massive, case possessing additional complexity in that {\em there is a radial submanifold which is connected to itself by a sequence of bicharacteristics limiting both backward and forward to various components of the radial set.} To deal with this systematically, we prove a general-purpose \textit{localized} radial point estimate (Theorem~\ref{thm:localized-rp-main}) similar to that due to Haber and Vasy \cite{Haber-Vasy}, which under some nondegeneracy assumptions allows one to propagate regularity completely microlocally into a radial set, i.e.\ into a single radial point belonging to a larger radial set. Our result is applicable to situations where the radial set is a source, sink, or saddle for the flow in normal directions and is valid under mild decay requirements on non-smooth error terms in the operator's symbol. The localized nature of the result deals with the above mentioned complexity as the parts of the relevant radial set connected via a chain of bicharacteristics, as described above, lie in different boundary components (spacelike vs.\ timelike boundary) of the same radial submanifold. As indicated already, a further difficulty, making the localization arguments within the radial set more challenging, is that in order to accommodate radiative spacetimes we {\em only assume that the metric and thus the wave operator have conormal, or symbolic, regularity} (rather than being smooth on the compactification); however, we point out that in the actual setting of radiative asymptotically Minkowski spacetimes some constructions simplify somewhat, see the remarks after the statement of Theorem~\ref{thm:localized-rp-main}. We hope this technical contribution of our paper is of independent interest.

\subsection{Propagators in curved spacetime and related work}
\label{sec:propagators}
Consider a time-oriented globally hyperbolic spacetime $(M,\g)$ and the Klein-Gordon operator $P=\Box_{\g}+m^2$ for some $m\in\R$. The retarded propagator $P_R^{-1}$ and advanced propagator $P_A^{-1}$ are the unique operators $P_R^{-1},P_A^{-1}:C_c^{\infty}(M)\to C^{\infty}(M)$ such that $P\circ P_R^{-1} = P\circ P_A^{-1}=I$ and $\mathrm{supp}(P_R^{-1}f)$ is contained in the causal future of $\mathrm{supp}(f)$ and $\mathrm{supp}(P_A^{-1}f)$ is contained in the causal past of $\mathrm{supp}(f)$ for all $f\in C_c^{\infty}(M)$. The existence and uniqueness of these propagators is closely related to the well-posedness of the Cauchy problem for $P$ \cite{B-G-P}.

In $(d+1)$-dimensional Minkowski spacetime, the Feynman propagator $P_+^{-1}$ and the anti-Feynman propagator $P_-^{-1}$ are defined by
\begin{equation}
P_+^{-1}f = P_R^{-1}(f_+) + P_A^{-1}(f_-),
\hspace{50pt}
P_-^{-1}f = P_R^{-1}(f_-) + P_A^{-1}(f_+),
\label{eq:Feynman-pos-freq}
\end{equation}
where $(f)_{\pm} = \mathcal{F}_t^{-1}(\theta_{\pm}\cdot \mathcal{F}_t(f))$, $\theta_{\pm}(\omega)=\mathbf{1}_{\{\mp\omega>0\}}$, $\mathcal{F}_t$ the Fourier transform in time, is the operation of taking the \textit{positive-/negative-frequency part} of $f$.\footnote{Sign conventions vary. Our conventions are chosen so that $u_+(t)=e^{-i\omega t}$ for $\omega>0$ is considered a positive-frequency oscillation and $u_-(t)=e^{+i\omega t}$ a negative-frequency one, as usual in QFT, but the Fourier transform, including in time, is defined using the usual PDE convention $\hat{u}(\omega) = \int u(t) e^{-i\omega t}\ dt$. Thus $\hat{u}_{\pm}$ is supported at $\mp \omega$ (hence the definition of $\theta_{\pm}$), and therefore the wavefront set corresponding to a purely \textit{positive}-frequency oscillation consists of momenta with a \textit{negative} coefficient in front of $dt$ and vice versa. Then with our conventions for $\Box$ and the associated Hamilton vector field, the latter flows forward in time in the component of the characteristic set where positive-frequency oscillations have wavefront sets, i.e. in usual coordinates in Minkowski space the temporal component is $-2\omega\frac{\partial}{\partial t}$. \label{note:signs}} 
One can check that $P_{\pm}^{-1}$ are indeed solution operators for $P$. This definition is motivated by quantum field theory, where the integral kernel of the Feynman propagator is the time-ordered two-point function of the free scalar field's vacuum state and appears in perturbative expressions for scattering amplitudes of interacting theories. For $m>0$, an equivalent definition is
\begin{equation}
P_+^{-1}f
 = 
 \lim_{\varepsilon\to 0^+} \mathcal{F}^{-1}\left( \frac{\mathcal{F}f(\omega,k)}{-\omega^2+k^2+m^2-i\varepsilon} \right)
 =
\lim_{\varepsilon\to 0^+} (P-i\varepsilon)^{-1} f,
\label{eq:Feynman-i-epsilon}
\end{equation}
where $\mathcal{F}$ is the full $(d+1)$-dimensional Fourier transform, $(\omega,k)$ are the dual variables to $(t,x)$, and in the second formula $(P-i\varepsilon)^{-1}$ can be understood in the context of the functional calculus for self-adjoint operators, since $P$ is essentially self-adjoint on the domain $C_c^{\infty}(\R^{d+1})$ with respect to $L^2(\R^{d+1})$. The other three propagators can be characterized in a similar manner as limits of Fourier multipliers which tend to $\frac{1}{-\omega^2+k^2+m^2}$ while avoiding the two poles at $\omega=\pm\sqrt{k^2+m^2}$ in different ways.

Since the global notion of positive/negative frequency does not meaningfully generalize to curved spacetimes without any symmetries, the Feynman propagator may not have a natural generalization to all spacetimes either. From the physical point of view, this is related to the fact that there is no distinguished vacuum state for QFT in a general curved spacetime, whereas in Minkowski space the vacuum is distinguished by invariance under the action of the global Poincar\'{e} group.

A broad generalization of the notion of a Feynman propagator to any globally hyperbolic spacetime was discovered by Duistermaat and H\"{o}rmander \cite{D-H}  (in fact under weaker conditions than global hyperbolicity). The characteristic set of the Klein-Gordon operator consists of the dual lightcones over every point of spacetime, which splits into a future-directed and a past-directed connected component if the spacetime is time-oriented and connected. By microlocal elliptic regularity, any solution operator for $P$, and in fact any parametrix (i.e. solution operator modulo smoothing operators), when acting on a distribution can only create new singularities (as measured by the wavefront set) within the characteristic set, whose positive- and negative-frequency components \textit{are} well-defined (identified with the past- and future-directed components by a choice of sign conventions in the definition of wavefront sets). In the characteristic set, singularities propagate along the Hamilton flow of the dual metric function, that is the lifted geodesic flow on the cotangent bundle. Duistermaat and H\"{o}rmander showed that a class of parametrices for $P$ modulo smoothing operators is uniquely specified by postulating, \textit{independently in each component} of the characteristic set, that the parametrix creates new singularities either \textit{only} downstream or \textit{only} upstream along the flow from where they were originally present. The direction of propagation for the retarded/advanced propagators follows from their support properties, with singularities in both components of the characteristic set propagating in the same direction with respect to time, or equivalently in opposite directions with respect to the Hamilton flow. The Feynman propagator in $\R^{d+1}$, meanwhile, propagates positive-frequency singularities to the future and negative-frequency singularities to the past, or equivalently all singularities forward along the Hamilton flow (given appropriate sign conventions); the anti-Feynman propagator has the opposite properties. Thus, in any globally hyperbolic spacetime, the four equivalence classes are called the retarded, advanced, Feynman, and anti-Feynman distinguished parametrices for $P$.\footnote{We warn the reader that the term ``propagation" is used in this context in at least three different ways, which leads to different ``directions of propagation" being associated to the same propagators. When one says that the retarded propagator propagates all singularities forward in time, or the Feynman propagator forward along the Hamilton flow, one means that when acting on a distribution, the operator creates new singularities only in the specified region relative to where the original distribution was singular. On the other hand, H\"{o}rmander's propagation of singularities theorem refers to propagation of \textit{information about} singularities: e.g. if $u$ is a solution to $Pu=f$, then in a region where $f$ is smooth, if $u$ is singular at a point of the characteristic set we know that it is also singular at all points on the same flow line within the region; if, for instance, $u$ is the retarded solution, i.e. known to be trivial in the far past, then this implies that $f$ has singularities somewhere \textit{in the past} along this flow line. In addition, in proving propagation theorems one actually tracks information about \textit{regularity} of solutions, which propagates in the opposite direction relative to information about singularities.} See \cite{I-S} for analogous constructions for vector and spinor fields.

This result was connected to quantum field theory by Radzikowski \cite{Radzikowski}, who gave a microlocal formulation of the Hadamard condition for states of QFT in curved spacetime, which specifies a class of states for which the expected value of the stress-energy tensor can be defined by analogy with the Minkowski vacuum. Radzikowski showed that the condition amounts to demanding that the two-point function only have singularities along the forward-pointing lightcone, thus generalizing the fact that the two-point function of the Minkowski vacuum is the positive-frequency part of the commutator function. For any Hadamard state, the time-ordered two-point function is a Feynman parametrix, in fact an exact solution operator acting on $C_c^{\infty}(M)$, which can thus be considered the \textit{Feynman propagator associated to the state}. This characterization led to many rigorous constructions in QFT in curved spacetime, including perturbative interacting theories starting with the work of Brunetti and Fredenhagen \cite{B-F}, by direct analogy with QFT in Minkowski space, with the vacuum state replaced by an arbitrarily chosen Hadamard state.

Since any globally hyperbolic spacetime admits many Hadamard states \cite{F-N-W,G-W-pseudo,Lewandowski}, Feynman propagators in this sense are not unique, which is to be expected because the preceding discussion was concerned almost completely with local/short-distance/UV considerations. The main use of the Feynman propagator in conventional QFT, however, is as a building block for scattering amplitudes, which describe the evolution of particle states in the far past into other particle states in the far future. To generalize this picture to curved spacetime, one presupposes the existence of an ``in-vacuum" state $\Omega_{\mathrm{in}}$ and an ``out-vacuum" state $\Omega_{\mathrm{out}}$, which are interpreted as having no particles as experienced by observers in the far past and far future respectively. Scattering phenomena are then described entirely in terms of these two states, and the role of the Feynman propagator is played not by any time-ordered expected value but rather by a time-ordered matrix element of the heuristic form $\frac{\langle \Omega_{\mathrm{in}}|T\phi(x)\phi(y)|\Omega_{\mathrm{out}}\rangle}{\langle\Omega_{\mathrm{in}} | \Omega_{\mathrm{out}}\rangle }$. One can make sense of this expression if there exist distinguished in-/out-vacua giving rise to particle interpretations related by a well-defined S-matrix, and it is expected that, among other settings, these constructions are possible for spacetimes which are ``sufficiently asymptotically flat" (see \cite{Wald-S-matrix} for the analysis when the metric is exactly Minkowski outside a compact set, and \cite[Sections 4.3-4.4]{Wald-book} for the general framework). If this is the case, then this matrix element is another object one could reasonably call \textit{the} Feynman propagator, one that in fact would be unique.

This motivates the search for solution operators to $P$ which are not only Feynman parametrices in the sense of Duistermaat-H\"{o}rmander but are canonically determined by the asymptotic structure of the spacetime. Such an operator should be distinguished by the fact that, when it acts e.g. on a compactly supported smooth function $f$, it outputs a solution $u$ to $Pu=f$ which is ``asymptotically positive-frequency" in the far future and ``asymptotically negative-frequency" in the far past.

One setting in which there are distinguished vacuum states is that of static spacetimes, which have a distinguished time-slicing. A rigorous construction of Feynman propagators in this case was given by Derezi\'{n}ski and Siemssen \cite{D-S-static}, who also considered generalizations to several other settings with a well-behaved time-slicing \cite{D-S-KG,D-S-QFT}.

More in the spirit of Duistermaat-H\"{o}rmander, in settings without an explicit time-slicing progress has been enabled by the development of microlocal tools for analysis on manifolds with boundary arising as compactifications of geometrically non-compact spaces, especially starting with the work of Melrose \cite{Melrose-APS,Melrose-AES}. The asymptotic structure of spacetime can be encoded in a choice of compactification with respect to which the metric, and therefore the wave operator, degenerates at infinity in a controlled way. The wave operator can then be considered as an element of a pseudodifferential calculus tailored to the compactification, with a generalized notion of wavefront set tracking frequencies at which functions fail to \textit{decay} at any given point of spacetime infinity in addition to the information tracked by the usual wavefront set, namely the (co-)directions (parametrized by the sphere at infinity in frequency space) in which functions fail to be smooth at any given point in the spacetime interior. Microlocal elliptic and propagation estimates extend to these settings; however, due to the compactness of the phase space, the Hamilton flow of an operator usually has critical, or \textit{radial}, points, to propagate regularity information into which one proves additional \textit{radial point estimates}. This enables one to prove results on smoothness \textit{and} decay of solutions $u$ to $Pu=f$ given such information about $f$, which, informally speaking, provides the tools to impose boundary conditions ensuring existence and uniqueness of solutions, thus defining inverses of $P$. Returning to wave equations, this allows one to specify solutions which, independently in the positive- and negative-frequency components of the characteristic set, have no wavefront set over either the asymptotic future or asymptotic past boundary, giving an invariant meaning to the notion of solutions which are asymptotically positive- or negative-frequency in those regions. 

The first constructions of this sort used the \textit{radial compactification} of Minkowski space, that is the compactification of $\R^{d+1}$ into a ball by adding a sphere at infinity corresponding to $r=\sqrt{x^2+t^2}=\infty$ in spacetime spherical coordinates, since this setting is amenable to tools developed earlier for Riemannian-signature problems, such as Melrose's \textit{boundary} (b-) and \textit{scattering} (sc-) pseudodifferential calculi. Thus, Gell-Redman, Haber, and Vasy \cite{GR-H-V} carried out this program to show that for small perturbations of Minkowski space in the sense of smooth \textit{scattering metrics} on the radial compactification, the massless wave operator is invertible between spaces encoding the Feynman conditions based on \textit{weighted b-Sobolev spaces}, identifying the inverse as the Feynman propagator. A similar analysis, but using sc-Sobolev spaces, applies to the Klein-Gordon operator with $m>0$. Besides being analytically simpler, at first glance the sc-framework is more natural than the b-framework for this purpose regardless of the mass since the wave operator is a sc-operator, but when $m=0$ it has a degenerate principal symbol, so it is convenient to analyze a rescaled version of the operator which is nondegenerate as a b-operator. (One can also work with a combined sc-b ps.d.o.\ algebra introduced in \cite{Vasy:Zero-energy,Vasy:Limiting-absorption-lag} for a related but different purpose; in a sense this is the systematic desingularization approach to the massless problem.) G\'{e}rard and Wrochna \cite{GW-FP1,GW-FP2} also constructed Feynman propagators for the Klein-Gordon equation with $m>0$ on perturbations of Minkowski space with error terms which are symbolic on the radial compactification which admit a well-behaved time-slicing; their method, related to work by B\"{a}r and Strohmaier on the Dirac equation \cite{BS}, uses a description of time-evolution with respect to the time-slicing and scattering data rather than propagation arguments.

The radial compactification, however, is poorly adapted to model many situations of physical interest in asymptotically flat spacetimes. Ultimately this is because one is interested in metrics which solve Einstein's equation, and their asymptotic behavior is indicated by solutions of the linearized equations, which even for very well-behaved (Schwartz) forcing produce singularities at certain parts of the radial compactification, corresponding to null infinity, i.e. the points where $\frac{r}{t}=\pm 1$ on the boundary of the radial compactification. Thus, in this paper, we consider spacetimes which may include radiation, which leads to more complicated behavior at null infinity; see Section~\ref{sec:geom} for details. On the other hand, Baskin, Doll, and Gell-Redman \cite{B-D-GR-causal,B-D-GR-Feynman} recently considered the massive Klein-Gordon equation on Minkowski space with asymptotically static spatially decaying potentials and/or first-order perturbations, which instead leads to more complicated behavior at the points where $\frac{r}{t}=0$ on the boundary of the radial compactification. They also use the general non-elliptic Fredholm framework (in a different pseudodifferential calculus) to define all four propagators. Since the regions in which these two cases require a resolution relative to the radial compactification are disjoint and the global results are proven by joining microlocal ones together in a modular fashion, it should be possible to combine the two approaches to treat problems which involve both of these kinds of behavior, as already pointed out in \cite{JMS}.

These methods for defining Feynman propagators can be roughly thought of as generalizing Eq.~(\ref{eq:Feynman-pos-freq}) since they directly prescribe direction of propagation for positive- and negative-frequency singularities. The other characterization, Eq.~(\ref{eq:Feynman-i-epsilon}), on the other hand, is generalized by the limiting absorption principle, which says that $\lim_{\varepsilon\to 0^+}(P-i\varepsilon)^{-1}$ exists in an appropriate sense and is in fact the Feynman propagator. In the radial-compactification setting, this follows from the fact that a similar Fredholm framework applies to the operators $P-i\varepsilon$ (though propagation of singularities is only allowed in one direction relative to the Hamilton flow) and the corresponding estimates are uniform in $\varepsilon>0$ near zero, which yields existence and uniform boundedness of the inverses $(P-i\varepsilon)^{-1}$, and the limit follows \cite{Vasy-SA}. This argument is very robust: once the Fredholm framework is established in our more general setting, the limiting absorption principle follows in an essentially identical manner. We note that the proof uses the fact that the Feynman (i.e. $\varepsilon=0$) inverse exists, so at least in this approach this is not an independent definition of the Feynman propagator. We also note that while $P-i\varepsilon$ for any $\varepsilon>0$ is elliptic at finite frequency, hence its invertibility can be shown using the simpler propagation of singularities at infinite frequency only (as in \cite{JMS}), establishing uniform boundedness of the inverses does require propagation in the full characteristic set of $P$ because the family $P-i\varepsilon$ is not \textit{uniformly} elliptic there as $\varepsilon\to 0^+$.

The limiting absorption principle is often discussed in the context of the question of essential self-adjointness of the wave operator with respect to the $L^2$ space of the metric, and the Feynman propagator is then interpreted as the boundary value of the resolvent \cite{D-S-static,Vasy-SA,NT-RPT,Taira,NT-static-compact,NT-KG}. For radiative perturbations of Minkowski space, the wave operator was shown to be essentially self-adjoint in \cite{JMS} using the same tools we use in this paper.

See also \cite{DG} for a review and comparison of different global notions of Feynman propagators relevant for various classes of spacetimes of physical interest.

\subsection{Plan of the paper}
The remainder of the paper is organized as follows.
\begin{itemize}
\item In Section~\ref{sec:geom}, we describe our assumptions on the spacetime and the examples motivating them. While our most general assumptions are rather abstract and refer explicitly to the bicharacteristic flow, the reader can keep in mind the most physically interesting case, general small perturbations of Minkowski space as a solution to the vacuum Einstein initial value problem, whose relationship to other classes we consider is discussed in Example~\ref{ex:perturbations}.

\item In Section~\ref{sec:de-sc}, we review Sussman's double-edge--scattering calculus (in its natural geometric generalization and including variable-order spaces) and prove a version of the propagation of singularities theorem and a localized radial point estimate which are applicable to wave operators of metrics of our class.

\item In Section~\ref{sec:Hamiltonian}, we describe in detail the Hamilton flow associated to the Klein-Gordon operator $P$ in its characteristic set, largely mirroring \cite[Section 4]{Sussman}.

\item In Section~\ref{sec:global}, we combine the microlocal results into a result on the global regularity of solutions to $Pu=f$, define the retarded, advanced, Feynman, and anti-Feynman Fredholm realizations of the operator, and show that under an additional causal structure assumption they are in fact invertible, thus defining propagators. Along the way, we attempt to provide a clearer explanation of the non-standard positive-commutator argument proving that solutions to $Pu=0$ satisfying the Feynman/anti-Feynman conditions must be Schwartz, based on a construction of Isozaki \cite{Isozaki}; this makes them amenable to energy estimates in the proof of invertibility, which is automatically the case for the retarded/advanced solutions because they are by definition Schwartz near one of the timelike infinities.

\item In Section~\ref{sec:LAP}, we prove the limiting absorption principle and discuss implications for the uniqueness of Feynman propagators.

\item Finally, in Appendix~\ref{sec:Minkowski-like}, we show that for metrics with a particular structure on the whole boundary of the compactification, generalizing perturbations of Minkowski space and similar to that considered in \cite{BVW,BVW-long-range,HV-eb} near null infinity, some of our assumptions on the metric are satisfied automatically, and in Appendix~\ref{sec:time-fcn} we show that the extra condition required for invertibility is always satisfied for these metrics.
\end{itemize}

\subsection{Acknowledgments}
The authors gratefully acknowledge support from the National Science Foundation under grant
numbers DMS-1953987, DMS-2247004 (AV) and PHY-2014215, PHY-2310429 (MM). AV is also grateful for support from the Simons Foundation via a Simons Fellowship. The authors are grateful to Jan Derezi\'{n}ski and Robert Wald for very helpful discussions on notions of Feynman propagators in curved spacetime, to Andrew Hassell, Qiuye Jia, Eva Silverstein, and Micha{\l} Wrochna for comments on a draft of the paper, and especially to Ethan Sussman for extensive discussions and comments.

\section{Geometric setting}
\label{sec:geom}
	Our analysis applies to asymptotically flat spacetimes whose metric is well-behaved on a compactification introduced by Sussman \cite{Sussman} and Hintz and Vasy \cite{HV-eb}, which contains boundary faces for each of past and future timelike infinity, past and future null infinity, and spacelike infinity. This compactification combines features of the radial compactification, which resolves points at timelike and spacelike but not null infinity, and the Penrose compactification (which provides the usual paradigm for studying asymptotically flat spacetimes with physical-space methods), which has the opposite features. The advantage of the radial compactification is that it is well-adapted to the use of microlocal tools; however, to treat radiative spacetimes, a resolution of null infinity is necessary because the metrics are only well-behaved after this resolution.

	In this section, we review this geometric setting. We do not restrict the topology and therefore follow the more coordinate-invariant description in \cite{HV-eb} but describe the metric in de,sc-terms rather than e,b-terms.

	\subsection{Topology and smooth structure}
	\label{sec:topology}
	
	For a smooth manifold with boundary or corners $M$, we always distinguish the space $C^{\infty}(M)$ of functions which are smooth up to and including the boundary and the space $C^{\infty}(M^{\circ})$ of functions which need only be smooth in the interior. A \textit{defining function} of a compact embedded {\em product-type (p-) or neat} hypersurface $N\subset M$ (i.e.\ one for which in local coordinates the boundary hypersurfaces and $N$ are simultaneously given by the vanishing of a subset of the coordinates) with orientable normal bundle is a function $\rho\in C^{\infty}(M)$ such that $N=\{\rho=0\}$ and $d\rho$ does not vanish on $N$. Any such hypersurface has a defining function, any two of which differ only by a nowhere-vanishing smooth factor. If $N$ is a boundary hypersurface of $M$, then we also demand that $\rho\geqslant 0$.
	
		Let $\tilde{\M}$ be an orientable compact smooth manifold with boundary; the spacetime is identified with the interior of $\tilde{\M}$, and we set $\dim\tilde{\M}=d+1$. As in \cite{BVW,BVW-long-range}, we assume that there is a specified embedded submanifold $Y\subset \partial\tilde{\M}$ of codimension 1 with orientable normal bundle. We fix a defining function $\rho$ of $\partial\tilde{\M}$ as well as $v\in C^{\infty}(\tilde{\M})$ which extends a defining function of $Y$ within $\partial\tilde{\M}$. We assume the following properties:
	\begin{itemize}
	\item $Y=Y^+\sqcup Y^-$, where $Y^+,Y^-$ are two disjoint embedded submanifolds of $\partial\tilde{\M}$ (which we identify as future and past null infinity).
	\item The set $\{\rho=0,\ v > 0\}\subset\partial \tilde{\M}$ is of the form $\tilde{I}^+ \sqcup \tilde{I}^-$ (which we identify as future and past timelike infinity), where $\partial \tilde{I}^{\pm}=Y^{\pm}$.
	\end{itemize}
	We denote the set $\{\rho=0,\ v < 0\}\subset\partial \tilde{\M}$ (which we identify as spacelike infinity) by $\tilde{I}^0$.
		
		Define the blowup $\M=[\tilde{\M};Y;\frac{1}{2}]$; this means we first perform the polar blowup of $Y\subset \tilde{\M}$ and then modify the smooth structure by adjoining the square root of a defining function of the front face (see \cite{Sussman,HV-eb} for motivation and details). We denote the closures of the lifts of $\tilde{I}^{\pm},\tilde{I}^0$ to $\M$ by $I^{\pm},I^0$ respectively and write $I^T=I^+\sqcup I^-$. We denote the lifts of $Y^{\pm}$ to $\M$ by $\scri^{\pm}$ and write $\scri=\scri^+\sqcup\scri^-$. $\M$ is a smooth manifold with corners; we denote $\mathcal{G}(\M)=\{I^-,\scri^-,I^0,\scri^+,I^+\}$ its set of boundary hypersurfaces. For the corners, we use the notations
		$$
		\scri^-_-=\scri^-\cap I^-,
		\hspace{30pt}
		\scri^-_+=\scri^-\cap I^0,
		\hspace{30pt}
		\scri^+_-=\scri^+\cap I^0,
		\hspace{30pt}
		\scri^+_+=\scri^+\cap I^+.
		$$

	\subsection{Neighborhoods of the corners}
	\label{sec:neighborhoods}
	
			There exist product neighborhoods of $Y^{\pm}$ in $\tilde{\M}$ of the form $U_{\scri^{\pm}}\simeq [0,\varepsilon)_{\rho}\times (-\varepsilon,\varepsilon)_v\times Y^{\pm}$. Let $\U_{\scri^{\pm}}$ denote the lifts of $U_{\scri^{\pm}}$ to $\M$. In the interior of $\U_{\scri^{\pm}}$, we can define new coordinates $u=\frac{v}{\rho}$ and $r=\frac{1}{\rho}$; note that $r\to + \infty$ at $\partial\M$, while $u$ extends smoothly to the interior of $\scri$ but goes to $+\infty$ at $I^T$ and to $-\infty$ at $I^0$.
		
		We define neighborhoods of the corners by
		$$
		\U_{\scri^+_+}=\{1/\varepsilon<r \leqslant +\infty,\ u_1<u\leqslant +\infty\} \cap \U_{\scri^+},
		\hspace{30pt}
		\U_{\scri^+_-}=\{1/\varepsilon<r \leqslant +\infty,\ -\infty\leqslant u<u_2\} \cap \U_{\scri^+},
		$$
		$$
		\U_{\scri^-_-}=\{1/\varepsilon<r \leqslant +\infty,\ u_1<u\leqslant +\infty\} \cap \U_{\scri^-},
		\hspace{30pt}
		\U_{\scri^-_+}=\{1/\varepsilon<r \leqslant +\infty,\ -\infty\leqslant u<u_2\} \cap \U_{\scri^-}
		$$
		for some $u_1,u_2\in\R$ with $u_1<u_2$. When convenient, we will combine these into
		$$
		\U_0=\U_{\scri^+_-}\sqcup\U_{\scri^-_+},
		\hspace{30pt}
		\U_T=\U_{\scri^+_+}\sqcup\U_{\scri^-_-}.
		$$
		Together, these neighborhoods cover a neighborhood of $\scri$ -- see Figure~\ref{fig:compactifications}.
		
		\begin{figure}
		\begin{center}
		\includegraphics{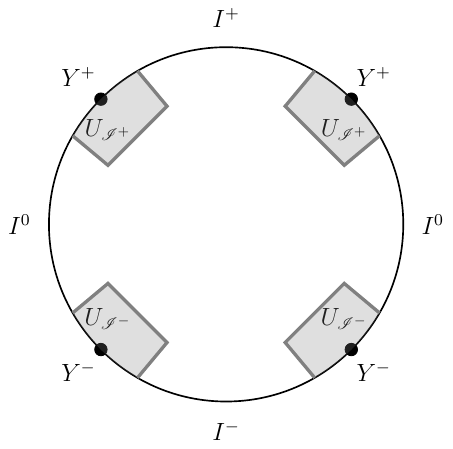}
		\hspace{30pt}
		\includegraphics{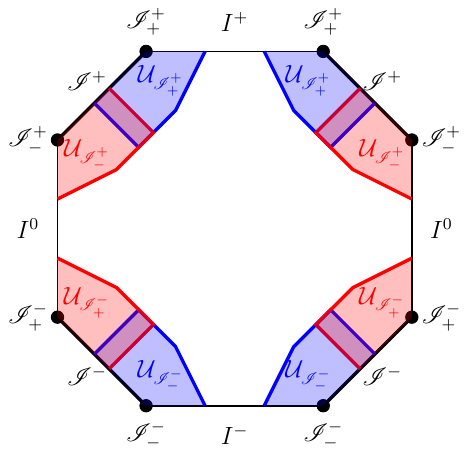}
		\end{center}
		\caption{$\tilde{\M}$ (left) and $\M$ (right) in the case of $d=1$ and trivial topology, with boundary hypersurfaces and corners and the neighborhoods defined above labeled. To get the $d=2$ picture, rotate about the central vertical axis.}
		\label{fig:compactifications}
		\end{figure}

	\subsection{Vector bundles}
	The compactified spacetime $\M$ can be equipped with several useful versions of the tangent and cotangent bundles. Below, let $\rho_0,x_0$ be any choice of local defining functions of $I^0,\scri$ respectively within $\U_0$ and $\rho_T,x_T$ any choice of local defining functions of $I^T,\scri$ within $\U_T$.
	
	A \textit{b-vector field} on any manifold with boundary or corners $M$ is a smooth vector field which is tangent to $\partial M$.  In local coordinates $(\rho_1,\ldots,\rho_k,y_1,\ldots,y_l)$, where $\rho_1,\ldots,\rho_k$ are defining functions of some boundary hypersurfaces and $y_1,\ldots,y_l$ extend a coordinate chart on their intersection, a b-vector field is a $C^{\infty}(M)$-linear combination of $\rho_1\frac{\partial}{\partial\rho_1},\ldots,\rho_k\frac{\partial}{\partial\rho_k},\frac{\partial}{\partial y_1},\frac{\partial}{\partial y_l}$. Such vector fields can be characterized as smooth sections of a vector bundle called the \textit{b-tangent bundle}, denoted ${}^{\mathrm{b}}TM$.
	
	A \textit{scattering vector field} on $M$ is a b-vector field which, as a smooth section of ${}^{\mathrm{b}}TM$, vanishes at $\partial M$. Thus, in coordinates as above, such a vector field is a $C^{\infty}(M)$-linear combination of $(\rho_1\ldots\rho_k)\rho_1\frac{\partial}{\partial\rho_1}$, ..., $(\rho_1\ldots\rho_k)\rho_k\frac{\partial}{\partial\rho_k}$, $(\rho_1\ldots\rho_k)\frac{\partial}{\partial y_1}$, $(\rho_1\ldots\rho_k)\frac{\partial}{\partial y_l}$. Such vector fields can be characterized as smooth sections of the \textit{scattering tangent bundle}, denoted ${}^{\mathrm{sc}}TM$.
		
In our case, $\M$ has two boundary hypersurfaces $\scri^{\pm}$ which come equipped with a natural fibration structure: points of $\scri^{\pm}$ belong to the same fiber if they blow down to the same point of $Y^{\pm}$, i.e.\ the fibration is the blow-down map restricted to the front face. An \textit{edge-b--vector field} on $\M$ is a b-vector field which is in addition tangent to these fibers at $\scri^{\pm}$. Thus, these are b-vector fields whose restrictions to $\U_0$ are $C^{\infty}(\M)$-linear combinations of $\rho_0\frac{\partial}{\partial\rho_0},x_0\frac{\partial}{\partial x_0},x_0\V(Y)$ and whose restrictions to $\U_T$ are $C^{\infty}(\M)$-linear combinations of $\rho_T\frac{\partial}{\partial\rho_T},x_T\frac{\partial}{\partial x_T},x_T\V(Y)$. Such vector fields can be characterized as smooth sections of the \textit{edge-b--tangent bundle}, denoted ${}^{\mathrm{e,b}}T\M$.

Finally, a \textit{double-edge--scattering vector field} on $\M$ is an e,b-vector field which, as a smooth section of ${}^{\mathrm{e,b}}T\M$, vanishes at $\partial \M$. Thus, these are sc-vector fields whose restrictions to $\U_0$ are $C^{\infty}(\M)$-linear combinations of $\rho_0^2x_0\frac{\partial}{\partial\rho_0},\rho_0 x_0^2\frac{\partial}{\partial x_0},\rho_0 x_0^2\V(Y)$ and whose restrictions to $\U_T$ are $C^{\infty}(\M)$-linear combinations of $\rho_T^2 x_T\frac{\partial}{\partial\rho_T},\rho_T x_T^2\frac{\partial}{\partial x_T},\rho_T x_T^2\V(Y)$. Such vector fields can be characterized as smooth sections of the \textit{double-edge--scattering tangent bundle}, denoted ${}^{\mathrm{de,sc}}T\M$.

The spaces of smooth vector fields in these categories are denoted $\Vb(M)$, $\V_{\mathrm{sc}}(M)$, $\V_{\mathrm{e,b}}(\M)$, $\V_{\mathrm{de,sc}}(\M)$ respectively, and the spaces of differential operators they generate $\mathrm{Diff}_{\mathrm{b}}(M)$, $\mathrm{Diff}_{\mathrm{sc}}(M)$, $\mathrm{Diff}_{\mathrm{e,b}}(\M)$, $\mathrm{Diff}_{\mathrm{de,sc}}(\M)$. It is important that the commutator of two vector fields of any one of these classes is again a vector field of the same class, and in the case of $\V_{\mathrm{sc}}(M)$ and $\V_{\mathrm{de,sc}}(\M)$ the commutator moreover vanishes everywhere on the boundary as a vector field of the same class.

There are also the dual cotangent bundles; relevant to our analysis are ${}^{\mathrm{sc}}T^*\tilde{\M}$ and $\ttm$, the dual bundles to ${}^{\mathrm{sc}}T\tilde{\M}$ and $\tm$ respectively. Smooth sections of ${}^{\mathrm{sc}}T^*\tilde{\M}$ are 1-forms which are smooth in $\tilde{\M}^{\circ}$, while near $\partial\tilde{\M}$ they are $C^{\infty}(\tilde{\M})$-linear combinations of $\frac{d\rho}{\rho^2}$ and $\frac{\Omega^1(\partial\tilde{\M})}{\rho}$ for any boundary-defining function $\rho$. Smooth sections of $\ttm$ are smooth as sections of ${}^{\mathrm{sc}}T^*\tilde{\M}$ away from null infinity, while their restrictions to $\U_0$ are $C^{\infty}(\M)$-linear combinations of $\frac{d\rho_0}{\rho_0^2x_0}$, $\frac{dx_0}{\rho_0x_0^2}$, and $\frac{\Omega^1(Y)}{\rho_0x_0^2}$, and their restrictions to $\U_T$ are $C^{\infty}(\M)$-linear combinations of $\frac{d\rho_T}{\rho_T^2 x_T}$, $\frac{dx_T}{\rho_T x_T^2}$, and $\frac{\Omega^1(Y)}{\rho_T x_T^2}$.

Note that a smooth de,sc--tensor field on $\M$ is smooth as a sc--tensor field on $\tilde{\M}$ except at $Y$. In this sense, de,sc-analysis on $\M$ is a refinement of sc-analysis on $\tilde{\M}$, allowing for more singular behavior at null infinity.

	\subsection{Spaces of functions with symbolic regularity}
	\label{sec:fcn-spaces-base}
	
	In this section, $M$ is any compact smooth manifold with corners and $\mathcal{G}(M)$ is the set of its boundary hypersurfaces. Besides $\M$, we will also be interested in symbols on the de,sc-phase space $\PP$ defined in Section~\ref{sec:de-sc}, where all the properties discussed below remain true.

The space of order-$\mathsf{m}$ symbols on $M$ conormal to the boundary is defined by
\begin{equation}
\symm{m} = \{ a\in C^{\infty}(M^{\circ})\ |\ \forall Q\in \db(M)\ Qa\in\rho^{\mathsf{-m}}L^{\infty}(M)\},
\label{eq:def-conormal-base}
\end{equation}
where $\rho$ is a collection of defining functions of the boundary hypersurfaces and we use multi-index notation: for $M=\M$, $\rho=(\rho_{I^-},\rho_{\scri^-},\rho_{I^0},\rho_{\scri^+},\rho_{I^+})$ and $\mathsf{m}=(m_{I^-},m_{\scri^-},m_{I^0},m_{\scri^+},m_{I^+})$ prescribes the growth orders of $a$ at every face. When we write a numerical constant as a multi-index, we mean that every entry is equal to that number, e.g. $\mathsf{1}=(1,1,1,1,1)$.

We have $\symm{m}=\rho^{\mathsf{-m}}\symm{0}$. One can check using the chain rule that if $\phi:\R\to\R$ is smooth on the range of $a\in\symm{0}$, then $\phi\circ a \in\symm{0}$.

A useful subspace is that of classical symbols, $\symmcl{m}=\rho^{\mathsf{-m}}C^{\infty}(M)$. One can also consider spaces of symbols which are classical at some boundary hypersurfaces but only conormal to others; thus, for $G\subset \mathcal{G}(M)$, we define $S^{\mathsf{m}}(M;G) = \rho^{\mathsf{-m}}S^{\mathsf{0}}(M;G)$, where
\begin{equation}
S^{\mathsf{0}}(M;G)
=
\{a\in C^{\infty}(M^{\circ})\ |\ \forall Q\in \mathrm{Diff}_{\mathrm{b}}(M;G)\ Qa\in L^{\infty}(M)\},
\end{equation}
and by $\mathrm{Diff}_{\mathrm{b}}(M;G)$ we denote the algebra of differential operators generated by smooth vector fields on $M$ which are tangent to all the boundary hypersurfaces in $G$. Note that with this definition,
$$
\mathrm{Diff}(M)=\mathrm{Diff}_{\mathrm{b}}(M;\varnothing),
\hspace{30pt}
\mathrm{Diff}_{\mathrm{b}}(M)=\mathrm{Diff}_{\mathrm{b}}(M;\mathcal{G}(M)),
$$
so
$$
\symmcl{m} = S^{\mathsf{m}}(M;\varnothing),
\hspace{30pt}
\symm{m} = S^{\mathsf{m}}(M;\mathcal{G}(M)).
$$

Finally, one can also define spaces of symbols with classical leading behavior at every face: for $\epsilon>0$,
\begin{equation}
S^{\mathsf{m}}_{\epsilon}(M)
=
\sum_{G\subset\mathcal{G}(M)}
\rho_G^{\epsilon} S^{\mathsf{m}}(M;G)
=
\rho^{\mathsf{-m}}\sum_{G\subset\mathcal{G}(M)}
\rho_G^{\epsilon} S^{\mathsf{0}}(M;G)
,
\end{equation}
where $\rho_G$ is a product of defining functions of all boundary hypersurfaces in $G$, with $\rho_{\varnothing}=1$. (The subspace $S^{\mathsf{m}}(M;G)$ corresponding to $G=\varnothing$ is $\symmcl{m}$, and the subspace corresponding to $G=\mathcal{G}(M)$ is $\symm{m-\epsilon}$). 

Note that symbols in $\soe(M)$ are continuous, thus their values are well-defined pointwise, up to and including the boundary. Moreover, if we allow ourselves to consider large $\epsilon$, we have $\soe(M)\subset C^k(M)$ for any natural number $k<\epsilon$. As one application, returning to small $\epsilon>0$, if one defines a new smooth structure on $M$ by fixing $N>\frac{k}{\epsilon}$ and adjoining the $N$-th roots of the defining functions of every boundary hypersurface (so these become the new defining functions), then symbols in $\soe(M)$ with respect to the original smooth structure become $C^k$ functions with respect to the new one.

For a nonempty corner $K=\bigcap_{\Gamma\in G} \Gamma$ for some $G\subset\mathcal{G}(M)$, it is convenient to also introduce the subspace of $\soe(M)$ consisting of the ``error terms" at $K$:
\begin{equation}
\mathcal{I}_K^{\epsilon}(M) = \{a \in \soe(M)\ |\ a|_K=0\} = \sum_{ \substack{G'\subset\mathcal{G}(M),\\ G'\cap G\neq\varnothing}} \rho_{G'}^{\epsilon} S^{\mathsf{0}}(M;G').
\end{equation}

These classes are all independent of the choice of defining functions. Symbols in any of these classes stay in the class under application of b-differential operators. Moreover, under multiplication
$$
\symm{m}\cdot\symm{n} \subset \symm{m+n},
\hspace{30pt}
\symmcl{m}\cdot\symmcl{n} \subset \symmcl{m+n},
$$
$$
S^{\mathsf{m}}(M;G)\cdot S^{\mathsf{n}}(M;G') \subset S^{\mathsf{m+n}}(M;G\cup G'),
$$
from which we get $
S^{\mathsf{m}}_{\epsilon}(M) \cdot S^{\mathsf{n}}_{\epsilon}(M)
\subset
S^{\mathsf{m+n}}_{\epsilon}(M)$. Additionally, the space of error terms at any corner is an ideal of $\soe(M)$, i.e. $\mathcal{I}_K^{\epsilon}(M) \cdot \soe(M) \subset \mathcal{I}_K^{\epsilon}(M)$.

\begin{lemma}
If $a\in\soe(M)$ does not vanish, then $a^{-1},\sqrt{|a|}\in\soe(M)$ as well. If $A$ is a square matrix with entries in $\soe(M)$ and nonvanishing determinant, then $A^{-1}$ has entries in $\soe(M)$ as well. If $A$ is a real symmetric positive definite matrix (with respect to the standard Euclidean inner product) with entries in $\symm{0}$ and eigenvalues bounded from below by a positive number, then its real symmetric positive definite square root $\sqrt{A}$ has entries in $\symm{0}$ as well.
\label{thm:symbol-prop}
\end{lemma}
We note that a symbol in $\soe(M)$ is automatically elliptic (see Section \ref{sec:de-sc}) at any boundary point where it does not vanish. 
\begin{proof}
Membership in $\soe(M)$ is a local property, and away from the boundary the result follows from properties of smooth functions, so we only need to consider an arbitrarily small neighborhood of any point $x\in \partial M$. Let $\mathcal{G}(x) =\{\Gamma_1,\ldots,\Gamma_k\}$ be the set of boundary hypersurfaces of $M$ to which $x$ belongs, let $\rho_1,\ldots,\rho_k$ be a set of defining functions of the corresponding boundary hypersurfaces, and for any $G\subset\mathcal{G}(x)$ let $\rho_G=\prod_{\Gamma_i\in G} \rho_i$. Then near $x$, we can write
$$
a = a_{\varnothing} + \sum_{\substack{G\in\mathcal{G}(x),\\ G \neq \varnothing}} \rho_G^{\epsilon} a_G,
\hspace{30pt}
a_G \in S^{\mathsf{0}}(M;G),\ 
a_{\varnothing}\in C^{\infty}(M).
$$
On a small enough neighborhood of $x$, $a\neq 0$ implies $a_{\varnothing}\neq 0$ and $|a_{\varnothing}|> \rho_G^{\epsilon} |a_G|$ for every $G$. Then we can write
\begin{equation}
a = a_{\varnothing} \left( 1 + \sum_{\substack{G\in\mathcal{G}(x),\\ G \neq \varnothing}} \rho_G^{\epsilon} a'_G
\right),
\hspace{30pt}
a'_G = \frac{a_G}{a_{\varnothing}} \in S^{\mathsf{0}}(M;G),
\hspace{30pt}
\rho_G^{\epsilon} |a_G'| <1.
\label{eq:soe-renormalized}
\end{equation}

Directly from the definition, we can verify using the chain rule that if $a_G\in S^{\mathsf{0}}(M;G)$ has a strictly positive lower bound, then $a_G^{-1},\sqrt{a_G}\in S^{\mathsf{0}}(M;G)$ as well. For any $a_G\in S^{\mathsf{0}}(M;G)$, we have $1+\rho_G^{\epsilon}a_G \in S^{\mathsf{0}}(M;G)$ as well and therefore
$$
\frac{1}{1+\rho_G^{\epsilon} a_G}
=
1
-
\rho_G^{\epsilon} \frac{a_G}{1+\rho_G^{\epsilon} a_G},
\hspace{30pt}
\sqrt{1+\rho_G^{\epsilon} a_G} 
=
1
+
\rho_G^{\epsilon} \frac{ a_G}{1 + \sqrt{1+\rho_G^{\epsilon} a_G}}
\hspace{30pt}
\in 
1 + \rho_G^{\epsilon} S^{\mathsf{0}}(M;G)
$$
as long as $\rho_G^{\epsilon}|a_G|<1$. Then the lemma's statements about $a^{-1}$ and $\sqrt{|a|}$ follow from the fact that any $a\in\soe(M)$ can be decomposed near $x$ as
$$
a = a_{\varnothing} \prod_{\substack{G\in\mathcal{G}(x),\\ G\neq\varnothing}} (1 + \rho_G^{\epsilon} a_G), \hspace{30pt} a_G\in S^{\mathsf{0}}(M;G),\ a_{\varnothing} \in C^{\infty}(M),
\hspace{30pt}
\rho_G^{\epsilon} |a_G| <1.
$$
To see that this decomposition holds, note that for any $n=1,\ldots,k-1$ and $a_G\in S^{\mathsf{0}}(M;G)$ with $\rho_G^{\epsilon}|a_G|<1$ we have
$$
1+ \sum_{\substack{G\in\mathcal{G}(x),\\ |G|\geqslant n}} \rho_G^{\epsilon} a_G
=
\prod_{\substack{G\in\mathcal{G}(x),\\ |G|= n}} (1+\rho_G^{\epsilon} a_G)
+ \sum_{\substack{G\in\mathcal{G}(x),\\ |G|> n}} \rho_G^{\epsilon} b_G
=
$$
$$
=
\prod_{\substack{G\in\mathcal{G}(x),\\ |G|= n}} (1+\rho_G^{\epsilon} a_G)
\left(
1
+
\prod_{\substack{G\in\mathcal{G}(x),\\ |G|= n}} (1+\rho_G^{\epsilon} c_G)
\sum_{\substack{G\in\mathcal{G}(x),\\ |G|> n}} \rho_G^{\epsilon} b_G
\right)
=
\prod_{\substack{G\in\mathcal{G}(x),\\ |G|= n}} (1+\rho_G^{\epsilon} a_G)
\left(
1 + \sum_{\substack{G\in\mathcal{G}(x),\\ |G|> n}} \rho_G^{\epsilon} d_G
\right),
$$
where $b_G,c_G,d_G\in S^{\mathsf{0}}(M;G)$ and we used the preceding result that $(1+\rho_G^{\epsilon} a_G)^{-1} \in 1+\rho_G^{\epsilon} S^{\mathsf{0}}(M;G)$. Possibly on a smaller neighborhood of $x$, we can ensure $\rho_G^{\epsilon}|d_G|<1$. Then in view of Eq.~(\ref{eq:soe-renormalized}), the decomposition follows by induction starting with the base case $n=k$, which is a special case of the preceding result since the sum then has only one term.

For matrices, if $A\in\soe(M)$, then its determinant and all cofactors are in $\soe(M)$, so if the determinant does not vanish then $A^{-1}\in\soe(M)$ as well. 

If $A\in\symm{0}$ is real symmetric positive definite everywhere on $M$, then if the eigenvalues are bounded from below by a positive number, there exists a compact subset $K$ of the positive real numbers containing all eigenvalues of $A$ at all points of $M$. (Membership in $\symm{0}$ implies boundedness from above). Then its real symmetric positive definite square root can be expressed as
$$
\sqrt{A}=\frac{1}{2\pi i}\oint_{\gamma} \sqrt{z}(zI-A)^{-1}\ dz,
$$
where $\gamma$ is a closed contour in the $\mathrm{Re}(z)>0$ half of the complex plane enclosing $K$ and $\sqrt{z}$ is defined to be holomorphic away from the negative real axis. Since $\gamma$ stays away from the spectrum, $\sup_{z\in\gamma,\ x\in M} |(zI-A(x))^{-1}|<\infty$. Then to estimate b-derivatives of $\sqrt{A}$, we can differentiate under the integral sign, getting integrals of $\sqrt{z}$ times sums of products of b-derivatives of $A$ and powers of $(zI-A)^{-1}$, all of which are bounded uniformly in $z\in\gamma$ and $x\in M$. Since the contour is fixed, the integrals are therefore also bounded uniformly in $x$, so $\sqrt{A}\in \symm{0}$.
\end{proof}

	\subsection{Metric assumptions}
	\label{sec:metric}
	We assume that the spacetime is equipped with a metric $\g$ satisfying the following properties:
		\begin{itemize}
		\item \textbf{Regularity:} $\g$ is a Lorentzian-signature ($-,+,\ldots,+$) section of $\mathrm{Sym}^2(\ttm)$, nondegenerate up to and including the boundary, with coefficients in $\soe(\M)$ for some $\epsilon>0$.
		
		Then the dual metric defines a smoothly varying quadratic form on the fibers of $T^*\M^{\circ}$, which gives rise to the associated Hamilton vector field on the phase space $T^*\M^{\circ}$ with its canonical symplectic structure. In Section~\ref{sec:de-sc}, we review how, after the vector field is rescaled, its flow extends to the \textit{compactified de,sc-phase space} $\PP$ constructed by radially compactifying the fibers of $\ttm$. We refer to integral curves of this extended flow within the characteristic set of $P=\Box_{\g}+m^2$ (which is the part in $\partial\PP$ of the closure within $\PP$ of the union of the \textit{mass shells} $\{\g_x^{-1}(\xi,\xi)=-m^2\}$ for $\xi\in{}^{\mathrm{de,sc}}T^*_x\M$) as mass-$m$ bicharacteristics or simply bicharacteristics.
		
		\item \textbf{Causal structure/non-trapping assumptions:} 
		\begin{itemize}
		\item $\tm$ equipped with $\g$ is time-orientable (up to and including the boundary).
		\item All null geodesics of $\g$ in $\M^{\circ}$ tend to $\scri^-$ in one direction and to $\scri^+$ in the other.
		\item All mass-$m$ bicharacteristics over $I^0\backslash \scri$ tend to $\scri^-$ in one direction and to $\scri^+$ in the other.
		\item All bicharacteristics at \textit{infinite} de,sc-frequency over $I^{\pm}\backslash \scri^{\pm}$ tend to $\scri^{\pm}$ in both directions.
		\item There exists a critical set $\mathcal{R}$ of the Hamilton flow in the mass-$m$ shell over $I^T$, located at finite de,sc-frequency, nondegenerate in the sense that the relevant hypotheses of Theorem~\ref{thm:localized-rp-main}, listed in paragraphs 2-4, are satisfied at each point of $\mathcal{R}$, and such that every mass-$m$ bicharacteristic at \textit{finite} de,sc-frequency over $I^T$ (except points of $\mathcal{R}$) tends to $\mathcal{R}$ in one direction and to infinite frequency in the other.
		\end{itemize}

		\item \textbf{Asymptotic flatness at $\scri$:} There exist defining functions $\rho_0$ of $I^0$, $x_0$ of $\scri$ within $\U_0$ and $\rho_T$ of $I^T$, $x_T$ of $\scri$ within $\U_T$ such that 
		\begin{itemize}
		\item For some $\tilde{u}_1,\tilde{u}_2,\varepsilon>0$ the region $\{\rho_0< \tilde{u}_2,\ x_0<\varepsilon\}$ of $\U_0$ is diffeomorphic to $[0,\tilde{u}_2)\times [0,\varepsilon) \times Y$ and the region $\{\rho_T< \tilde{u}_1,\ x_T<\varepsilon\}$ of $\U_T$ is diffeomorphic to $[0,\tilde{u}_1)\times [0,\varepsilon) \times Y$, the diffeomorphisms being implemented by the respective defining functions on the first two factors and restricting on $\scri$ to the blowdown map on the third factor, and $\scri$ is fully covered by these two regions;
		
		\item With respect to these product decompositions, $\g$ as a de,sc-metric satisfies
		\begin{equation}
\label{eq:metric-assumptions}
		\g|_{\scri\cap \U_0} = 2\frac{d\rho_0^2}{\rho_0^4x_0^2} + 4\frac{d\rho_0\os dx_0}{\rho_0^3x_0^3} + \frac{\h}{\rho_0^2x_0^4},
		\hspace{30pt}
		\g|_{\scri\cap \U_T} = -2\frac{d\rho_T^2}{\rho_T^4x_T^2} - 4\frac{d\rho_T\os dx_T}{\rho_T^3x_T^3} + \frac{\h}{\rho_T^2x_T^4},
		\end{equation}
		where $\h$ is a smooth Riemannian metric on the $Y$ factor.
		\end{itemize}
		
		\end{itemize}

We note that the last assumption on the form of the metric at $\scri$ is, for example, always satisfied for metrics such as those studied in \cite{BVW,HV-semilinear} on $\tilde{\M}$, which, near $Y$, have the special form
$$
\g = -v\frac{d\rho^2}{\rho^4} + \frac{d\rho\os dv}{\rho^3} + \frac{\h}{\rho^2}
$$
for $\h\in C^{\infty}(\tilde{\M};\mathrm{Sym}^2 (T^*\tilde{\M}))$ which induces a Riemannian metric on $Y$, and it remains satisfied if, as considered in \cite{HV-eb}, one allows error terms in $\mathcal{I}_{\scri^+}^{\epsilon}(\M)\cap\mathcal{I}_{\scri^-}^{\epsilon}(\M)$ (though the non-trapping assumptions will impose additional restrictions at other boundary hypersurfaces). See the next section for details.
		
		When working in coordinates on the product neighborhoods of $\scri$ in $\U_0$, $\U_T$ (which by an abuse of notation we will usually just call $\U_0$, $\U_T$ themselves), we will often drop the subscripts on $\rho_0,x_0$ and $\rho_T,x_T$.

\subsection{Examples of spacetimes}
\label{sec:examples}
The preceding abstract definitions are motivated by the following basic examples.

\begin{example}[Minkowski metric]
\label{ex:Minkowski}
The Minkowski metric on $\R_t\times\R^d_x$ is, in Cartesian, resp.\ spherical coordinates on the $\R^d_x$ factor,
$$\g_0
=
-dt^2+\sum_{i=1}^d dx_i^2
=
-dt^2+dr^2+r^2\h_{S^{d-1}}
,$$
where $\h_{S^{d-1}}$ is the usual metric on the sphere. Take $\tilde{\M}$ to be the radial compactification of $\R^{d+1}$ and $Y$ to be the part of $\partial \tilde{\M}$ in the closure of $|t|=r$. Before passing to the blowup, we note that a global boundary-defining function of $\tilde{\M}$ can be taken equal to $\rho=\frac{1}{\sqrt{t^2+r^2}}$ in a neighborhood of the boundary, and a global defining function of $Y$ within $\partial\tilde{\M}$ is $v= \frac{t^2-r^2}{t^2+r^2}$. Then a calculation shows that near the boundary,
\begin{equation}
\g_0 = -v\frac{d\rho^2}{\rho^4} + \frac{d\rho\os dv}{\rho^3} +\frac{v}{4(1-v^2)} \frac{dv^2}{\rho^2} + \frac{1-v}{2}\frac{\h_{S^{d-1}}}{\rho^2}.
\label{eq:Minkowski-sc}
\end{equation}
The expressions for the last two terms become singular near $r=0$; however, transforming into coordinates $y_i=\frac{x_i}{t}$ which are well-behaved there, one can check that $\frac{v}{4(1-v^2)} \frac{dv^2}{\rho^2} + \frac{1-v}{2}\frac{\h_{S^{d-1}}}{\rho^2} \in C^{\infty}(\tilde{\M};\mathrm{Sym}^2({}^{\mathrm{sc}}T^*\tilde{\M}))$. This calculation shows that $\g$ is a smooth sc-metric on $\tilde{\M}$.

Turning our attention to the blowup, near $Y$ we can take $\rho=\frac{1}{r}$ and $v=\frac{|t|-r}{r}$, so our $r$ agrees with $r$ in the notation of Section~\ref{sec:neighborhoods}, and in the same notation $u=|t|-r$, matching standard notation near $\scri^+$. Then in $\U_{\scri^+_-}$, a choice of defining functions of $I^0$ and $\scri^+$ respectively is $\rho_0 = \frac{1}{r-t+u_2}$ and $x_0 = \sqrt{\frac{r-t+u_2}{r}}$ for some $u_2\in\R$. Similarly, in $\U_{\scri^+_+}$, a choice of defining functions of $I^+$ and $\scri^+$ respectively is $\rho_T =\frac{1}{t-r-u_1}$ and $x_T = \sqrt{\frac{t-r-u_1}{r}}$ for some $u_1\in\R$. This yields
$$
\g_0|_{\U_{\scri^+_-}} = (2-x_0^2)\frac{d\rho_0^2}{\rho_0^4x_0^2} + 4\frac{d\rho_0\os dx_0}{\rho_0^3 x_0^3} +\frac{\h_{S^{d-1}}}{\rho_0^2x_0^4}
,
\hspace{30pt}
\g_0|_{\U_{\scri^+_+}} = -(2+x_T^2)\frac{d\rho_T^2}{\rho_T^4x_T^2} - 4\frac{d\rho_T\os dx_T}{\rho_T^3x_T^3}+ \frac{\h_{S^{d-1}}}{\rho_T^2 x_T^4},$$
so the metric satisfies Eq.~(\ref{eq:metric-assumptions}). One gets similar expressions near past null infinity. The non-trapping assumptions are also satisfied; for the analysis at spacelike and timelike infinity, see Appendix~\ref{sec:Minkowski-like}.
\end{example}

\begin{example}[Exterior Schwarzschild metric]
\label{ex:Schwarzschild}
Let $\M'$ be the radial compactification of $\R^{3+1}$, and consider a metric on $\R_{t'}\times\R^3$ which is smooth as a sc-metric on $\M'$, matches $\g_0$ on all of $\partial\M'$, and in some region of the form $\{r>R,\ r>C|t'|\}$ for some $R>2m>0$, $C>1$ matches the exterior Schwarzschild metric
\begin{equation}
\g_m 
= 
-\left(1-\frac{2m}{r}\right)dt'^2 + \left(1-\frac{2m}{r}\right)^{-1}dr^2 + r^2\h_{S^2}.
\label{eq:Schwarzschild}
\end{equation}
We do not call this compactification $\tilde{\M}$ because modifications will be necessary before blowup into $\M$. Let $Y$ again be the part of $\partial\M'$ in the closure of $|t'|=r$, and consider again $\rho=\frac{1}{\sqrt{t'^2+r^2}}$ and $v=\frac{t'^2-r^2}{t'^2+r^2}$ near $\partial\M'$. This yields
\begin{equation}
\label{eq:Schwarzschild-sc}
\g_m = - \left(v - \sqrt{\frac{2}{1-v}}2m\rho + \mathcal{O}(\rho^2)\right) \frac{d\rho^2}{\rho^4} 
+ \Big( 1 + \mathcal{O}(\rho)\Big) \frac{d\rho\os dv}{\rho^3} 
+ \frac{v+\mathcal{O}(\rho)}{4(1-v^2)}\frac{dv^2}{\rho^2} + \frac{1-v}{2}\frac{\h_{S^2}}{\rho^2}.
\end{equation}
This calculation confirms that $\g_m$ is indeed a smooth sc-metric on $\M'$ and decays to the Minkowski metric at $\partial\M'$ away from $\{\frac{r}{t'}=0\}$.\footnote{As discussed in \cite[Example 4.1]{JMS}, the assumption that the metric decays to $\g_0$ even at the points of $\partial\M'$ where $\frac{r'}{t}=0$ can be interpreted as a finite lifetime assumption for any massive bodies responsible for the long-range gravitational field, since such a body would have to eventually decompose by dispersing all of its matter along timelike trajectories in order for the metric to decay to flatness.}

Turning to preparations for the blowup, one defines $\tilde{\M}$ to again be the radial compactification of $\R^{3+1}$ but with a different identification of the spacetime with the interior, which accounts for the logarithmic divergence of the radial null geodesics in Schwarzschild spacetime from those in Minkowski spacetime. Namely, we can define a diffeomorphism $\varphi:\R^{3+1}\to\R^{3+1}$, understood as $\tilde{\M}^{\circ}\to\M'^{\circ}$, by $(r,t,\omega)\mapsto (r,t',\omega)$ in spherical coordinates, where $t'(t,r) = t + 2m \chi_1(r)\chi_2\left(\frac{t}{r}\right) \ln\left(\frac{r}{2m}-1\right)$ for $\chi_1,\chi_2\in C^{\infty}(\R)$ chosen so that
$$
\chi_1(r) =
\begin{cases}
0,\ r\leqslant r_0,\\
\text{monotone increasing},\ r_0< r < r_1,\\
1,\ r \geqslant r_1,
\end{cases}
\hspace{30pt}
\chi_2(s) =
\begin{cases}
0,\ ||s|-1|>\varepsilon,\\
-1,\ -1-\frac{1}{2}\varepsilon < s < -1+\frac{1}{2}\varepsilon,\\
1,\ 1-\frac{1}{2}\varepsilon < s < 1+\frac{1}{2}\varepsilon
\end{cases}
$$
for some $r_1>r_0>2m$ and small $\varepsilon>0$ and such that $2m \frac{\ln\left(\frac{r}{2m}-1\right)}{r}|\chi_2'(s)|<1$ for $r>r_0$, ensuring that $t'$ is monotone in $t$ for any fixed $r$ and the map is thus indeed bijective. Then we consider $\g_m$ as a metric on the interior of $\tilde{\M}$ by pulling it back by $\varphi$, so formally it is given by the same expression Eq.~(\ref{eq:Schwarzschild}) with $t'$ understood as a function of $(t,r)$.

In a neighborhood of $Y^{\pm}$ where $\chi(t/r)=\pm 1$, $Y^{\pm}$ is defined within $\partial\tilde{\M}$ by $v=\frac{|t|-r}{r} = \frac{|t'|-r^*}{r}$, where $r^*=r+2m\ln\left(\frac{r}{2m}-1\right)$ is the conventional \textit{tortoise coordinate}. We define $\M$ by the blowup procedure from $\tilde{\M}$. To check that $(\M,\varphi_*\g_m)$ satisfies our assumptions, we need to check its form as a de,sc-metric at null infinity and as a sc-metric on the support of $\chi_2(r/t)$ at timelike and spacelike infinity. For the latter, one can calculate directly from the definition of $t'(t,r)$ that $dt'=dt \mod S^{\mathsf{-1+\epsilon}}(\tilde{\M};{}^{\mathrm{sc}}T^*\tilde{\M})$ for any $\epsilon>0$, so from Eq.~(\ref{eq:Schwarzschild}) we get $\g_m=\g_0 \mod S^{\mathsf{-1+\epsilon}}(\tilde{\M};\mathrm{Sym}^2({}^{\mathrm{sc}}T^*\tilde{\M}))$.

Turning to null infinity, working within the neighborhood of $\scri^+$ where $v=\frac{t-r^*}{r}$, in $\U_{\scri^+_-}$ a choice of defining functions of $I^0$ and $\scri^+$ respectively is $\rho_0=\frac{1}{r^*-t+u_2}$, $x_0=\sqrt{\frac{r^*-t+u_2}{r}}$. Similarly, in $\U_{\scri^+_+}$ a choice of defining functions of $I^+$ and $\scri^+$ respectively is $\rho_T = \frac{1}{t-r^*-u_1}$, $x_T=\sqrt{\frac{t-r^*-u_1}{r}}$. This yields
$$
\g_m|_{\U_{\scri^+_-}} = \g_0|_{\U_{\scri^+_-}} + 2m\rho_0 x_0^4 \frac{d\rho_0^2}{\rho_0^2 x_0^4},
\hspace{50pt}
\g_m|_{\U_{\scri^+_+}} = \g_0|_{\U_{\scri^+_+}} + 2m\rho_T x_T^4 \frac{d\rho_T^2}{\rho_T^2 x_T^4}.
$$
One gets similar expressions near past null infinity. Thus, the metric satisfies Eq.~(\ref{eq:metric-assumptions}) at $\scri$ and in fact decays to the same form also at $I^T$ and $I^0$ near the corners. Thus we conclude that $\g_m$ satisfies all our assumptions and in fact matches the Minkowski metric exactly on $\partial\M$.

See the next examples and \cite{BVW-long-range,HV-stability,JMS} for more details on correctly radially compactifying long-range asymptotically flat spacetimes so as to get a meaningful null infinity after blowup. The approach of \cite{BVW-long-range,HV-stability} is formally different: instead of giving up the identification of Schwarzschild spacetime with Minkowski spacetime prescribed by Schwarzschild coordinates $(t,r,\omega)$ (natural e.g. from the perspective of the initial value problem), one modifies the smooth structure of $\M'$ near null infinity before the blowup in a way dependent on the mass. For us, the approach of \cite{JMS} described above has the conceptual advantage that such spacetimes decay to the exact Minkowski metric at every boundary hypersurface of the \textit{same} compactification $\M$, which makes it clear that the Hamilton flow at spacetime infinity is exactly the same.

\end{example}

\begin{example}[Lorentzian scattering spaces]
Baskin, Vasy, and Wunsch \cite{BVW,BVW-long-range} define the class of \textit{(long-range) Lorentzian scattering spaces}, which generalize Eqs.~(\ref{eq:Minkowski-sc}-\ref{eq:Schwarzschild-sc}) by analogy with the notion of scattering metrics in Riemannian signature due to Melrose \cite{Melrose-AES}. The topological assumptions on the compactified spacetime $(\M',Y)$ are the same as for $(\tilde{\M},Y)$ in Section~\ref{sec:topology}; we will write $v'$ instead of $v$ for the defining function of $Y$. $\g$ is then required to be a smooth scattering metric on $\M'$ which in a product neighborhood $U_{\scri}\simeq [0,\varepsilon)_{\rho}\times (-\varepsilon,\varepsilon)_{v'}\times Y$ of $Y$ has the form
\begin{equation}
\label{eq:Lorentzian-scattering-metrics}
\g = -\Big(v'-m\rho\Big)\frac{d\rho^2}{\rho^4} + \frac{d\rho\os \theta'}{\rho^3} + \frac{\h'}{\rho^2},
\end{equation}
where $m\in\R$ is constant, $\theta'\in  C^{\infty}(\M';T^*\M')$ with $\theta'|_Y=dv'$, and $\h'\in C^{\infty}(\M';\mathrm{Sym}^2(T^*\M'))$ restricts to a Riemannian metric on the $Y$ factors (which may vary from factor to factor). For the analysis of \cite{BVW,BVW-long-range}, it is important that the form of the $\mathcal{O}(\rho)\frac{d\rho^2}{\rho^4}$ term near $Y$ is constrained. This class includes asymptotically flat spacetimes without any radiation escaping through null infinity, which are characterized by a mass parameter $m$ -- this can be understood as either an ADM mass or a Bondi mass which in these examples remains constant along null infinity. In particular, the Kerr metric is of this form near null infinity, as checked in \cite[Appendix A]{BVW-long-range}.

Similarly to the previous example, let $\tilde{\M}$ be another copy of $\M'$ in which we denote the defining function of $Y$ by $v$; we can define a diffeomorphism $\varphi:\tilde{\M}^{\circ}\to \M'^{\circ}$ which is the identity outside $U_{\scri}$ while within $U_{\scri}$ it is given by $(\rho,v,\omega)\mapsto (\rho,v',\omega)$ with respect to the product decomposition, where $v'(\rho,v)=v-m\chi_1(\rho)\chi_2(v)\rho\ln\rho$ for $\chi_1,\chi_2\in C^{\infty}(\R)$ chosen so that
$$
\chi_1(\rho) =
\begin{cases}
0,\ \rho\geqslant \varepsilon,\\
\text{monotone increasing},\ \frac{1}{2}\varepsilon < \rho <\varepsilon,\\
1,\ \rho \leqslant \frac{1}{2}\varepsilon,
\end{cases}
\hspace{30pt}
\chi_2(v) =
\begin{cases}
0,\ |v|>\varepsilon,\\
1,\ |v|<\frac{1}{2}\varepsilon.
\end{cases}
$$
Note that $v'=v \mod S^{-1+\epsilon}(\tilde{\M})$ for any $\epsilon>0$, so for any function $f\in C^{\infty}(\M')$ we have $f\circ \varphi = f \mod S^{-1+\epsilon}(\tilde{\M})$. We then consider $(\tilde{\M},\varphi_*\g)$. Direct calculation shows that in the neighborhood of $Y$ where $v'=v-m\rho\ln\rho$, we have
$$
(-v+a)\frac{d\rho^2}{\rho^4}
+(1+b)\frac{d\rho \os dv}{\rho^3}
+\sum_{i=1}^{d-1}c_i\frac{d\rho \os dy_i}{\rho^3}
+h_{vv}\frac{dv^2}{\rho^2}
+\sum_{i=1}^{d-1} h_{vi} \frac{dv\os dy_i}{\rho^2}
+\sum_{i,j=1}^{d-1} h_{ij} \frac{dy_i\ dy_j}{\rho^2},
$$
where $h_{vv},h_{vi},h_{ij}\in C^{\infty}(\tilde{\M})$, $a\in v S^{-1+\epsilon}(\tilde{\M}) + S^{-2+\epsilon}(\tilde{\M})$, and $b,c_i\in v C^{\infty}(\tilde{\M}) + S^{-1+\epsilon}(\tilde{\M})$; note again the strong restriction on lower-order terms in the coefficient of $\frac{d\rho^2}{\rho^4}$. We define $\M$ by the blowup procedure from $\tilde{\M}$.

In $\U_T$, a choice of defining functions of $I^T$ and $\scri$ respectively is $\rho_T=\frac{2\rho}{v-u_1\rho}$, $x_T=\sqrt{\frac{1}{2}(v-u_1\rho)}$. Similarly, in $\U_0$, a choice of defining functions of $I^0$ and $\scri$ respectively is $\rho_0=\frac{2\rho}{-v+u_2\rho}$, $x_0=\sqrt{\frac{1}{2}(-v+u_2\rho)}$. Then in $\U_T$ we have $\rho=\rho_T x_T^2$, $v\in x_T^2 C^{\infty}(\M)$, $\frac{d\rho}{\rho^2}\in x_T^{-1} C^{\infty}(\M;\ttm)$, $\frac{dv}{\rho} \in x_T C^{\infty}(\M;\ttm)$, and $\frac{dy_i}{\rho}\in C^{\infty}(\M;\ttm)$, and similarly in $\U_0$. A calculation then shows that as a de,sc-metric, $\varphi^*\g \in \soe(\M;\ttm)$ and satisfies Eq.~(\ref{eq:metric-assumptions}) at null infinity. 

The non-trapping assumptions in general need to be verified separately; in particular, since the form of $\theta'$ on $\partial\M'$ away from $Y$ is not constrained, one would not expect the structure assumptions on timelike and spacelike infinity to be satisfied in the absence of additional assumptions. See \cite[Section 7]{BVW-long-range} for more details about the compactification (using the change-of-smooth-structure approach).
\end{example}

\begin{example}[Perturbations of Minkowski space and Minkowski-like and asymptotically Minkowski metrics]
\label{ex:perturbations}
Metrics on $\R^{3+1}$ which are small perturbations of Minkowski spacetime in the sense of solutions to the vacuum Einstein initial value problem generally \textit{do} include radiation escaping through null infinity and therefore are \textit{not} of Lorentzian scattering form because they are more singular at $Y$ from the perspective of the radial compactification. From the blown-up perspective, this is because the metric can have different behavior at different retarded times near null infinity. However, for initial data given by small mass asymptotically Schwarzschild data plus small terms decaying slightly faster than Schwarzschild-like data, in $S^{-1-\delta}$ for the first fundamental form, $S^{-2-\delta}$ for the second fundamental form, $\delta>0$, Hintz and Vasy \cite{HV-stability} showed that the corresponding solutions of Einstein's equation are \textit{symbolic} on the blowup $\M$ if it is constructed as in the Schwarzschild example, with the mass parameter set by the ADM mass of the initial data, and decay at every boundary hypersurface to the Minkowski metric, with error terms in $\bigcap_{\epsilon>0}S^{(-1,-2,-1,-2,-1)+\mathsf{\epsilon}}(\M;\ttm)$ (the stronger decay at $\scri$ being due to the square-root blowup in the definition of $\M$). Note that the results of \cite{HV-stability} are stated using $L^2$-based, rather than $L^\infty$-based symbol spaces, explicitly denoted by $H_{\mathrm{b}}$ Sobolev spaces in \cite{HV-stability}, but for all $\epsilon>0$, up to $\epsilon$ shift in the weight, these two types of spaces are contained in each other. (Note also that in \cite{HV-stability} the Sobolev spaces are measured relative to a scattering density, such as the Euclidean metric, so decay order $r$ corresponds there corresponds to decay order $r+3/2$ for the $L^\infty$-based spaces, i.e.\ symbolic order $-(r+3/2)$, up to $\epsilon$ shifts.) This means that the Hamilton flow structure over the boundary is identical to the Minkowski case, so the non-trapping assumptions there are also satisfied. For small enough perturbations, the non-trapping property of null geodesics is preserved as well. Therefore, our analysis applies to such perturbations.
 
Generalizing this, we will call a metric $\g$ on $\M$ satisfying the properties in Section~\ref{sec:topology} \textit{Minkowski-like} if
\begin{equation}
\label{eq:Minkowski-like}
\g = -v\frac{d\rho^2}{\rho^4} + \frac{d\rho \os dv}{\rho^3} + \frac{\h}{\rho^2}
\mod S^{\mathsf{-\epsilon}}(\M;\mathrm{Sym}^2(\ttm)) 
\end{equation}
for some $\epsilon>0$, where $\h\in C^{\infty}(\tilde{\M};\mathrm{Sym}^2(T^*\tilde{\M}))$ induces a smooth Riemannian metric on $Y$ (cf. Eq.~(\ref{eq:Minkowski-sc})). This generalizes Lorentzian scattering metrics at null infinity similarly to the class of ``admissible metrics" in \cite{HV-eb} but is more restrictive since it requires decay to this form at all boundary hypersurfaces, not only at null infinity. As we show in Appendix~\ref{sec:Minkowski-like}, this form ensures that the structure assumption on finite frequencies over timelike infinity is satisfied; at spacelike infinity and at infinite frequencies over timelike infinity, it ensures the absence of critical points of the Hamilton flow, but the absence of extended trapped bicharacteristics there needs to be checked separately.

We will further say that a Minkowski-like metric on $\R^{d+1}$ is \textit{asymptotically Minkowski} if it decays to $\g_0$ at every boundary hypersurface of $\M$ as defined in Example~\ref{ex:Minkowski}. As discussed, perturbations of Minkowski space in $3+1$ dimensions are indeed asymptotically Minkowski and in fact satisfy Eq.~(\ref{eq:Minkowski-like}) for any $\epsilon<1$. Asymptotically Minkowski metrics automatically satisfy all our assumptions except the non-trapping assumption on null geodesics.
\end{example}

For a relatively concrete non-vacuum example which allows for radiation escaping through null infinity, see \cite{JMS}. We are not aware of connected spacetimes satisfying our assumptions which are not topologically $\R^n$ or of a proof that such examples do not exist; however, there are topologically nontrivial examples with trapping which is nevertheless mild enough that a Fredholm theory can be set up, as shown by Amar \cite{Amar:Normally-Hyperbolic} for a warped product of Minkowski space with a sphere or torus. The asymptotically Minkowski setting is the main case in which our results are of physical interest, and in this case the Hamilton flow at spacetime infinity is the same as computed by Sussman \cite{Sussman}, who considered asymptotically Minkowski metrics with error terms in $S^{-2}_{\mathrm{cl}}(\M)$ (though the proofs of propagation and radial point estimates require more care for the more general error terms we consider).\footnote{The main goal of Sussman's paper was to establish asymptotics for solutions to the Klein-Gordon equation, which requires more structure than the Fredholm setup for propagators, hence the more stringent conditions on the metric.}

\begin{remark}
The assumption on the structure of the Hamilton flow at timelike infinity is a generalization, natural from the point of view of propagation-of-singularities arguments, of the structure arising from identifying points of timelike infinity with families of timelike geodesics which have the same asymptotic momentum, as is common in the physics literature (see e.g. \cite{CL,C-G-W}). Indeed, a point in the (de,)sc-cotangent space over the interior of $I^{\pm}$ represents an asymptotic (de,)sc-momentum, which in Minkowski spacetime corresponds exactly to the usual notion of momentum, and it being a fixed point of the mass-$m$ bicharacteristic flow corresponds to the trajectory of a mass-$m$ particle going off to infinity (or coming in from infinity) with that asymptotic momentum.
\end{remark}

\subsection{Wave operator}
\label{sec:wave-op}
With our sign convention, the wave operator is given in any local coordinates $(x_1,\ldots,x_{d+1})$ on $\M^{\circ}$ by 
\begin{equation}
\Box_{\g}
=
-\operatorname{div} \circ \nabla 
=
-
\frac{1}{\sqrt{|\det\g|}}\sum_{\mu,\nu=1}^{d+1} \partial_{\mu} (\sqrt{|\det\g|}g^{\mu\nu}\partial_{\nu}),
\label{eq:wave-op}
\end{equation}
so $\Box_{\g}
=
-\sum_{\mu,\nu=1}^{d+1} g^{\mu\nu}\partial_{\mu}\partial_{\nu}
\mod \mathrm{Diff}^1(\M^{\circ})$,
where $g^{\mu\nu}$ are the matrix elements of the dual metric to $\g$ with respect to the chosen coordinates. 

Consider a neighborhood of a point of $\partial\M\backslash \scri$ on which we have coordinates $(\rho,y_1,\ldots,y_d)$, where $\rho$ is a boundary-defining function and $y_1,\ldots,y_d$ extend coordinates on the boundary. Then with respect to these coordinates, $\g$ has the form
$$
\g = 
\begin{pmatrix}
\rho^{-4} \tilde{g}_{00} & \rho^{-3} \tilde{g}_{01} & \hdots & \rho^{-3}\tilde{g}_{0d} \\
\rho^{-3} \tilde{g}_{01} & \rho^{-2} \tilde{g}_{11} & \hdots & \rho^{-2}\tilde{g}_{1d} \\
\vdots & \vdots & \ddots & \vdots \\
\rho^{-3} \tilde{g}_{0d} & \rho^{-2} \tilde{g}_{1d} & \ldots & \rho^{-2}\tilde{g}_{dd} 
\end{pmatrix},
\hspace{30pt}
\tilde{g}_{\mu\nu} \in \soe(\M),
\hspace{30pt}
\det\g \in \rho^{-2d-4}\soe(\M).
$$
Using cofactors, we calculate that the dual metric in these coordinates has the form
$$
\g^{-1} = 
\begin{pmatrix}
\rho^4 \tilde{g}^{00} & \rho^3 \tilde{g}^{01} & \hdots & \rho^3\tilde{g}^{0d} \\
\rho^3 \tilde{g}^{01} & \rho^2 \tilde{g}^{11} & \hdots & \rho^2\tilde{g}^{1d} \\
\vdots & \vdots & \ddots & \vdots \\
\rho^3 \tilde{g}^{0d} & \rho^2 \tilde{g}^{1d} & \ldots & \rho^2\tilde{g}^{dd} 
\end{pmatrix},
\hspace{30pt}
\tilde{g}^{\mu\nu} \in \soe(\M).
$$
This calculation shows that away from null infinity, the dual metric defines a nondegenerate bilinear form on the fibers of the (de,)sc-tangent bundle with coefficients in $\soe(\M)$.

At null infinity, similarly consider local coordinates $(\rho_0,x_0,y_1,\ldots,y_{d-1})$ on the product neighborhood $\U_0$, where $\rho_0$, $x_0$ are the defining functions of $I^0$, $\scri$ respectively which yield Eq.~(\ref{eq:metric-assumptions}), and $y_1,\ldots,y_d$ are local coordinates on the $Y$ factors. With respect to these, $\g$ has the form
$$
\g =
\begin{pmatrix}
\rho_0^{-4} x_0^{-2} \tilde{g}_{\rho\rho} & \rho_0^{-3} x_0^{-3} \tilde{g}_{\rho x} & \rho_0^{-3} x_0^{-3} \tilde{g}_{\rho,1} & \hdots & \rho_0^{-3} x_0^{-3} \tilde{g}_{\rho,d-1} \\
\rho_0^{-3} x_0^{-3} \tilde{g}_{\rho x} & \rho_0^{-2} x_0^{-4} \tilde{g}_{xx} & \rho_0^{-2} x_0^{-4} \tilde{g}_{x,1} & \hdots & \rho_0^{-2} x_0^{-4} \tilde{g}_{x,d-1} \\
\rho_0^{-3} x_0^{-3} \tilde{g}_{\rho, 1} & \rho_0^{-2} x_0^{-4} \tilde{g}_{x,1} & \rho_0^{-2} x_0^{-4} \tilde{g}_{11} & \hdots & \rho_0^{-2} x_0^{-4} \tilde{g}_{1,d-1} \\
\vdots & \vdots & \vdots & \ddots & \vdots \\
\rho_0^{-3} x_0^{-3} \tilde{g}_{\rho, d-1} & \rho_0^{-2} x_0^{-4} \tilde{g}_{x,d-1} & \rho_0^{-2} x_0^{-4} \tilde{g}_{1,d-1} & \hdots & \rho_0^{-2} x_0^{-4} \tilde{g}_{d-1,d-1}
\end{pmatrix},
$$
$$
\tilde{g}_{\rho\rho} = \tilde{g}_{\rho x} = 2,
\hspace{30pt}
\tilde{g}_{xx} = \tilde{g}_{\rho,i} = \tilde{g}_{x,i} = 0,
\hspace{30pt}
\tilde{g}_{ij} = h_{ij},
\hspace{30pt}
 \text{ all }\mod \mathcal{I}_{\scri^+}^{\epsilon}(\M) \cap \mathcal{I}_{\scri^-}^{\epsilon}(\M)
$$
for $i,j=1,\ldots,d-1$, where $h_{ij}(y)$ are the matrix elements of $\h$ with respect to $(y_1,\ldots,y_{d-1})$, and $\rho_0^{2d+4} x_0^{4d+2} \det \g = -4\det\h \mod \mathcal{I}_{\scri^+}^{\epsilon}(\M) \cap \mathcal{I}_{\scri^-}^{\epsilon}(\M)$. Then again we calculate using cofactors and the properties of symbols described in Section~\ref{sec:fcn-spaces-base} that the dual metric has the form
$$
\g^{-1} = 
\begin{pmatrix}
\rho_0^4 x_0^2 \tilde{g}^{\rho\rho} & \rho_0^3 x_0^3 \tilde{g}^{\rho x} & \rho_0^3 x_0^3 \tilde{g}^{\rho,1} & \hdots & \rho_0^3 x_0^3 \tilde{g}^{\rho,d-1} \\
\rho_0^3 x_0^3 \tilde{g}^{\rho x} & \rho_0^2 x_0^4 \tilde{g}^{xx} & 
\rho_0^2 x_0^4 \tilde{g}^{x,1} & \hdots & \rho_0^2 x_0^4 \tilde{g}^{x,d-1} \\
\rho_0^3 x_0^3 \tilde{g}^{\rho, 1} & \rho_0^2 x_0^4 \tilde{g}^{x,1} & 
\rho_0^2 x_0^4 \tilde{g}^{11} & \hdots & \rho_0^2 x_0^4 \tilde{g}^{1,d-1} \\
\vdots & \vdots & \vdots & \ddots & \vdots \\
\rho_0^3 x_0^3 \tilde{g}^{\rho, d-1} & \rho_0^2 x_0^4 \tilde{g}^{x,d-1} & \rho_0^2 x_0^4 \tilde{g}^{1,d-1} & \hdots & \rho_0^2 x_0^4 \tilde{g}^{d-1,d-1}
\end{pmatrix},
$$
$$
\tilde{g}^{\rho x} = -\tilde{g}^{x x} = \frac{1}{2},
\hspace{30pt}
\tilde{g}^{\rho \rho} = \tilde{g}^{\rho,i} = \tilde{g}^{x,i} = 0,
\hspace{30pt}
\tilde{g}^{ij} = h^{ij},
\hspace{30pt}
 \text{ all }\mod \mathcal{I}_{\scri^+}^{\epsilon}(\M) \cap \mathcal{I}_{\scri^-}^{\epsilon}(\M)
$$
for $i,j=1,\ldots,d-1$, where $h^{ij}$ are the matrix elements of the dual metric to $\h$ with respect to $(y_1,\ldots,y_{d-1})$. The analogous calculation in coordinates $(\rho_T,x_T,y_1,\ldots,y_{d-1})$ on $\U_T$ yields the same form with $\rho_0, x_0$ replaced by $\rho_T, x_T$ and
$$
-\tilde{g}^{\rho x} = \tilde{g}^{x x} = \frac{1}{2},
\hspace{30pt}
\tilde{g}^{\rho \rho} = \tilde{g}^{\rho,i} = \tilde{g}^{x,i} = 0,
\hspace{30pt}
\tilde{g}^{ij} = h^{ij},
\hspace{30pt}
 \text{ all }\mod \mathcal{I}_{\scri^+}^{\epsilon}(\M) \cap \mathcal{I}_{\scri^-}^{\epsilon}(\M).
$$
This shows that the dual metric defines a nondegenerate bilinear form on the fibers of the de,sc-tangent bundle of $\M$ globally, with coefficients in $\soe(\M)$.

Considering Eq.~(\ref{eq:wave-op}) in any of these special coordinate systems and using the fact that de,sc-vector fields vanish on all of $\partial\M$ as b-vector fields, one can check that for $\g\in\soe(\M;\mathrm{Sym}^2(\ttm))$ the coefficients of all first-order terms in $\Box_{\g}$ must vanish on $\partial\M$ to first order, so $\Box_{\g}=-\sum_{\mu,\nu=1}^{d+1} g^{\mu\nu}\partial_{\mu}\partial_{\nu} \mod S^{\mathsf{-1}}\mathrm{Diff}_{\mathrm{de,sc}}^1(\M)$, where $-\sum_{\mu,\nu=1}^{d+1} g^{\mu\nu}\partial_{\mu}\partial_{\nu} \in \soe\mathrm{Diff}_{\mathrm{de,sc}}^2(\M)$; this will mean that the dual metric function is the \textit{de,sc-principal symbol} of $\Box_{\g}$. In coordinates on $\U_0$ and $\U_T$ as above,
$$
\Box_{\g}
=
-(\rho_0^2 x_0 \partial_{\rho_0})(\rho_0 x_0^2 \partial_{x_0})
+\frac{1}{2}(\rho_0 x_0^2 \partial_{x_0})^2
-\sum_{i,j=1}^{d-1} h^{ij} (\rho_0 x_0^2 \partial_{y_i}) (\rho_0 x_0^2 \partial_{y_j})
+ P_0 \mod \rho_0 x_0 \mathrm{Diff}_{\mathrm{de,sc}}(\U_0),
$$
$$
\Box_{\g}
=
(\rho_T^2 x_T \partial_{\rho_T})(\rho_T x_T^2 \partial_{x_T})
-\frac{1}{2}(\rho_T x_T^2 \partial_{x_T})^2
-\sum_{i,j=1}^{d-1} h^{ij} (\rho_T x_T^2 \partial_{y_i}) (\rho_T x_T^2 \partial_{y_j})
+ P_T \mod \rho_T x_T \mathrm{Diff}_{\mathrm{de,sc}}(\U_T),
$$
where $P_0,P_T\in\mathrm{Diff}_{\mathrm{de,sc}}(\M)$ are strictly second-order with coefficients in $\mathcal{I}_{\scri^+}^{\epsilon}(\M) \cap \mathcal{I}_{\scri^-}^{\epsilon}(\M)$.

\section{Double-edge--scattering calculus}
\label{sec:de-sc}
	The wave operator $\Box_{\g}$ and the Klein-Gordon operator $P=\Box_{\g}+m^2$ are de,sc-differential operators. Their microlocal analysis is therefore based on the algebra of de,sc-pseudodifferential operators, developed on the compactification of $\R^{d+1}$ described in Example~\ref{ex:Minkowski} for this purpose by Sussman \cite{Sussman}. In this section, we review the construction and properties of this algebra in our general topological setting and including weighted Sobolev spaces with microlocally varying orders; we only sketch the relevant steps since the constructions are direct analogues of those used to define other fully symbolic pseudodifferential calculi (such as the sc-calculus) on manifolds with boundary, and moreover the de,sc-calculus fits into the very general framework of pseudodifferential operators on manifolds with scaled bounded geometry recently proposed by Hintz -- see \cite[Section 1.2.4]{Hintz-SBG}. 
	
	We also discuss propagation of singularities for more general operators than considered in \cite{Sussman} in order to deal with the more general class of metrics, and we prove a localized version of Sussman's radial point estimates in the spirit of Haber-Vasy \cite{Haber-Vasy}, which allows for a finer and more unified treatment of propagation of singularities.
		
	\subsection{Development of the calculus}
\subsubsection{Coordinate charts, Schwartz functions, and tempered distributions}	
\label{sec:charts}

De,sc-operators have different behavior at null infinity than at the other boundary hypersurfaces of $\M$. Correspondingly, we fix a finite cover of $\M$ by smooth local coordinate charts $\phi_i: U_i\to V_i\subset \R^{d-1}_y\times [0,+\infty)_{\rho_I}\times [0,+\infty)_{\rho_{\scri}}$ such that within each $U_i$, $\rho_I\circ\phi_i$ is a defining function of $I^T\sqcup I^0$ and $\rho_{\scri}\circ\phi_i$ is a defining function of $\scri$. For simplicity, we can ensure that each $V_i$ is a copy of one of $(1,2)^{d+1}$, $(1,2)^d\times [0,1)$, $(1,2)^{d-1}\times [0,1)\times (1,2)$, and $(1,2)^{d-1}\times [0,1)^2$. We also fix a partition of unity on $\M$ by functions $\chi_i\in C_c^{\infty}(U_i)$, and another set of functions $\psi_i\in C_c^{\infty}(U_i)$ such that $\psi_i|_{\mathrm{supp}(\chi_i)}=1$.

For any compact manifold with boundary or corners $M$, we denote by $\dot{C}^{\infty}(M)$ the space of smooth functions on $M$ which vanish at the boundary to infinite order, which is a generalization of the Schwartz space. For our spacetime, we write $\sch=\dot{C}^{\infty}(\M)$ for short. A Fr\'{e}chet-space topology on $\sch$ is defined by the family of seminorms
\begin{equation}
\|u\|_{\sch,N} = \max_{|\alpha|,|\mathsf{m}|\leqslant N} \sum_i \sup_{x\in U_i} \rho^{-\mathsf{m}} |\partial^{\alpha}(\chi_i u\circ \phi_i^{-1})(\phi_i(x))|
\label{eq:Schwartz-seminorms}
\end{equation}
		for $N\in\N$. 
		
		We denote its dual space, which is a generalization of the space of tempered distributions, by $\sch'$. On the compactification of Minkowski space defined in Example~\ref{ex:Minkowski}, $\sch$ and $\sch'$ thus defined indeed correspond to the usual spaces of Schwartz functions and tempered distributions on $\R^{d+1}$. We consider $\sch'$ to be equipped with the weak-* topology. Sufficiently regular functions on $\M$ are identified with elements of $\sch'$ via the $L^2(\M,\g)$ inner product, which we take to be linear in the second variable and antilinear in the first. When we are only considering topological vector space properties (that is whenever we are not taking inner products or adjoints), we write $\LL=L^2(\M,\g)$ since the metric is inessential: the positive density defined by any nonvanishing section of the top exterior power of $\ttm$ which is continuous up to and including the boundary gives rise to the same $L^2$ space.

	\subsubsection{Phase space and symbols}
	Symbols of de,sc-pseudodifferential operators on $\M$ are conormal symbols on the compactified phase space $\mathcal{P}$, defined as the fiber-radial compactification of $\ttm$. $\mathcal{P}$ is a manifold with corners which has a boundary hypersurface $\Gamma_f$ at fiber infinity in addition to the boundary faces inherited from $\M$, for which by an abuse of notation we will use the same symbols. We denote by $\mathcal{G}(\PP)=\{\Gamma_f,I^-,\scri^-,I^0,\scri^+,I^+\}$ the set of boundary hypersurfaces. By a \textit{corner} we mean any nonempty intersection of boundary hypersurfaces.
	
	$\mathcal{P}$ has a different smooth structure than the fiber-radial compactification of $T^*\M$ due to a rescaling of momenta which is singular at $\partial\M$.
	\begin{itemize}
	\item Away from null infinity, $\PP$ is identical to the fiber-radial compactification of ${}^{\mathrm{sc}}T^*\tilde{\M}$. Thus, let $(\rho,y_1,\ldots,y_d)$ be local coordinates near a point of the interior of $I^{T}$ or $I^0$, where $\rho$ is a defining function of the boundary and the $y_i$ extend local coordinates on the boundary; let $(\tilde{\xi}, \tilde{\eta}_1,\ldots,\tilde{\eta}_d)$ be the canonical dual variables in $T^*\M$. Then smooth coordinates on $\ttm$ over this coordinate neighborhood (including the boundary) are given by $(\rho,y,\xi,\eta)$, where $\tilde{\xi}\ d\rho + \sum_{i=1}^d \tilde{\eta}_i\ dy_i = \xi \frac{d\rho}{\rho^2} + \sum_{i=1}^d \eta_i \frac{dy_i}{\rho}$, so
	\begin{equation}
\xi = \rho^2\tilde{\xi},
\hspace{30pt}
\eta_i = \rho\tilde{\eta}_i.
\label{eq:freq-transform-sc}
	\end{equation}
We consider $\xi,\eta_i$ the ``sc-dual" variables to $\rho,y_i$. By definition of the radial compactification, a defining function of fiber infinity in $\PP$ over this neighborhood is $\rho_f = \frac{1}{\sqrt{1+\xi^2+|\eta|^2}}$.

\item Turning to null infinity, let $(y_1,\ldots,y_{d-1})$ be local coordinates on $Y$, so $(\rho_0,x_0,y_1,\ldots,y_{d-1})$ is a local coordinate chart on a subset of $\U_0$. Let $(\tilde{\zeta}_0,\tilde{\xi}_0,\tilde{\eta}_1,\ldots,\tilde{\eta}_{d-1})$ be the canonical dual variables in $T^*\M$. Then smooth coordinates on $\ttm$ over this coordinate neighborhood (including the boundary) are given by $(\rho_0,x_0,y,\zeta_0,\xi_0,\eta)$, where $\tilde{\zeta}_0\ d\rho_0 + \tilde{\xi}_0\ dx_0 + \sum_{i=1}^{d-1} \tilde{\eta}_i\ dy_i =  \zeta_0 \frac{d\rho_0}{\rho_0^2x_0} + \xi_0 \frac{dx_0}{\rho_0x_0^2} + \sum_{i=1}^{d-1} \eta_i \frac{dy_i}{\rho_0 x_0^2}$, so
	\begin{equation}
\zeta_0 = \rho_0^2x_0\tilde{\zeta}_0,
\hspace{30pt}
\xi_0 = \rho_0x_0^2\tilde{\xi}_0,
\hspace{30pt}
\eta_i = \rho_0x_0^2\tilde{\eta}_i.
	\label{eq:freq-transform}
	\end{equation}
We consider $\zeta_0,\xi_0,\eta_i$ the ``de,sc-dual" variables to $\rho_0,x_0,y_i$. A defining function of fiber infinity in $\mathcal{P}$ over this neighborhood is $\rho_f=\frac{1}{\sqrt{1+\zeta_0^2+\xi_0^2+|\eta|^2}}$. Analogous considerations apply to $\U_T$, where we denote the de,sc-dual variables to $\rho_T,x_T,y_i$ by $\zeta_T,\xi_T,\eta_i$.
\end{itemize}

	As in Eq.~(\ref{eq:def-conormal-base}), we define phase-space conormal symbols
\begin{equation}
\sym{m} = \{ a\in C^{\infty}(\PP^{\circ})\ |\ \forall Q\in \db(\PP)\ Qa\in\rho^{\mathsf{-m}}L^{\infty}(\PP)\},
\label{eq:def-conormal-phase}
\end{equation}
where we also keep track of decay in the fiber variables: $\rho=(\rho_f,\rho_{I^-},\rho_{\scri^-},\rho_{I^0},\rho_{\scri^+},\rho_{I^+})$ is a collection of some globally defined defining functions of the boundary hypersurfaces of $\PP$ and $\mathsf{m}=(m_f,m_{I^-},m_{\scri^-},m_{I^0},m_{\scri^+},m_{I^+})$. When it is convenient to consider the order at fiber infinity separately from the rest, we use the splitting $\rho=(\rho_f,\rho_{base})$ and $\mathsf{m}=(m,\mathsf{m}_{base})$. We write $\sym{\infty}=\bigcup_{\mathsf{m}\in\R^6} \sym{m}$ and $\sym{-\infty}=\bigcap_{\mathsf{m}\in\R^6} \sym{m}$. 

A Fr\'{e}chet-space topology on $\sym{m}$ is defined by the family of seminorms
$$
\|a\|_{S^{\mathsf{m}},N} = \max_{k+l+|\alpha|\leqslant N} \sum_i \sup_{x\in U_i} \rho^{\mathsf{m}} |(\rho_I \partial_{\rho_I})^k (\rho_{\scri}\partial_{\rho_{\scri}})^l \partial_y^{\alpha} (\chi_i a\circ \phi_i^{-1})(\phi_i(x))|
$$
for $N\in\N$. An important fact is that for any orders $\mathsf{m}<\mathsf{m'}$, the residual space $\sym{-\infty}$ is dense in $\sym{m}$ in the topology of $\sym{m'}$, though not in that of $\sym{m}$. Multiplication of symbols and application of operators in $\mathrm{Diff}_{\mathrm{b}}(\PP)$ are continuous operations between symbol spaces of appropriate orders.

The other symbol classes defined in Section~\ref{sec:fcn-spaces-base} are also well-defined for $\PP$. Symbols in any of these classes stay in the class under application of operators in $\mathrm{Diff}_{\mathrm{b}}(\PP)$. For b-, e,b-, etc. vector fields with symbolic (rather than smooth) coefficients, we use the notation $S^{\mathsf{m}}\Vb(\PP)$, etc.

	\subsubsection{Pseudodifferential operators}
	\label{sec:psiDO}
The algebra $\ps$ of de,sc-pseudodifferential operators is defined by reduction via coordinate charts to the model $\Psi_{\mathrm{de,sc}}([0,+\infty)_{\rho_I}\times [0,+\infty)_{\rho_{\scri}} \times \R^{d-1}_y)$. The latter can be defined as the quantization of the Lie algebra of vector fields generated by $\rho_I^2\rho_{\scri}\frac{\partial}{\partial\rho_I}$, $\rho_I\rho_{\scri}^2\frac{\partial}{\partial\rho_{\scri}}$, $\rho_I\rho_{\scri}^2\frac{\partial}{\partial y_i}$. The de,sc-phase space over $[0,+\infty)^2\times \R^{d-1}$ is defined analogously to $\PP$, but is noncompact at the $\rho_I\to\infty$, $\rho_{\scri}\to\infty$ ends. One fixes a quantization map $\mathrm{Op}_{\mathrm{model}}$ from conormal symbols on this phase space which are supported over compact sets in the base space to operators on $\dot{C}^{\infty}([0,+\infty)^2\times \R^{d-1})$, with polynomials in the de,sc-dual momenta to $\rho_I$, $\rho_{\scri}$, $y_i$ with fiberwise-constant coefficients mapping to linear combinations of the corresponding compositions of $-i\rho_I^2\rho_{\scri}\frac{\partial}{\partial\rho_I}$, $-i\rho_I\rho_{\scri}^2\frac{\partial}{\partial\rho_{\scri}}$, $-i\rho_I\rho_{\scri}^2\frac{\partial}{\partial y_i}$ respectively with the same coefficients. See \cite[Section 2.3]{Sussman} for details; $(\varrho_{\mathrm{nf}},\varrho_{\mathrm{Of}},\theta)$ in the reference correspond to our $(\rho_{\scri},\rho_I,y)$.

For $\M$, $\ps[m]$ can be defined as the set of continuous linear operators $A:\sch\to\sch$ of the form $\mathrm{Op}(a)+A_{-\infty}$, where $a\in\sym{m}$, $A_{-\infty}\in\mathcal{L}(\sch';\sch)$, and $\mathrm{Op}(a):\sch\to\sch$ is given by
\begin{equation}
\mathrm{Op}(a)(u)(x) =\sum_{ i\ |\ x\in U_i} \Big( \mathrm{Op}_{\mathrm{model}}(a_i) (u_i)\Big)  (\phi_i(x)),
\label{eq:Op}
\end{equation}
where $u_i$ is the function on the model base space defined by $(\psi_i u)\circ\phi_i^{-1}$ on $V_i$ and zero outside and $a_i$ is the symbol on the model phase space defined by $\chi_i a\circ \phi_i^*$ over $V_i$ and zero outside (where $\phi_i^*$ is the pullback of one-forms, taking points of the cotangent space over the model space to points of $T^*\M$). We write $\ps=\bigcup_{\mathsf{m}\in\R^6}\ps[m]$ and $\ps[-\infty]=\bigcap_{\mathsf{m}\in\R^6}\ps[m]$. It can be shown that any operator $A\in\ps$ continuously maps $\sch'\to \sch'$, and $\ps[-\infty]=\mathcal{L}(\sch';\sch)$.

There is no canonical bijective quantization map $\sym{\infty}\to\ps$ since $\mathrm{Op}$ depends on the choice of coordinate charts and partition of unity, but the following are independent of choices:
\begin{itemize}
\item surjective linear \textit{principal symbol} maps $\sigma_{\mathsf{m}}:\ps[m]\to\sym{m}/\sym{m-1}$ defined by $\sigma_{\mathsf{m}}(\mathrm{Op}(a)+A_{-\infty})=[a]$ and satisfying $\sigma_{\mathsf{m}}(A)=[0]$ if and only if $A\in\ps[m-1]$;
\item a notion of \textit{essential support} $\wfs(A)\subset\partial\PP$ such that $\wfs(\mathrm{Op}(a) + A_{-\infty})$ is the set of points of $\partial\PP$ which \textit{do not} have a neighborhood in which $|a|\leqslant C_{\mathsf{N}}\rho^{\mathsf{N}}$ for every $\mathsf{N}$ (which implies such a bound for any b-derivatives of $a$ as well), and $\wfs(A)=\varnothing$ if and only if $A\in\ps[-\infty]$.
\end{itemize}
We often define an operator in $\ps[m]$ by specifying its symbol without specifying the exact quantization map; it should be understood that really this determines the operator's principal symbol (and therefore its class in $\ps[m]/\ps[m-1]$) as well as its essential support, but not the exact operator. For de,sc-differential operators, which are local, the symbol over any coordinate chart is obtained by the usual prescription of replacing partial derivatives with $i$ times the corresponding de,sc-dual variables.

$\ps$ satisfies the algebraic properties expected of a fully symbolic pseudodifferential operator algebra, summarized below. Here $A^*$ is defined a priori as the Fr\'{e}chet-space adjoint to $A:\sch\to \sch'$; $\{a,b\}$ is the Poisson bracket of symbols $a$, $b$ considered as functions on $T^*\M^{\circ}$ with its canonical symplectic structure.
\begin{table}[H]
\centering
\renewcommand{\arraystretch}{1.5}
\begin{tabular}{|c|c|c|c|}
\hline
If ... & then ... & with principal symbol & and essential support \\
\hline
$A,B\in\ps[m]$ & $A+B\in\ps[m]$ & $\sigma_{\mathsf{m}}(A) +\sigma_{\mathsf{m}}(B)$ & $\subset\wfs(A)\cup\wfs(B)$ \\
\hline
$A\in\ps[m_1]$, & $AB\in\ps[m_1+m_2]$ & $\sigma_{\mathsf{m_1}}(A)\sigma_{\mathsf{m_2}}(B)$ & \multirow{2}{*}{$\subset \wfs(A)\cap\wfs(B)$} \\
\cline{2-3}
$B\in\ps[m_2]$  & $[A,B]\in\ps[m_1+m_2-1]$ & $-i\{\sigma_{\mathsf{m_1}}(A),\sigma_{\mathsf{m_2}}(B)\}$ &  \\
\hline
$A\in\ps[m]$ & $A^*\in\ps[m]$ & $\overline{\sigma_{\mathsf{m}}(A)}$ & $ = \wfs(A)$ \\
\hline
\end{tabular}
\caption{
Algebraic properties of pseudodifferential operators.}
\label{tab:psiDO}
\end{table}

The topology on $\sym{m}$ induces a Fr\'{e}chet-space topology on $\ps[m]$, with seminorms given by $\|A\|_{\Psi^{\mathsf{m}},N}=\inf \Big( \|a\|_{S^{\mathsf{m}},N} + \|K_{A_{-\infty}}\|_{\dot{C}^{\infty},N}\Big)$ taken over $a\in\sym{m}$, $A_{-\infty}\in\ps[-\infty]$ such that $A=\mathrm{Op}(a)+A_{-\infty}$. Here $K_{A_{-\infty}}\in \dot{C}^{\infty}(\M\times\M)$ is the Schwartz kernel of $A_{-\infty}$ and the seminorms $\|\bullet\|_{\dot{C}^{\infty},N}$ on $\dot{C}^{\infty}(\M\times\M)$ are defined similarly to Eq.~(\ref{eq:Schwartz-seminorms}) (using any finite atlas of $\M\times\M$, since the properties of the particular charts we chose in Section~\ref{sec:charts} were only important for issues of quantization). The global $A_{-\infty}$ term is irrelevant when one is working microlocally, defining operators by quantizing symbols which are supported in arbitrarily small regions of $\PP$. Just like for symbols, for any orders $\mathsf{m}<\mathsf{m}'$, the residual space $\ps[-\infty]$ is dense in $\ps[m]$ in the topology of $\ps[m']$, though not in that of $\ps[m]$. The basic operations in Table~\ref{tab:psiDO} are all continuous with respect to these topologies.

Beyond principal symbols, it is important that for any $a\in\sym{m}$, $b\in\sym{n}$ there exist sequences of symbols $c_i\in\sym{m+n-i}$, $d_i\in\sym{m-i}$ for $i\in\N_0$ such that for any $N$ there are $R_N\in\ps[m+n-N]$, $R_N'\in\ps[m-N]$ such that
$$
\mathrm{Op}(a)\mathrm{Op}(b)
=
\sum_{i=0}^{N-1} \mathrm{Op}(c_i) + R_N,
\hspace{30pt}
\mathrm{Op}(a)^*
=
\sum_{i=0}^{N-1} \mathrm{Op}(d_i) + R_N'.
$$
Moreover, $c_i$ and $R_N$ depend continuously, in their respective spaces, on $(a,b)$, and $d_i$ and $R_N'$ similarly depend continuously on $a$. Consistently with Table~\ref{tab:psiDO}, we can take $c_0(a,b)=ab$, $c_1(a,b)-c_1(b,a)=-i\{a,b\}$, and $d_0(a)=\bar{a}$. The continuous dependence of the lower-order terms ensures that symbolic identities like those used to prove propagation theorems give rise to operator identities modulo error terms which are not only one order lower (at every face) but also bounded in those lower-order spaces in terms of the symbols involved. These expansions are also used to prove many of the basic properties of pseudodifferential operators which we now review.

	\subsubsection{Ellipticity, Sobolev spaces, and microlocalization}
	
An operator $A\in\ps[m]$ with symbol $a$ is called elliptic of order $\mathsf{m}$ at a point $\alpha\in\partial\PP$ if there is a neighborhood of $\alpha$ in which $|a|\geqslant C\rho^{\mathsf{-m}}$ for some $C>0$. This property only depends on the principal symbol $[a]\in\sym{m}/\sym{m-1}$, so it does not require specification of a quantization map. Any point at which $A$ is elliptic belongs to $\wfs(A)$. If $a\in S^{\mathsf{m}}_{\epsilon}(\PP)$, then $\tilde{a}=\rho^{\mathsf{m}}a$ is continuous on $\PP$ up to and including the boundary, and ellipticity at $\alpha$ is equivalent to $\tilde{a}(\alpha)\neq 0$.

One often uses that for any disjoint compact subsets $K_1,K_2\subset\partial\PP$, there exists $Q\in\ps[0]$ such that $\wfs(Q)\cap K_1 = \wfs(I-Q)\cap K_2 = \varnothing$, and in particular such $Q$ is elliptic on $K_2$.

The set of points in $\partial\PP$ where $A$ is \textit{not} elliptic is called the characteristic set of $A$, denoted $\Sigma^{\mathsf{m}}(A)$. $A$ is called globally elliptic if it is elliptic on all of $\partial\PP$. If $A\in\ps[m]$ is elliptic on a compact set $K\subset\partial\PP$, one can construct a \textit{microlocal elliptic parametrix for $A$ on $K$}, that is an operator $B\in\ps[-m]$ such that $K\cap\wfs(I-AB)=K\cap\wfs(I-BA)=\varnothing$.

To study mapping properties of operators in $\ps$, weighted de,sc-Sobolev spaces are defined for any set of orders $\mathsf{s}\in\R^6$ by
\begin{equation}
\sob{s} = \{u\in \sch'\ |\ \forall A\in\ps[s]\ Au\in \LL\}.
\label{eq:Sobolev-def}
\end{equation}
They are Hilbert spaces with the norm $\|u\|_{\mathsf{s}}=\|\Lambda_{\mathsf{s}}u\|_{\LL}$, where $\Lambda_{\mathsf{s}}$ can be taken to be any globally elliptic operator in $\ps[s]$ which is invertible on $\sch'$ (which always exist), since one can show that the resulting norms are all equivalent. Going forward, we drop the subscript on $L^2$ norms.

The dense embedding $\sob{s}\subset\sob{s'}$ for $\mathsf{s}\geqslant\mathsf{s'}$ is compact if and only if $\mathsf{s}>\mathsf{s'}$. We have
$$
\bigcap_{\mathsf{s}\in\R^6} \sob{s} = \sch,
\hspace{30pt}
\bigcup_{\mathsf{s}\in\R^6} \sob{s} = \sch',
$$
and the family of Sobolev norms generates the standard topology on $\sch$, which is dense in any finite intersection of Sobolev spaces with the topology of simultaneous convergence in all of them. The $L^2(\M,\g)$ pairing on $\sch\times\sch$, which extends to a continuous sesquilinear pairing on $\sch\times\sch'$ or $\sch'\times\sch$ (the dual pairing modified by conjugation), also defines a continuous sesquilinear pairing on $\sob{s}\times\sob{-s}$ for any $\mathsf{s}$, thereby identifying the dual space to $\sob{s}$ with $\sob{-s}$. An operator $A\in\ps[m]$ continuously maps $\sob{s}\to \sob{s-m}$ for any $\mathsf{s}$; moreover, the map $(A,u)\mapsto Au$ is jointly continuous $\ps[m]\times\sob{s}\to\sob{s-m}$.

For any $u\in \sch'$, the order-$\mathsf{s}$ wavefront set of $u$ is defined as
$$
\wf{s}(u) = \{\alpha\in\partial\PP\ |\ \nexists Q\in\ps[0]\text{ elliptic at }\alpha\text{ such that }Qu\in \sob{s}\}.
$$
The wavefront set satisfies the expected properties:
\begin{itemize}
\item $\wf{s}(u) = \varnothing$ if and only if $u\in \sob{s}$;
\item Microlocal elliptic regularity: $\wf{s}(u) \subset \wf{s-m}(Au) \cup \Sigma^{\mathsf{m}}(A)$.
\item Microlocality: $\wf{s-m}(Au)\subset \wfs(A)\cap\wf{s}(u)$ for $A\in\ps[m]$.
\end{itemize}

\subsubsection{Variable-order spaces}
One can also define Sobolev spaces $\sob{s}$ with microlocally varying orders, i.e. $\mathsf{s}\in C^{\infty}(\PP;\R^6)$. First, one defines the variable-order symbol spaces
$$
S^{\mathsf{s}}_{\mathrm{var}}(\PP)=\{a\in C^{\infty}(\PP^{\circ}) \cap \rho^{\mathsf{-s}} L^{\infty}(\PP)\ |\ \forall\epsilon>0,\ \forall Q\in\mathrm{Diff}_{\mathrm{b}}(\PP)\ Qa \in \rho^{\mathsf{-s-\epsilon}}L^{\infty}(\PP)\},
$$
where allowing for slightly worse decay after differentiation ensures that $\rho^{\mathsf{-s}}\subset S^{\mathsf{s}}_{\mathrm{var}}(\PP)$ despite the fact that when $\mathsf{s}$ is variable, derivatives of $\rho^{\mathsf{-s}}$ may contain factors which grow logarithmically at $\partial\PP$ (see e.g. \cite[Section 2]{HJSV}). We have $S^{\mathsf{s}}_{\mathrm{var}}(\PP) \subset \sym{s'}$ for any constant order $\mathsf{s'}>\sup \mathsf{s}$. The space $\ps[s]\subset\ps$ is then defined as for constant orders using Eq.~(\ref{eq:Op}), and variable-order Sobolev spaces by Eq.~(\ref{eq:Sobolev-def}).

For variable $\mathsf{s}$, a principal symbol map independent of the particular choice of quantization can be defined as a map $\sigma_{\mathsf{s}}:\ps[s] \to S^{\mathsf{s}}_{\mathrm{var}}(\PP) / \bigcap_{\epsilon>0} S^{\mathsf{s-1+\epsilon}}_{\mathrm{var}}(\PP)$. Ellipticity of an operator $A\in\ps[s]$ at a point of $\partial\PP$ is well-defined, and the microlocal elliptic parametrix construction goes through, analogously to the constant-order notion.

$\sob{s}$ is a Hilbert space with the squared norm $\|u\|_{\mathsf{s}}^2  = \|\Lambda_{\mathsf{s}}u\|^2 + \|u\|_{\mathsf{s''}}^2$ for any globally elliptic $\Lambda_{\mathsf{s}}\in\ps[s]$ and any constant order $\mathsf{s''}<\inf\mathsf{s}$, the resulting norms all being equivalent. The analogues of the properties discussed in the previous section regarding dense and compact embeddings, duality, Sobolev boundedness of pseudodifferential operators, as well as the definition and properties of wavefront sets, hold for variable orders. For compact embedding of Sobolev spaces the strict inequality on orders only needs to hold in a pointwise sense.

While we will use variable-order spaces to define propagators, we will only need the orders to vary in the region where the Klein-Gordon operator is elliptic. Therefore, when we discuss propagation of singularities in the characteristic set, we will state all results in the constant-order setting. In the end, to combine these results into a global statement in terms of variable-order spaces, we will use that if $\mathsf{s}$ equals a constant $\mathsf{s_0}$ in an open neighborhood of a compact set $K\subset\PP$, then for any $u\in\sch'$, $\wf{s}(u)\cap K = \wf{s_0}(u)\cap K$.

\subsubsection{Some useful lemmas}
\label{sec:lemmas}

We now review a few more detailed results which are required for proofs of propagation of singularities and radial point estimates. Below all orders are assumed to be constant unless stated otherwise.

One often needs to employ arguments approximating an operator by better-behaved ones. The topology on $\ps[m]$ allows us to consider bounded families of operators $A_t\in \ps[m]$ for $t\in (0,1)$. For such a family, the \textit{joint essential support} is defined as
$$\wfs(\{A_t\}) = \{\alpha\in\partial\PP\ |\ \nexists Q\in\ps[0] \text{ elliptic at }\alpha:\ \forall N\in\N\ \sup_{t\in (0,1)} \|QA_t\|_{\mathcal{L}(H^{\mathsf{-N}}_{\mathrm{de,sc}};H^{\mathsf{N}}_{\mathrm{de,sc}})}<\infty\}.$$
If the family is defined by $A_t=\mathrm{Op}(a_t)+A_{-\infty}$ with $A_{-\infty}\in\ps[-\infty]$, the existence of such $Q$ is ensured (independently of coordinate choices) if in some neighborhood of $\alpha$, $\sup_{t\in(0,1)} |a_t|\leqslant C_{\mathsf{N}}\rho^{\mathsf{N}}$ for every $\mathsf{N}$; however, adding a general $t$-dependent error in $\ps[-\infty]$ is not allowed.
 
Note that $\bigcup_{t\in(0,1)}\wfs(A_t)\subset\wfs(\{A_t\})$; in general the latter can be much larger than the former. If the family converges to $A\in\ps[m]$ as $t\to 0^+$ in the topology of $\ps[m']$ for some $\mathsf{m'}\geqslant\mathsf{m}$, then necessarily $\wfs(A)\subset\wfs(\{A_t\})$. Analogues of the properties in the last column of Table~\ref{tab:psiDO} hold for the joint essential support of the families $A_t+B_t$, $A_tB_t$, $[A_t,B_t]$, $A_t^*$.

For any set of orders $\mathsf{s}$ (which may be variable) and $t$ running through some index set $I$, we also define the \textit{joint order-$\mathsf{s}$ wavefront set} of a family $u_t\in\sch'$ bounded in $\sob{-N}$ for some $N$ as
$$
\wf{s}(\{u_t\}) = \{\alpha\in\partial\PP\ |\ \nexists Q\in\ps[0] \text{ elliptic at } \alpha \text{ such that } \sup_{t\in I} \|Qu_t\|_{\mathsf{s}}<\infty\}.
$$
$\wf{s}(\{u_t\})=\varnothing$ if and only if the family $u_t$ is bounded in $\sob{s}$.

\begin{lemma}
\label{thm:uniform-lemma}
Let $A_t$ for $t\in (0,1)$ be a bounded family in $\ps[m]$. Then for any $Q\in \ps[n]$ elliptic on $\wfs(\{A_t\})$ and any set of orders $\mathsf{N}$, there exists $C>0$ such that any $u\in\sob{-N}$ with $Qu\in\sob{s-n}$ satisfies
$$
\|A_t u\|_{\mathsf{s-m}}\leqslant C \Big( \|Qu\|_{\mathsf{s-n}} + \|u\|_{\mathsf{-N}}\Big).
$$
In particular, for any $u\in\sch'$ with $\wf{s}(u)\cap\wfs(\{A_t\})=\varnothing$, the family $A_tu$ is bounded in $\sob{s-m}$. If in addition $\lim_{t\to 0^+}A_t= A\in\ps[m]$ in $\ps[m']$ for some $\mathsf{m'}\geqslant\mathsf{m}$, then $\lim_{t\to 0^+}A_tu= Au$ in $\sob{s-m}$.
\end{lemma}
\begin{proof}
Take $Q'\in\ps[-n]$ which is a microlocal parametrix for $Q$ on $\wfs(\{A_t\})$. Then we can write $A_tu = A_tQ'(Qu) + A_t(I-Q'Q)u$, where
\begin{itemize}
\item $Qu\in\sob{s-n}$ and the family $A_tQ'$ is bounded in $\ps[m-n]$, so $\|A_tQ'(Qu)\|_{\mathsf{s-m}} \leqslant C \|Qu\|_{\mathsf{s-n}}$;
\item $\wfs(\{A_t(I-Q'Q)\})=\varnothing$, so $A_t(I-Q'Q)$ is uniformly bounded in $\ps[-(s-m)-N]$ for every $\mathsf{N}$ and thus $\|A_t(I-Q'Q)u\|_{\mathsf{s-m}}\leqslant C\|u\|_{\mathsf{-N}}$.
\end{itemize}
The boundedness statement then follows from the fact that for $u\in\sch'$ as specified, one can always find $Q\in\ps[0]$ with $\wfs(Q)\cap\wf{s}(u)=\wfs(I-Q)\cap\wfs(\{A_t\})=\varnothing$.

If in addition $A_t\to A\in\ps[m]$ in $\ps[m']$, then $A_tv\to Av$ in $\sob{s-m}$ for any $v\in\sob{s+(m'-m)}$. Since we only know $Qu\in\sob{s}$, this does not directly imply that $A_t(Qu)\to A(Qu)$ in $\sob{s-m}$; but because $\sob{s+(m'-m)}$ is dense in $\sob{s}$, the family $A_t$ is bounded in $\ps[m]$, and the limit is in $\ps[m]$, this convergence follows from an ``$\varepsilon/3$ argument" -- see the proof of \cite[Lemma 4.39]{Hintz-notes}. We also have $A_t(I-Q)\to A(I-Q)$ in $\ps[m']$, and since necessarily $\wfs(I-Q)\cap\wfs(A)=\varnothing$, the limit is in $\ps[-\infty]$. Then by the same argument, $(A_t(I-Q))u\to (A(I-Q))u$ in $\sch$. Adding the terms together, we get $A_tu\to Au$ in $\sob{s-m}$.
\end{proof}

An important approximation construction is the following. Fix $\varphi\in C^{\infty}(\R)$ such that $\varphi(t)=0$ for $t<1$ and $\varphi(t)=1$ for $t>2$. Then if we take $J_t=\mathrm{Op}(q_t)$ for $q_t=\prod_{\Gamma_i\in\mathcal{G}(\PP)} \varphi(\rho_i/t)$, where $\rho_i$ is a defining function of $\Gamma_i$, then $J_t\in\ps[-\infty]$ for every $t$, the family $J_t$ for $t\in (0,1)$ is bounded in $\ps[0]$, and $\lim_{t\to 0^+}J_t=I$ in $\ps[\epsilon]$ for any $\epsilon>0$. We refer to a family with these properties as an \textit{approximation of the identity} below.

Recall that the $L^2(\M,\g)$ pairing $\langle u,v\rangle=\overline{\langle v,u\rangle}$ is well-defined for $(u,v)$ in $\sob{-s}\times\sob{s}$ for any $\mathsf{s}$. The integration-by-parts formula $\langle Au,v\rangle=\langle u,A^*v\rangle$ holds for any $A\in\ps$ when $(u,v)$ is in $\sch\times\sch'$ or $\sch'\times\sch$. In addition, if holds for any $(u,v)\in\sch'\times\sch'$ when $A\in\ps[-\infty]$, since in that case $A$ has Schwartz kernel in $\dot{C}^{\infty}(\M\times \M)$. This generalizes as follows.

\begin{lemma}[Microlocal conditions for pairing and integration by parts]\ 
\begin{enumerate}
\item Consider $u,v\in \sch'$ such that $u\in\sob{-s}$ and $\wf{}(u)\cap\wf{s}(v)=\varnothing$ for some set of orders $\mathsf{s}$. Then for any such $\mathsf{s}$ and any choice of $Q\in\ps[0]$ such that $\wfs(Q)\cap\wf{}(u) = \wfs(I-Q)\cap\wf{s}(v)=\varnothing$, the expression $\langle Qu,v\rangle + \langle u,(I-Q^*)v\rangle$ gives the same result, which one takes as the definition of $\langle u,v\rangle$ and $\overline{\langle v,u\rangle}$.
\item Let $A_t$ for $t\in (0,1)$ be a bounded family in $\ps[m]$. Then for any $Q_1,Q_2\in\ps[0]$ elliptic on $\wfs(\{A_t\})$ and any set of orders $\mathsf{N}$, there exists $C>0$ such that any $u,v\in\sob{-N}$ with $Q_1u\in\sob{m-s}$ and $Q_2v\in\sob{s}$ satisfy
$$
|\langle A_tu,v\rangle| \leqslant C\Big( \|Q_1u\|_{\mathsf{m-s}} + \|u\|_{\mathsf{-N}} \Big)  \Big(\|Q_2v\|_{\mathsf{s}} + \|v\|_{\mathsf{-N}}\Big).
$$
In particular, for any $u,v\in\sch'$ with $\wfs(\{A_t\})\cap\Big(\wf{m-s}(u)\cup\wf{s}(v)\Big)=\varnothing$, the family $\langle A_tu,v\rangle$ is bounded in $\C$. If in addition $A_t\to A\in \ps[m]$ in $\ps[m']$ for some $\mathsf{m'}\geqslant\mathsf{m}$, then $\langle A_tu,v\rangle\to \langle Au,v\rangle$ in $\C$.
\item If $A\in\ps[m]$ and $\wfs(A)\cap \Big(\wf{m-s}(u)\cup\wf{s}(v)\Big)$, then $\langle Au,v\rangle = \langle u, A^*v\rangle$.
\end{enumerate}
\label{thm:int-by-parts}
\end{lemma}
\begin{proof}\ 
\begin{enumerate}
\item 
Note that for $Q$ as required, $Qu\in\sch$ and $(I-Q^*)v\in\sob{s}$, so the pairings in both terms are well-defined. Consider two different pairs $(\mathsf{s_1},Q_1)$ and $(\mathsf{s}_2,Q_2)$ as required. 

Fix an approximation of the identity $J_t$. Then by Lemma~\ref{thm:uniform-lemma}, $(I-Q_1^*)J_t v\to (I-Q_1^*)v$ in $\sob{s_1}$ and $(I-Q_2^*)J_t v\to (I-Q_2^*)v$ in $\sob{s_2}$. Then using continuity of the pairing between dual Sobolev spaces, we can write
$$
\Big( \langle Q_1u,v\rangle + \langle u,(I-Q_1^*)v\rangle \Big)
-
\Big( \langle Q_2u,v\rangle + \langle u,(I-Q_2^*)v\rangle \Big)
=
$$
$$
=
\langle (Q_1-Q_2)u,v \rangle
+
\lim_{t\to 0^+} \Big(
\langle u,(I-Q_1^*) J_t v\rangle 
-
\langle u,(I-Q_2^*) J_t v\rangle 
\Big)
=
$$
$$
=
\langle (Q_1-Q_2)u,v\rangle -
\lim_{t\to 0^+} \langle u,(Q_1^*-Q_2^*) J_t v\rangle
=
\lim_{t\to 0^+} \langle (I-J_t^*)(Q_1-Q_2)u,v\rangle,
$$
where we could integrate by parts freely because $(Q_1^*-Q_2^*)J_t\in\ps[-\infty]$. Now since $(Q_1-Q_2)u\in\sch$, we have $(I-J_t^*)(Q_1-Q_2)u\to 0$ in $\sch$, so the whole expression is zero.

\item Fix any $Q\in\ps[0]$ such that $\wfs(Q)\cap \wfs(\{A_t\})=\varnothing$ and $\wfs(I-Q)\cap\Sigma^{\mathsf{0}}(Q_2)=\varnothing$. Then $\langle A_tu,v\rangle = \langle QA_tu,v\rangle + \langle A_tu, (I-Q^*)v\rangle$, where
\begin{itemize}
\item $Q A_t$ is uniformly bounded in $\ps[-2N]$ for every $\mathsf{N}$, so $|\langle QA_tu,v\rangle|\leqslant C\|u\|_{\mathsf{-N}}\|v\|_{\mathsf{-N}}$;
\item By Lemma~\ref{thm:uniform-lemma}, $\|A_tu\|_{\mathsf{-s}} \leqslant C\Big(\|Q_1u\|_{\mathsf{m-s}} + \|u\|_{\mathsf{-N}}\Big)$ and $\|(I-Q^*)v\|_{\mathsf{s}}\leqslant C\Big(\|Q_2v\|_{\mathsf{s}} + \|v\|_{\mathsf{-N}}\Big)$.
\end{itemize}
The boundedness statement follows analogously to Lemma~\ref{thm:uniform-lemma}.

If in addition $A_t\to A$ as assumed, then necessarily $\wfs(Q)\cap\wfs(A)=\varnothing$, so we also have $\langle Au,v\rangle = \langle QAu,v\rangle + \langle Au,(I-Q^*)v\rangle$. By Lemma~\ref{thm:uniform-lemma}, $A_tu\to Au$ in $\sob{-s}$ and $QA_tu\to QAu$ in $\sch$, so we have $\langle QA_tu,v\rangle \to \langle QAu,v\rangle$ and $\langle A_tu,(I-Q^*)v\rangle \to \langle Au,(I-Q^*)v\rangle$.

\item Fix an approximation of the identity $J_t$. Then by the previous part, since $J_tA\in\ps[-\infty]$,
$$
\langle Au,v\rangle
=
\lim_{t\to 0^+} \langle (J_tA)u,v\rangle
=
\lim_{t\to 0^+} \langle u, (J_tA)^*v\rangle
=
\langle u, A^*v\rangle.
$$

\end{enumerate}
\end{proof}

The positive-commutator estimate used in the proof of radial point estimates relies on the following additional integration-by-parts result. The $\tau=0$ case is the one relevant for the main real principal type estimate, but for the limiting absorption principle we are also interested in adding a constant imaginary term $-i\tau$, where we are primarily interested in $\tau$ small.

\begin{lemma}[Commutator lemma, cf.\ Lemma~3.4 of \cite{Haber-Vasy}]
\label{thm:comm-lemma}
Consider $P\in\ps[m]$ with real-valued principal symbol, $A=\check{A}^*\check{A}$ for some $\check{A}\in\ps[s-\frac{m-1}{2}]$, and $u\in \sch'$ such that
$$
\wfs(\check{A})\cap\Big(\wf{s}(u)\cup\wf{s-m+1}((P-i\tau)u)\Big)
=
\varnothing
$$
for some $\tau\in\R$. Then in the $\tau=0$ case, $
-i\Big\langle \big([P,A]-(P-P^*)A \big)u,u\Big\rangle
=
2\operatorname{Im}\langle Au,Pu\rangle$,
whereas if $\tau\neq 0$, then $\check{A}u\in \LL$ and
\begin{equation}
-i\Big\langle \big([P,A]-(P-P^*)A \big)u,u\Big\rangle
=
2\operatorname{Im}\langle Au,(P-i\tau)u\rangle
+2\tau\|\check{A}u\|^2.
\end{equation}

\end{lemma}
\begin{proof}
Note that since $P$ has real principal symbol, hence the same principal symbol as $P^*$, we have $[P,A]-(P-P^*)A\in \ps[2s]$, so the pairings are all well-defined since this operator's essential support is contained in $\wfs(\check{A})$. For regular enough $u$, the result follows from the calculation
$$
-i\Big\langle \big([P,A]-(P-P^*)A \big)u,u\Big\rangle
=
-i\langle (P^*A-AP)u,u\rangle
=
-i\Big(\langle Au,Pu\rangle - \langle Pu,Au\rangle\Big)
=
2\operatorname{Im}\langle Au,Pu\rangle
=
$$
$$
=
2\operatorname{Im}\langle Au,(P-i\tau)u\rangle
+2\tau\|\check{A}u\|^2.
$$
However, the second and fourth equalities above cannot be justified (and some of the pairings need not be well-defined) under only the assumptions made. Instead, fix an approximation of the identity $\check{J}_t$. Then $J_t=\check{J}_t^*\check{J}_t$ is also an approximation of the identity. By Lemma~\ref{thm:int-by-parts}, part 2,
$$
\lim_{t\to 0^+} \langle (P^*\check{A}^*J_t\check{A}-\check{A}^*J_t\check{A}P)u,u\rangle 
=
$$
$$
=
\lim_{t\to 0^+}\Big( \langle J_t(P^*A-AP)u,u\rangle + \langle [I-J_t,P^*\check{A}^*]\check{A} u,u\rangle + \langle [\check{A}^*, I-J_t]\check{A}P u,u\rangle \Big)
=
$$
$$
=
\langle (P^*A-AP)u,u\rangle
=
\Big\langle \big([P,A]-(P-P^*)A \big)u,u\Big\rangle.
$$
On the other hand,
$$
2\operatorname{Im}\langle Au,(P-i\tau)u\rangle
=
-i\Big(\langle Au,(P-i\tau)u\rangle - \langle (P-i\tau)u,Au\rangle\Big)
=
$$
$$
=
-i \lim_{t\to 0^+} \Big(
\langle\check{A}^*J_t\check{A}u,(P-i\tau)u\rangle
- \langle (P-i\tau)u,\check{A}^*J_t\check{A}u\rangle\Big)
=
$$
$$
=
 \lim_{t\to 0^+} \Big( 
-i \langle (P^*\check{A}^*J_t\check{A}-\check{A}^*J_t\check{A}P)u,u\rangle 
+2\tau \|\check{J}_t\check{A} u\|^2 \Big), 
$$
which establishes the $\tau=0$ result. For $\tau\neq 0$, note that the existence of both limits above means that the family $\check{J}_t\check{A}u$ is bounded in $\LL$ for small enough $t$. Then it contains a sequence that converges weakly to some limit in $\LL$; since $\lim_{t\to 0^+}\check{J}_t\check{A}u=\check{A} u$ in the weaker $\sch'$ topology, the limit must be $\check{A}u$. Thus, $\check{A}u\in \LL$, and the $\tau\neq 0$ result follows.
\end{proof}

The extra regularity result we get in the $\tau\neq 0$ case, which is helpful in proving a radial point estimate for $P-i\tau$ which is uniform in $\tau\geqslant 0$ near zero and eventually establishing the limiting absorption principle, can be restated separately as follows.

\begin{cor}
\label{thm:complex}
Consider $P\in\ps[m]$ with real-valued principal symbol. For any $Q,Q',Q''\in\ps[0]$ with $Q',Q''$ elliptic on $\wfs(Q)$ and any set of orders $\mathsf{N}$, there exists $C>0$ such that all $\tau\in\R\backslash\{0\}$ and $u\in\sob{-N}$ such that $Q'(P-i\tau)u\in\sob{s-m+1}$, $Q''u\in\sob{s}$ satisfy
\begin{equation}
\|Qu\|_{\mathsf{s-\frac{m-1}{2}}}^2
\leqslant
\frac{C}{|\tau|} \Big(
\|Q'(P-i\tau)u\|_{\mathsf{s-m+1}}^2
+\|Q''u\|_{\mathsf{s}}^2
+(1+|\tau|)\|u\|_{\mathsf{-N}}^2
\Big).
\label{eq:complex-reg}
\end{equation}
In particular, $\alpha\notin \Big(\wf{s}(u) \cup \wf{s-m+1}((P-i\tau)u)\Big)$ for some $\tau\in\R\backslash\{0\}$ implies $\alpha\notin \wf{s-\frac{m-1}{2}}(u)$.
\end{cor}
\begin{proof}
We can take $\check{A}\in\ps[s-\frac{m-1}{2}]$ elliptic on $\wfs(Q)$ and such that $Q',Q''$ are elliptic on $\wfs(\check{A})$. The finiteness of the norms on the right-hand side of Eq.~(\ref{eq:complex-reg}) implies that $(\check{A},u,\tau)$ satisfy the assumption of Lemma~\ref{thm:comm-lemma}. Then using the preceding lemmas, we can estimate
$$
\|Qu\|_{\mathsf{s-\frac{m-1}{2}}}^2
\leqslant
C\Big(
\|\check{A}u\|^2 + \|u\|_{\mathsf{-N}}^2
\Big)
\leqslant
$$
$$
\leqslant
C\Big(\frac{1}{|\tau|}
\left| \left\langle ([P,A]-(P-P^*)A)u,u\right\rangle \right|
+ \frac{1}{|\tau|} |\langle \check{A}^*\check{A} u,(P-i\tau)u\rangle |
+ \|u\|_{\mathsf{-N}}^2
\Big)
\leqslant
$$
$$
\leqslant
\frac{C}{|\tau|}\Big(
\| Q''u \|_{\mathsf{s}}^2
+ \|Q'(P-i\tau)u\|_{\mathsf{s-m+1}}^2
+ (1+|\tau|)\|u\|_{\mathsf{-N}}^2
\Big).
$$
The wavefront set conclusion follows from the fact that we can take $Q$ elliptic at $\alpha$ with $\wfs(Q)\cap\Big(\wf{s}(u) \cup \wf{s-m+1}((P-i\tau)u)\Big)=\varnothing$.
\end{proof}

	\subsection{Propagation of singularities}
	\label{sec:PoS}
	We are interested in determining regularity of solutions $u$ to an equation $Pu=f$ given information about regularity of $f$. Wherever $P$ is elliptic, such information is provided by elliptic regularity. Within the characteristic set, on the other hand, information about microlocal regularity of $u$ is propagated along the Hamilton vector field associated to the principal symbol of $P$ (if the symbol is real-valued).
	
	Given a symbol $p\in\sym{m}$, the associated Hamilton vector field $H_p$ on $T^*\M^{\circ}$ is defined by $\omega(H_p,V)=dp(V)$ for all vector fields $V$, where $\omega$ is the canonical symplectic form on $T^*\M^{\circ}$, i.e.
	\begin{equation}
H_p(x,\tilde{\xi})=	\sum_{i=1}^{d+1} \left( \frac{\partial p}{\partial \tilde{\xi}_i}\frac{\partial}{\partial x_i} - \frac{\partial p}{\partial x_i}\frac{\partial}{\partial \tilde{\xi}_i} \right)
	\label{eq:Hp-canonical}
	\end{equation}
for any local coordinates $(x_1,\ldots,x_{d+1})$ on $\M^{\circ}$, where $(\tilde{\xi}_1,\ldots,\tilde{\xi}_{d+1})$ are the canonical dual variables. 

In local coordinates $(\rho,y,\xi,\eta)$ on $\PP$ away from null infinity as defined in Eq.~(\ref{eq:freq-transform-sc}), this becomes
\begin{equation}
H_p = \rho \Bigg[
\rho \frac{\partial p}{\partial\xi} \frac{\partial}{\partial\rho}
+ \sum_{i=1}^d \frac{\partial p}{\partial \eta_i} \frac{\partial}{\partial y_i}
- \left( \rho\frac{\partial p}{\partial\rho} + \sum_{i=1}^d \eta_i\frac{\partial p}{\partial \eta_i} \right) \frac{\partial}{\partial \xi}
+ \sum_{i=1}^d \left(\eta_i \frac{\partial p}{\partial\xi} - \frac{\partial p}{\partial y_i}\right) \frac{\partial}{\partial \eta_i}
\Bigg]
\label{eq:Hp-coordinates-general-sc}
\end{equation}
Since $\frac{\partial}{\partial\xi},\frac{\partial}{\partial\eta_i}\in\rho_f\Vb(\PP)$, $\rho\frac{\partial}{\partial\rho},\frac{\partial}{\partial y}\in\Vb(\PP)$ in such a neighborhood, and the order of a symbol does not change under application of b-vector fields, by considering each term above we conclude that $H_p$ is a section of ${}^{\mathrm{b}}T\PP$ with coefficients in $\sym{m-1}$ away from null infinity.

Turning to null infinity, in the coordinates $(\rho_0,x_0,y,\zeta_0,\xi_0,\eta)$ on $\mathcal{P}$ over $\U_0$ defined in Eq.~(\ref{eq:freq-transform}),
\begin{equation}
\begin{split}
H_p
			=
			\rho_0x_0\Bigg[
			\left(
			(\xi	_0-\zeta_0)\frac{\partial p}{\partial\zeta_0}
			-x_0\frac{\partial p}{\partial x_0}
			-2 \sum_{i=1}^{d-1} \eta_i\frac{\partial p}{\partial\eta_i}
			\right)
			\frac{\partial}{\partial\xi_0}
			+
			\left(
			(\zeta_0-\xi_0)\frac{\partial p}{\partial\xi_0}
			-\rho_0\frac{\partial p}{\partial\rho_0}
			-\sum_{i=1}^{d-1} \eta_i\frac{\partial p}{\partial\eta_i}
			\right)
			\frac{\partial}{\partial\zeta_0}
			+
			\\
			+
			x_0\frac{\partial p}{\partial\xi_0}
			\frac{\partial}{\partial x_0}
			+
			\rho_0\frac{\partial p}{\partial\zeta_0}
			\frac{\partial}{\partial\rho_0}
			+
			x_0\sum_{i=1}^{d-1} \frac{\partial p}{\partial\eta_i}\frac{\partial}{\partial y_i}
			+
			\sum_{i=1}^{d-1}
			\left(
			2\eta_i\frac{\partial p}{\partial\xi_0}
			+\eta_i\frac{\partial p}{\partial\zeta_0}
			-x_0\frac{\partial p}{\partial y_i}
			\right)
			\frac{\partial}{\partial\eta_i}
			\Bigg]
			.
\end{split}
\label{eq:Hp-coordinates-general}
\end{equation}
The calculation over $\U_T$ is completely identical, with $\rho_0$ and $\zeta_0$ replaced by $\rho_T$ and $\zeta_T$. Since $\frac{\partial}{\partial\zeta_0},\frac{\partial}{\partial\xi_0},\frac{\partial}{\partial\eta_i}\in\rho_f\V_{\mathrm{e,b}}(\mathcal{P})$, $\rho_0\frac{\partial}{\partial \rho_0},x_0\frac{\partial}{\partial x_0},x_0\frac{\partial}{\partial y_i}\in\V_{\mathrm{e,b}}(\mathcal{P})$ (and similarly for the coordinate vector fields over $\U_T$)\footnote{$\PP$ inherits the fibration of $\scri^{\pm}$ defined by the blowdown map $\M\to\tilde{\M}$, so $\V_{\mathrm{e,b}}(\mathcal{P})$ is defined analogously to $\V_{\mathrm{e,b}}(\M)$ as the space of b-vector fields on $\PP$ which are in addition tangent to the fibers of $\scri^{\pm}$.}, we conclude that $H_p$ is globally a section of ${}^{\mathrm{e,b}}T\mathcal{P}$ with coefficients in $\sym{m-1}$. Therefore, fixing any collection $\rho$ of defining functions of the boundary hypersurfaces, we can write $H_p = \rho^{\mathsf{-m+1}}\Hp$, where we call $\Hp\in S^{\mathsf{0}}\V_{\mathrm{e,b}}(\mathcal{P})$ the rescaled Hamilton vector field. In the general considerations that follow, we will only explicitly use that $H_p$ is a b- (as opposed to e,b-) vector field on $\PP$, but the additional vanishing of the $\frac{\partial}{\partial y_i}$ components at null infinity is reflected in the dynamics.

	\underline{Now we restrict attention to $p\in S^{\mathsf{m}}_{\epsilon}(\PP)$.} We call $\tilde{p}=\rho^{\mathsf{m}}p\in \soe(\PP)$ the rescaled principal symbol. $\tilde{p}$ is continuous on $\PP$ up to and including the boundary, and its restriction to the interior of any corner is smooth. We denote 
$$
\Sigma=\Sigma^{\mathsf{m}}(P) \text{ and }\tilde{\Sigma}=\{\alpha\in\PP\ |\ \tilde{p}(\alpha)=0\}, \text{ so }\Sigma = \tilde{\Sigma}\cap\partial\PP.
$$
Since $\tilde{p}$ is not necessarily smooth at the boundary, $\tilde{\Sigma}$ is not necessarily a product-type (p-) or neat smooth submanifold of $\PP$. However, if $\Gamma$ is the intersection of all boundary hypersurfaces containing some $\alpha\in\Sigma$ (so automatically $\alpha\in\Gamma^{\circ}$), then $\tilde{p}|_{\Gamma}$ is smooth nearby, so if $d(\tilde{p}|_{\Gamma})(\alpha)\neq 0$, then $\Sigma\cap\Gamma$ is a smooth submanifold of $\Gamma$ near $\alpha$.
	
	Note that Eqs.~(\ref{eq:Hp-coordinates-general-sc})-(\ref{eq:Hp-coordinates-general}) show that in coordinates, the coefficients of $H_p$ as a b-vector field are the results of acting on $p$ with some vector fields in $S^{\mathsf{-1}}\Vb(\PP)$. Therefore, if a symbol is classical at some boundary hypersurface, the coefficients of its Hamilton vector field are classical at that face as well. Since the correspondence $p\mapsto H_p$ is linear, we conclude that $\Hp\in \soe(\PP)\V_{\mathrm{e,b}}(\PP)$ for $p\in S^{\mathsf{m}}_{\epsilon}(\PP)$. In particular, it is also continuous on $\PP$ up to and including the boundary, with smooth restriction to the interior of any corner, and in fact $\Hp|_{\Gamma}\in \soe\Vb(\Gamma)$ for any corner $\Gamma$. (This is a slight abuse of notation, since the restriction is really a section of $T\PP|_{\Gamma}$, but since it is tangent to all boundary hypersurfaces, we can consider it as a section of $T\Gamma$). Then the integral curves (some of which may be single critical points) of $\Hp$ are well-defined, and any curve with a point in $\Gamma^{\circ}$ lies entirely in $\Gamma^{\circ}$ (though it may limit to a point on the boundary of $\Gamma$ as the natural parameter goes to $\pm\infty$)\footnote{The flow cannot reach the boundary of $\Gamma$ in finite parameter time because $\Hp$ is a b-vector field, so its component transverse to any boundary face is bounded by a constant times its defining function and the flow can only approach the boundary exponentially slowly.}. This allows one to consider the dynamics of the Hamilton flow one corner at a time.
	
	Eq.~(\ref{eq:Hp-canonical}) implies that $H_pp=0$, so $\Hp\tilde{p} = \left(\sum_{i=1}^{|\mathcal{G}(\PP)|} m_i\rho_i^{-1}(\Hp\rho_i)\right)\tilde{p}$. Since the prefactor is in $\soe(\PP)$, we conclude that $\Hp \tilde{p}|_{\tilde{\Sigma}}=0$, so any integral curve of $\Hp$ with a point in $\tilde{\Sigma}$ lies entirely in $\tilde{\Sigma}$, and similarly for $\Sigma$. We call the integral curves of $\Hp$ within $\Sigma$ bicharacteristics; the curves themselves are independent of the choice of boundary-defining functions used to rescale $H_p$ into $\Hp$, though the curves' natural parametrization depends on this choice. 
	
	If $d(\tilde{p}|_{\Gamma})(\alpha)=0$ at some point $\alpha\in\Sigma$, then $V\tilde{p}(\alpha)=0$ for any $V\in\Vb(\PP)$, and hence $\rho^{\mathsf{-m}}V p(\alpha) = \rho^{\mathsf{-m}}V(\rho^{\mathsf{m}})\tilde{p}(\alpha) + V\tilde{p}(\alpha)=0$; this then implies that $\Hp(\alpha)=0$. Thus any bicharacteristic consisting of more than one point is automatically contained in a region where $d(\tilde{p}|_{\Gamma})$ does not vanish and $\Sigma\cap\Gamma$ is locally a smooth submanifold of $\Gamma$.
	
	H\"{o}rmander's propagation of singularities theorem \cite{Hormander,D-H} for operators of real principal type, as well as the version with a sign-definite imaginary part, carries over to the de,sc-setting. The following version is sufficient for our purposes.
	  
	\begin{theorem}[Propagation of singularities]
	\label{thm:pos}
	Let $P\in\ps[m]$ with real-valued principal symbol $p\in \soe(\PP)$. Consider distinct $\alpha,\beta\in\Sigma$ lying on the same bicharacteristic $\gamma$ in a corner $\Gamma$ of $\PP$, with $\beta$ downstream from $\alpha$ with respect to $\Hp$; denote the segment of $\gamma$ between them $\gamma_{\alpha\beta}$.
	\begin{enumerate}
	
	\item For $\tau$ in a bounded subset of $\R$, either all non-negative or all non-positive, consider a family $u_{\tau}$ bounded in $\sob{-N}$ such that $\wf{s-m+1}(\{(P-i\tau)u_{\tau}\})\cap \gamma_{\alpha\beta}=\varnothing$ for some sets of orders $\mathsf{s}$, $\mathsf{N}$. If $\tau \geqslant 0$, then $\alpha\notin\wf{s}(\{u_{\tau}\})$ implies $\beta\notin\wf{s}(\{u_{\tau}\})$; if $\tau \leqslant 0$, then $\beta\notin\wf{s}(\{u_{\tau}\})$ implies $\alpha\notin\wf{s}(\{u_{\tau}\})$.
	
	\item For any neighborhoods $U$ of $\gamma_{\alpha\beta}$ and $V$ of $\alpha$ within $\partial\PP$, there exists a neighborhood $W$ of $\gamma_{\alpha\beta}$ within $\partial\PP$ such that for any $Q,Q',Q''\in\ps[0]$ with $Q'$ elliptic on $U$, $Q''\in\ps[0]$ elliptic on $V$, and $\wfs(Q)\subset W$ and any set of orders $\mathsf{N}$, there exists $C>0$ such that all $\tau\geqslant 0$ and $u\in\sob{-N}$ with $Q'(P-i\tau)u\in\sob{s-m+1}$, $Q''u\in\sob{s}$ satisfy
\begin{equation}
\label{eq:pos-estimate}
\|Qu\|_{\mathsf{s}} \leqslant C\Big( \|Q'(P-i\tau) u_{\tau}\|_{\mathsf{s-m+1}} + \|Q''u\|_{\mathsf{s}} + (1+|\tau|)\|u\|_{\mathsf{-N}} \Big).
\end{equation}
The analogous statement holds for $\tau\leqslant 0$ if $\beta$ is upstream from $\alpha$.
	\end{enumerate}
	\end{theorem}

	The proof is based on the standard positive-commutator argument, which is based on constructing a symbol $a$ (the ``commutant") supported near $\gamma_{\alpha\beta}$ and monotone non-increasing/non-decreasing along $H_p$ wherever $u$ is not assumed to be microlocally $H^{\mathsf{s}}$-regular (i.e. away from a neighborhood of either $\alpha$ or $\beta$, where we assume regularity). Using the identity $H_pa=\{p,a\}=\sigma(i[P,A])$, one then shows that, assuming regularity of $Pu$ on $\gamma_{\alpha\beta}$, regularity of $u$ is propagated forward/backward along $\Hp$ in the region where $H_pa$ is definite of the chosen sign, allowing one to conclude regularity in all regions downstream/upstream from the region where it is assumed. This is described in detail for classical sc-operators in \cite[Section 4]{Grenoble-notes}, which for de,sc-operators with \textit{classical} symbols applies with minimal modifications. Since the de,sc-calculus is symbolic at every boundary face, the presence of more faces and higher-order corners does not introduce any difficulties.
	
	For $p$ which is classical only to leading order, the Hamilton flow on $\PP$ cannot necessarily be straightened by a choice of smooth local coordinates, which is the first step in the approach of \cite{Grenoble-notes}. However, since $\Hp$ deviates from a smooth b-vector field only by terms which vanish at $\Gamma$, one can work in coordinates in which that vector field is straightened, and then the commutant construction of \cite{Grenoble-notes} only requires a slight modification to work. Namely, for a fully straightened vector field one can fully separate variables when constructing $a$, so the monotonicity along $H_p$ is built into a choice of factor dependent only on the variable along the flow, while in the transverse variables one localizes using any bump function; in our case, the vector field may have nonzero (but bounded) components in the transverse directions, and correspondingly the factors localizing in those variables need to be modified to slope strongly enough towards $\gamma_{\alpha\beta}$ so that all nearby bicharacteristics still cross the level sets of $a$ in the same direction. A proof of propagation of singularities where one similarly needs to leave ``extra room for errors" by making the propagation region narrower as one moves along the flow appeared already in the work of Melrose and Sj\"{o}strand~\cite{MS}.
	
	Because we only require $p$ to be classical at leading order, and correspondingly the components of $\Hp$ can decay to their boundary values quite slowly and irregularly (with an $\mathcal{O}(\rho^{\mathsf{\epsilon}})$ envelope and possibly oscillations), some care needs to be taken to ensure that uncontrolled lower-order terms in $H_p$ do not give rise to leading-order terms in $H_p a$. To do this, we localize using small powers of boundary-defining functions (which are not smooth but symbolic on $\PP$), the scale being set by $\epsilon$, making the leading-order terms large enough to dominate over the large errors. We note that if the non-classical error terms in $\tilde{p}$ are assumed to decay at \textit{all} of the boundary faces (which in our context is the case for e.g. asymptotically Minkowski metrics), then an alternative approach is to use a notion of principal symbol which tracks decay order in increments smaller than $\epsilon$. Our approach, however, only requires each error term in $\tilde{p}$ to decay at \textit{some} face containing $\Gamma$.
	
	Further steps necessary to regularize the argument as well as treat operators with nonzero skew-adjoint part and/or a sign-definite imaginary part are standard, though we note that adding a constant imaginary part $i\tau$ (which we are interested in because of the limiting absorption principle) is somewhat different from the more common scenario where one adds a microlocalized operator of the same order as $P$ to act as a complex absorption (see \cite[Section 4.5]{Grenoble-notes}). We warn the reader that we reserve $\epsilon$ to stand for the order of error terms of symbols in $\soe(\PP)$ but also use $\varepsilon$ for various other small parameters.
	
	\begin{proof}[Proof of Theorem~\ref{thm:pos}]
	
	We focus on the case when $\tau\geqslant 0$ and $\beta$ is downstream from $\alpha$; the other case is analogous, requiring only some sign changes. Fix neighborhoods $U,V\subset\partial\PP$ as described. We aim to show the existence of $W\subset\partial\PP$ as required.
	
	 Fix any defining functions $\rho_1,\ldots,\rho_k$ of the boundary hypersurfaces forming $\Gamma$. We denote by $p_1$ the principal symbol of $-i(P-P^*)\in\ps[m-1]$. We use multi-index notation with $\rho=(\rho_1,\ldots,\rho_k)$ and write $\tilde{p}=\rho^{\mathsf{m}}p$, $\tilde{p}_1=\rho^{\mathsf{m-1}}p_1$, and $\Hp=\rho^{\mathsf{m-1}}H_p$.
	
\begin{itemize}

\item \textbf{Choice of coordinates.}
	 
	  Take any smooth extension $\tilde{p}_0$ of $\tilde{p}|_{\Gamma}$ near $\gamma_{\alpha\beta}$ and set $p_0=\rho^{\mathsf{-m}}\tilde{p}_0$. Then we can choose coordinates which straighten $\tilde{H}_{p_0}$ in a neighborhood of $\gamma_{\alpha\beta}$. Thus, consider a neighborhood $\tilde{U}$ of $\gamma_{\alpha\beta}$ in $\PP$ such that $\tilde{U}\cap\partial\PP \subset U$ and neither $\Hp$ nor $\tilde{H}_{p_0}$ vanishes on $\tilde{U}$ and on which there exist coordinates $(\rho_1,\ldots,\rho_k,q,y_1,\ldots,y_m)$ such that $(\rho_1,\ldots,\rho_k)$ are defining functions of the boundary hypersurfaces containing $\gamma_{\alpha\beta}$, $\tilde{H}_{p_0}=\frac{\partial}{\partial q}$ on all of $\tilde{U}$, $y|_{\gamma_{\alpha\beta}}=0$, and $q(\alpha)=0$, $q(\beta)=s$ for some $s>0$. Then, since $\tilde{p}=\tilde{p}_0 \mod\mathcal{I}^{\epsilon}_{\Gamma}(\PP)$, we have
	$$
\Hp = h_q\frac{\partial}{\partial q} + \sum_{i=1}^k h_{\rho_i}\rho_i\frac{\partial}{\partial \rho_i} + \sum_{i=1}^m h_{y_i}\frac{\partial}{\partial y_i},
	$$
where $h_q=1\mod \mathcal{I}_{\Gamma}^{\epsilon}\Vb(\PP)$ and $h_{\rho_i},h_{y_i}\in\mathcal{I}_{\Gamma}^{\epsilon}\Vb(\PP)$.

\item \textbf{Defining the commutant.} Fix $N>\frac{2}{\epsilon}$, and let $\tilde{\rho}_i= \sqrt[N]{\rho_i}$ for $i=1,\ldots,k$. Let $z^2=\sum_{i=1}^k\tilde{\rho}_i^2 + \sum_{i=1}^m y_i^2$. For $\varepsilon>0$ small enough, we denote
$$
\tilde{U}_{\varepsilon} = \{ z^2<2\varepsilon,\ -2\varepsilon<q<s+2\varepsilon\},
\hspace{30pt}
\tilde{V}_{\varepsilon} = \{z^2<2\varepsilon,\ |q|<2\varepsilon\}
$$
and $U_{\varepsilon}=\tilde{U}_{\varepsilon}\cap\partial\PP$, $V_{\varepsilon}=\tilde{V}_{\varepsilon}\cap\partial\PP$. Consider $\varepsilon>0$ to be fixed later, small enough that $\tilde{U}_{\varepsilon}\subset\tilde{U}$ and $V_{\varepsilon}\subset V$. Choose $\phi,\psi,\chi\in C^{\infty}(\R)$ such that:
\begin{itemize}
\item $\phi(t)=0$ for $|t|>\varepsilon$ and $\phi(t)=1$ for $|t|<\frac{1}{2}\varepsilon$, strictly increasing on $[-\varepsilon,-\frac{1}{2}\varepsilon]$ and strictly decreasing on $[\frac{1}{2}\varepsilon,\varepsilon]$, and such that $\sqrt{|\phi'\phi|}$ is smooth;
\item $\psi(t)=0$ for for $t<-\varepsilon$ and $\psi(t)=1$ for $t>\varepsilon$, strictly increasing on $[-\varepsilon,\varepsilon]$, and such that $\sqrt{|\psi'\psi|}$ is smooth;
\item $\chi(t)=e^{\frac{\digamma}{t-(s+\varepsilon)}}$ for $t<s+\varepsilon$ and $\chi(t)=0$ for $t\geqslant s+\varepsilon$, for some $\digamma>0$ to be fixed later.
\end{itemize}

Let $\hat{z}=z^2+\frac{\varepsilon}{2s} q$; note that $z^2$ and therefore $\hat{z}$ are in $\sym{0}$ (unlike $z=\sqrt{z^2}$). For small enough $\varepsilon$ we have $\mathrm{supp}(\phi(\hat{z})) \cap \{-2\varepsilon<q<s+2\varepsilon\} \subset 
\{|z|^2<2\varepsilon\}$. Then for any set of orders $\mathsf{r}$, we can define $\check{a}\in \sym{r-\frac{m-1}{2}}$ by
$$
\check{a}=\rho^{\mathsf{-r+\frac{m-1}{2}}} \phi(\hat{z})\psi(q)\chi(q)
$$ 
on $\tilde{U}_{\varepsilon}$ and zero outside. We set $a=\check{a}^2\in\sym{2r-m+1}$. We denote
$$
\tilde{W}_{\varepsilon} = \left\{ -\varepsilon<q<s+\varepsilon,\ z^2<\varepsilon \left(1-\frac{q}{2s}\right) \right\} \subset \tilde{U}_{\varepsilon},
\hspace{30pt}
W_{\varepsilon} = \tilde{W}_{\varepsilon}\cap\partial\PP.
$$
We have $\mathrm{supp}(a)=\overline{\tilde{W}_{\varepsilon}}$.

\item \textbf{Action of vector field on commutant.} Fix any $\delta>0$. Then
$$
H_p a + p_1a + \delta\rho^{\mathsf{-m+1}}a
=
\rho^{\mathsf{-m+1}} (\Hp a + \tilde{p}_1a + \delta a)
=
\rho^{\mathsf{-2r}} \phi(\hat{z}) \psi(q)\chi(q) \cdot 
$$
$$
\cdot
\Bigg(
\left[
2h_q
\Big(\psi'(q)\chi(q) + \psi(q)\chi'(q)\Big)
+
\left(\sum_{i=1}^k h_{\rho_i} (-2r_i+m_i-1) + \tilde{p}_1 + \delta \right)
\psi(q)\chi(q)
\right]
\cdot \phi(\hat{z})
+
$$
$$
+
\left(\frac{\varepsilon h_q}{s} +\frac{4}{N}\sum_{i=1}^k h_{\rho_i} \tilde{\rho}_i^2
+4\sum_{i=1}^m h_{y_i} y_i\right) \cdot \phi'(\hat{z})
\psi(q) \chi(q) 
\Bigg).
$$
We consider each of the resulting terms.
\begin{itemize}
\item The term
$$
e = \rho^{\mathsf{-2r}} \cdot 2h_q \psi'(q)\psi(q)\chi(q)^2 \phi(\hat{z})^2 \in \sym{2r}
$$
containing $\psi'(q)$ is supported in $\tilde{V}_{\varepsilon}$, so its contribution to the estimate will be controlled by the regularity assumption on $u$ near $\alpha$.

\item Taking $\varepsilon$ small enough and $\digamma$ large enough, the remaining terms on the second line can be written as $-b^2$, where
$$
b = \rho^{\mathsf{-r}}\phi(\hat{z}) \psi(q)
\sqrt{\left|
2h_q\cdot \chi'(q)\chi(q)
+\left(\sum_{i=1}^k h_{\rho_i} (-2r_i+m_i-1) + \tilde{p}_1 + \delta \right) \chi(q)^2
\right|
}
\in\sym{r}.
$$
To see this, note that for $\varepsilon$ small enough, $h_q>\frac{1}{2}$ on $\tilde{U}_{\varepsilon}$. On the other hand, $\sum_{i=1}^k h_{\rho_i} (-2r_i+m_i-1) + \tilde{p}_1 + \delta $ is bounded on $\tilde{U}_{\varepsilon}$. Then since $\chi(q)=-\frac{(s+\varepsilon-q)^2}{\digamma}\chi'(q)$, by taking $\digamma$ large enough we can ensure that the expression under the absolute value is a strictly negative smooth multiple of $\chi'(q)^2$ on the support of the localizing prefactors, so the square root is smooth.

\item Considering the last line, for $\varepsilon$ small enough we have $\frac{\varepsilon h_q}{s}>\frac{\varepsilon}{2s}$ on $\tilde{U}_{\varepsilon}$. Since $|h_{\rho_i}|,|h_{y_i}|=\mathcal{O}(\rho^{\mathsf{\epsilon}})$ as $\rho\to 0^+$ and $\rho_i^{\frac{2}{N}}< 2\varepsilon$ on $\tilde{U}_{\varepsilon}$ with $\frac{2}{N}<\epsilon$, we have $\sup_{\tilde{U}_{\varepsilon}} |h_{\rho_i}|,\sup_{\tilde{U}_{\varepsilon}} |h_{y_i}| =o(\varepsilon)$ as $\varepsilon\to 0^+$. Thus, for $\varepsilon$ small enough the expression in parentheses in the last line is strictly positive on the support of the localizing factors, so the last line can be written as $-f^2$, where
$$
f = \rho^{\mathsf{-r}} \psi(q) \chi(q) \sqrt{\left| \left( \frac{\varepsilon h_q}{s} + \frac{4}{N}\sum_{i=1}^k h_{\rho_i}\tilde{\rho}_i^2 + \frac{4}{N}\sum_{i=1}^m h_{y_i}y_i \right) \cdot \phi'(\hat{z})\phi(\hat{z}) \right|} \in\sym{r}.
$$
This term will come with the right sign to be discarded for the estimate.
\end{itemize}

Thus, we conclude that
\begin{equation}
H_pa + p_1a + \delta\rho^{\mathsf{-m+1}}a
=
-b^2 + e -f^2.
\end{equation}

\item \textbf{Regularization.}
For $t\in (0,1)$, we define a regularizer family
\begin{equation}
\label{eq:pos-regularizer}
\chi_t=\frac{1}{\prod_{i=1}^k\left( 1 + t\rho_i^{-1}\right)^{K_i}} = \frac{\rho^{\mathsf{K}}}{\prod_{i=1}^k (t+\rho_i)^{K_i}}
\end{equation}
for some multi-index $\mathsf{K}=(K_1,\ldots,K_k)>\mathsf{0}$ to be fixed later. The family $\chi_t$ is bounded in $\sym{0}$ and converges to the constant $1$ as $t\to 0^+$ in any positive-order symbol space.

We set $\check{a}_t=\chi_t\check{a}$ and $a_t=\check{a}_t^2$. Then
\begin{equation}
\label{eq:pos-commutator-symbolic-reg}
(H_p + p_1 +\delta\rho^{\mathsf{-m+1}})a_t
=
- b_t^2 + e_t - f_t^2,
\end{equation}
where $e_t=\chi_t^2 e\in \sym{2r-2K}$, $f_t=\chi_t f\in \sym{r-K}$, and
$$
b_t = \rho^{\mathsf{-r}}\phi(\hat{z}) \psi(q)
\sqrt{\left|
2h_q\cdot \chi'(q)\chi(q)
+\left(\sum_{i=1}^k h_{\rho_i} \left(-2r_i+m_i-1+\frac{2K_it}{t+\rho_i}\right) + \tilde{p} + \delta \right) \chi(q)^2
\right|
}
$$
is in $\sym{r-K}$, as long as inclusion of the $\frac{2K_it}{t+\rho_i}$ terms does not change the sign of the expression under the absolute value. Since $\frac{2K_it}{t+\rho_i}\leqslant 2K_i$, for any fixed $\mathsf{K}$ this will indeed be true if we take $\varepsilon$ small enough and $\digamma$ large enough so that $|\chi'(q)|$ is sufficiently larger than $\chi(q)$.

\item \textbf{Quantization and estimates.}
We define operators $\check{A}_t=\mathrm{Op}(\check{a}_t)$, $A_t=\check{A}_t^*\check{A}_t$, $B_t=\mathrm{Op}(b_t)$, $E_t=\mathrm{Op}(e_t)$, $F_t=\mathrm{Op}(f_t)$, $\Lambda=\mathrm{Op}(\rho^{\mathsf{-\frac{m-1}{2}}})$. The symbolic relation Eq.~(\ref{eq:pos-commutator-symbolic-reg}) means that
\begin{equation}
i[P,A_t] 
-i(P-P^*)A_t
+ \delta (\Lambda\check{A}_t)^*(\Lambda\check{A}_t)
= - B_t^*B_t + E_t - F_t^* F_t + R_t
\end{equation}
for some $R_t\in \ps[2r-2K-1]$, where $\wfs(\{R_t\})\subset \overline{W_{\varepsilon}}$ and the explicit quantization map ensures that the family $R_t$ is bounded in $\ps[2r-1]$ (see the end of Section~\ref{sec:psiDO} and the proof of \cite[Lemma 8.16]{Hintz-notes}). 

For any $u\in\sch'$ and $\tau \geqslant 0$, if we fix any $\mathsf{r}$ and any $\mathsf{K}$ large enough, we can calculate
$$
\|B_tu\|^2
=
- \left\langle \Big(i[P,A_t]-i(P-P^*)\Big)u, u\right\rangle
-\delta \|\Lambda\check{A}_t u\|^2
+\langle E_tu,u\rangle
-\|F_t u\|^2
+\langle R_tu,u\rangle
\leqslant
$$
$$
\leqslant
 2 |\langle \check{A}_tu, \check{A}_t (P-i\tau)u\rangle|
- \delta \|\Lambda\check{A}_t u\|^2
+ \langle E_tu,u\rangle
+\langle R_tu,u\rangle.
$$
The pairings and integrations by parts are justified by Lemma~\ref{thm:int-by-parts} due to the choice of large enough $\mathsf{K}$, and we used Lemma~\ref{thm:comm-lemma} for the commutator term. There exist $C,C',C''>0$ such that
$$
2 |\langle \check{A}_tu, \check{A}_t (P-i\tau)u\rangle|
- \delta \|\Lambda\check{A}_t u\|^2
\leqslant
2C\|\check{A}_tu\|_{\mathsf{\frac{m-1}{2}}} 
\|\check{A}_t(P-i\tau)u\|_{\mathsf{-\frac{m-1}{2}}}
- \delta\|\Lambda\check{A}_t u\|^2
\leqslant
$$
$$
\leqslant
C'\delta\|\check{A}_tu\|^2_{\mathsf{\frac{m-1}{2}}} 
+
\frac{C^2}{C'\delta}\|\check{A}_t(P-i\tau)u\|^2_{\mathsf{-\frac{m-1}{2}}}
- \delta\|\Lambda\check{A}_t u\|^2
\leqslant
 \frac{C^2}{C'\delta} \|\check{A}_t(P-i\tau)u\|_{\mathsf{-\frac{m-1}{2}}}^2
 + C''\delta \|u\|_{\mathsf{-N}}^2,
$$
where the first step used the continuity of the pairing between dual Sobolev spaces, the second step is justified for any $C'>0$ by the AM-GM inequality, and the third step uses the fact that $\Lambda$ is globally elliptic, so there exists $C'>0$ such that $C'\|\check{A}_tu\|^2_{\mathsf{\frac{m-1}{2}}}\leqslant \|\Lambda \check{A}_tu\|^2 + \|\check{A}_t u\|_{\mathsf{-N'}}^2$ for any $u$. Thus,
$$
\|B_t u\|^2
\leqslant
 C \Big(
  \|\check{A}_t(P-i\tau)u\|_{\mathsf{-\frac{m-1}{2}}}^2
+ \|u\|_{\mathsf{-N}}^2
\Big)
+ |\langle E_tu,u\rangle|
+|\langle R_tu,u\rangle|.
$$
We bound the remaining terms on the right-hand side uniformly in $t$ using Lemmas~\ref{thm:uniform-lemma} and~\ref{thm:int-by-parts}, fixing $\mathsf{r}\leqslant\mathsf{s}$ and any $Q',Q'',Q'''_{\mathsf{r}}\in\ps[0]$ with $Q'$ elliptic on $U$, $Q''$ elliptic on $V$, and $Q'''_{\mathsf{r}}$ elliptic on $\overline{W_{\varepsilon}}$:
\begin{itemize}
\item The family $R_t$ is bounded in $\ps[2r-1]$ and $\wfs(\{R_t\})\subset \overline{W_{\varepsilon}}$, so
$$
|\langle R_tu,u\rangle| \leqslant C \Big( \|Q'''_{\mathsf{r}} u\|_{\mathsf{r-\frac{1}{2}}}^2 + \|u\|_{\mathsf{-N}}^2 \Big).
$$

\item The family $E_t$ is bounded in $\ps[2r]$, hence also in $\ps[2s]$, and $\wfs(\{E_t\})\subset V$, so
$$
|\langle E_t u,u\rangle| \leqslant C \Big( \|Q''u\|_{\mathsf{s}}^2 + \|u\|_{\mathsf{-N}}^2 \Big).
$$

\item The family $\check{A}_t$ is bounded in $\ps[r-\frac{m-1}{2}]$, hence also in $\ps[s-\frac{m-1}{2}]$, and $\wfs(\{\check{A}_t\})\subset U$, so
$$
\|\check{A}_t(P- i\tau)u\|_{\mathsf{-\frac{m-1}{2}}}^2 
\leqslant 
C'\Big( \|Q'(P- i\tau) u\|_{\mathsf{s-m+1}}^2 + \|(P- i\tau)u\|_{\mathsf{-N'}}^2 \Big)
\leqslant 
$$
$$
\leqslant
C\Big( \|Q'(P- i\tau) u\|_{\mathsf{s-m+1}}^2 + (1+|\tau|^2)\|u\|_{\mathsf{-N}}^2 \Big).
$$
\end{itemize}

\item \textbf{Regularity conclusions and iteration.}
From the preceding calculations we conclude that, if $Q'((P-i\tau)u)\in\sob{s-m+1}$, $Q''u\in\sob{s}$, and $u\in\sob{-N}$ as assumed and if in addition $Q'''_{\mathsf{r}} u\in\sob{r-\frac{1}{2}}$, then the family $B_tu$ for $t\in (0,1)$ is bounded in $\LL$. Since $\lim_{t\to 0^+}B_t= B$ in $\ps[r']$ for any $\mathsf{r'}>\mathsf{r}$, we have $B_t u\to Bu$ in $\sch'$ and therefore $Bu\in \LL$ by the Banach-Alaoglu theorem, with the estimate
$$
\|Bu\|^2 \leqslant C\Big( \|Q'(P - i\tau)u\|^2_{\mathsf{s-m+1}} + \|Q''u\|^2_{\mathsf{s}} + \|Q'''_{\mathsf{r}} u\|^2_{\mathsf{r-\frac{1}{2}}} + (1+|\tau|^2)\|u\|^2_{\mathsf{-N}} \Big),
$$
where $C>0$ is independent of $\tau\geqslant 0$ and $u$. 

By construction, $B$ is elliptic of order $\mathsf{r}$ on $W_{\varepsilon}$.
Then any $Q_{\mathsf{r}}\in\ps[0]$ with $\wfs(Q)\subset W_{\varepsilon}$ satisfies
\begin{equation}
\|Q_{\mathsf{r}}u\|_{\mathsf{r}}^2 \leqslant C\Big( \|Q'(P - i\tau)u\|^2_{\mathsf{s-m+1}} + \|Q''u\|^2_{\mathsf{s}} + \|Q'''_{\mathsf{r}} u\|^2_{\mathsf{r-\frac{1}{2}}} + (1+|\tau|^2)\|u\|^2_{\mathsf{-N}} \Big).
\label{eq:pos-estimate-iteration}
\end{equation}
To get the estimate Eq.~(\ref{eq:pos-estimate}), we iterate the above argument starting with the $\sob{-N}$ global regularity assumption. Indeed, we can take any $\varepsilon'\in (0,\varepsilon)$ and fix any $\varepsilon_0,\ldots,\varepsilon_n$ so that $\varepsilon'=\varepsilon_n<\ldots<\varepsilon_0=\varepsilon$.
Then we start by taking $\mathsf{r}=\mathsf{s-\frac{n}{2}}$ for $n$ large enough that $\mathsf{s-\frac{n+1}{2}}\leqslant\mathsf{-N}$ and $Q_{\mathsf{r}}'''=I$, so the $Q_{\mathsf{r}}'''u$ term is bounded by $\|u\|_{\mathsf{-N}}^2$ and thus the right-hand side of Eq.~(\ref{eq:pos-estimate-iteration}) matches Eq.~(\ref{eq:pos-estimate}) and is finite. At each step, after establishing Eq.~(\ref{eq:pos-estimate-iteration}) with $\varepsilon=\varepsilon_i$, we choose a particular $Q_{\mathsf{r}}$ which is elliptic on $\overline{W_{\varepsilon_{i+1}}}\subset W_{\varepsilon_i}$ and repeat the argument with $(\mathsf{r+\frac{1}{2}}, \varepsilon_{i+1}, Q_{\mathsf{r}})$ in place of $(\mathsf{r}, \varepsilon_i, Q_{\mathsf{r}}''')$. Since $\mathsf{s},\mathsf{N}$ are fixed, the number of steps until we reach $\mathsf{r}=\mathsf{s}$ and the resulting constant are independent of $u$ and $\tau$, establishing the uniform estimate Eq.~(\ref{eq:pos-estimate}) for $W=W_{\varepsilon'}$.

The qualitative version (part 1) of the theorem follows from the estimate because under the regularity assumptions, there exist $Q',Q''\in\ps[0]$ with $Q'$ elliptic on $\gamma_{\alpha\beta}$ and $Q''$ elliptic at $\alpha$ such that the family $Q'(P-i\tau)u_{\tau}$ is bounded in $\sob{s-m+1}$ and the family $Q''u_{\tau}$ is bounded in $\sob{s}$, so the right-hand side of Eq.~(\ref{eq:pos-estimate}) with $u=u_{\tau}$ is uniformly bounded for $\tau$ in a bounded set. Then taking $U$ and $V$ to be the elliptic sets of $Q',Q''$ respectively and any $Q\in\ps[0]$ elliptic on the corresponding $W$, which contains $\beta$, we conclude that $Qu_{\tau}$ is bounded in $\sob{s}$ and therefore $\beta\notin\wf{s}(\{u_{\tau}\})$.

\end{itemize}

	\end{proof}
	
See Figure~\ref{fig:pos} for an illustration of the various regions defined and used in the proof.
	
\begin{figure}
\begin{center}
\includegraphics{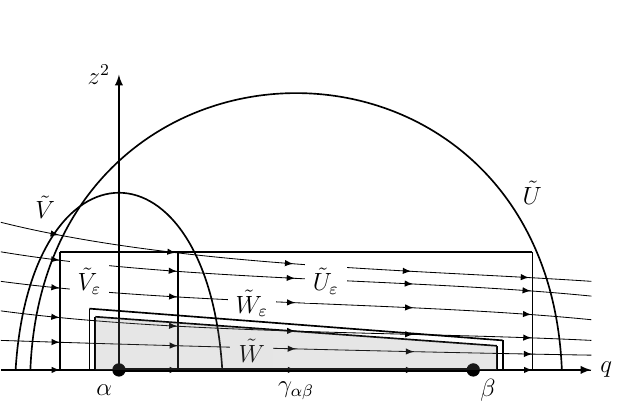}
\end{center}
\caption{Schematic illustration of propagation setup in proof of Theorem~\ref{thm:pos}. The slope of the upper boundary of $\tilde{W}$ ensures that the Hamilton flow only crosses it from the interior of $\tilde{W}$, and it similarly crosses all nearby level sets of the cutoff $\phi(\hat{z})$ in the same direction.}
\label{fig:pos}
\end{figure}

	\subsection{Localized radial point estimate}
	 The rescaled Hamilton vector field of the symbol of $\Box_{\g}+m^2$ has several sets of critical points in the characteristic set, or radial sets (see Section~\ref{sec:Hamiltonian}), where the propagation of singularities theorem is not helpful. One additionally needs to prove radial point estimates which allow one to conclude microlocal regularity at these points and thereby link the propagation results in different corners together into a global regularity theory. The first radial point estimate is due to Guillemin and Schaeffer~\cite{GS} for isolated radial points (allowing for various types of isolated radial points); extended sets of radial points were first systematically treated by Melrose~\cite{Melrose-AES} and localized radial point estimates were first systematically studied in \cite{Haber-Vasy}; in both cases these radial sets were source or sink manifolds. In this section, we prove a general localized version of these estimates which includes Sussman's estimates \cite[Propositions 5.6-5.14, part of Theorems 5-6]{Sussman}. As Sussman points out in \cite[Example 1.6 and Figure 8]{Sussman}, Klein-Gordon Green's functions in Minkowski spacetime provide a concrete example where a sharp description of a distribution's singularities is given by a fully microlocalized description of propagation.
	 
	Like in the previous section, consider $P\in\ps[m]$ with principal symbol $p\in S^{\mathsf{m}}_{\epsilon}(\PP)$. Let $\alpha\in\partial\PP$ be a critical point of $\Hp$; since $\Hp$ is continuous, this means simply $\Hp(\alpha)=0$. Let $\Gamma_1,\ldots,\Gamma_k$ be the boundary hypersurfaces containing $\alpha$. We denote $\Gamma=\bigcap_{i=1}^k\Gamma_i$, so $\alpha\in\Gamma^{\circ}$. As discussed above, $\Hp|_{\Gamma}$ is smooth on $\Gamma^{\circ}$ as a section of ${}^{\mathrm{b}}T\PP|_{\Gamma}$.
	 
Under certain nondegeneracy conditions on the Hamilton flow, the propagation of singularities of solutions to $Pu=f$ into $\alpha$ is governed by the linearization of $\Hp$ at $\alpha$, appropriately understood, as we now discuss.

\subsubsection{Linearization}
\label{sec:linearization-def}
	
	By the linearization of a smooth vector field $V$ on a smooth manifold $M$ at a critical point $\alpha$ we mean the following. Consider the bundle section $V:M\to TM$. The derivative of $V$ at $\alpha$ maps $DV(\alpha):T_{\alpha}M\to T_{(\alpha,0)}TM$. The zero section $\mathcal{Z}$ and the fiber $T_{\alpha}M$ are invariantly defined submanifolds of $TM$ which intersect transversely at $(\alpha,0)$, so there is a canonical splitting
  $$T_{(\alpha,0)}TM=T_{(\alpha,0)}\mathcal{Z}\oplus T_{(\alpha,0)}T_{\alpha}M.$$
  Denote $\pi_{fib}$ the projection onto the second subspace along the first, and let $\varphi: T_{(\alpha,0)}T_{\alpha}M\to T_{\alpha}M$ be the usual identification of a vector space with its tangent space at zero. Then the linearization of $V$ at $\alpha$ is the map
  $$LV(\alpha)=\varphi\circ \pi_{fib}\circ DV(\alpha):T_{\alpha}M\to T_{\alpha}M.$$
We think of this as a linear vector field on $T_{\alpha}M$. 

For any local smooth coordinate chart $(q_1,\ldots,q_m)$ on $M$ near $\alpha$, let $V=\sum_{i=1}^m v_i\frac{\partial}{\partial q_i}$ for some smooth functions $v_i$. Then the linearization is given by
  \begin{equation}
  LV(\alpha)\left(\sum_{i=1}^m a_i\frac{\partial}{\partial q_i}\right)=\sum_{i,j=1}^m a_j\frac{\partial v_i}{\partial q_j}(\alpha) \frac{\partial}{\partial q_i},
  \label{eq:linearization-coord}
  \end{equation}
or, abusing notation to interpret $LV(\alpha)$ as a section of $T\R^m_q$ via the map $a_i\mapsto q_i$,
   \begin{equation}
  LV(\alpha)(q)=\sum_{i,j=1}^m \frac{\partial v_i}{\partial q_j}(\alpha) q_j\frac{\partial}{\partial q_i}.
  \end{equation}

\begin{remark}
There is also a dual perspective on the linearization which we now recall. Let $L'V(\alpha):T_{\alpha}^*M\to T_{\alpha}^*M$ denote the adjoint of $LV(\alpha)$. We can make use of the identification $T^*_{\alpha}M\simeq \mathcal{I}_{\alpha}/\mathcal{I}^2_{\alpha}$, where $\mathcal{I}_{\alpha}=\{f\in C^{\infty}(M)\ |\ f(\alpha)=0\}$, via $[f]\mapsto df$. Then in coordinates as above, we have
$$
L'V(\alpha)([f])\left(\sum_{j=1}^m a_j\frac{\partial}{\partial q_j}\right)
=
\sum_{i=1}^m \frac{\partial f}{\partial q_i}(\alpha)\ dq_i \left(\sum_{j,k=1}^m a_k\frac{\partial v_j}{\partial q_k}(\alpha)\frac{\partial}{\partial q_j}\right)
=
\sum_{i,k=1}^m a_k\frac{\partial f}{\partial q_i}(\alpha)\frac{\partial v_i}{\partial q_k}(\alpha)
=
$$
$$
=
\left(\sum_{i,k=1}^m \frac{\partial f}{\partial q_i}(\alpha)\frac{\partial v_i}{\partial q_k}(\alpha) dq_k\right)\left(\sum_{j=1}^m a_j\frac{\partial}{\partial q_j}\right)
=
\left(\sum_{k=1}^m 
\frac{\partial (Vf)}{\partial q_k}(\alpha)\right)
\left(\sum_{j=1}^m a_j\frac{\partial}{\partial q_j}\right)
=
[Vf]\left(\sum_{j=1}^m a_j\frac{\partial}{\partial q_j}\right),
$$
where we used the fact that $v_i(\alpha)=0$. Thus, $L'V(\alpha)[f]=[Vf]$.

Note that, for coordinates as above, the vector $\frac{\partial}{\partial q_i}$ is an eigenvector of $LV(\alpha)$ with eigenvalue $\lambda$ if and only if $\frac{\partial v_j}{\partial q_i}(\alpha)=\lambda\delta_{ij}$ for all $j$. On the dual side, the one-form $dq_i$ is an eigenvector of $L'V(\alpha)$ with eigenvalue $\lambda$ if and only if $\frac{\partial v_i}{\partial q_j}(\alpha)=\lambda\delta_{ij}$ for all $j$. Then we see that the set $\frac{\partial}{\partial q_1},\ldots,\frac{\partial}{\partial q_m}$ is an eigenbasis for $LV(\alpha)$ if and only if the set $dq_1,\ldots,dq_m$ is an eigenbasis for $L'V(\alpha)$, and the eigenvalues match. The following results can be stated and proven using either of these pictures; we will use the tangent-space version.
\end{remark}

\begin{remark}
The fact that we are dealing with non-classical symbols and correspondingly $\Hp$ is not smooth on $\PP$ forces us to think about its linearization in a somewhat awkward way, treating the dynamics within any corner $\Gamma$ (where $\Hp|_{\Gamma}$ \textit{is} smooth) and transverse to $\Gamma$ differently: we only ever formally define the linearization of $\Hp|_{\Gamma}$ rather than $\Hp$. One way to sidestep this would be to modify the smooth structure on $\PP$ by adjoining $\sqrt[N]{\rho_i}$ for every defining function $\rho_i$ of the faces forming $\Gamma$ and large enough $N$, which would ensure that symbols in $\soe(\PP)$ have some finite number of continuous derivatives. Since the definitions of linearization above make sense as long as the vector field is $C^1$, this would allow us to talk about the linearization of $\Hp$ without restricting to $\Gamma$, though it would only act on the tangent space at $\alpha$ defined using the modified smooth structure. 

While approaches using finite differentiability can be convenient, essentially because they allow one to use Taylor's theorem, it is difficult to integrate them seamlessly into symbolic constructions, where it is \textit{conormal} regularity which is relevant and which therefore needs to be checked for any terms arising from Taylor expansions. See the proof of Proposition~\ref{thm:Hp-coord}, where we use the change of smooth structure described above.

Another approach is to consider the ``smooth part" of $\Hp$ near $\Gamma$, that is any smooth extension of $\Hp|_{\Gamma}$ to $\PP$, similarly to the proof of Theorem~\ref{thm:pos}; one can check that the linearization of any such extension will have the same eigenvalues, but the eigenvectors transverse to $\Gamma$ will depend on the choice of extension and carry no particular meaning, since the bicharacteristics limiting to $\Gamma$ need not have well-defined tangent directions as they approach the corner. Below, we instead define the ``eigenvalues" for directions transverse to $\Gamma$ as the corresponding coefficients of $\Hp(\alpha)$ as a b-vector, which is an invariant description of the same objects.
\end{remark}

Let $N\subset M$ be a smooth submanifold containing $\alpha$ to which $V$ is tangent. We can always choose coordinates $q_1,\ldots,q_m$ in a neighborhood of $\alpha$ such that in this neighborhood $N$ is defined by $q_1=\ldots=q_k=0$. Then in these coordinates,
$$V=\sum_{i,j=1}^k a_{ij}q_j\frac{\partial}{\partial q_i}+\sum_{i=k+1}^m b_i \frac{\partial}{\partial q_i}$$
for some smooth $a_{ij},b_i$, while $T_{\alpha}N$ is the span of $\frac{\partial}{\partial q_{k+1}},\ldots,\frac{\partial}{\partial q_m}$. Then direct computation using Eq.~(\ref{eq:linearization-coord}) shows that $LV(\alpha)(T_{\alpha}N)\subset T_{\alpha}N$. Thus, the tangent space to an invariant manifold of $V$ containing $\alpha$ is itself invariant under the action of the linearization. In particular, this allows us to consider the restriction of $L(\Hp|_{\Gamma})(\alpha)$ to $T_{\alpha}(\Sigma\cap\Gamma)$.

\subsubsection{Statement of theorem and idea of proof}
\label{sec:rp-statement}
Recall that whether a point $\alpha\in\partial\PP$ is a critical point of $\Hp$ is independent of the choice of boundary-defining functions for rescaling of $H_p$ into $\Hp$.

\begin{theorem}
\label{thm:localized-rp-main}
Let $P\in\ps[m]$ with real-valued principal symbol $p\in S^{\mathsf{m}}_{\epsilon}(\PP)$. We denote by $p_1$ the principal symbol of $-i(P-P^*)\in\ps[m-1]$.

Let $\alpha\in\Sigma$ be a critical point of $\Hp$. Let $\Gamma_1,\ldots,\Gamma_k$ be the boundary hypersurfaces of $\mathcal{P}$ containing $\alpha$; we denote $\Gamma=\bigcap_{i=1}^k\Gamma_i$. Let $\rho_1,\ldots,\rho_k$ be any defining functions of $\Gamma_1,\ldots,\Gamma_k$. We use multi-index notation with $\rho=(\rho_1,\ldots,\rho_k)$ and write $\tilde{p}=\rho^{\mathsf{m}}p$, $\tilde{p}_1=\rho^{\mathsf{m-1}}p_1$, and $\Hp=\rho^{\mathsf{m-1}}H_p$. Let
 $$
\hat{p}_1 = \limsup_{\alpha'\to\alpha} \tilde{p}_1(\alpha'),
\hspace{30pt}
\check{p}_1 = \liminf_{\alpha'\to\alpha} \tilde{p}_1(\alpha').
 $$

For $i=1,\ldots,k$, let $\lambda_i=(\rho_i^{-1}\ d\rho_i)(\Hp(\alpha))$ (understood as a b-1-form and a b-vector respectively). We order the faces so that there exists some $k_{\parallel}\leqslant k$ such that $\lambda_i=0$ if and only if $i>k_{\parallel}$.

 We assume that
\begin{itemize}
	\item There exist symbols $\check{q}_1,\ldots,\check{q}_l\in \soe(\PP)$ with $d(\tilde{p}|_{\Gamma})(\alpha),d(\check{q}_1|_{\Gamma})(\alpha),\ldots,d(\check{q}_l)|_{\Gamma})(\alpha)$ linearly independent such that in some neighborhood of $\alpha$ in $\Sigma$, the critical set of $\Hp$ is 
	$$R=\{\tilde{p}=0,\ \rho_1=\ldots=\rho_{k_{\parallel}}=0,\ \check{q}_1=\ldots=\check{q}_l=0\}.$$
	\item The restriction of $L(\Hp|_{\Gamma})(\alpha)$ to $T_{\alpha}(\Sigma\cap\Gamma)$ has an eigenbasis with real eigenvalues. The eigenspace of zero is $T_{\alpha}(R\cap\Gamma)$.
	\end{itemize}	
For $\tau$ in a bounded subset of $\R$, either all non-negative or all non-positive, consider a family $u_{\tau}$ bounded in $\sob{-N}$ such that $\alpha\notin \wf{s-m+1}(\{(P-i\tau)u_{\tau}\})$ for some sets of orders $\mathsf{s}$, $\mathsf{N}$ and $\alpha\notin \wf{s'}(\{u_{\tau}\})$ for some $\mathsf{s'}\leqslant\mathsf{s}$.
	\begin{enumerate}
	\item Assume $\tau\geqslant 0$, $\sum_{i=1}^k (-2s_i+m_i-1)\lambda_i + \hat{p}_1 < 0$, and $\sum_{i=1}^k (-2s'_i+m_i-1)\lambda_i + \hat{p}_1 < 0$. Then if $\alpha$ has a punctured neighborhood in which $\wf{s}(\{u_{\tau}\})$ is disjoint from every bicharacteristic limiting to $\alpha$ in the forward direction, then $\alpha\notin \wf{s}(\{u_{\tau}\})$.
	\item Assume $\tau\leqslant 0$, $\sum_{i=1}^k (-2s_i+m_i-1)\lambda_i  + \check{p}_1 > 0$ and $\sum_{i=1}^k (-2s'_i+m_i-1)\lambda_i  + \check{p}_1 > 0$. Then if $\alpha$ has a punctured neighborhood in which $\wf{s}(\{u_{\tau}\})$ is disjoint from every bicharacteristic limiting to $\alpha$ in the backward direction, then $\alpha\notin \wf{s}(\{u_{\tau}\})$.
	\end{enumerate}
\end{theorem}

\begin{remark}
If $\lambda_1,\ldots,\lambda_{k_{\parallel}}$ are not all of the same sign (i.e. $\alpha$ is a saddle point for the flow transverse to $\Gamma$), the assumption that $\alpha\notin\wf{s'}(\{u_{\tau}\})$ for some $\mathsf{s'}$ satisfying the inequality in 1 or 2 is automatically satisfied for any family $u_{\tau}$ bounded in some Sobolev space, since there are solutions of either inequality with all orders arbitrarily low. If $\lambda_1,\ldots,\lambda_{k_{\parallel}}$ are all of the same sign (i.e. $\alpha$ is a source or sink for the flow transverse to $\Gamma$), this is true for one inequality (bounding $\sum_{i=1}^k |\lambda_i| s_i$ from above) but not the other (bounding it from below) -- these are the below- and above-threshold estimates respectively. On the other hand, in the above-threshold case the assumption on bicharacteristics limiting to $\alpha$ is often vacuously satisfied since the signs of the $\lambda_i$ rule out the existence of bicharacteristics limiting to $\alpha$ in the relevant direction which are not contained in $\Gamma$.
\end{remark}

\begin{remark}
The condition that the linearization have an eigenbasis with real eigenvalues is satisfied for radial points in many cases of interest (including all radial points which will appear in this paper). It is also satisfied in the general setting previously considered by Haber and Vasy \cite{Haber-Vasy}, who proved a localized radial point estimate at source/sink Lagrangian submanifolds of radial points.
\end{remark}

\begin{remark}
In asymptotically Minkowski spacetimes, including small perturbations of Minkowski spacetime as solutions to the Einstein equations, every radial set is in fact a product-type or neat smooth submanifold of some corner of $\PP$, so the symbols $\check{q}_i$ can be taken to be smooth, which simplifies the proofs. The rescaled symbol $\tilde{p}$, however, is generally not smooth, though the non-smooth error terms in it vanish at all base infinity faces.
\end{remark}

The proof is again based on a positive-commutator argument, where the commutant $a$ now needs to be supported near the point $\alpha$ and $H_pa$ needs to have a definite sign in its neighborhood. If $\alpha$ is a source or sink within $\Sigma$, this is accomplished simply by taking $a$ to be a bump function centered at $\alpha$ in all variables, multiplied by a product of powers of boundary-defining functions. The required growth order of $a$ is fixed by the amount of regularity we wish to conclude, and the inequalities on the orders in radial point estimates arise from requiring compatibility of this order with the monotonicity requirement. Specifically, if we wish to conclude that $u$ has any level of regularity at $\alpha$ above a certain threshold, we can infer it directly from knowledge of the corresponding regularity of $Pu$ as long as we know a priori that $u$ possesses regularity at least slightly above the threshold; if $Pu$ only possesses below-threshold regularity, however, to infer the corresponding regularity at $\alpha$ for $u$ we need to assume regularity in a punctured neighborhood of $\alpha$.
		
	This construction can be generalized in two directions. First, if $\alpha$ is instead a nondegenerate saddle point of the flow, then the same construction works provided we assume regularity of $u$ at the stable/unstable manifold of $\alpha$, depending on the direction of regularity propagation. Second, if $\alpha$ belongs to a larger $\Hp$-invariant submanifold $R\subset\Sigma$ which, as a whole, is a source/sink for the flow near $R$ but may support a nontrivial flow within itself, we can modify the construction to conclude regularity at all of $R$ simultaneously, essentially by separating variables tangent and transverse to $R$ and using a bump function in the transverse variables only. For the below-threshold case, a regularity assumption on all bicharacteristics limiting to $R$ is necessary. Finally, if the flow in the directions transverse to $R$ is instead of hyperbolic type with $R$ having well-defined stable and unstable manifolds (locally in the transverse variables, but globally in $R$), the argument also goes through if we add an assumption on regularity of $u$ on the stable/unstable manifold, depending on the direction of propagation.
	
	When $\alpha$ belongs to an invariant manifold $R$, in general we cannot expect to prove a localized result on regularity at $\alpha$ assuming only regularity \textit{near $\alpha$ but away from $R$}, because singularities can propagate into $\alpha$ along $R$. An important exception is the case where $R$ is known to be a critical set, so the directions along $R$ are neutral for the flow. Microlocal regularity at a point is an open condition, proven using estimates in a finite neighborhood of $\alpha$; since bicharacteristics arbitrarily close to a critical set can flow in any direction, in general this prevents us from ruling out propagation of singularities in neutral directions. However, additional nondegeneracy assumptions on $\Hp$ near $\alpha$ restrict the relative speeds of the flow in the tangential and transverse directions, allowing one to prove regularity at $\alpha$ without any assumptions on regularity in $R$. In \cite{Haber-Vasy}, this argument is carried out for $R$ which is a source/sink Lagrangian manifold of radial points. Here we formulate assumptions more directly in terms of the linearized dynamics near $\alpha$ and allow for saddle-type flow in the transverse directions.
	
	Compared to the estimates global in $R$, we need to modify $a$ by a factor localizing in the neutral variables. The technical challenge is to construct this factor so that all nearby bicharacteristics cross its level sets in the same direction. If the flow was truly neutral in those variables, a simple bump function in them would suffice; instead, under our assumptions, bicharacteristics which hit $R$ any finite distance away from $\alpha$ may do so at a finite (but bounded) angle relative to the ``normal" (in adapted coordinates) and furthermore curve at a finite (but bounded) rate as soon as they leave $R$. We therefore construct a localizer by modifying a bump function of the neutral coordinates by adding dependence on the transverse coordinates $z$: first an $\mathcal{O}(|z|)$ correction to ensure that at $R$ all bicharacteristics are exactly tangent to the level sets, then an $\mathcal{O}(|z|^2)$ correction to ensure that the level sets curve strongly enough away from $R$ to keep all nearby bicharacteristics crossing them in the same direction. See Figure~\ref{fig:localization}. As in the proof of propagation of singularities, we use small powers of boundary-defining functions in the localizers in order to control the large errors which arise unless the error terms in $\tilde{p}$ decay at \textit{every} face containing $\alpha$.
	
	\begin{figure}
	\begin{center}
	\includegraphics[width=0.47\textwidth]{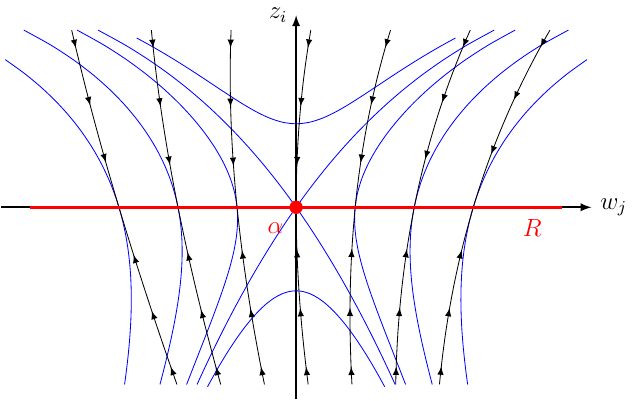}
	\hspace{10pt}
	\includegraphics[width=0.47\textwidth]{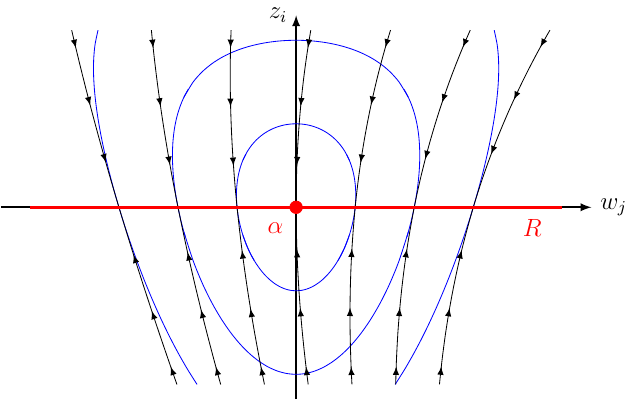}
	\caption{Schematic illustration of Hamilton flow and commutant construction near a radial point $\alpha$ belonging to an extended radial set $R$. In this two-dimensional slice, $z_i$ represents a defining function of $R$ (in this case corresponding to a negative eigenvalue) and $w_j$ a coordinate along $R$. The commutant is localized to a neighborhood of $\alpha$ using a cutoff along $R$ which has level sets like the blue curves in the figure, so all bicharacteristics near $\alpha$ cross them in the same direction, the most relevant region being a neighborhood of an annular subset of $R$ around $\alpha$ (the ``cutoff region"). In the figure on the left, the cutoff is constructed so that bicharacteristics do not enter the cutoff region from the sides; in the figure on the right, it is constructed so they do not exit through the sides. See Figure~\ref{fig:rad-pts-prop} for more details.
	\label{fig:localization}}
	\end{center}
	\end{figure}

	Finally, we note that only the flow within $\Sigma$ is relevant. To emphasize this, we do not make any assumptions on $L\Hp$ away from $T\Sigma$. This means that nearby integral curves off $\Sigma$ are not controlled in the way described above; however, any $\tilde{p}$-dependent terms in $\Hp$ are necessarily $\mathcal{O}(\tilde{p})$ and their contributions to the estimates are directly controlled by regularity assumptions on $Pu$. When proving estimates global in $R$, nontrivially $\tilde{p}$-dependent terms in $\Hp a$ arise only from localizing in $\tilde{p}$ and are actually supported away from $\Sigma$, so they are often treated using elliptic estimates. However, in our case localization along $R$ contributes additional $\tilde{p}$-dependent terms which may only vanish linearly at $\Sigma$. For simplicity and to emphasize the more general situation, we combine these terms into one.
	
	Beyond the symbolic construction just outlined, the quantization and estimates (as well as further technical modifications necessary to regularize the argument and to treat operators with nonzero skew-adjoint part and/or a sign-definite imaginary part) are standard. While we formulate the result in the de,sc-setting, we emphasize that the argument will apply in any pseudodifferential calculus which is symbolic at every boundary face containing the radial point of interest.

\subsubsection{Proof of Theorem~\ref{thm:localized-rp-main}}
We first reduce the problem to a model form by choosing coordinates adapted to the flow.

\begin{lemma}
\label{thm:coordinates}
Assume the hypotheses of Theorem~\ref{thm:localized-rp-main}. Let $\tilde{p}_0$ be any smooth function extending $\tilde{p}|_{\Gamma}$. Then there exist smooth functions $q_1,\ldots,q_l$, $y_1,\ldots,y_m$ such that $x=(\tilde{p}_0,\rho_1,\ldots,\rho_k,q_1,\ldots,q_l,y_1,\ldots,y_m)$ is a local coordinate chart on $\PP$ near $\alpha$, $x(\alpha)=0$, and
	\begin{itemize}
	\item the partial derivatives with respect to the $q$ and $y$ coordinates at $\alpha$ are eigenvectors of $L(\Hp|_{\Gamma})(\alpha)$;
	\item there exist extensions $\tilde{q}_1,\ldots,\tilde{q}_l\in \soe(\PP)$ of $q_1|_{\Gamma},\ldots,q_l|_{\Gamma}$ such that in a neighborhood of $\alpha$ in $\PP$, we have $R=\{\tilde{p}=0,\rho_1=\ldots=\rho_{k_{\parallel}}=0,\tilde{q}_1=\ldots=\tilde{q}_l=0\}$.
	\end{itemize}
      \end{lemma}

      \begin{remark} If the $\check{q}_i$ are all smooth, as for instance in the asymptotically Minkowski setting, then the $\tilde{q}_i$ are such as well, as follows from the proof below. However, $\tilde p$ need not be smooth even in this case. \end{remark}
	
	\begin{proof}
	As one of the hypotheses of Theorem~\ref{thm:localized-rp-main}, there exist $\check{q}_1,\ldots,\check{q}_l\in \soe(\PP)$ with $d(\tilde{p}|_{\Gamma})(\alpha)$, $d(\check{q}_1|_{\Gamma})(\alpha)$,..., $d(\check{q}_l|_{\Gamma})(\alpha)$ linearly independent such that in some neighborhood of $\alpha$ in $\Sigma$, the critical set of $\Hp$ is 
	$$R=\{\tilde{p}=0,\ \rho_1=\ldots=\rho_{k_{\parallel}}=0,\ \check{q}_1=\ldots=\check{q}_l=0\}.$$
	We start by choosing any smooth extensions $\hat{q}_1,\ldots,\hat{q}_l$ of $\check{q}_1|_{\Gamma},\ldots,\check{q}_l|_{\Gamma}$ and any smooth $\hat{y}_1,\ldots,\hat{y}_l$ extending a local coordinate chart on $R\cap\Gamma$ (which is a smooth submanifold of $\Sigma\cap\Gamma$) such that $\hat{y}(\alpha)=0$. Then $(\tilde{p}_0,\rho,\hat{q},\hat{y})$ is a local coordinate chart near $\alpha$ on $\PP$.
	
	$T_{\alpha}(R\cap\Gamma)$ is spanned by the $\frac{\partial}{\partial\hat{y}_i}$, which are automatically eigenvectors with eigenvalue zero. Let us extend this to an eigenbasis $\left(v_1,\ldots,v_l,\frac{\partial}{\partial \hat{y}_1},\ldots, \frac{\partial}{\partial \hat{y}_m}\right)$ of $T_{\alpha}(\Sigma\cap\Gamma)$, which is spanned by $\left(\frac{\partial}{\partial \hat{q}_1},\ldots,\frac{\partial}{\partial\hat{q}_l},\frac{\partial}{\partial\hat{y}_1},\ldots,\frac{\partial}{\partial\hat{y}_m}\right)$. The linear transformation $A$ on $T_{\alpha}(\Sigma\cap \Gamma)$ defined by $A\left(\frac{\partial}{\partial \hat{q}_i}\right)=v_i$ and $A\left(\frac{\partial}{\partial \hat{y}_i}\right)=\frac{\partial}{\partial \hat{y}_i}$ is invertible. In the coordinate basis, its matrix $[A]$ and its inverse have the form
	$$
	[A]
	=
	\begin{pmatrix}
	B & 0_{l\times m} \\
	C & I_{m\times m}
	\end{pmatrix},
	\hspace{30pt}
	[A]^{-1}
	=
	\begin{pmatrix}
	B^{-1} & 0_{l\times m} \\
	-CB^{-1} & I_{m\times m}
	\end{pmatrix}.
	$$
Then we define $(q,y)=[A]^{-1}(\hat{q},\hat{y})$, so
$$
\hat{q} = Bq,
\hspace{30pt}
\hat{y} = Cq + y.
$$
 Then in the new coordinates $(\tilde{p}_0,\rho,q,y)$, we have $\frac{\partial}{\partial q_i}=v_i$ and $\frac{\partial}{\partial y_i}=\frac{\partial}{\partial\hat{y}_i}$, so the first requirement of the lemma is satisfied.

To check the second requirement, note that since $\check{q}_i-\hat{q}_i$ vanishes on $\Gamma$, there exist $q_1',\ldots,q_l'\in \mathcal{I}_{\Gamma}^{\epsilon}(\PP)$ such that
$$
\check{q}
 = 
 \hat{q} + q'
  =
   Bq + q'
   =B(q+B^{-1}q'),$$
so $\check{q}=0$ if and only if $ q+B^{-1}q'=0$. Therefore we can take $\tilde{q}=q+B^{-1}q'$, which is in $\soe(\PP)$, to be the symbol defining $R$.
	\end{proof}

One would like to use $(\tilde{p},\rho,\tilde{q},y)$ as coordinates to analyze the Hamilton flow because the characteristic set and radial set are simply described in terms of these symbols in a full neighborhood of $\alpha$ in $\PP$, not just in $\Gamma$. Unfortunately, unlike $(\tilde{p}_0,\rho,q,y)$, these symbols need not extend smoothly to the boundary, so it is not immediately clear that this approach is viable. The following proposition shows that it is, using the smooth coordinates $(\tilde{p}_0,\rho,q,y)$ as an auxiliary tool.

\begin{prop}
\label{thm:Hp-coord}
Assume the hypotheses of Theorem~\ref{thm:localized-rp-main}, and consider coordinates $(\tilde{p}_0,\rho,q,y)$ and symbols $\tilde{q}$ as in Lemma~\ref{thm:coordinates}. Let $\lambda_i = (\rho_i^{-1}\ d\rho_i)(\Hp(\alpha))$ as before, and let $\mu_i$ denote the eigenvalue of $\frac{\partial}{\partial q_i}$ in the context of Lemma~\ref{thm:coordinates}. Let $\tilde{\rho}_i=\sqrt[N]{\rho_i}$ for some natural number $N>\frac{2}{\epsilon}$. Let $z=(\tilde{\rho}_1,\ldots,\tilde{\rho}_{k_{\parallel}},\tilde{q}_1,\ldots,\tilde{q}_l)$ and $w=(\tilde{\rho}_{k_{\parallel}+1},\ldots,\tilde{\rho}_k,y_1,\ldots,y_m)$. 

There exists a neighborhood $U$ of $\alpha$ in $\PP$ such that $(\tilde{p},\rho,\tilde{q},y)$ is a smooth coordinate chart on $U\backslash \partial\PP$. The rescaled Hamilton vector field on $U\backslash \partial\PP$ has the form
\begin{equation}
\begin{split}
\Hp
=
\tilde{p}V
+\sum_{i=1}^{k_{\parallel}} \left(
h_{\rho_i}^{(R)}
+ \sum_{j=1}^{k_{\parallel}+l} h_{\rho_i}^{(z_j)}z_j
\right)\rho_i\frac{\partial}{\partial\rho_i}
+
\sum_{i=k_{\parallel}+1}^k \left( 
\sum_{j=1}^l h_{\rho_i}^{(\tilde{q}_j)}\tilde{q}_j
+ \sum_{j,r=1}^{k_{\parallel}+l} h_{\rho_i}^{(z_j,z_r)}z_jz_r
 \right) \rho_i\frac{\partial}{\partial\rho_i}
+
\\
+ \sum_{i=1}^l \left( 
\sum_{j=1}^l h_{\tilde{q}_i}^{(\tilde{q}_j)}\tilde{q}_j
+ \sum_{j,r=1}^{k_{\parallel}+l} h_{\tilde{q}_i}^{(z_j,z_r)}z_jz_r
 \right) \frac{\partial}{\partial \tilde{q}_i}
+
\sum_{i=1}^m \left( 
\sum_{j=1}^l h_{y_i}^{(\tilde{q}_j)}\tilde{q}_j
+ \sum_{j,r=1}^{k_{\parallel}+l} h_{y_i}^{(z_j,z_r)}z_jz_r
 \right) \frac{\partial}{\partial y_i},
\end{split}
\label{eq:Hp-coord}
\end{equation}
where $V\in S^{\mathsf{0}}\Vb(\PP)$ and
\begin{itemize}
\item  the leading terms $h_{\rho_i}^{(R)}$, $h_{\rho_i}^{(\tilde{q}_j)}$, $h_{\tilde{q}_i}^{(\tilde{q}_j)}$, $h_{y_i}^{(\tilde{q}_j)}$ are functions only of $w$, i.e. they are independent of $\tilde{p}$ and $z$, and belong to $\soe(\PP)$; 
\item the error terms $h_{\rho_i}^{(z_j)}$, $h_{\rho_i}^{(z_j,z_r)}$, $h_{\tilde{q}_i}^{(z_j,z_r)}$, $h_{y_i}^{(z_j,z_r)}$ are functions of $z$ and $w$ only, i.e. they are independent of $\tilde{p}$, and belong to $\sym{0}$.
\end{itemize}
In addition,
\begin{equation}
h_{\rho_i}^{(R)}(0)=\lambda_i,
\hspace{45pt}
h_{\tilde{q}_i}^{(\tilde{q}_j)}(0)=\mu_i\delta_{ij},
\hspace{45pt}
h_{y_i}^{(\tilde{q}_j)}(0)=0.
\label{eq:Hp-leading-terms}
\end{equation}

\end{prop}

\begin{proof}
We have
$$
\tilde{p} = \tilde{p}_0 + R_p,
\hspace{30pt}
\tilde{q}_i = q_i + R_i,
\hspace{30pt}
R_p,R_i\in \mathcal{I}_{\Gamma}^{\epsilon}(\PP).
$$
Symbols in $\mathcal{I}_{\Gamma}^{\epsilon}(\PP)$ remain in $\mathcal{I}_{\Gamma}^{\epsilon}(\PP)$ under the action of b-derivatives. Therefore, in $(\tilde{p}_0,\rho,q,y)$ coordinates,
$$
\frac{\partial \tilde{p}}{\partial \tilde{p}_0} = 1 \mod \mathcal{I}_{\Gamma}^{\epsilon}(\PP),
\hspace{30pt}
\frac{\partial \tilde{q}_i}{\partial q_j} = \delta_{ij} \mod \mathcal{I}_{\Gamma}^{\epsilon}(\PP),
\hspace{30pt} 
\frac{\partial\tilde{p}}{\partial q_j},\frac{\partial\tilde{q}_i}{\partial \tilde{p}_0} \in \mathcal{I}_{\Gamma}^{\epsilon}(\PP).
$$
Thus, there exists a neighborhood $U$ of $\alpha$ in $\PP$ on which the matrix of partial derivatives of $\tilde{p},\tilde{q}$ at fixed $\rho,y$ is invertible. Then since $\tilde{p}$ and $\tilde{q}$ are smooth away from the boundary, by the implicit function theorem, $(\tilde{p},\rho,\tilde{q},y)$ is a smooth coordinate chart on $U\backslash\partial\PP$.

Considering the full coordinate transformation, we find
$$
\frac{\partial'}{\partial'\tilde{p}_0}
 = 
 \left(1+\frac{\partial' R_p}{\partial' \tilde{p}_0}\right)\frac{\partial}{\partial\tilde{p}}
 +\sum_{j=1}^l \frac{\partial' R_j}{\partial'\tilde{p}_0}\frac{\partial}{\partial\tilde{q}_j},
\hspace{50pt}
\rho_i\frac{\partial'}{\partial' \rho_i}
=
\rho_i\frac{\partial' R_p}{\partial'\rho_i}\frac{\partial}{\partial\tilde{p}}
+ \rho_i\frac{\partial}{\partial\rho_i}
+ \sum_{j=1}^l \rho_i\frac{\partial' R_j}{\partial'\rho_i}\frac{\partial}{\partial\tilde{q}_j},
$$
$$
\frac{\partial'}{\partial' q_i}
=
\frac{\partial' R_p}{\partial' q_i}\frac{\partial}{\partial\tilde{p}}
+\sum_{j=1}^l \left(\delta_{ij} + \frac{\partial' R_j}{\partial' q_i}\right) \frac{\partial}{\partial \tilde{q}_j},
\hspace{50pt}
\frac{\partial'}{\partial' y_i} = \frac{\partial' R_p}{\partial' y_i}\frac{\partial}{\partial\tilde{p}} + \sum_{j=1}^l \frac{\partial' R_j}{\partial' y_i}\frac{\partial}{\partial\tilde{q}_j}
+\frac{\partial}{\partial y_i},
$$
where here and below in the proof the primed derivatives are with respect to $(\tilde{p}_0,\rho,q,y)$ coordinates and the unprimed ones with respect to $(\tilde{p},\rho,\tilde{q},y)$. We have $\frac{\partial' R_{\bullet}}{\partial'\tilde{p}_0},\frac{\partial 'R_{\bullet}}{\partial' q_i},\frac{\partial' R_{\bullet}}{\partial' y_i} \in \mathcal{I}_{\Gamma}^{\epsilon}(\PP)$ and $\rho_i\frac{\partial' R_{\bullet}}{\partial'\rho_i} \in \mathcal{I}_{\Gamma_i}^{\epsilon}(\PP)$ (due to the fact that zeroth-order classical terms at $\Gamma_i$ stay zeroth-order under the action of $\frac{\partial'}{\partial'\rho_i}$). Then, inverting the transformation, we get
$$
\frac{\partial}{\partial\tilde{p}} = \frac{\partial'}{\partial'\tilde{p}_0}
 \mod \mathcal{I}_{\Gamma}^{\epsilon}\Vb(\PP),
\hspace{50pt}
\rho_i \frac{\partial}{\partial\rho_i} = \rho_i \frac{\partial'}{\partial'\rho_i} \mod \mathcal{I}_{\Gamma_i}^{\epsilon}\Vb(\PP),
$$
$$
\frac{\partial}{\partial\tilde{q}_i} = \frac{\partial'}{\partial' q_i}
 \mod \mathcal{I}_{\Gamma}^{\epsilon}\Vb(\PP), 
\hspace{50pt}
\frac{\partial}{\partial y_i} = \frac{\partial'}{\partial' y_i}
 \mod \mathcal{I}_{\Gamma}^{\epsilon}\Vb(\PP).
$$
These computations tell us that, even though $\tilde{p}$ and $\tilde{q}_i$ do not extend smoothly to the boundary, the smoothly varying basis $(\frac{\partial}{\partial\tilde{p}},\rho_i\frac{\partial}{\partial\rho_i},\frac{\partial}{\partial\tilde{q}_j},\frac{\partial}{\partial y_r})$ of $T(U\backslash\partial\PP)$ extends to a $\soe$-regular basis of ${}^{\mathrm{b}}T\PP|_U$, and moreover at $\Gamma$ it matches the original coordinate basis. Then we can write
\begin{equation}
\Hp
=
h_{\tilde{p}}\frac{\partial}{\partial\tilde{p}}
+\sum_{i=1}^k h_{\rho_i}\rho_i\frac{\partial}{\partial\rho_i}
+\sum_{i=1}^l h_{\tilde{q}_i}\frac{\partial}{\partial \tilde{q}_i}
+\sum_{i=1}^m h_{y_i}\frac{\partial}{\partial y_i},
\label{eq:Hp-apriori}
\end{equation}
where $h_{\tilde{p}},h_{\rho_i},h_{\tilde{q}_i},h_{y_i}\in \soe(\PP)$, the expression being valid on all of $U$ including the boundary.

Now let $\PPP$ denote the manifold $\PP$ equipped with a modified smooth structure obtained by adjoining $\tilde{\rho}_i=\sqrt[N]{\rho_i}$ for $i=1,\ldots,k$ to the original smooth structure, and let $\tilde{U}$ denote the neighborhood $U$ as a subspace of $\PPP$. Then $(\tilde{p},z,w)$ is a $C^2$ coordinate chart on $\tilde{U}$ (including the boundary), and $h_{\bullet} \in S^{\mathsf{0}}_{N\epsilon}(\PPP) \subset C^2(\PPP)$. Then by Taylor's theorem with integral remainder, expanding around $\tilde{\Sigma}$,
$$
h_{\bullet}(\tilde{p},z,w)
=
h_{\bullet}(0,z,w)
+
\tilde{p}\cdot \int_0^1 \frac{\partial h_{\bullet}}{\partial\tilde{p}}(t\tilde{p},z,w)\ dt.
$$
Since $\Hp$ is tangent to $\tilde{\Sigma}$, we must have $h_{\tilde{p}}(0,z,w)=0$. For the rest of the coefficients, we further expand the first term around $R$ to second order:
$$
h_{\bullet}(0,z,w)
=
h_{\bullet}(0,0,w)
+
\sum_{i=1}^{k_{\parallel}+l} z_i \cdot \frac{\partial h_{\bullet}}{\partial z_i}(0,0,w)
+
\sum_{i,j=1}^{k_{\parallel}+l} z_i z_j \cdot
\int_0^1 (1-t)\frac{\partial^2 h_{\bullet}}{\partial z_i\partial z_j}(0,tz,w)\ dt.
$$
The fact that $R$ is a critical set means that $h_{\tilde{p}}(0,0,w)=h_{\tilde{q}_i}(0,0,w)=h_{y_i}(0,0,w)=0$, and $h_{\rho_i}(0,0,w)=0$ as well for $i>k_{\parallel}$.

Since the transformation from $(\tilde{p},\rho,\tilde{q},y)$ to $(\tilde{p},\tilde{\rho},\tilde{q},y)$ only affects the $\rho_i$ and $\tilde{\rho}_i$ coordinates, the calculations already done above expressing the derivatives with respect to $\tilde{p}$ and $\tilde{q}_i$ in terms of the original coordinate basis show that $\frac{\partial}{\partial\tilde{p}},\frac{\partial}{\partial \tilde{q}_i}\in \soe\Vb(\PP)$, so $\frac{\partial h_{\bullet}}{\partial\tilde{p}},\frac{\partial h_{\bullet}}{\partial \tilde{q}_i},\frac{\partial^2 h_{\bullet}}{\partial\tilde{q}_i\partial\tilde{q}_j}\in \soe(\PP)$. Meanwhile, away from the boundary we have 
$$\tilde{\rho}_i \frac{\partial}{\partial\tilde{\rho}_i} = N\rho_i \frac{\partial}{\partial\rho_i} = N\rho_i \frac{\partial'}{\partial'\rho_i} \mod \mathcal{I}_{\Gamma_i}^{\epsilon}\Vb(\PP).$$
Then we calculate that
\begin{equation}
\frac{\partial h_{\bullet}}{\partial\tilde{\rho}_i},\frac{\partial^2 h_{\bullet}}{\partial \tilde{\rho}_i\partial\tilde{q}_j} \in \rho_i^{-\frac{1}{N}} \mathcal{I}_{\Gamma_i}^{\epsilon}(\PP),
\hspace{45pt}
 \frac{\partial^2 h_{\bullet}}{\partial\tilde{\rho}_i\partial\tilde{\rho}_j} \in \rho_i^{-\frac{1}{N}}\rho_j^{-\frac{1}{N}} \Big(\mathcal{I}_{\Gamma_i}^{\epsilon}(\PP) \cap \mathcal{I}_{\Gamma_j}^{\epsilon}(\PP)\Big) ,
\label{eq:rho-derivatives}
\end{equation}
which again uses the fact that zeroth-order classical terms at $\Gamma_i$ stay zeroth-order under the action of $\frac{\partial'}{\partial'\rho_i}$. Since $\frac{2}{N}<\epsilon$ and hence symbols in these classes extend continuously to $\partial\PP$, these conclusions apply to the derivatives even when they are taken at boundary points, despite the fact that the coordinate transformation is degenerate at the boundary and $\frac{\partial}{\partial\tilde{\rho}_i}$ only makes sense as a tangent vector on $\PP$ over the interior. In particular, since $\frac{\partial h_{\bullet}}{\partial\tilde{\rho}_i}$ vanishes at $\Gamma_i$ (since $\epsilon>\frac{1}{N}$), we see that $\frac{\partial h_{\bullet}}{\partial\tilde{\rho}_i}(0,0,w)$ for $i\leqslant k_{\parallel}$ vanishes identically.

Thus, combining these results, we can rewrite Eq.~(\ref{eq:Hp-apriori}) in the form Eq.~(\ref{eq:Hp-coord}), where
\begin{equation}
V
=
\int_0^1 \left(\frac{\partial h_{\tilde{p}}}{\partial\tilde{p}}(t\tilde{p},z,w) \frac{\partial}{\partial\tilde{p}}
+\sum_{i=1}^k \frac{\partial h_{\rho_i}}{\partial\tilde{p}}(t\tilde{p},z,w) \rho_i \frac{\partial}{\partial\rho_i} 
+\sum_{i=1}^l \frac{\partial h_{\tilde{q}_i}}{\partial\tilde{p}}(t\tilde{p},z,w) \frac{\partial}{\partial \tilde{q}_i} 
+\sum_{i=1}^m \frac{\partial h_{y_i}}{\partial\tilde{p}}(t\tilde{p},z,w)\frac{\partial}{\partial y_i}\right) dt,
\label{eq:Hp-coeffs-1}
\end{equation}
\begin{equation}
h_{\rho_i}^{(R)}(w) = h_{\rho_i}(0,0,w),
\hspace{45pt}
h_{\bullet}^{(\tilde{q}_j)}(w) = \frac{\partial h_{\bullet}}{\partial \tilde{q}_j}(0,0,w),
\label{eq:Hp-coeffs-2}
\end{equation}
\begin{equation}
h_{\rho_i}^{(z_j)} (z,w) = \int_0^1 \frac{\partial h_{\rho_i}}{\partial z_j}(0,tz,w)\ dt,
\hspace{45pt}
h_{\bullet}^{(z_j,z_r)}(z,w) = \int_0^1 (1-t)\frac{\partial^2 h_{\bullet}}{\partial z_j \partial z_r}(0,tz,w)\ dt.
\label{eq:Hp-coeffs-3}
\end{equation}
The values of the leading-order coefficients at $\alpha$ given in Eq.~(\ref{eq:Hp-leading-terms}) follow from the assumptions on the linearization of $\Hp$ at $\alpha$ and the way we adapted the coordinates to the eigenvectors. 

We have already established that the functions in Eq.~(\ref{eq:Hp-coeffs-2}) are in $\soe(\PP)$ and the derivatives inside the integrals in Eqs.~(\ref{eq:Hp-coeffs-1}) and (\ref{eq:Hp-coeffs-3}) are in $S^{\mathsf{0}}_{\epsilon-\frac{2}{N}}(\PP)$. Since the integrands are all continuous, we can differentiate under the integral sign and use the fact that the derivatives in the integrands are symbols to verify that the integrals stay bounded under repeated application of b-derivatives, so the integrals all belong to $\sym{0}$. This concludes the proof.
\end{proof}

We now prove the main part of the theorem, describing the propagation of singularities in terms of the coordinates just introduced. Assume the hypotheses of Theorem~\ref{thm:localized-rp-main} on $P$ and $\alpha$, and consider coordinates $(\tilde{p}_0,\rho,q,y)$ and symbols $\tilde{q}$ as in Lemma~\ref{thm:coordinates} and notation $\tilde{\rho}$, $z$, $w$ as in Proposition~\ref{thm:Hp-coord}. Let $\lambda_i$, $\mu_i$ be as before. Denote 
$$ 
z_{\pm}^2 = \sum_{\substack{1\leqslant i\leqslant k,\\ \pm\lambda_i>0}} \tilde{\rho}_i^2  + \sum_{\substack{1\leqslant i \leqslant l,\\ \pm\mu_i>0}} \tilde{q}_i^2.
$$
For $\varepsilon>0$, let
$$
\tilde{U}_{\varepsilon} = \{|\tilde{p}|< 2\varepsilon,\ z_+^2<2\varepsilon,\ z_-^2<2\varepsilon,\ |w|^2<2\varepsilon\},
\hspace{30pt}
\tilde{V}_{\varepsilon}^{\pm} = \{ |\tilde{p}|<2\varepsilon,\ \frac{1}{4}\varepsilon < z_{\pm}^2 < 2\varepsilon,\ z_{\mp}^2 < 2\varepsilon,\ |w|^2 < 2\varepsilon \},
$$
and $U_{\varepsilon}=\tilde{U}_{\varepsilon}\cap\partial\PP$, $V_{\varepsilon}=\tilde{V}_{\varepsilon}\cap\partial\PP$.
	
\begin{prop}
\label{thm:localized-rp-estimates}
Consider two sets of orders $\mathsf{s'}\leqslant \mathsf{s}$. Assume $\sum_{i=1}^k (-2s_i+m_i-1)\lambda_i + \hat{p}_1 < 0$ and $ \sum_{i=1}^k (-2s'_i+m_i-1)\lambda_i + \hat{p}_1 < 0$. Then there exists $\varepsilon_0>0$ such that for any $\varepsilon\in (0,\varepsilon_0)$, there exists $\varepsilon'\in (0,\varepsilon)$ such that for any $Q,Q',Q''\in\ps[0]$ with $Q'$ elliptic on $U_{\varepsilon}$, $Q''$ elliptic on $V^-_{\varepsilon}$, and $\wfs(Q)\subset U_{\varepsilon'}$ and any set of orders $\mathsf{N}$, there exists $C>0$ such that all $\tau\geqslant 0$ and $u\in\sob{-N}$ with $Q'(P-i\tau)u\in\sob{s-m+1}$, $Q''u\in\sob{s}$, and $Q'u\in\sob{s'}$ satisfy
\begin{equation}
\label{eq:rp-estimate}
\|Qu\|_{\mathsf{s}} \leqslant C\Big( \|Q'(P - i\tau)u\|_{\mathsf{s-m+1}} + \|Q''u\|_{\mathsf{s}} + \|Q'u\|_{\mathsf{s'}} + (1+|\tau|)\|u\|_{\mathsf{-N}} \Big).	
\end{equation}
The analogous statement holds for $\tau\leqslant 0$ if instead $\sum_{i=1}^k (-2s_i+m_i-1)\lambda_i + \check{p}_1 > 0$ and $ \sum_{i=1}^k (-2s'_i+m_i-1)\lambda_i + \check{p}_1 > 0$ and $V^-_{\varepsilon}$ is replaced by $V^+_{\varepsilon}$.
\end{prop}
	
	\begin{proof} We focus on the case $\tau\geqslant 0$, $\sum_{i=1}^k (-2s_i+m_i-1)\lambda_i + \hat{p}_1 <0$, and $\sum_{i=1}^k (-2s_i'+m_i-1)\lambda_i + \hat{p}_1 <0$; the other case is analogous, requiring only some sign changes on which we comment later. Fix $\varepsilon_0>0$ small enough that the conclusions of Proposition~\ref{thm:Hp-coord} apply on $\tilde{U}_{\varepsilon_0}$ and any $\varepsilon\in (0,\varepsilon_0)$. We aim to show the existence of $\varepsilon'$ as required if $\varepsilon_0$ is small enough.
	
	\begin{itemize}
	
\item \textbf{Defining the commutant.} Let $\phi:\R\to\R$ be a smooth function such that $\phi(t)=0$ for $|t|>\varepsilon$, $\phi(t)=1$ for $|t|<\frac{1}{2}\varepsilon$, monotone increasing on $[-\varepsilon,-\frac{1}{2}\varepsilon]$ and monotone decreasing on $[\frac{1}{2}\varepsilon,\varepsilon]$, and such that $\sqrt{|\phi'\phi|}$ is smooth. On $\tilde{U}_{\varepsilon}$, define
\begin{equation}
\hat{w}
=
|w|^2+\sum_{i=1}^l C_i(w)\tilde{q}_i + \sum_{i=1}^{k_{\parallel}+l} D_i z_i^2,
\label{eq:scaling-factor}
\end{equation}
with functions $C_i$ and constants $D_i$ to be fixed later. (The role of the $C_i$ is to correct for the fact that $\frac{\partial}{\partial q_i}$ are only eigenvectors of the linearization at $\alpha$ itself, not necessarily at other nearby points of $R$; see also the discussion and figures in Section~\ref{sec:rp-statement}). We will choose them so that $|C_i|<\frac{\sqrt{\varepsilon}}{2l}$ and $|D_i|<\frac{1}{2(k_{\parallel}+l)}$; then $|\hat{w}-|w|^2|<\varepsilon$ on $\{z_-^2,z_+^2\leqslant\varepsilon\}\subset \{\forall i\ z_i^2\leqslant\varepsilon\}$, so
\begin{equation}
\mathrm{supp}(\phi(\hat{w}))\cap \mathrm{supp}(\phi(z_+^2)) \cap \mathrm{supp}(\phi(z_-^2)) \subset \left\{ |w|^2\leqslant 2\varepsilon\right\}.
\label{eq:w-supp}
\end{equation}
We will also ensure $C_i\in\sym{0}$. Then for any set of orders $\mathsf{r}$, we can define a symbol $\check{a}\in \sym{r-\frac{m-1}{2}}$ by
\begin{equation}
\check{a}
=
\rho^{\mathsf{-r+\frac{m-1}{2}}}
\phi(\tilde{p})
\phi(z_+^2)
\phi(z_-^2)
\phi(\hat{w})
\label{eq:commutant}
\end{equation}
in $\tilde{U}_{\varepsilon}$ and $\check{a}=0$ outside $\tilde{U}_{\varepsilon}$. We consider $\mathsf{r}$ which also satisfies the inequality $\sum_{i=1}^k (-2r_i+m_i-1)\lambda_i + \hat{p}_1 <0$ and set $a=\check{a}^2\in \sym{2r-m+1}$. We denote
$$
\tilde{W}_{\varepsilon} = \left\{|\tilde{p}|<\varepsilon,\ z_-^2<\varepsilon,\ z_+^2<\varepsilon,\ \hat{w}<\varepsilon \right\} \subset \tilde{U}_{\varepsilon},
\hspace{30pt}
W_{\varepsilon} = \tilde{W}_{\varepsilon}\cap\partial\PP.
$$
We have $\mathrm{supp}(a)=\overline{\tilde{W}_{\varepsilon}}$.

\item \textbf{Action of Hamilton vector field on commutant.}
On $\tilde{U}_{\varepsilon}$, let us write
$$
\Hp = \tilde{p}V + \sum_{i=1}^k h_{\rho_i}\rho_i\frac{\partial}{\partial\rho_i} + \sum_{i=1}^l h_{\tilde{q}_i}\frac{\partial}{\partial \tilde{q}_i} + \sum_{i=1}^m h_{y_i}\frac{\partial}{\partial y_i},
$$
where $V\in \soe\Vb(\PP)$ and $h_{\rho_i},h_{\tilde{q}_i},h_{y_i}\in \soe(\PP)$ as found in Proposition~\ref{thm:Hp-coord}. Fix a constant $\delta\in \left( 0, |\sum_{i=1}^k(-2r_i+m_i-1)\lambda_i+\hat{p}|\right)$. Then we calculate
$$
H_pa + p_1a + \delta \rho^{\mathsf{-m+1}} a
=
\rho^{\mathsf{-m+1}}\Big(\Hp a + \tilde{p}_1a + \delta a\Big)
=
$$
$$
=
p\rho^{\mathsf{1}}(Va)
+
\rho^{\mathsf{-2r}}
\Bigg(
\left(
\sum_{i=1}^{k} (-2r_i+m_i-1)
h_{\rho_i}
+ \tilde{p}_1
+ \delta
\right)
 \cdot
\phi(\tilde{p})^2\phi(z_+^2)^2\phi(z_-^2)^2
\phi(\hat{w})^2
+
$$
$$
+
2
\left(
\frac{2}{N}\sum_{\substack{1\leqslant i\leqslant k,\\ \lambda_i>0}}  h_{\rho_i} \tilde{\rho}_i^2
+
2\sum_{\substack{1\leqslant i\leqslant l,\\ \mu_i>0}}  h_{\tilde{q}_i}\tilde{q}_i
\right)
\cdot
\phi(\tilde{p})^2\phi'(z_+^2)\phi(z_+^2)\phi(z_-^2)^2
\phi(\hat{w})^2
+
$$
$$
+
2
\left(
\frac{2}{N} \sum_{\substack{1\leqslant i\leqslant k,\\ \lambda_i<0}}  h_{\rho_i} \tilde{\rho}_i^2
+
2\sum_{\substack{1\leqslant i\leqslant l,\\ \mu_i<0}}  h_{\tilde{q}_i}\tilde{q}_i
\right)
\cdot
\phi(\tilde{p})^2\phi(z_+^2)^2\phi'(z_-^2)\phi(z_-^2)
\phi(\hat{w})^2
+
$$
$$
+
2\Bigg[
\sum_{i=k_{\parallel}+1}^k h_{\rho_i} \left( \frac{2}{N} \tilde{\rho}_i^2 + \sum_{j=1}^l \rho_i \frac{\partial C_j}{\partial \rho_i}(w) \tilde{q}_j\right)
+
\sum_{i=1}^m h_{y_i} \left( 2y_i + \sum_{j=1}^l \frac{\partial C_j}{\partial y_i}(w) \tilde{q}_j \right)
+
\sum_{i=1}^l h_{\tilde{q}_i} C_i(w)
+
$$
$$
+
\frac{2}{N} \sum_{i=1}^{k_{\parallel}} h_{\rho_i}  D_i \tilde{\rho}_i^2 
+
2 \sum_{i=1}^l h_{\tilde{q}_i}  D_{k_{\parallel}+i}\tilde{q}_i 
\Bigg]
\cdot
\phi(\tilde{p})^2\phi(z_+^2)^2\phi(z_-^2)^2
\phi'(\hat{w})
\phi(\hat{w})
\Bigg).
$$
Now we consider each of the resulting terms. We write $\tilde{x}=(\tilde{p},z,w)$.

\begin{itemize}

\item Since $V$ is a b-vector field and $a\in\sym{2r-m+1}$, we have $\rho^{\mathsf{1}} Va \in\sym{2r-m}$. We write $\rho^{\mathsf{1}} Va=\tilde{a}$; the first term is then $\tilde{a}p$, and its contribution will be controlled using the regularity assumption on $(P- i\tau)u$.

\item Assuming $\sum_{i=1}^{k_{\parallel}} (-2r_i+m_i-1)\lambda_i + \hat{p}_1 < 0$ and $\delta\in \left( 0, |\sum_{i=1}^k(-2r_i+m_i-1)\lambda_i+\hat{p}_1|\right)$, since $h_{\rho_i}(\alpha)=\lambda_i$, for $\varepsilon$ small enough the expression in parentheses in the second term is negative on all of $\tilde{U}_{\varepsilon}$. Then
$$b=
 \rho^{\mathsf{-r}}
 \sqrt{\left|
\sum_{i=1}^k (-2r_i+m_i-1)h_{\rho_i}
+ \tilde{p}_1
+ \delta
 \right|} \phi(\tilde{p})\phi(z_+^2)\phi(z_-^2)\phi(\hat{w})$$
(extended to be zero outside of $\tilde{U}_{\varepsilon}$) is a symbol of order $\mathsf{r}$ elliptic on $W_{\varepsilon}$, and the second term equals $- b^2$.

\item Since $h_{\rho_i}=\lambda_i+o(1)$ and $h_{\tilde{q}_i}=\mu_i\tilde{q}_i+o(|\tilde{x}|)$ as $\tilde{x}\to 0$, the expression in parentheses in the third term is 
$$
\frac{2}{N}\sum_{i=1}^{k_{\parallel}} \lambda_i\tilde{\rho}_i^2 + 2\sum_{i=1}^l \mu_i \tilde{q}_i^2 + o(|\tilde{x}|^2)
\geqslant
2\min_{\substack{1\leqslant i\leqslant k_{\parallel},\\1\leqslant j\leqslant l}}(\lambda_i/N,\mu_j)\cdot z_+^2 + o(|\tilde{x}|^2).
$$
On the support of the term, $z_+^2\geqslant \frac{1}{2}\varepsilon$. On the other hand, $\sup_{U_{\varepsilon}}|\tilde{x}|^2=\mathcal{O}(\varepsilon)$ as $\varepsilon\to 0^+$. Therefore, for small enough $\varepsilon$ the expression above is strictly positive on the support. By the same argument, in the fourth term the expression in parentheses is strictly negative on the support for small enough $\varepsilon$. Therefore, we can define
$$
e_{\pm}
=
\rho^{\mathsf{-r}}
\sqrt{\mp 2\phi'(z_{\pm}^2)\phi(z_{\pm}^2)
\left(
\frac{2}{N}\sum_{\substack{1\leqslant i\leqslant k,\\ \pm\lambda_i>0}}  h_{\rho_i} \tilde{\rho}_i^2
+
2\sum_{\substack{1\leqslant i\leqslant l,\\ \pm\mu_i>0}}  h_{\tilde{q}_i}\tilde{q}_i
\right)}
\phi(\tilde{p})\phi(z_{\mp}^2)\phi(\hat{w}),
$$
which are symbols of order $\mathsf{r}$. $e_{\pm}^2$ is supported in $V^{\pm}_{\varepsilon}$. Therefore, the contribution of the $e_-^2$ term will be controlled by the regularity assumption on $u$ in $V^-_{\varepsilon}$. The $e_+^2$ term will come with the right sign to be ignored for the estimate. If we were propagating in the other direction, the roles would be reversed.

\item To deal with the last term, which is responsible for localization in the neutral directions, we use the full form of the coefficients found in Proposition~\ref{thm:Hp-coord}. The expression in brackets becomes

\begin{equation}
\begin{split}
  \sum_{j=1}^l 
  \tilde{q}_j \cdot
  \left( 
  \sum_{i=1}^l C_i h_{\tilde{q}_i}^{(\tilde{q}_j)}
  +
  \frac{2}{N} \sum_{i=k_{\parallel}+1}^k  h_{\rho_i}^{(\tilde{q}_j)} \tilde{\rho}_i^2
  + 2 \sum_{i=1}^m h_{y_i}^{(\tilde{q}_j)} y_i
  \right) 
  +
  \\
+
\frac{2}{N}\sum_{i=1}^{k_{\parallel}}  D_i \tilde{\rho}_i^2 
\cdot 
\left(
h_{\rho_i}^{(R)}
+ \sum_{j=1}^{k_{\parallel}+l} h_{\rho_i}^{(z_j)} z_j
\right)
+
\\
+
\sum_{j,r=1}^l \tilde{q}_j\tilde{q}_r \cdot 
\left(
2 D_{k_{\parallel}+r} h_{\tilde{q}_r}^{(\tilde{q}_j)}
+ \sum_{i=k_{\parallel}+1}^k h_{\rho_i}^{(\tilde{q}_j)} \rho_i \frac{\partial C_r}{\partial\rho_i}
+ \sum_{i=1}^m h_{y_i}^{(\tilde{q}_j)} \frac{\partial C_r}{\partial y_i}
\right)
+
\\
+
\sum_{j,r=1}^{k_{\parallel}+l} 
z_j z_r \cdot
\Bigg[
\sum_{i=1}^l  \Big(C_i + 2D_{k_{\parallel}+i} \tilde{q}_i \Big)  h_{\tilde{q}_i}^{(z_j,z_r)} 
+
\sum_{i=k_{\parallel}+1}^k 
\left(\frac{2}{N}\tilde{\rho}_i^2
+\sum_{s=1}^l \rho_i\frac{\partial C_s}{\partial\rho_i} \tilde{q}_s \right) h_{\rho_i}^{(z_j,z_r)}
+
\\
+
\sum_{i=1}^m 
\left( 2y_i + \sum_{s=1}^l \tilde{q}_s
 \frac{\partial C_s}{\partial y_i}
 \right) h_{y_i}^{(z_j,z_r)}
\Bigg]
  .
\end{split}
\label{eq:localization-along-R}
\end{equation}

First, we choose the functions $C_i(w)$ so that the first line vanishes identically on $\{ |w|^2< 2\varepsilon\}$; this condition is equivalent to the system of linear equations
$$
\sum_{i=1}^l h_{\tilde{q}_i}^{(\tilde{q}_j)}(w) \cdot C_i(w)
=
-\left(
\frac{2}{N} \sum_{i=k_{\parallel}+1}^k  h_{\rho_i}^{(\tilde{q}_j)}(w) \tilde{\rho_i}^2
  + 2 \sum_{i=1}^m h_{y_i}^{(\tilde{q}_j)}(w) y_i
  \right)
$$
for $j=1,\ldots,l$. The right-hand side and the coefficients on the left-hand side are in $\soe(\PP)$ and depend only on $w$, and by Eq.~(\ref{eq:Hp-leading-terms}) the matrix of coefficients approaches a diagonal matrix with nonzero diagonal entries as $w\to 0$. Therefore, for $\varepsilon$ small enough the solution exists and is a $\soe(\PP)$-regular function of $w$ on $\{|w|^2<2\varepsilon\}$. Eq.~(\ref{eq:Hp-leading-terms}) implies that the right-hand side is $o(|w|)$ as $w\to 0$, so we have $C_i(w)=o(|w|)$ as $w\to 0$ and therefore for $\varepsilon$ small enough we will have $|C_i|<\frac{\sqrt{\varepsilon}}{2l}$ on $\{|w|^2<2\varepsilon\}$ as needed.

Next, we note that the expression on the second through last lines of Eq.~(\ref{eq:localization-along-R}) is of the form $\sum_{j,r=1}^{k_{\parallel}+l} Q_{jr}(\tilde{x})z_j z_r$ with $Q_{jr}\in \sym{0}$ and $\lim_{\tilde{x}\to 0}Q_{jr}(\tilde{x})= 2 \mu_j' D_j \delta_{jr}$, where $\mu_j'=\lambda_j/N$ for $j=1,\ldots, k_{\parallel}$ and $\mu_j'=\mu_{j-k_{\parallel}}$ for $j=k_{\parallel}+1,\ldots,k_{\parallel}+l$. (Besides the results in Eq.~(\ref{eq:Hp-leading-terms}), this uses the fact that, since $C_i\in \soe(\PP)$ and $C_i(w)=o(|w|)$, the values of the $C_i$ all vanish at zero, and so do the b-derivatives $\rho_i\frac{\partial C_r}{\partial\rho_i}$). Then we choose the $D_j\in \left(-\frac{1}{2(k_{\parallel}+l)},\frac{1}{2(k_{\parallel}+l)}\right)$ so that the numbers $\mu_j' D_j$ are all negative (the same sign as the $-b^2$ term; if we were propagating in the backward direction, we would make the opposite choice). Then in a neighborhood of zero the symmetric matrix $-Q$ is positive definite, so it has a symmetric square root $\check{Q}\in \sym{0}$ and the expression can be written as $-\sum_{j=1}^{k_{\parallel}+l} \check{z}_j^2$, where $\check{z}_j=\sum_{r=1}^{k_{\parallel}+l}\check{Q}_{jr}z_r\in \sym{0}$. We define
$$
f_j(x)
=
\rho^{\mathsf{-r}}\phi(\tilde{p})\phi(z_+^2)\phi(z_-^2)\check{z}_j
\sqrt{2 |\phi'(\hat{w}) \phi(\hat{w})|},
$$
which are symbols of order $\mathsf{r}$.
\end{itemize}

With these definitions, we have
\begin{equation}
(H_p + p_1 + \delta\rho^{\mathsf{m-1}})a = \tilde{a}p - b^2 - e_+^2 + e_-^2 - \sum_{j=1}^{k_{\parallel}+l}f_j^2.
\label{eq:commutator-symbolic}
\end{equation}

\item \textbf{Regularization.}
We define $\chi(t)$ as in Eq.~(\ref{eq:pos-regularizer}) and again set $\check{a}_t=\chi_t\check{a}$ and $a_t=\check{a}_t^2$. Then
\begin{equation}
(H_p + p_1 + \delta\rho^{\mathsf{m-1}})a_t
=
\tilde{a}_tp - b_t^2 - e_{+,t}^2 + e_{-,t}^2 - \sum_{j=1}^{k_{\parallel}+l} f_{j,t}^2,
\label{eq:commutator-symbolic-reg}
\end{equation}
where $\tilde{a}_t=\chi_t^2\tilde{a}\in \sym{2r-m-2K}$, $e_{\pm,t}=\chi_t e_{\pm}\in \sym{r-K}$, $f_{j,t}=\chi_t f_j\in \sym{r-K}$, and
$$
b_t=
 \chi_t\rho^{-\mathsf{r}}
 \sqrt{\left|
 \sum_{i=1}^k \left(-2r_i+m_i-1+\frac{2K_it}{t+\rho_i}\right)h_{\rho_i}
+ \tilde{p}_1
+ \delta 
 \right|}
 \cdot
 \phi(\tilde{p})\phi(z_+^2)\phi(z_-^2)\phi(\hat{w}) \in \sym{r-K},
$$
as long as inclusion of the $\frac{2K_it}{t+\rho_i}$ terms does not change the sign of the expression under the absolute value on $\tilde{U}_{\varepsilon}$. Since $\frac{2K_it}{t+\rho_i}(\alpha)=2K_i$, this will be true for small enough $\delta$ and $\varepsilon$ as long as $ \sum_{i=1}^{k_{\parallel}} (-2(r_i-K_i)+m_i-1)\lambda_i + \tilde{p}_1(\alpha)< 0$, i.e. $\mathsf{r-K}$ satisfies the same inequality as that required of $\mathsf{s}$ in the hypotheses of the theorem. We fix $\mathsf{K}$ such that $\mathsf{r-K}\leqslant\mathsf{s'}$ while still satisfying the inequality and choose $\delta$ and $\varepsilon$ small enough correspondingly.

\item \textbf{Quantization and estimates.} We define operators $\check{A}_t=\mathrm{Op}(\check{a}_t)$, $A_t=\check{A}_t^*\check{A}_t$, $\tilde{A}_t=\mathrm{Op}(\tilde{a}_t)$, $B_t=\mathrm{Op}(b_t)$, $E_{\pm,t}=\mathrm{Op}(e_{\pm,t})$, $F_{j,t}=\mathrm{Op}(f_{j,t})$, $\Lambda=\mathrm{Op}(\rho^{\mathsf{-\frac{m-1}{2}}})$. The symbolic relation Eq.~(\ref{eq:commutator-symbolic-reg}) means that
\begin{equation}
i[P,A_t] 
-i(P-P^*)A_t
+ \delta (\Lambda\check{A}_t)^*(\Lambda\check{A}_t)
= \tilde{A}_tP 
 - B_t^*B_t - E_{+,t}^*E_{+,t} + E_{-,t}^*E_{-,t} - \sum_{j=1}^{k_{\parallel}+l} F_{j,t}^*F_{j,t} + R_t
\label{eq:commutator-quantized}
\end{equation}
for some $R_t\in \ps[2r-2K-1]$, where the family $R_t$ is bounded in $\ps[2r-1]$ and $\wfs(\{R_t\})\subset\overline{W_{\varepsilon}}$. 

Fixing $\mathsf{r}\leqslant\mathsf{s}$, for any $u\in\sch'$ such that 
\begin{equation}
\label{eq:wf-assumption-iteration}
\Big( \wf{s'}(u)\cup
\wf{s-m+1}((P- i\tau)u)
\Big)
\cap U_{\varepsilon}=\varnothing
\end{equation}
for some $\tau\in\R$, thanks to the choice of large enough $\mathsf{K}$ above we can calculate
$$
\|B_tu\|^2
=
- \left\langle \Big(i[P,A_t]-i(P-P^*)\Big)u, u\right\rangle
+ \langle (P-i\tau)u, \tilde{A}_t^*u\rangle
- i\tau\langle u,\tilde{A}_t^*u\rangle
-\delta \|\Lambda\check{A}_t u\|^2
-
$$
$$
- \|E_{+,t}u\|^2
+ \|E_{-,t}u\|^2
-\sum_{j=1}^{k_{\parallel}+l}\|F_{j,t} u\|^2
+\langle R_tu,u\rangle
\leqslant
$$
$$
\leqslant
 2|\langle \check{A}_tu, \check{A}_t (P+i\tau)u\rangle|
+ |\langle (P-i\tau)u, \tilde{A}_t^*u\rangle|
+ |\tau| |\langle u,\tilde{A}_t^*u\rangle|
- \delta \|\Lambda\check{A}_t u\|^2
+ \|E_{-,t}u\|^2
+ |\langle R_tu,u\rangle|,
$$
where the $i\tau\langle u,\tilde{A}_t^*u\rangle$ term should be understood to be absent in the $\tau=0$ case, where the pairing may not be well-defined; for $\tau\neq 0$, on the other hand, it is well-defined since $\wf{r-\frac{m-1}{2}-K}(u)\cap U_{\varepsilon}=\varnothing$ by Corollary~\ref{thm:complex} while $\tilde{A}_t^*\in \ps[2r-m-2K]$ with $\wfs(\tilde{A}_t^*)\subset\overline{W_{\varepsilon}}\subset U_{\varepsilon}$. Proceeding as in the proof of Theorem~\ref{thm:pos}, we conclude
$$
\|B_t u\|^2
\leqslant
 C \Big(
 \|\check{A}_t(P-i\tau)u\|_{\mathsf{-\frac{m-1}{2}}}^2
+\|u\|_{\mathsf{-N}}^2
\Big)
 +|\langle (P-i\tau)u, \tilde{A}_t^*u\rangle|
 +|\tau| |\langle u,\tilde{A}_t^*u\rangle|
+ \|E_{-,t}u\|^2
+|\langle R_tu,u\rangle|
.
$$
Fix any $Q',Q'',Q'''_{\mathsf{r}}\in\ps[0]$ with $Q'$ elliptic on $U_{\varepsilon}$, $Q''$ elliptic on $V_{\varepsilon}^-$, and $Q'''_{\mathsf{r}}$ elliptic on $\overline{W_{\varepsilon}}$. The first, fourth, and fifth terms on the right-hand side are bounded uniformly in $t$ in a manner analogous to the proof of Theorem~\ref{thm:pos}. Turning to the two terms involving $\tilde{A}_t^*$:
 
\begin{itemize}

\item The family $\tilde{A}_t$ is bounded in $\ps[2r-m]$ and $\wfs(\{\tilde{A}_t\})\subset\overline{W_{\varepsilon}} \subset U_{\varepsilon}$, so
$$
|\langle(P- i\tau)u,\tilde{A}_t^*u\rangle | \leqslant
C \Big( \|Q'(P- i\tau)u\|_{\mathsf{s-m+1}}^2 + \|Q'''_{\mathsf{r}} u\|_{\mathsf{r-1}}^2 + \|u\|_{\mathsf{-N}}^2 \Big).
$$

\item Introducing an auxiliary $\tilde{Q}\in\ps[0]$ elliptic on $\overline{W_{\varepsilon}}$ but such that $Q_{\mathsf{r}}'''$ is elliptic on $\wfs(\tilde{Q})$, and combining Lemma~\ref{thm:int-by-parts} and Corollary~\ref{thm:complex},
$$
|\tau| |\langle u,\tilde{A}_t^*u\rangle |
\leqslant
C|\tau| \Big( \|\tilde{Q} u\|_{\mathsf{r-\frac{m}{2}}}^2 + \|u\|_{\mathsf{-N}}^2 \Big)
\leqslant
C \Big( \|Q'(P-i\tau)u\|_{\mathsf{s-m+1}}^2 + \|Q'''_{\mathsf{r}} u\|_{\mathsf{r-\frac{1}{2}}}^2 + (1+|\tau|)\|u\|_{\mathsf{-N}}^2 \Big)
.
$$
\end{itemize}

\item \textbf{Regularity conclusions and iteration.}
As in the proof of Theorem~\ref{thm:pos}, we conclude that $Bu\in \LL$ with the estimate
$$
\|Bu\|^2 \leqslant C\Big( \|Q'(P - i\tau)u\|^2_{\mathsf{s-m+1}} + \|Q''u\|^2_{\mathsf{s}} + \|Q'''_{\mathsf{r}} u\|^2_{\mathsf{r-\frac{1}{2}}} + (1+|\tau|^2)\|u\|^2_{\mathsf{-N}} \Big),	
$$
where $C>0$ is independent of $\tau\geqslant 0$ and $u$.

Compared to the propagation estimate, the preceding argument required that $\mathsf{r}$ satisfied the same inequality as $\mathsf{s}$ in the hypotheses of the theorem. Therefore, while we can still iterate the argument to improve the $\|Q'''_{\mathsf{r}}u\|_{\mathsf{r-\frac{1}{2}}}^2$ term, we can only do so while staying at every step within the set of orders satisfying the inequality. This set is a half-space in $\R^k$, so for instance the straight-line path between $\mathsf{s'}$ and $\mathsf{s}$ stays within the set; thus, we can start with $\mathsf{r}=\mathsf{s'}$ and iterate along this path until we reach $\mathsf{r}=\mathsf{s}$ in a fixed number of steps, arriving at the estimate Eq.~(\ref{eq:rp-estimate}) with any $\varepsilon' \in(\frac{1}{4}\varepsilon,\frac{1}{2}\varepsilon)$ (the upper bound due to $U_{\varepsilon'}\subset W_{2\varepsilon'}$, and the lower bound ensuring the non-sign-definite term stays supported in $\tilde{V}_{\varepsilon}^-$ throughout the iteration).
\end{itemize}
	\end{proof}

See Figure~\ref{fig:rad-pts-prop} for an illustration of the various regions defined and used in the proof.
	
\begin{figure}
\begin{center}
\includegraphics[width=0.47\textwidth]{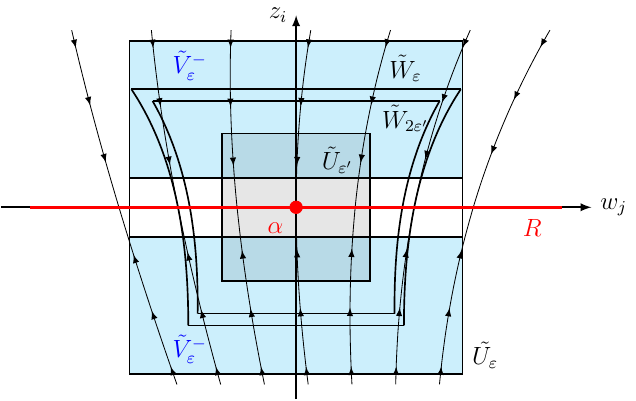}
\hspace{10pt}
\includegraphics[width=0.47\textwidth]{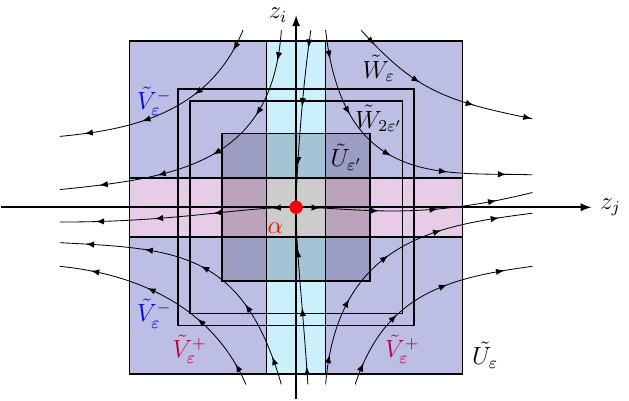}
\end{center}
\caption{Schematic illustration of propagation setup in proof of Proposition~\ref{thm:localized-rp-estimates}, for propagation into $\alpha$ in the forward direction (from $V_{\varepsilon}^-$, shaded blue). In these two two-dimensional slices of the same higher-dimensional region, $z_i$ represents a defining function of $R$ corresponding to a negative eigenvalue, $z_j$ a defining function corresponding to a positive eigenvalue, and $w_j$ a coordinate along $R$. The figure on the left illustrates the localization along $R$; the figure on the right illustrates a cross-section transverse to $R$ in which the flow is of saddle type. The regions $\tilde{W}_{\varepsilon}$ are constructed so that, for a range of values of $\varepsilon$, the flow only enters them through the top and bottom faces (cf. figure on the left in Figure~\ref{fig:localization}). This choice ensures that bicharacteristics cross level sets of the cutoffs along $R$ in the same direction as they cross those of the cutoff on the unstable side of $R$ (left/right faces of $\tilde{W}_{\varepsilon}$ in the figure on the right) and opposite to the direction that they cross the cutoff on the stable side (top/bottom faces). This means that if a priori regularity on the unstable side is not required for the estimate (i.e. if one is propagating forward), then the cutoff along $R$ will also not contribute any new terms which require a priori assumptions. To propagate in the backward direction, one instead constructs the cutoff along $R$ as in the figure on the right in Figure~\ref{fig:localization}, corresponding to opposite signs of the parameters $D_j$.
\label{fig:rad-pts-prop}}
\end{figure}
	
We can now finish the proof of Theorem~\ref{thm:localized-rp-main} by analyzing the dynamics of the Hamilton flow near $\alpha$. Concretely, we show that if for given $(u,\tau)$ with $\pm\tau\geqslant 0$ and $\alpha\notin \wf{s-m+1}((P-i\tau)u)$ there \textit{does not} exist arbitrarily small $\varepsilon>0$ such that $\wf{s}(u)\cap V^{\mp}_{\varepsilon}=\varnothing$, then there necessarily exists a bicharacteristic on which $u$ is singular limiting to $\alpha$ in the right direction, i.e. the hypotheses of Theorem~\ref{thm:localized-rp-main} must be violated. Thus the hypotheses of Theorem~\ref{thm:localized-rp-main} imply that if in addition $\alpha\notin\wf{s'}(u)$, then we can apply Proposition~\ref{thm:localized-rp-estimates} with some $Q',Q''$ such that the right-hand side of Eq.~(\ref{eq:rp-estimate}) is finite and $Q$ elliptic at $\alpha$, concluding $\alpha\notin\wf{s}(u)$.

\begin{figure}
\begin{center}
\includegraphics[width=0.47\textwidth]{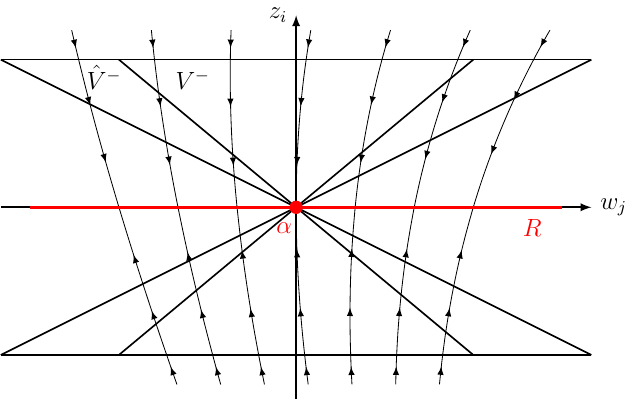}
\hspace{10pt}
\includegraphics[width=0.47\textwidth]{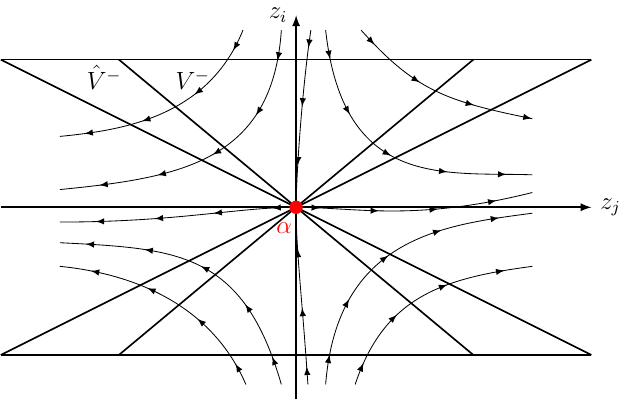}
\end{center}
\caption{Schematic illustration of the regions $V^-$, $\hat{V}^-$ defined in the proof of Theorem~\ref{thm:localized-rp-main}. As in Figure~\ref{fig:rad-pts-prop}, $z_i$ represents a defining function of $R$ corresponding to a negative eigenvalue, $z_j$ a defining function corresponding to a positive eigenvalue, and $w_j$ a coordinate along $R$. Note the flow only enters both regions through the top and bottom faces, and every integral curve in the regions either limits to $\alpha$ or exits the regions in finite time through the sides.
\label{fig:limiting-bichar}}
\end{figure}
	
\begin{proof}[Proof of Theorem~\ref{thm:localized-rp-main}]
Fix $\varepsilon_0$ as in Proposition~\ref{thm:localized-rp-estimates} and small enough that $\wf{s-m+1}((P-i\tau)u)\cap U_{\varepsilon_0}=\varnothing$, so we can appeal to elliptic regularity and propagation of singularities in this region. Below we always work within $\Sigma$, which we keep implicit when specifying regions in coordinates. We can do this because we know the Hamilton flow stays within $\Sigma$.

Define
$$
V^{\pm} = 
\left(\bigcup_{\varepsilon\in (0,\varepsilon_0)} \overline{V^{\pm}_{\varepsilon}}\right) 
\cap
\{z_{\pm}^2 \leqslant \frac{1}{4}\varepsilon_0\} = \{0<z_{\pm}^2 \leqslant \frac{1}{4}\varepsilon_0,\ z_{\mp}^2 \leqslant 8z_{\pm}^2,\ |w|^2 \leqslant 8z_{\pm}^2\},
$$
as well as the larger set $\hat{V}^{\pm} =  \{0<z_{\pm}^2\leqslant \frac{1}{4}\varepsilon_0,\ z_{\mp}^2 \leqslant  9 z_{\pm}^2,\ |w|^2 \leqslant 9 z_{\pm}^2\}$. For any $C>0$ (the values of interest for us being $C=\frac{1}{8}$ or $\frac{1}{9}$), we have
$$
\Hp z_{\pm}^2,\ \Hp (z_{\pm}^2-C|w|^2) = \pm 2\sum_{\substack{1\leqslant i\leqslant k_{\parallel}+l,\\ \pm\mu'_i>0 }} |\mu'_i|z_i^2 + o(|\tilde{x}|^2),
$$
$$
\Hp(z_{\pm}^2-C z_{\mp}^2) = 
\pm 2\sum_{\substack{1\leqslant i\leqslant k_{\parallel}+l,\\ \pm\mu'_i>0 }} |\mu'_i|z_i^2
\pm 2C \sum_{\substack{1\leqslant i\leqslant k_{\parallel}+l,\\ \mp\mu'_i>0 }} |\mu'_i|z_i^2 + o(|\tilde{x}|^2),
$$
where we again use the notation $\tilde{x}=(\tilde{p},z,w)$. On $\hat{V}^{\pm}$, we have $|\tilde{x}|^2 = z_+^2+z_-^2+|w|^2\leqslant 19 z_{\pm}^2$, so for $\varepsilon_0$ small enough the $o(|\tilde{x}|^2)$ terms become negligible and the expressions above have the $\pm$ sign on all of $\hat{V}^{\pm}$. Moreover, they are bounded from below by a positive constant on any set of the form $\{\varepsilon''\leqslant z_{\pm}^2 \leqslant \frac{1}{4}\varepsilon_0\}\cap \hat{V}^{\pm}$ for $\varepsilon''\in (0,\frac{1}{4}\varepsilon_0)$. See Figure~\ref{fig:limiting-bichar}.

Focusing on $V^-$ (the analysis for $V^+$ is analogous), this means that for small $\varepsilon_0$ bicharacteristics do not enter $V^-$ or $\hat{V}^-$ except through the $z_-^2=\frac{1}{4}\varepsilon_0$ face, and in addition $z_-^2$ decreases along the flow in $\tilde{V}^-$ at a rate (with respect to the flow parameter) that is bounded from below by a positive constant on any set of the form $\{\varepsilon''\leqslant z_-^2 \leqslant \frac{1}{4}\varepsilon_0\}\cap \hat{V}^-$. Then in the forward direction, the bicharacteristic through any point of $V^-$ or $\hat{V}^-$ must either limit to $\alpha$ or exit $V^-$ or $\tilde{V}^-$ respectively in finite time at a lower value of $z_-^2$, and in the backward direction it must reach $\{z_-^2=\frac{1}{4}\varepsilon_0\}$ in finite time while staying within $V^-$ or $\hat{V}^-$ respectively, with $z_-^2$ monotone along the flow.

Let us assume that for every $\varepsilon\in (0,\frac{1}{4}\varepsilon_0)$ there exists a point $\alpha_{\varepsilon}\in\wf{s}(u)\cap V_{\varepsilon}^-$. These points must lie in $\Sigma$ by elliptic regularity. Let $\alpha_{\varepsilon}'$ be the intersection of the bicharacteristic through $\alpha_{\varepsilon}$ with the set $\{z_-^2=\frac{1}{4}\varepsilon_0\}\cap V^-$. The family $\alpha_{\varepsilon}'$ necessarily contains a sequence converging to some $\alpha'\in \{z_-^2=\frac{1}{4}\varepsilon_0\} \cap V^-$. By propagation of singularities, $\alpha_{\varepsilon}'\in\wf{s}(u)\cap\Sigma$ for all $\varepsilon$, so $\alpha'\in \wf{s}(u)\cap\Sigma$ since both sets are closed.

We claim that the bicharacteristic $\gamma$ starting at $\alpha'$ cannot exit $V^-$ at any point and therefore limits to $\alpha$. To see this, assume that it exits at some parameter time $t$ after starting at $\alpha'$. Take small $\delta>0$ such that at time $t+\delta$, the bicharacteristic has not yet exited $\hat{V}^-$. Then there exists a neighborhood $U$ of $\gamma(t+\delta)$ contained in the interior of $\hat{V}^-\backslash V^-$ on which $z_-^2$ is bounded from below by a positive constant. Then the preimage of $U$ by the time-$(t+\delta)$ flow is a neighborhood of $\alpha'$, and on images under the flow at all intermediate times $z_-^2$ is no smaller than in $U$; but any neighborhood of $\alpha'$ contains points among the $\alpha_{\varepsilon}$ which are known to flow to arbitrarily small values of $z_-^2$ while staying within $V^-$, so they cannot flow into $U$ without going through lower values of $z_-^2$ on the way.

Thus, the bicharacteristic through $\alpha'$ limits to $\alpha$ while staying within $V^-$, so by propagation of singularities the whole bicharacteristic starting from $\alpha'$ is in $\wf{s}(u)$.
\end{proof}

\subsection{Dynamics near nondegenerate radial sets}
\label{sec:dynamics}
In this section, we prove a few additional results about the flow near radial sets which will be helpful in rigorously establishing the global flow structure. We call a radial point $\alpha$ satisfying the hypotheses of Theorem~\ref{thm:localized-rp-main} a \textit{nondegenerate radial point} for short.

\begin{prop}
Consider a boundary hypersurface $\Gamma_i\in\mathcal{G}(\PP)$ with defining function $\rho_i$ and a nondegenerate radial point $\alpha\in \Gamma_i$. Let $\lambda_i$ be $\rho_i^{-1}\ d\rho_i(\Hp(\alpha))$ as before. If $\pm \lambda_i>0$, then any bicharacteristic limiting to $\alpha$ in the forward($+$)/backward($-$) direction is contained in $\Gamma_i$.
\label{thm:bdry-eval}
\end{prop}

\begin{proof}
If $\lambda_i>0$, then there exists a neighborhood $U$ of $\alpha$ on which $\Hp\rho_i \geqslant \frac{1}{2}\lambda_i \rho_i$. Assuming a bicharacteristic limits to $\alpha$ in the forward direction, it must at some point enter $U$ and stay in it for all later times. Thus, let $\alpha'\in U$ be a point such that the bicharacteristic starting at $\alpha'$ stays in $U$ for all time in the forward direction; if $\alpha'\notin \Gamma_i = \{\rho_i=0\}$, then $\rho_i$ is monotone increasing everywhere along the part of this bicharacteristic downstream from $\alpha'$, so it cannot limit to $\alpha\in \Gamma_i$ in the forward direction. The proof in the $\lambda_i<0$ case is analogous.
\end{proof}

For the proofs of the next two propositions, recall that for any nondegenerate radial point $\alpha$ and $\varepsilon>0$ small enough (depending on $\alpha$), using the coordinates introduced in the previous section, we define special neighborhoods of $\alpha$ by $\tilde{W}_{\varepsilon}^{\pm} = \{|\tilde{p}|<\varepsilon,\ z_-^2<\varepsilon,\ z_+^2<\varepsilon,\ \hat{w}_{\pm}<\varepsilon\}$, where we now introduce a $\pm$ index to distinguish between the constructions used in the proof of Proposition~\ref{thm:localized-rp-estimates} to propagate regularity forward ($-$) or backward ($+$) along the flow, which require different signs for the $D_i$ parameters in Eq.~(\ref{eq:scaling-factor}) defining $\hat{w}$. The whole $\tilde{p}=z_+^2=z_-^2=0$ region in any of these neighborhoods consists of radial points. By construction, bicharacteristics only enter $\tilde{W}_{\varepsilon}^-$ through the $z_-^2=\varepsilon$ face and only leave it through the $z_+^2=\varepsilon$ and $\hat{w}=\varepsilon$ faces; they only enter $\tilde{W}_{\varepsilon}^+$ through the $z_-^2=\varepsilon$ and $\hat{w}=\varepsilon$ faces and only leave it through the $z_+^2=\varepsilon$ face (since $\Hp\hat{w}$ was made to be sign-definite away from the radial set; see Figure~\ref{fig:localization} and the discussion of the localization along $R$ in the proof of Proposition~\ref{thm:localized-rp-estimates}). We will write $W_{\varepsilon}^{\pm}(\alpha)$ to make explicit the dependence on $\alpha$ below.

\begin{prop}
Let $R$ be a compact set of nondegenerate radial points. If a bicharacteristic limits to $R$ in either the forward or backward direction, it must in fact limit to some point $\alpha\in R$ in that direction.
\label{thm:limit-to-pt}
\end{prop} 

Here $R$ is allowed to be part of a larger connected radial set, so it does not have to coincide, for any of its points, with the set we called $R$ in the assumptions of Theorem~\ref{thm:localized-rp-main}.

\begin{proof}
We consider the case of a bicharacteristic $\gamma$ limiting to $R$ in the forward direction; the opposite case is similar. Since $R$ is compact, we can cover $R$ by a finite number of neighborhoods $W_{\varepsilon}^+(\alpha)$ for points $\alpha\in R$. Let $U$ be a small neighborhood of $R$ which is also covered by this finite collection of neighborhoods and which does not intersect any of their $z_+^2=\varepsilon$ faces. If $\gamma$ limits to $R$, it must at some point enter $U$ and stay in it for all later times; then since by construction it cannot exit one of the neighborhoods after this point, there must be at least one of the neighborhoods in which it remains for all later times. Thus, $\gamma$ must limit to the part of $R$ in the closure of that neighborhood. Since $\varepsilon$ can be taken arbitrarily small, this proves that $\gamma$ must limit to a single point of $R$.
\end{proof}

\begin{prop}
Let $\Gamma = \Gamma_1\cap\ldots\cap\Gamma_k$ be a corner of $\partial\PP$, and assume that $R\subset\Gamma$ is a compact set of nondegenerate radial points which has a neighborhood in $\Gamma$ in which there are no other radial points. Assume that there exists $k_0\leqslant k$ such that for every point $\alpha\in R$, we have $\pm \lambda_i>0$ for all $i\leqslant k_0$ and additionally all eigenvalues $\mu_j$ satisfy $\pm \mu_j \geqslant 0$, with $\lambda_i$ and $\mu_j$ defined as before. Then there exists a neighborhood $U$ of $R$ in $\Sigma\cap\Gamma_{k_0+1}\cap\ldots\cap\Gamma_k$ from which every bicharacteristic limits to $R$ in the backward($+$)/forward($-$) direction.
\label{thm:source-sink}
\end{prop} 

\begin{proof}
We consider the case $\lambda_i<0$ for $i\leqslant k_0$, $\mu_j\leqslant 0$; the opposite case is similar. As in the proof of Proposition~\ref{thm:limit-to-pt}, consider a finite cover of $R$ by neighborhoods $\tilde{W}_{\varepsilon}^+(\alpha)$ for points $\alpha\in R$; taking these neighborhoods small enough, we can ensure that all radial points in them belong to $R$. Under the assumptions made, in the notation of the previous section, $z_+^2 = 0$ on $\Gamma_{k_0+1}\cap\ldots\cap\Gamma_k$, so the flow in $\Sigma\cap \Gamma_{k_0+1}\cap\ldots\cap\Gamma_k$ cannot exit any of the $\tilde{W}_{\varepsilon}^+$ neighborhoods. Thus, any bicharacteristic starting in the union of these neighborhoods remains in it for all later times, and in fact remains in the same neighborhood in which it started.

From Eq.~(\ref{eq:Hp-coord}), when $\tilde{p}=z_+^2=0$, we have $\Hp z_-^2 > C z_-^2$ with $C>0$ on $\tilde{W}_{\varepsilon}^+$ for $\varepsilon>0$ small enough. Thus, away from any neighborhood of zero $z_-^2$ is monotone decreasing at a rate bounded from below, so a bicharacteristic which stays in a $\tilde{W}_{\varepsilon}^+$ neighborhood for all time after some point must limit to the part of $R$ in that neighborhood.
\end{proof}

We say that a radial set $R$ as in Proposition~\ref{thm:source-sink} is a \textit{local source($+$)/sink($-$) for the flow in $\Gamma_{k_0+1}\cap\ldots\cap\Gamma_k$}.

\section{Microlocal analysis of the Klein-Gordon operator}
\label{sec:Hamiltonian}
	We now review the properties of $P=\Box_{\g}+m^2$ in the framework of the previous section. The results in this section are essentially due to Sussman \cite[Section 4]{Sussman}, extended to apply to our more general class of metrics and expressed in a more invariant manner. We use different phase-space coordinates but retain the notation of \cite{Sussman} for the names of the radial sets.
	
	Recall that we denote by $\tilde{\Sigma}$ the zero set of the rescaled principal symbol of $P$. This depends on a choice of representative of the principal symbol class (though all choices agree on $\partial\PP$ and therefore have the same zero set $\Sigma$ there). The most convenient (because it is homogeneous in the fibers) representative of the principal symbol of $\Box_{\g}$ is the dual metric (see Section~\ref{sec:wave-op}), which is a nondegenerate quadratic form on the fibers of $\ttm$ of Lorentzian signature. $\tilde{\Sigma}$ is then the closure in $\PP$ of the set of points in $\ttm$ where the dual metric equals $-m^2$ (the \textit{mass shell}).
	
	Basic facts about Lorentzian-signature quadratic forms imply that $\tilde{\Sigma}$ has two connected components in every fiber of $\ttm$, one inside each component of the dual light-cone over that point, and these components limit to the light-cones at fiber infinity. Since we assume the spacetime is time-orientable, this means that $\tilde{\Sigma}$ (and therefore also $\Sigma$) globally over $\M$ splits into two connected components per connected component of $\M$, one of which contains only future-directed momenta and the other only past-directed momenta (for any given choice of time orientation).
	
	As we will see, our non-trapping assumptions essentially determine the structure of the Hamilton flow in $\Sigma$ everywhere except over null infinity.	We now describe the characteristic set and Hamilton flow over null infinity in detail, working in local coordinates $(\rho_0,x_0,y_1,\ldots,y_{d-1})$ on $\U_0$ and $(\rho_T,x_T,y_1,\ldots,y_{d-1})$ on $\U_T$, where $(y_1,\ldots,y_{d-1})$ are local coordinates on $Y$. We drop the subscripts on $\rho_0,x_0,\rho_T,x_T$. We warn the reader that we will also soon introduce functions denoted $\varrho$, which will always be local defining functions of \textit{fiber} infinity.

	\subsection{Principal symbol over null infinity}
	
	 Let $(\zeta,\xi,\eta_1,\ldots,\eta_{d-1})$ be the de,sc-dual variables to $(\rho,x,y_1,\ldots,y_{d-1})$ in either $\U_0$ or $\U_T$. We use the notation $\|\eta\|_{\h}^2=\sum_{i,j=1}^{d-1} h^{ij}\eta_i\eta_j$, where $h^{ij}$ are the matrix elements of the dual metric to $\h$ in coordinates $(y_1,\ldots,y_{d-1})$.
	 
	Over $\U_0$, the de,sc-principal symbol of $P=\Box_{\g}+m^2$ is
	\begin{equation}
	\label{eq:symbol-U0}
		p
		=
		\xi\zeta-\frac{1}{2}\xi^2+\|\eta\|^2_{\h}+m^2+Q_0
		=
		-\frac{1}{2}(\xi-\zeta)^2+\frac{1}{2}\zeta^2+\|\eta\|^2_{\h}+m^2+Q_0,
	\end{equation}
		where $Q_0$ is a quadratic form in $\zeta,\xi,\eta$ with coefficients in $\mathcal{I}_{\scri^-}^{\epsilon}(\M)\cap\mathcal{I}_{\scri^+}^{\epsilon}(\M)$.
		 Over $\U_T$, the symbol is
		\begin{equation}
		\label{eq:symbol-UT}
		p=-\xi\zeta+\frac{1}{2}\xi^2+\|\eta\|_{\h}^2+m^2+Q_T
		=\frac{1}{2}(\xi-\zeta)^2-\frac{1}{2}\zeta^2+\|\eta\|_{\h}^2+m^2+Q_T,		
		\end{equation}
		where $Q_T$ is a quadratic form in $\zeta,\xi,\eta$ with coefficients in $\mathcal{I}_{\scri^-}^{\epsilon}(\M)\cap\mathcal{I}_{\scri^+}^{\epsilon}(\M)$.

	\subsection{Characteristic set over null infinity}

Recall that we write $\Gamma_f$ for the boundary hypersurface of $\PP$ at fiber infinity. We write $\pi$ for the bundle projection $\ttm\to\M$ extended to $\PP$.
	
		\subsubsection{Over $\U_0$}
		From Eq.~(\ref{eq:symbol-U0}), since $Q_0$ vanishes at null infinity, we see that there is $C>0$ such that $p>C(1+\xi^2+\zeta^2+\eta^2)$ on the set $\{x=0,\xi=\zeta\}\subset \pi^{-1}\U_0\backslash\Gamma_f$. Therefore, for any choice of defining functions for rescaling, $\tilde{p}$ does not vanish in a neighborhood of the set $\{x=0,\ \xi=\zeta\}\subset\pi^{-1}\U_0\backslash\Gamma_f$ or its closure in $\PP$. Then we define two regions of $\pi^{-1}\U_{\scri^+_-}$ and, similarly, two regions of $\pi^{-1}\U_{\scri^-_+}$ by
		$$
		\V_{\scri^+_-}^{\pm}
		=
		\overline{
		\{\pm(\xi-\zeta)>\varepsilon\}}\cap\pi^{-1}\U_{\scri^+_-}
		,
		\hspace{30pt}
		\V_{\scri^-_+}^{\pm}
		=
		\overline{
		\{\mp(\xi-\zeta)>\varepsilon\}}\cap \pi^{-1}\U_{\scri^-_+}
		$$
for small $\varepsilon>0$, the closures being taken in $\PP$. For $\varepsilon$ small enough, over a neighborhood of null infinity $\tilde{\Sigma}\cap\pi^{-1}\U_0$ is contained in the union of these four regions.

		To describe the characteristic set at fiber infinity as well, we introduce new fiber coordinates which extend smoothly to fiber infinity in these four phase-space regions:
\begin{equation}
		\varrho=\frac{1}{|\xi-\zeta|},
		\hspace{30pt}
		\omega=\frac{\zeta}{\xi-\zeta},
		\hspace{30pt}
		\theta=\frac{\eta}{\xi-\zeta}.
\label{eq:coord-phase-U0}
\end{equation}
		In all four regions, $\varrho\geqslant 0$ is a defining function of fiber infinity. In terms of these coordinates, over $\U_0$
		\begin{equation}
		p=\frac{1}{2\varrho^2}(\omega^2+2\|\theta\|^2_{\h}+2m^2\varrho^2-1 + \tilde{Q}_0),
		\label{eq:char-set}
		\end{equation}
		where $\tilde{Q}_0 \in \mathcal{I}_{\scri^-}^{\epsilon}(\PP) \cap \mathcal{I}_{\scri^+}^{\epsilon}(\PP)$.
		
		\subsubsection{Over $\U_T$}
		From Eq.~(\ref{eq:symbol-UT}), since $Q_T$ vanishes at null infinity, we see that there is $C>0$ such that $p>C(1+\xi^2+\zeta^2+\eta^2)$ on the set $\{x=0,\zeta=0\}\subset \pi^{-1}\U_T\backslash\Gamma_f$. Therefore, $\tilde{p}$ does not vanish in a neighborhood of the set $\{x=0,\ \zeta=0\}\subset\pi^{-1}\U_T\backslash\Gamma_f$ or its closure in $\PP$. Then we define two regions of $\pi^{-1}\U_{\scri^+_+}$ and, similarly, two regions of $\pi^{-1}\U_{\scri^-_-}$ by
		 $$
		 \V_{\scri^+_+}^{\pm}=\overline{\{\pm\zeta>\varepsilon\}}\cap\pi^{-1}\U_{\scri^+_+},
		 \hspace{30pt}
		 \V_{\scri^-_-}^{\pm}=\overline{\{\mp\zeta>\varepsilon\}}\cap\pi^{-1}\U_{\scri^-_-}.
		 $$
		For $\varepsilon>0$ small enough, over a neighborhood of null infinity $\tilde{\Sigma}\cap\pi^{-1}\U_T$ is contained in the union of these four regions.
		
		Again, we introduce new fiber coordinates which extend smoothly to fiber infinity in these four phase-space regions:
		\begin{equation}
		\varrho=\frac{1}{|\zeta|},
		\hspace{30pt}
		\omega=\frac{\xi-\zeta}{\zeta},
		\hspace{30pt}
		\theta=\frac{\eta}{\zeta}.
		\label{eq:coord-phase-UT}
		\end{equation}
		In all four regions, $\varrho\geqslant 0$ is again a defining function of fiber infinity. In terms of these coordinates, over null infinity in $\U_T$ the symbol again has the expression 		\begin{equation}
		p=\frac{1}{2\varrho^2}(\omega^2+2\|\theta\|^2_{\h}+2m^2\varrho^2-1 + \tilde{Q}_T),
		\end{equation}
		where $\tilde{Q}_T \in \mathcal{I}_{\scri^-}^{\epsilon}(\PP)\cap \mathcal{I}_{\scri_+}^{\epsilon}(\PP)$.
		
	\subsection{Hamilton vector field over null infinity}
	
			Starting from Eqs.~(\ref{eq:Hp-coordinates-general}), (\ref{eq:symbol-U0}) and calculating the derivatives, we find that over $\U_0$ at finite frequencies,
			\begin{equation}
			\begin{split}
			H_p
			=
			\rho x\Bigg(
			\Big(
			(\xi-\zeta)\xi
			-4\|\eta\|_{\h}^2
			+Q_1^0
			\Big)\frac{\partial}{\partial\xi}
			+
			\Big(
			(\zeta-\xi)^2
			-2\|\eta\|_{\h}^2
			+Q_2^0
			\Big)\frac{\partial}{\partial\zeta}
			+\\
			+
			\Big(
			\zeta-\xi
			+ L_1^0
			\Big)x\frac{\partial}{\partial x}
			+
			\Big( \xi + L_2^0 \Big)\rho\frac{\partial}{\partial\rho}
			+
			\sum_{i,j=1}^{d-1}
			\Big( 2h^{ij}\eta_j +L_{3,i}^0 \Big) 			
			x\frac{\partial}{\partial y_i}
			+
			\sum_{i=1}^{d-1}
			\Big((2\zeta-\xi)\eta_i + Q_{3,i}^0 \Big) \frac{\partial}{\partial\eta_i}
			\Bigg),
			\end{split}		
			\label{eq:Hp-U0}
			\end{equation}
where $Q_1^0, Q_2^0, Q_{3,i}^0$ are quadratic forms and $L_1^0, L_2^0, L_{3,i}^0$ linear forms in $(\zeta,\xi,\eta)$ with coefficients in $\mathcal{I}_{\scri^-}^{\epsilon}(\M) \cap \mathcal{I}_{\scri^+}^{\epsilon}(\M)$.

			Similarly, over $\U_T$ at finite frequencies,		
			\begin{equation}
			\begin{split}
			H_p
			=
			\rho x\Bigg(
			-\Big(
			(\xi-\zeta)\xi
			+4\|\eta\|_{\h}^2
			+Q_1^T
			\Big)
			\frac{\partial}{\partial\xi}
			-
			\Big(
			(\xi-\zeta)^2
			+2\|\eta\|_{\h}^2
			+Q_2^T
			\Big)
			\frac{\partial}{\partial\zeta}
			+\\
			+
			\Big(\xi-\zeta + L_1^T\Big)
			x\frac{\partial}{\partial x}
			-
			\Big(\xi + L_2^T\Big) \rho\frac{\partial}{\partial\rho}
			+
			\sum_{i,j=1}^{d-1}
			\Big( 2h^{ij}\eta_j + L_{3,i}^T \Big)
			x\frac{\partial}{\partial y_i}
			+
			\sum_{i=1}^{d-1}
			\Big(
			(\xi-2\zeta)\eta_i
			+ Q_{3,i}^T
			\Big)
			\frac{\partial}{\partial\eta_i}
			\Bigg),
			\end{split}
			\label{eq:Hp-UT}
			\end{equation}			
			where $Q_1^T, Q_2^T, Q_{3,i}^T$ are quadratic forms and $L_1^T, L_2^T, L_{3,i}^T$ linear forms in $(\zeta,\xi,\eta)$ with coefficients in $\mathcal{I}_{\scri^-}^{\epsilon}(\M) \cap \mathcal{I}_{\scri^+}^{\epsilon}(\M)$.
			
			 Now we rewrite this in terms of the coordinates in Eqs.~(\ref{eq:coord-phase-U0}), (\ref{eq:coord-phase-UT}) valid on the compactification up to fiber infinity. On $\V_{\scri_+^-}^{\mp}$ and on $\V_{\scri_-^+}^{\pm}$, the result is
			\begin{equation}
			\begin{split}
			H_p
			=
			\pm\varrho^{-1}\rho x\Bigg(
			-x\frac{\partial}{\partial x}
			+
			\Big(
			\omega+1
			\Big)\rho\frac{\partial}{\partial\rho}
			-
			\Big(
			\omega-2\|\theta\|_{\h}^2
			\Big)
			\varrho\frac{\partial}{\partial\varrho}
			-\\
			-\Big(\omega-1\Big)\Big(\omega+1-2\|\theta\|_{\h}^2			\Big)
			\frac{\partial}{\partial\omega}
			+
			\Big(
			2\|\theta\|_{\h}^2-1 \Big)\sum_{i=1}^{d-1} \theta_i
			\frac{\partial}{\partial\theta_i}
			\Bigg)
			\mod \varrho^{-1}\rho x(\mathcal{I}_{\scri^-}^{\epsilon}\cap\mathcal{I}_{\scri^+}^{\epsilon})\Vb(\PP).
			\end{split}
			\label{eq:Hp-spacelike-corner}
			\end{equation}
			 On $\V_{\scri^-_-}^{\mp}$ and on $\V_{\scri^+_+}^{\pm}$, the result is
			\begin{equation}
			\begin{split}
			H_p
			=\pm
			\varrho^{-1}\rho x\Bigg(
			\omega x\frac{\partial}{\partial x}
			-\Big(\omega+1\Big)\rho\frac{\partial}{\partial\rho}
			+\Big(\omega^2+2\|\theta\|_{\h}^2\Big)\varrho\frac{\partial}{\partial\varrho}
			+
			\Big(\omega-1\Big)\Big(\omega(\omega+1)+2\|\theta\|_{\h}^2\Big)\frac{\partial}{\partial\omega}
			+\\
			+\Big(\omega^2+\omega-1+2\|\theta\|_{\h}^2\Big) \sum_{i=1}^{d-1} \theta_i\frac{\partial}{\partial\theta_i}\Bigg)
			\mod \varrho^{-1}\rho x(\mathcal{I}_{\scri^-}^{\epsilon}\cap\mathcal{I}_{\scri^+}^{\epsilon})\Vb(\PP).
			\end{split}
			\label{eq:Hp-timelike-corner}
			\end{equation}
			
Since in $\Sigma\cap \scri$ we have $|\omega|<1$ in any of these coordinate charts, by checking that the signs of the $\rho\frac{\partial}{\partial\rho}$ terms agree on the overlap of charts we can see that the overlapping pairs are $(\V_{\scri^{\pm}_{\pm}}^+,\V_{\scri^{\pm}_{\mp}}^+)$ and $(\V_{\scri^{\pm}_{\pm}}^-,\V_{\scri^{\pm}_{\mp}}^-)$. Moreover, from the signs of the $x\frac{\partial}{\partial x}$ terms over spacelike infinity near the corners with null infinity, we see that the non-trapping assumption at spacelike infinity requires that bicharacteristics in the component of $\Sigma$ in $\V_{\scri^-_+}^{\pm}$ flow to $\V_{\scri^+_-}^{\pm}$. Thus, the four charts with the $+$ superscript cover (near null infinity) one of the components of $\Sigma$, which we call $\Sigma^+$, and those with the $-$ superscript cover the other component, which we call $\Sigma^-$. We identify these as the positive- and negative-frequency components respectively (which could instead be fixed by a choice of time orientation and a convention for whether future- or past-directed covectors represent positive frequencies; see footnote~\ref{note:signs}).
			
We see that
	\begin{itemize}
	\item Near null infinity in $\V_{\scri^-_-}^{\pm} \cup \V_{\scri^-_+}^{\pm}$, the flow has the same structure as in $\V_{\scri^+_-}^{\pm} \cup \V_{\scri^+_+}^{\pm}$, but with reversed direction (relative to the corresponding timelike infinity face) and consequently reversed stability properties of any radial points.
	
	\item Near null infinity in $\V_{\scri^{\pm}_-}^+ \cup \V_{\scri^{\pm}_+}^+$, the flow has the same structure as in $\V_{\scri^{\pm}_-}^- \cup \V_{\scri^{\pm}_+}^-$, but with reversed direction (relative to the corresponding timelike infinity face) and consequently reversed stability properties of any radial points.
	
	\end{itemize}
	
	Therefore, for simplicity and clarity we can restrict attention to $\V_{\scri^+_+}^+$ and $\V_{\scri^+_-}^+$, i.e. the flow in the positive-frequency component of the characteristic set near future null infinity.
			
	\subsection{Radial sets}
We now find the radial sets of $\Hp$ in the characteristic set over null infinity. When defining $\Hp$, in this section we rescale $H_p$ using $\varrho^{-1}\rho_0 x_0$ over $\U_0$ and using $\varrho^{-1}\rho_T x_T$ over $\U_T$.
	
	\subsubsection{Over the corner of null and spacelike infinity}
	In $\V^+_{\scri^+_-}$, restricted to $\{x_0=0\}$ and the characteristic set $\{2\|\theta\|_{\h}^2=1-\omega^2-2m^2\varrho^2\}$, we have
	\begin{equation}
	\begin{split}
	\Hp=
			\Big(\omega+1\Big)\rho_0\frac{\partial}{\partial\rho_0}
			-
			\Big(\omega^2+\omega+2m^2\varrho^2-1\Big)\varrho\frac{\partial}{\partial\varrho}
			-
			\Big(\omega-1\Big)\Big(\omega^2+\omega+2m^2\varrho^2\Big)\frac{\partial}{\partial\omega}
			-
			\\
			-
			\Big(\omega^2+2m^2\varrho^2\Big)\sum_{i=1}^{d-1} \theta_i\frac{\partial}{\partial\theta_i}.
	\end{split}	
	\end{equation}
	\begin{itemize}
	\item Away from fiber infinity ($\{\varrho=0\}$), for the last term to vanish we need $\theta=0$, i.e. $2m^2\varrho^2=1-\omega^2$. This yields
	$$\Hp=
			\Big(\omega+1\Big)\rho_0\frac{\partial}{\partial\rho_0}
			-
			\omega\varrho\frac{\partial}{\partial\varrho}
			-
			\Big(\omega^2-1\Big)\frac{\partial}{\partial\omega},$$
			which never vanishes for $\varrho>0$. So all of the radial points are at fiber infinity, where we have
			$$\Hp=
			\Big(\omega+1\Big)\rho_0\frac{\partial}{\partial\rho_0}
			-\omega\Big(\omega-1\Big)\Big(\omega+1\Big)\frac{\partial}{\partial\omega}
			-\omega^2 \sum_{i=1}^{d-1} \theta_i\frac{\partial}{\partial\theta_i}.$$
	\item Away from spacelike infinity ($\{\rho_0=0\}$), for radial points we need $\omega=-1$. This yields $\Hp=-\sum_{i=1}^{d-1} \theta_i\frac{\partial}{\partial\theta_i}$, so we see that the radial points are $\{\varrho=0,\ \omega=-1,\ \theta=0,\ \rho_0,y\ \text{arbitrary}\}$.

	\item At the corner of spacelike infinity and fiber infinity, we have
	$$\Hp=
			-\omega\Big(\omega-1\Big)\Big(\omega+1\Big)\frac{\partial}{\partial\omega}
			-\omega^2 \sum_{i=1}^{d-1}\theta_i\frac{\partial}{\partial\theta_i}.$$
		We can read off the radial points:
		\begin{itemize}
		\item $\{\varrho=0,\ \omega=-1,\ \theta=0,\ \rho_0=0,\ y=\text{any}\}$, which is the limiting set at fiber infinity of the first radial set found above;
		\item $\{\varrho=0,\ \omega=1,\ \theta=0,\ \rho_0=0,\ y=\text{any}\}$;
		\item $\{\varrho=0,\ \omega=0,\ \|\theta\|_{\h}^2=\frac{1}{2},\ \rho_0=0,\ y=\text{any}\}$.
		\end{itemize}
	\end{itemize}

\subsubsection{Over the corner of null and timelike infinity}
	In $\V^+_{\scri^+_+}$, restricted to $\{x_T=0\}$ and the characteristic set $\{2\|\theta\|_{\h}^2=1-\omega^2-2m^2\varrho^2\}$, we have
	\begin{equation}
	\Hp
			=
			-\Big(\omega+1\Big)\rho_T\frac{\partial}{\partial\rho_T}
			+\Big(1-2m^2\varrho^2\Big)\varrho\frac{\partial}{\partial\varrho}
			+
			\Big(\omega-1\Big)\Big(\omega+1-2m^2\varrho^2\Big)\frac{\partial}{\partial\omega}
			+\Big(\omega-2m^2\varrho^2\Big)  \sum_{i=1}^{d-1}\theta_i\frac{\partial}{\partial\theta_i}.
	\end{equation}
	\begin{itemize}
	\item Away from fiber infinity, for the second term to vanish we need $\varrho=\frac{1}{\sqrt{2}m}$, which on the characteristic set also implies $\omega=0$ and $\theta=0$. This yields $\Hp=-\rho_T\frac{\partial}{\partial\rho_T}$, so we see that the radial points are $\{\varrho=\frac{1}{\sqrt{2}m},\ \omega=0,\ \theta=0,\ \rho_T=0,\ y=\text{any}\}$. Since this is a radial set in $I^T\backslash\Gamma_f$, it must be part of the radial set $\mathcal{R}$ over timelike infinity whose existence is assumed in the causal structure/non-trapping assumptions on $\g$.
	
	\item Now we restrict to fiber infinity, where
	$$\Hp
			=
			-\Big(\omega+1\Big)\rho_T\frac{\partial}{\partial\rho_T}
			+
			\Big(\omega^2-1\Big)\frac{\partial}{\partial\omega}
			+\omega \sum_{i=1}^{d-1} \theta_i\frac{\partial}{\partial\theta_i}.$$
			We can read off the radial points:
			\begin{itemize}
			\item $\{\varrho=0,\ \omega=1,\ \theta=0,\ \rho_T=0,\ y=\text{any}\}$;
			\item $\{\varrho=0,\ \omega=-1,\ \theta=0,\ \rho_T=\text{any},\ y=\text{any}\}$.
			\end{itemize}
	\end{itemize}
	
	\subsubsection{Complete list of radial points}
	In $\V_{\scri^+_-}^+$, the radial sets are (using Sussman's notation)
	\begin{itemize}
		\item $\mathcal{K}_+^+=\{x_0=0,\ \rho_0=0,\ y=\text{any},\ \varrho=0,\ \omega=1,\ \theta=0\}$.
		\item $\mathcal{A}_+^+=\{x_0=0,\ \rho_0=0,\ y=\text{any},\ \varrho=0,\ \omega=0,\ \|\theta\|_{\h}^2=\frac{1}{2}\}$.
		\item $\mathcal{N}_+^+\cap \V_{\scri^+_-}^+=\{x_0=0,\ \rho_0=\text{any},\ y=\text{any},\ \varrho=0,\ \omega=-1,\ \theta=0\}$.
	\end{itemize}
	In $\V_{\scri^+_+}^+$, the radial sets are
	\begin{itemize}
		\item $\mathcal{R}_+^+\cap\scri^+_+=\{x_T=0,\ \rho_T=0,\ y=\text{any},\ \varrho=\frac{1}{\sqrt{2}m},\ \omega=0,\ \theta=0\}$.
		\item $\mathcal{C}_+^+=\{x_T=0,\ \rho_T=0,\ y=\text{any},\ \varrho=0,\ \omega=1,\ \theta=0\}$.
		\item $\mathcal{N}_+^+\cap \V_{\scri^+_+}^+=\{x_T=0,\ \rho_T=\text{any},\ y=\text{any},\ \varrho=0,\ \omega=-1,\ \theta=0\}$.
	\end{itemize}
The extended radial set $\mathcal{N}^+_+$ can also be described as the set of limit points in $\Sigma^+\cap\scri^+$ of the set of points of $\ttm|_{\scri^+}$ where $\xi=\eta=0$, i.e. of the span of $\frac{d\rho}{\rho^2 x}$, in any of the coordinates we used (``momenta strictly along the fibers of null infinity"), or indeed as the image, under the metric, of multiples of $x^2\rho\partial_x$, i.e.\ it corresponds to waves propagating normally into null infinity (under the Hamilton flow). Over the interior of $\scri^+$, this span can be defined without reference to local coordinates as the image of the 0-cotangent bundle (one-forms dual to smooth vector fields vanishing on $\partial\M$) under its natural inclusion into the de,sc-cotangent bundle (since $\frac{d x}{\rho x^2}$ and $\frac{dy_i}{\rho x^2}$ are too singular at $\scri^+$ to define 0-one-forms).
	
	Due to the symmetries noted above, there are also radial sets analogous to the first three in each of $\V_{\scri^-_+}^-$, $\V_{\scri^-_+}^+$, and $\V_{\scri^+_-}^-$ and radial sets analogous to the last three in each of $\V_{\scri^-_-}^-$, $\V_{\scri^-_-}^+$, and $\V_{\scri^+_+}^-$. They are denoted using the same letters but different subscripts and superscripts $\pm$: the superscript indicates the component of the characteristic set, whereas the subscript indicates past or future null infinity. By our non-trapping assumptions, there are no radial points elsewhere in $\partial\PP$ except a radial set $\mathcal{R}$ contained in $I^T\backslash \Gamma_f$, whose components in $\Sigma^{\alpha}\cap I^{\beta}$ for $\alpha,\beta\in \{+,-\}$ we also denote by $R^{\alpha}_{\beta}$.

	\subsection{Linearization of the flow at radial points}
	We now compute, starting from Eqs.~(\ref{eq:Hp-spacelike-corner}), (\ref{eq:Hp-timelike-corner}), the linearizations of $\Hp$ restricted to the relevant corners at all of the radial points found above and find their eigenvectors and eigenvalues, along with the ``eigenvalues" in the directions transverse to the corner (i.e. the coefficients of the transverse components of $\Hp$ at the radial point as a b-vector, which are true eigenvalues if the symbol is classical). We denote $\lambda_{\Gamma_i}$ the coefficient of the component transverse to hypersurface $\Gamma_i$. In expressions for the linearization, we abuse notation by identifying $L(\Hp|_{\Gamma})(\alpha):T_{\alpha}\Gamma\to T_{\alpha}\Gamma$ with a linear vector field on $\R^{\dim\Gamma}$. However, we still consider the eigenvectors as vectors in $T_{\alpha}\Gamma\subset T_{\alpha}\PP$.
	
	We note that the eigenvalues do depend on the choice of boundary-defining functions for rescaling $H_p$, and if $p_1\neq 0$ we need to take care to use the same rescaling for both in order to get the correct inequalities in Theorem~\ref{thm:localized-rp-main}. However, since our operator $P$ is symmetric (hence $p_1=0$) and a choice of rescaling only affects the eigenvalues by multiplying them all by the same factor, the inequalities obtained do not depend on this choice. We will continue to use the same choices of rescaling used in the preceding calculations.
	
		\subsubsection{Near the set $\mathcal{K}^+_+$}
	Let $\alpha=(x_0=0,\rho_0=0,y=y^{\alpha},\varrho=0,\omega=1,\theta=0)\in\mathcal{K}_+^+ \subset \scri^+_- \cap\Gamma_f$. Then as a b-vector,
	$$
\Hp(\alpha) = -x_0\frac{\partial}{\partial x_0} + 2\rho_0\frac{\partial}{\partial\rho_0} - \varrho\frac{\partial}{\partial\varrho},
	$$
	so the relevant coefficients are $ \lambda_{I^0}=2$, $\lambda_{\scri^+}=-1$, $\lambda_{\Gamma_f}=-1$. Meanwhile,
	$$
	L(\Hp|_{\scri^+_+\cap\Gamma_f})(\alpha)
	=
	-2(\omega-1)\frac{\partial}{\partial\omega}
	- \sum_{i=1}^{d-1} \theta_i\frac{\partial}{\partial\theta_i}.
	$$
	The eigenvector-eigenvalue pairs are $ \left( \frac{\partial}{\partial \omega},-2\right)$, $ \left(\frac{\partial}{\partial\theta_i},-1\right)$, $\left(\frac{\partial}{\partial y_i},0\right)$. The vector $\frac{\partial}{\partial\omega}$ is transverse to $\Sigma\cap \scri^+_-\cap\Gamma_f$.
	
	We see that $\mathcal{K}^+_+$ is a local sink for the flow in $I^0$, any bicharacteristic limiting to it in the forward direction is contained in $I^0$, and any bicharacteristic limiting to it in the backward direction is contained in $\scri^+\cap\Gamma_f$.
	
	\subsubsection{Near the set $\mathcal{A}^+_+$}
	Let $\alpha=(x_0=0,\rho_0=0,y=y^{\alpha},\varrho=0,\omega=0,\theta=\theta^{\alpha}\in \frac{1}{\sqrt{2}}S^{d-1}_h)\in\mathcal{A}_+^+ \subset \scri^+_-\cap\Gamma_f$. Then as a b-vector,
	$$
\Hp(\alpha) = -x_0\frac{\partial}{\partial x_0} + \rho_0\frac{\partial}{\partial\rho_0} + \varrho\frac{\partial}{\partial\varrho},	
	$$
	so the relevant coefficients are $ \lambda_{\scri^+} = -1$, $ \lambda_{I^0}=1,\ \lambda_{\Gamma_f}=1$. 
	
	Because $\mathcal{A}$ is not located at $\theta=0$, it is convenient to change coordinates before considering the linearization. Without loss of generality, let us assume that $\theta^{\alpha}_{d-1}\neq 0$; then we can replace the coordinates $(\theta_1,\ldots,\theta_{d-1})$ with $(\frac{\theta_1}{\theta_{d-1}},\ldots,\frac{\theta_{d-2}}{\theta_{d-1}},\|\theta\|_{\h}^2)$, so $\sum_{i=1}^{d-1}\theta_i \frac{\partial}{\partial\theta_i} = 2\|\theta\|_{\h}^2\frac{\partial}{\partial\|\theta\|_{\h}^2}$. Then
	
	$$
	L(\Hp|_{\scri^+_-\cap\Gamma_f})(\alpha)
	=
	\left(\omega-2\left(\|\theta\|_{\h}^2-\frac{1}{2}\right)\right)\frac{\partial}{\partial\omega}
	+
	2\left(\|\theta\|_{\h}^2-\frac{1}{2}\right)\frac{\partial}{\partial \|\theta\|_{\h}^2}.
	$$
	The eigenvector-eigenvalue pairs are $ \left(\frac{\partial}{\partial\omega}, 1\right)$, $\left(\frac{\partial}{\partial\|\theta\|_{\h}^2}-2\frac{\partial}{\partial\omega},2\right)$, $\left(\frac{\partial}{\partial y_i},0\right)$, $\left(\frac{\partial}{\partial(\theta_i/\theta_{d-1})},0\right)$. The vector $\frac{\partial}{\partial\|\theta\|_{\h}^2}-2\frac{\partial}{\partial\omega}$ is transverse to $\Sigma\cap\scri^+_-\cap\Gamma_f$.
	
	We see that $\mathcal{A}^+_+$ is a local source for the flow in $\scri^+$, any bicharacteristic limiting to it in the forward direction is contained in $I^0\cap\Gamma_f$, and any bicharacteristic limiting to it in the backward direction is contained in $\scri^+$.
	
		\subsubsection{Near the set $\mathcal{N}_+^+$}
	Let $\alpha=(x_T=0,\rho_T=\rho_T^{\alpha},y=y^{\alpha},\varrho=0,\omega=-1,\theta=0)\in\mathcal{N}_+^+\cap\V^+_{\scri^+_+} \subset \scri^+\cap\Gamma_f$. Then as a b-vector,
	$$
\Hp(\alpha) = -x_T\frac{\partial}{\partial x_T} + \varrho\frac{\partial}{\partial \varrho},
	$$
	so the relevant coefficients are $\lambda_{\scri^+}=-1$, $\lambda_{\Gamma_f}=1$, and if $\alpha$ is at the corner with $I^0$, i.e. $\rho_T^{\alpha}=0$, then also $\lambda_{I^+}=0$. Meanwhile,
	$$
	L(\Hp|_{\scri^+\cap\Gamma_f})(\alpha)
	=
	-\rho_T^{\alpha}(\omega+1)\frac{\partial}{\partial\rho_T}
	+2(\omega+1)\frac{\partial}{\partial\omega}
	-\sum_{i=1}^d \theta_i\frac{\partial}{\partial\theta_i}.
	$$
	The eigenvector-eigenvalue pairs are $ \left( \frac{\partial}{\partial\omega}-\frac{1}{2}\rho_T^{\alpha}\frac{\partial}{\partial\rho_T}, 2\right)$, $ \left(\frac{\partial}{\partial\theta_i}, -1\right)$, $\left(\frac{\partial}{\partial\rho_T},0\right)$, $ \left(\frac{\partial}{\partial y_i},0\right)$. The vector $\frac{\partial}{\partial\omega}-\frac{1}{2}\rho_T^{\alpha}\frac{\partial}{\partial\rho_T}$ is transverse to $\Sigma\cap\scri^+\cap\Gamma_f$.

	The analysis near a point $\alpha=(x_0=0,\rho_0=\rho_0^{\alpha},y=y^{\alpha},\varrho=0,\omega=-1,\theta=0)\in\mathcal{N}^+_+\cap\V_{\scri^+_-}^+$ goes analogously and gives the same results, with $I^+$ and $\rho_T$ replaced by $I^0$ and $-\rho_0$ respectively.
	
	We see that $\mathcal{N}^+_+$ is a local sink for the flow in $\Gamma_f$, any bicharacteristic limiting to it in the forward direction is contained in $\Gamma_f$, and any bicharacteristic limiting to it in the backward direction is contained in $\scri^+$.

	\subsubsection{Near the set $\mathcal{C}^+_+$}
	Let $\alpha=(x_T=0,\rho_T=0,y=y^{\alpha},\varrho=0,\omega=1,\theta=0)\in\mathcal{C}_+^+\subset\scri^+_+\cap\Gamma_f$. Then as a b-vector,
	$$\Hp(\alpha) = x_T\frac{\partial}{\partial x_T} - 2\rho_T\frac{\partial}{\partial\rho_T} + \varrho\frac{\partial}{\partial\varrho},$$
	so the relevant coefficients are $\lambda_{\scri^+}=1$, $\lambda_{\Gamma_f}=1$, and $\lambda_{I^+}=-2$. Meanwhile,
	$$
	L(\Hp|_{\scri^+_+\cap\Gamma_f})(\alpha)
	=
	2(\omega-1)\frac{\partial}{\partial\omega}
	+ \sum_{i=1}^{d-1} \theta_i\frac{\partial}{\partial\theta_i}.
	$$
	The eigenvector-eigenvalue pairs are $\left(\frac{\partial}{\partial\omega},2\right)$, $\left(\frac{\partial}{\partial\theta_i},1\right)$, $\left(\frac{\partial}{\partial y_i},0\right)$. The vector $\frac{\partial}{\partial\omega}$ is transverse to $\Sigma\cap\scri^+_+\cap\Gamma_f$.
	
	We see that $\mathcal{C}^+_+$ is a local source for the flow in $I^+$, any bicharacteristic limiting to it in the forward direction is contained in $\scri^+\cap\Gamma_f$, and any bicharacteristic limiting to it in the backward direction is contained in $I^+$.

	\subsubsection{Near the set $\mathcal{R}_+^+$}
	Let $\alpha=(x_T=0,\rho_T=0,y=y^{\alpha},\varrho=\frac{1}{\sqrt{2}m},\omega=0,\theta=0)\in\mathcal{R}_+^+\cap\scri^+ \subset \scri^+_+$. Then as a b-vector, $\Hp(\alpha) = -\rho_T\frac{\partial}{\partial\rho_T}$, so the relevant coefficients are $\lambda_{I^+}=-1$ and $\lambda_{\scri^+}=0$. Meanwhile,
	$$
	L(\Hp|_{\scri^+_+})(\alpha)
	=
	-\omega\frac{\partial}{\partial\omega}
	-\sum_{i=1}^{d-1} \theta_i\frac{\partial}{\partial\theta_i}.
	$$
	The eigenvector-eigenvalue pairs are $\left(\frac{\partial}{\partial\theta_i},-1\right)$, $\left(\frac{\partial}{\partial\omega},-1\right)$, $\left(\frac{\partial}{\partial\varrho},0\right)$, $\left(\frac{\partial}{\partial y_i},0\right)$. The vector $\frac{\partial}{\partial\varrho}$ is transverse to $\Sigma\cap\scri^+_+$, so despite its zero eigenvalue the radial point is nondegenerate.
	
	The postulated radial set $\mathcal{R}^+_+$ is entirely located in $I^+\backslash \Gamma_f$ (at finite frequency). Since every point in $\mathcal{R}^+_+\backslash \scri^+$ is contained in $I^+$ and no other boundary faces, the sum $\sum_{i=1}^k (-2s_i+m_i-1)\lambda_i$ in the radial point estimate at any such point only contains one term, with $\lambda_{I^+}$ as an overall factor. Consequently, once we know that Theorem~\ref{thm:localized-rp-main} is applicable, further details of the linearization are not important because the applicable estimate is determined just by the sign of $\lambda_{I^+}$, that is whether we are considering a source or a sink of the flow; in the case of $\mathcal{R}^+_+$, the non-trapping assumptions imply that it is a sink. The result is the typical below- and above-threshold estimates with threshold value $s_{I^+}=-\frac{1}{2}$ (since $m_{I^+}=0$ for $P$).

	\subsection{Global structure of the flow}
	At null infinity, the flow is known exactly and is the same as computed by Sussman \cite{Sussman}. It is pictured in Figure~\ref{fig:Hamilton-flow}.

\begin{figure}
\begin{center}
\includegraphics{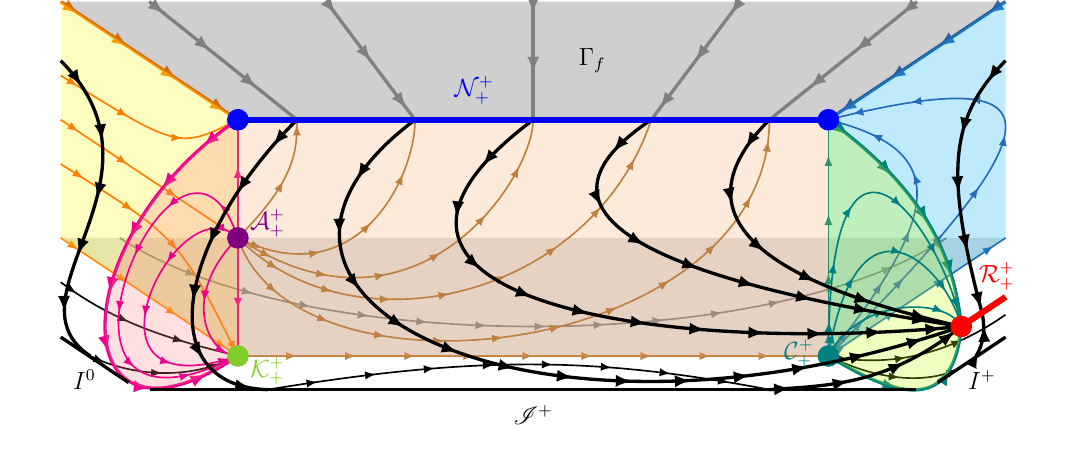}
\end{center}
\caption{Illustration of the flow of $\Hp$ in $\Sigma^+$ over a neighborhood of $\scri^+$. Corners of $\PP$ are color-coded: pink is $\scri^+_-$, green is $\scri^+_+$, brown is $\scri^+\cap\Gamma_f$, yellow is $I^0\cap\Gamma_f$, light blue is $I^+\cap\Gamma_f$. The four three-dimensional regions separated by the corners are the parts of $\Sigma^+$ over $\scri^+$, $I^0$, $I^+$, and at fiber infinity $\Gamma_f$. Only the characteristic set is shown; the $y$ and $\theta$ variables have been suppressed, which does not lose information on $\Sigma^+\cap\scri^+$, where the flow is neutral in $y$ and radial in $\theta$ and $\|\theta\|_{\h}^2$ is a function of the other coordinates by Eq.~(\ref{eq:char-set}). Off $\scri^+$, where this is no longer the case for general metrics of our class, the figure is only an indication of the flow structure (faithful at $I^0$ and/or $I^+$ if the metric is asymptotically Minkowski there). The curved surface in front (not shaded) is the set $\{\theta=0\}$ in $\Sigma^+\cap\scri^+$. To avoid overloading the figure, bicharacteristics are only shown on the corners and the front, top, and bottom surfaces, which is enough to reconstruct the flow structure at and near $\scri^+$. For more figures illustrating various parts of the flow, see \cite{Sussman,HV-eb,JMS}. Color scheme for radial sets is borrowed from \cite{Sussman}.}
\label{fig:Hamilton-flow}
\end{figure}
	
	The information provided by the flow over null infinity, the non-trapping assumptions, and the nonzero eigenvalues at the radial sets over timelike infinity would be enough for our purposes. Nevertheless, we go into a bit more detail to show that every bicharacteristic over $\M^{\circ}$, $(I^0)^{\circ}$, or $(I^{\pm})^{\circ}$ (in addition to $\scri^{\pm}$, where one can check this directly) actually limits to radial points in both directions and provide a more complete description of the global flow structure (in particular justifying the depiction of the flow in Figure~\ref{fig:Hamilton-flow} off $\scri$).

	Consider a bicharacteristic $\gamma$ in $\Sigma^+$ over $(I^0)^{\circ}$ which limits to $\scri^+$ in the forward direction, as stipulated in the non-trapping assumptions. Since $\Sigma^+ \cap \scri^+_-$ is a compact set, $\gamma$ must have a limit point in it for the forward flow. We know that every point of $\Sigma^+ \cap \scri^+_-$ except those in $\mathcal{A}^+_+$, $\mathcal{N}^+_+ \cap I^0$, or the bicharacteristics going from the former to the latter belongs to a bicharacteristic $\gamma'$ which limits in the forward direction to $\mathcal{K}^+_+$. For any neighborhood of $\mathcal{K}^+_+$, any point on such $\gamma'$ has a neighborhood which is taken by the flow to the chosen neighborhood of $\mathcal{K}^+_+$. Since $\mathcal{K}^+_+$ is a local sink for the flow over $I^0$, this means that points of $\gamma'$ cannot be limit points of bicharacteristics over $I^0$ because any bicharacteristic which gets close enough to the point must limit to $\mathcal{K}^+_+$. Thus we conclude that $\gamma$ must limit either to $\mathcal{K}^+_+$ or to the union of $\mathcal{A}^+_+$, $\mathcal{N}^+_+\cap I^0$, and the bicharacteristics connecting them (since it is impossible to have limit points in both of these sets without also having limit points elsewhere).
	
	If $\gamma$ is at finite frequency, the latter option is not possible because the coefficient of the $\varrho \frac{\partial}{\partial\varrho}$ component of $\Hp$ as a b-vector field is positive on that set, which prevents such limiting behavior (cf. Proposition~\ref{thm:bdry-eval}). Thus, any finite-frequency bicharacteristics in $\Sigma^+\cap (I^0)^{\circ}$ limit to $\mathcal{K}^+_+$ in the forward direction. At infinite frequency, on the other hand, $\mathcal{N}^+_+ \cap I^0$ is a local sink for the flow in $I^0\cap\Gamma_f$; therefore, by the same argument any infinite-frequency $\gamma$ which does not limit to $\mathcal{K}^+_+$ must in fact limit to either $\mathcal{A}^+_+$ or $\mathcal{N}^+_+$. Thus we conclude that $\gamma$ must limit to a radial point.

An analogous argument (simpler because there are only two radial sets, one of which is a local sink and the other a local source for the flow in $I^+\cap\Gamma_f$) shows that any infinite-frequency bicharacteristic in $\Sigma^+$ over $I^+$ limits to $\mathcal{C}^+_+$ in the backward direction and $\mathcal{N}^+_+ \cap I^+$ in the forward direction. Similarly, using the fact that $\mathcal{C}^+_+$ is a sink in all directions in $I^+$ while $\Hp\varrho >0$ at finite frequencies near $\mathcal{N}^+_+\cap I^+$, we deduce that any finite-frequency bicharacteristic in $\Sigma^+$ over $(I^{\pm})^{\circ}$ limits to $\mathcal{C}^+_+$ in the backward direction (and, as stipulated by the non-trapping assumptions, to $\mathcal{R}^+_+$ in the forward direction).

We turn to bicharacteristics over the spacetime interior. Since $\mathcal{N}^+_+$ is a local sink for the flow in $\Gamma_f$, reasoning as before using the flow over $\scri^+$ we conclude that any bicharacteristic $\gamma$ in $\Sigma^+$ over $\M^{\circ}$ must limit in the forward direction either to $\mathcal{N}^+_+$ or to the union of $\mathcal{A}^+_+$, $\mathcal{K}^+_+$, $\mathcal{C}^+_+$, the bicharacteristics going from $\mathcal{A}^+_+$ to $\mathcal{K}^+_+$, and the bicharacteristics going from $\mathcal{K}^+_+$ to $\mathcal{C}^+_+$. In the latter case, since $\Hp \rho_0>0$ away from $I^0$ near $\mathcal{A}^+_+$, $\mathcal{K}^+_+$, and the bicharacteristics connecting them, we conclude that $\gamma$ cannot have limit points in that part of the set. Further, since $\Hp \rho_0>0$ on the part of a neighborhood of the bicharacteristics from $\mathcal{K}^+_+$ to $\mathcal{C}^+_+$ in $\V_{\scri^+_-}^+$, and $\Hp \rho_T<0$ on the part in $\V_{\scri^+_+}^+$, we conclude that $\gamma$ also cannot have limit points on those bicharacteristics. Since $\gamma$ cannot limit to $\mathcal{C}^+_+$ due to the $\scri^+$ eigenvalue, we see that all bicharacteristics over the spacetime interior must limit to $\mathcal{N}$.
	
	We conclude that the Hamilton flow connects the radial sets over null and timelike infinity in the order shown in Figure~\ref{fig:connectivity-schematic}. 	At least at the level of detail of our analysis, we cannot rule out the existence of bicharacteristics over $\M^{\circ}$ which limit to $\mathcal{N}\cap(I^0 \sqcup I^T)$ or bicharacteristics over $(I^T)^{\circ}$ which limit to $\mathcal{R}\cap \scri$, since this would require analyzing details of the dynamics in directions corresponding to zero eigenvalues of the linearized flow. The connectivity of radial sets over $\scri^+_-$ with those over $\scri^-_+$ via the flow over $I^0$ is not directly specified by the non-trapping assumptions, but it is possible that it is uniquely determined by topological considerations; in any case, it is not important for our purposes. The main point which allows microlocal propagation estimates to be combined into a global result is the fact that the flow defines a partial order on the set of points of $\Sigma$, with any point being downstream from the global source $\mathcal{R}^+_-\sqcup\mathcal{R}^-_+$ and upstream from the global sink $\mathcal{R}^+_+\sqcup\mathcal{R}^-_-$ through a finite chain of bicharacteristics alternating with radial points.

\begin{figure}
\centering
\includegraphics{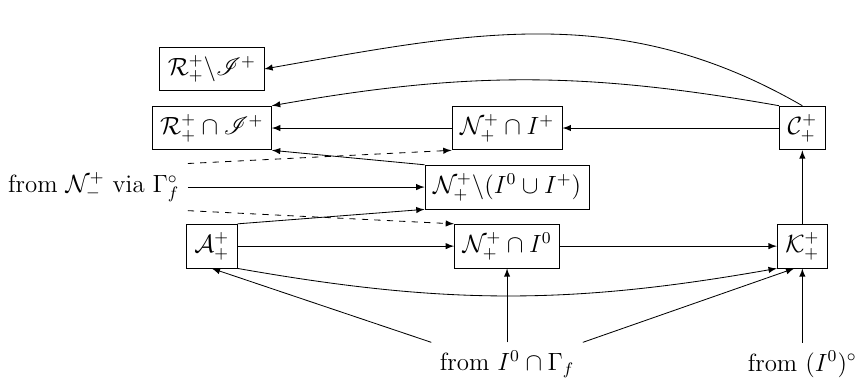}
\caption{Connectivity of the radial sets over $\scri^+$ and $I^+$ via the Hamilton flow in $\Sigma^+$. The dashed lines indicate that bicharacteristics limiting to $\mathcal{N}^+_+\cap (I^0 \cup I^+)$ from the interior may or may not exist.}
\label{fig:connectivity-schematic}
\end{figure}

The diagram illustrates the necessity of a localized radial point result: since there is a path $(\mathcal{N}\cap I^0) \to \mathcal{K} \to \mathcal{C} \to (\mathcal{N}\cap I^+)$ along bicharacteristics, one needs to already know regularity at one part of $\mathcal{N}$ before that knowledge can be propagated into another part of $\mathcal{N}$. Sussman deals with this using radial point estimates localized to extended subsets of $\mathcal{N}$ adjacent to either the spacelike or timelike corner \cite[Propositions 5.8 and 5.9]{Sussman}. Theorem~\ref{thm:localized-rp-main} is a pointwise-localized general-purpose version of such a result.

\section{Global regularity result and distinguished propagators}
\label{sec:global}
We now combine the propagation results into a statement about Fredholm mapping properties of $P$ between weighted Sobolev spaces. We define the advanced, retarded, Feynman, and anti-Feynman realizations of the Klein-Gordon operator as Fredholm operators, thereby extending the construction for Lorentzian scattering spaces to include radiative spacetimes.
The line of reasoning originates in \cite{Vasy-AH-KdS} (see \cite{Grenoble-notes} for an expository account) and is now a standard tool for establishing a Fredholm theory for non-elliptic operators.

The starting point is a result on the global Sobolev regularity of solutions to $Pu=f$.

\begin{prop}
\label{thm:regularity}
Let $\mathsf{s}=(s_f,\mathsf{s}_{base})$ be variable orders which, in a neighborhood of any connected component of the characteristic set, are constant and satisfy simultaneously either the inequalities
\begin{equation}
\begin{cases}
s_{I^+} < -\frac{1}{2} < s_{I^-}, \\
-s_{\scri^+} +2s_{I^+} < s_f-1 < -s_{\scri^-} +2s_{I^-}, \\
s_{\scri^+} < s_f-1 < s_{\scri^-}, \\
s_{\scri^+} < 2s_{I^0} - s_f + 1 < s_{\scri^-}, \\
s_{\scri^+} < s_{I^0} + s_f -\frac{1}{2} < s_{\scri^-}
\end{cases}
\label{eq:thresholds}
\end{equation}
or the opposite inequalities (choice of inequalities is independent for different components of $\Sigma$). Consider another variable order $\mathsf{r}$ such that near any connected component of $\Sigma$ where $s_{I^{\pm}}>-\frac{1}{2}$ we have constant $r_{I^{\pm}}>-\frac{1}{2}$ as well, with no restrictions on the other orders in $\mathsf{r}$. Then any $u\in \sob{r}$ such that $Pu\in H_{\mathrm{de,sc}}^{s_f-1,\mathsf{s}_{base}+1}(M)$ is in fact in $\sob{s}$.
\end{prop}

\begin{remark}
Solutions to the system Eq.~(\ref{eq:thresholds}) do exist: one can choose any $s_{I_0}$ and $s_f$; then any $s_{\scri^+}<\min(s_f-1,2s_{I^0}-s_f+1,s_{I_0}+s_f-\frac{1}{2})$ and any $s_{\scri^-}>\max(s_f-1,2s_{I^0}-s_f+1,s_{I_0}+s_f-\frac{1}{2})$; then any $s_{I^+}<\min(-\frac{1}{2},-\frac{1}{2}+\frac{1}{2}s_f+\frac{1}{2}s_{\scri^+})$  and any $s_{I^-}>\max(-\frac{1}{2},-\frac{1}{2}+\frac{1}{2}s_f+\frac{1}{2}s_{\scri^-})$. One similarly gets solutions to the system with all inequalities reversed in any set of connected components of $\Sigma$. 
\label{rmk:existence}
\end{remark}

\begin{proof}
By elliptic regularity, $Pu\in H_{\mathrm{de,sc}}^{s_f-1,\mathsf{s}_{base}+1}(M)$ implies that $\wf{s}(u)\subset\wf{s+1}(u)\subset\Sigma$.

The first inequality in Eq.~(\ref{eq:thresholds}) along with the existence of $\mathsf{r}$ as assumed allows us to apply Theorem~\ref{thm:localized-rp-main} to the component of $\mathcal{R}$ in each component of $\Sigma$ at which the relevant order of $\mathsf{s}$ is above $-\frac{1}{2}$. This is the above-threshold source/sink estimate, where the hypothesis on the bicharacteristics limiting to/from the radial point is vacuously satisfied, which lets us conclude that the component of $\mathcal{R}$ in question is disjoint from $\wf{s}(u)$.

The remaining inequalities, obtained from the eigenvalue calculations in the previous section and the fact that $P^*=P$, allow us to propagate this knowledge of regularity throughout the characteristic set, in each component either in the order indicated by Figure~\ref{fig:connectivity-schematic} (i.e. first using Theorem~\ref{thm:pos} to propagate regularity along all bicharacteristics starting from a radial set where $\sob{s}$-regularity has just been established, then using Theorem~\ref{thm:localized-rp-main} at a radial set when microlocal $\sob{s}$-regularity has been established at all bicharacteristics ending at that radial set, then repeating) or in the opposite order. The conclusion is that $\wf{s}(u)\cap\Sigma=\varnothing$, so, combined with elliptic regularity, $u\in\sob{s}$.
\end{proof}

For $\mathsf{s}$ as in Proposition~\ref{thm:regularity}, define the Hilbert spaces
$$
\mathcal{X}^{\mathsf{s}}=\{u\in \sob{s}\ |\ Pu\in  H_{\mathrm{de,sc}}^{s_f-1,\mathsf{s}_{base}+\mathsf{1}}(M)\},
\hspace{30pt}
\mathcal{Y}^{\mathsf{s}}=H_{\mathrm{de,sc}}^{s_f-1,\mathsf{s}_{base}+\mathsf{1}}(M),
$$
where we write $\|f\|_{\mathcal{Y}^{\mathsf{s}}}=\|f\|_{s_f-1,\mathsf{s}_{base}+\mathsf{1}}$ and $\mathcal{X}^{\mathsf{s}}$ is endowed with the squared norm $\|u\|_{\mathcal{X}^{\mathsf{s}}}^2 = \|u\|^2_{\mathsf{s}}+\|Pu\|^2_{\mathcal{Y}^{\mathsf{s}}}$.
\begin{prop}
$P:\mathcal{X}^{\mathsf{s}}\to\mathcal{Y}^{\mathsf{s}}$ is a Fredholm operator.
\label{thm:Fredholm}
\end{prop}
\begin{proof}
Since Eq.~(\ref{eq:thresholds}) is an open condition, there always exists $\mathsf{r}<\mathsf{s}$ as required in Proposition~\ref{thm:regularity}. Then by Proposition~\ref{thm:regularity} there is an inclusion map
 $$
\iota:\{u\in\sob{r}\ |\ Pu\in H_{\mathrm{de,sc}}^{s_f-1,\mathsf{s}_{base}+1}(M)\} \to \sob{s},
 $$
where the first space is complete in the squared norm $\|u\|_{\mathsf{r}}^2 + \|Pu\|^2_{\mathcal{Y}^{\mathsf{s}}}$. For any sequence which converges in both of these spaces, the limits must coincide since the identity map is continuous in the weaker $\sch'$ topology. Thus, the graph of $\iota$ is closed, so by the closed graph theorem $\iota$ is continuous. This implies the inequality
 $$
\|u\|_{\mathsf{s}} \leqslant C \Big(\|u\|_{\mathsf{r}} + \|Pu\|_{\mathcal{Y}^{\mathsf{s}}}\Big)
\Rightarrow
\|u\|_{\mathcal{X}^{\mathsf{s}}} \leqslant C' \Big(\|u\|_{\mathsf{r}} + \|Pu\|_{\mathcal{Y}^{\mathsf{s}}}\Big)
 $$
for any $u\in\mathcal{X}^{\mathsf{s}}$ (which could also be obtained by directly combining elliptic estimates, propagation estimates, and radial point estimates in different regions of $\partial\PP$ using a microlocal partition-of-unity argument which parallels the proof of Proposition~\ref{thm:regularity}). Since $\sob{s}$ (and therefore $\mathcal{X}^{\mathsf{s}}$) is compactly embedded in $\sob{r}$, this means $P:\mathcal{X}^{\mathsf{s}}\to\mathcal{Y}^{\mathsf{s}}$ is semi-Fredholm, i.e. $\ker P|_{\mathcal{X}_{\mathsf{s}}}$ is finite-dimensional and $\operatorname{Ran} P|_{\mathcal{X}^{\mathsf{s}}}\subset \mathcal{Y}^{\mathsf{s}}$ is closed.
 
Note that thanks to the fact that $P=P^*$ and therefore the subprincipal symbol does not enter the inequalities, if $\mathsf{s}=(s_f,\mathsf{s}_{base})$ satisfies Eq.~(\ref{eq:thresholds}), then $\mathsf{s'}=-(s_f-1,\mathsf{s}_{base}+1)$ satisfies the opposite inequalities (and vice versa), so Proposition~\ref{thm:regularity} applies to it as well. Thus, by the same argument, any $u\in\sob{s'}$ such that $Pu\in\sob{-s}$ satisfies $\|u\|_{\mathsf{s'}}\leqslant C\Big(\|u\|_{\mathsf{r'}}+ \|Pu\|_{\mathsf{-s}}\Big)$ for some fixed $\mathsf{r'}<\mathsf{s'}$, and $P|_{\sob{s'}}$ has finite-dimensional kernel.
 
Now consider $v$ in the annihilator of $\operatorname{Ran} P|_{\mathcal{X}^{\mathsf{s}}}$, i.e. $v\in \sob{s'} \simeq (\mathcal{Y}^{\mathsf{s}})^*$ such that $\langle Pu,v\rangle_{L^2(\M,\g)} =0$ for all $u\in\mathcal{X}^{\mathsf{s}}$. Then in particular for any $u\in \sob{s+1}$ we have $
\langle u,Pv\rangle = \langle Pu,v\rangle =0$, so $Pv=0$ in $\sob{-s-1}$. Thus, $\mathrm{Ann}(\operatorname{Ran} P|_{\mathcal{X}^{\mathsf{s}}}) \subset \ker P|_{\sob{s'}}$, which is finite-dimensional. Therefore, $\operatorname{Ran} P|_{\mathcal{X}^{\mathsf{s}}}$ has finite-dimensional annihilator in $(\mathcal{Y}^{\mathsf{s}})^*$, so $P:\mathcal{X}^{\mathsf{s}}\to\mathcal{Y}^{\mathsf{s}}$ is Fredholm.
\end{proof}

As a result, if $\M$ has $n$ connected components, we obtain $4^n$ Fredholm realizations of the Klein-Gordon operator $P$ distinguished by the direction of propagation of singularities of solutions to $Pu=f$ in each of the $2n$ components of the characteristic set, as determined by which of $s_{I^{\pm}}$ is lower. If we fix a time orientation and demand that singularities be always propagated towards the future/towards the past, the corresponding operator is the \textit{retarded/advanced} Klein-Gordon operator. On the other hand, if we demand that singularities are always propagated forward/backward along the Hamilton flow, the corresponding operator is the \textit{Feynman/anti-Feynman} Klein-Gordon operator. This means ``positive-frequency" singularities, i.e. those corresponding to momenta in $\Sigma^+$ in the wavefront set, are propagated forward in time and ``negative-frequency" ones backward in time.

Finally, we show that under the additional assumption that the spacetime admits a time function which is in a certain sense nondegenerate with respect to the compactification, these realizations of $P$ are not only Fredholm but in fact invertible. The inverses can thus be identified as the retarded, advanced, Feynman, and anti-Feynman \textit{propagators}. In the retarded/advanced cases, invertibility is a consequence of energy estimates (see Theorem~\ref{thm:invertibility}). The Feynman/anti-Feynman cases reduce to the retarded/advanced ones via the following analogue of the result \cite[Proposition 7]{Vasy-SA} in the sc-setting, with essentially the same proof based on a construction of Isozaki \cite[Lemma 4.5]{Isozaki}.

\begin{prop}
If $Pu=0$ for $u\in\sob{s'}$, where $\mathsf{s'}$ satisfies Eq.~(\ref{eq:thresholds}) on $\Sigma^+$ and the opposite inequalities on $\Sigma^-$, or vice versa, then $u\in\sch$. (In other words, the kernel of the Feynman or anti-Feynman realization of $P$ is a subspace of $\sch$.)
\label{thm:Feynman-kernel}
\end{prop}

\begin{proof}
We focus on the Feynman case; the anti-Feynman case is completely analogous.

Propagating regularity as in the proof of Proposition~\ref{thm:regularity}, $Pu\in\sch$ for $u\in\sob{s'}$ implies $u\in \sob{s''}$ for any $\mathsf{s''}$ with $s_{I^{\pm}}''<-\frac{1}{2}$ and $\wf{}(u)\subset\mathcal{R}^+_+\sqcup\mathcal{R}^-_-$, but we cannot propagate Schwartz regularity into the global sink $\mathcal{R}^+_+\sqcup\mathcal{R}^-_-$ using Theorem~\ref{thm:localized-rp-main}: the a priori regularity assumption ($u\in\sob{s'}$) is below-threshold at the sink, and the below-threshold estimate has an absolute limit on the amount of regularity it can provide (in this case, $s_{I^{\pm}}''<-\frac{1}{2}$). However, when $Pu=0$ identically, one can prove an additional positive-commutator estimate yielding above-threshold regularity at the sink, and then the above-threshold version of Theorem~\ref{thm:localized-rp-main} yields $u\in\sch$.

Fix $s\in (-\frac{1}{2},0)$ and let $\mathsf{s}=(\frac{1}{2},s,-\frac{1}{2},-\frac{1}{2},-\frac{1}{2},s)$, so we know $u\in\sob{s-\frac{1}{2}}$. Fix $\phi\in C^{\infty}(\R)$ such that $\phi(r)=0$ for $r\leqslant 1$, $\phi(r)>0$ for $r>1$, and $\phi(r)=1$ for $r\geqslant 2$. Define $a_t \in C^{\infty}(\M)$  by 
\begin{equation}
a_t (x)
= 
\int_0^{\rho_T(x)}
\phi(r/ t)^2 r^{-2s-2}\ dr
=
\begin{cases}
0,\ \rho_T\leqslant t,\\
\text{monotone increasing with }\rho_T,\ t<\rho_T<2t,\\
C_t-\frac{1}{2s+1}\rho_T^{-2s-1},\ \rho_T\geqslant 2t.
\end{cases}
\label{eq:Isozaki-commutant}
\end{equation}
See Figure~\ref{fig:Isozaki-commutant}. We digress to clarify the idea behind this choice. 

\begin{figure}
\begin{center}
\includegraphics{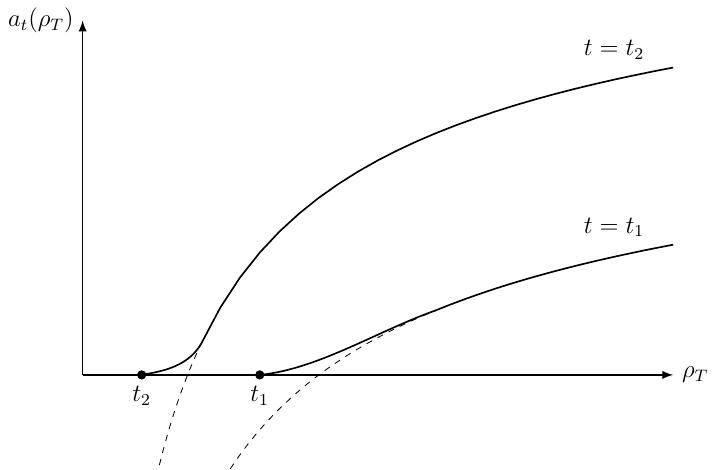}
\end{center}
\caption{Plot of representatives of the commutant family $a_t$ for two values of $t$. The dashed lines show the $C_t-\frac{1}{2s+1}\rho_T^{-2s-1}$ profiles. For any fixed positive value of $\rho_T$, the value $a_t(\rho_T)$ goes to infinity as $t\to 0^+$ but the derivative $a_t'(\rho_T)$ stays constant once $2t<\rho_T$.}
\label{fig:Isozaki-commutant}
\end{figure}

In the proof of usual radial point estimates, to conclude order-$\mathsf{s}$ regularity at a component $R$ of $\mathcal{R}$, one constructs a commutant $a$ from $\rho_T^{-2s-1}$ (since $P$ is zeroth-order at $I^{\pm}$) multiplied by bump functions localizing to a neighborhood of $R$. When $s$ crosses over the threshold, $a$ becomes growing rather than decaying at $I^{\pm}$; then the usual regularization procedure requires one to approximate it by symbols which grow at $I^{\pm}$ as well, and because of this the pairings and integrations by parts used in the proof cannot be justified unless $u$ is a priori known to have some above-threshold decay at $R$. For instance, if one attempts to regularize by simply multiplying $a$ by a factor supported away from $\{\rho_T=0\}$ so that the regularized version is supported away from $\mathcal{R}$, hence where $u$ is Schwartz, then this will change the monotonicity of $a$ along $H_p$ in the cutoff region and the positive-commutator proof will not work.

The key idea leading to Eq.~(\ref{eq:Isozaki-commutant}) is that in the estimate for the microlocal $\sob{s}$-norm of $u$ arising from the usual proof, all terms involving $A$ explicitly are paired with $Pu$, so if $Pu=0$ identically and $P=P^*$, then it is truly only the commutator $[P,A]$ (corresponding to the derivative $H_pa$), rather than $A$ itself (corresponding to the value of $a$), that matters. Then instead of cutting $a$ off to zero near $\mathcal{R}$ one can try leveling it off in a monotone fashion to equal any other constant $C_t$ there, where to approximate $a$ we need $C_t\to\infty$ as the regularization parameter $t$ goes to zero. Nevertheless, this does not quite work, since to obtain the estimate one first needs to rewrite the commutator term using Lemma~\ref{thm:comm-lemma}, and $a_t$ being constant, thus zeroth-order, at $I^{\pm}$ is just barely not enough for this as it would require $u$ to have a priori decay exactly at threshold.

Combining the two ideas works, however: since the nonzero value of $a$ is only an obstruction to rewriting the commutator term and elsewhere only the derivative $H_pa$ is relevant, we can level off $a$ as just described and then additionally shift it by $C_t$ to make it zero near $\mathcal{R}$. Then one can apply Lemma~\ref{thm:comm-lemma} since the support of $a_t$ is disjoint from the radial set, and using $Pu=0$ one gets rid of any terms involving $A$ explicitly. Since the shift at any given point goes to infinity as $t\to 0^+$, the family $a_t$ no longer converges to anything and is in fact not uniformly bounded in any symbol space, but only its derivative, which \textit{is} uniformly bounded, is relevant for the estimate.

The unboundedness of $a_t$, however, limits this construction's range of applicability. First of all, one cannot add arbitrary localizing factors to $a_t$ after the constant shift, since they will generally contribute terms to $H_pa_t$ which will be unbounded as $t\to 0^+$. Second, for Lemma~\ref{thm:comm-lemma} to be applicable one requires that $a_t=0$ simultaneously in all regions where $u$ is not a priori above-threshold. As a consequence, one cannot deal with one radial set at a time, since one must ensure that before the constant shift is applied, the commutant-to-be is in fact equal to \textit{the same constant} near all radial sets where regularity is a priori unknown. Moreover, this means that it is not particularly useful to add any localizing factors to $\rho_T^{-2s-1}$ even before the constant shift, since one will in any case need a priori above-threshold regularity of $u$ everywhere except near the radial sets of interest; therefore, $a_t$ in Eq.~(\ref{eq:Isozaki-commutant}) is constructed starting from $a=\rho_T^{-2s-1}$ globally. Finally, since for the positive-commutator argument to work $H_pa$ needs to have the same definite sign near all of the radial sets of interest, one can only use this to propagate into either all sources or all sinks, so the proof applies only to the Feynman/anti-Feynman cases.

Returning to the proof, we can consider $a_t$ as a smooth function on $\PP$ which is constant in the fibers. Then
$$
H_p a_t = \rho_f^{-1}\rho_{base}^{\mathsf{1}}\Hp a_t = h_{\rho_T}\rho^{\mathsf{-2s}}\cdot \phi(\rho_T/t)^2,
\hspace{30pt}
h_{\rho_T} = \rho_T^{-1} \Hp\rho_T \in \sym{0}.
$$
Since $\mathcal{R}^+_+\sqcup\mathcal{R}^-_-$ is a sink for the flow, $h_{\rho_T}<0$ in a neighborhood of it. Thus, we can decompose
$$
h_{\rho_T}\rho^{\mathsf{-2s}} = -b^2 + e, 
$$
where $b\in \sym{s}$ is elliptic at $\mathcal{R}^+_+\sqcup\mathcal{R}^-_-$ and $e\in\sym{2s}$ has $\mathrm{supp}(e)\cap (\mathcal{R}^+_+\sqcup\mathcal{R}^-_-)=\varnothing$. We fix $B=\mathrm{Op}(b)$, $E=\mathrm{Op}(e)$ and define regularized versions $B_t =  B \phi(\rho_T/t) $, $E_t = \phi(\rho_T/t) E \phi(\rho_T/t)$, where $\phi(\rho_T/t)$ is understood as a multiplication operator and is an approximation of the identity in the sense defined in Section~\ref{sec:lemmas}. Then, treating $a_t$ as a multiplication operator,
$$
i[P,a_t] = -B^*_t B_t + E_t + F_t,
$$
where $F_t$ for $t\in (0,1)$ is a bounded family in $\ps[2s-1]$\footnote{The simplest way to see this boundedness is to note that $[P,a_t]=[P,a_t-C_t]$ and $a_t-C_t$ is bounded in $\sym{2s-m+1}$ with $\mathsf{m}$ the order of $P$, so boundedness of the error term follows exactly as in the proofs of propagation theorems. This incidentally also shows that, if desired, one can formulate this proof without introducing any unbounded families of operators, namely taking the commutant to be $a_t-C_t$ (that is, the ``leveled-off" version of $a=\rho_T^{-2s-1}$ in the discussion at the beginning of the proof) from the start and using the fact that $[P,a_t-C_t]=[P,a_t]$ to justify the required manipulations.}.

Since $Pu=0$ and $a_t u\in \sch$ for any $t$, we have $
\langle i[P,a_t]u,u\rangle
=
\langle iP(a_t u),u\rangle
=
\langle i a_t u,Pu\rangle
=0$. Also, since $B_t u\in\sch$ for any $t$, we have $\langle B_t^* B_t u,u\rangle=\|B_t u\|^2$. Therefore,
$$
\|B_t u\|^2 = -\langle i[P,a_t]u,u\rangle + \langle E_t u,u\rangle + \langle F_t u,u\rangle = \langle E_t u,u\rangle + \langle F_t u,u\rangle.
$$
The first term is bounded uniformly in $t$ because $\wfs(\{E_t\})\cap \wf{}(u)=\varnothing$ while $E_t$ is bounded in $\ps[2s]$, and likewise for the second term because $u\in\sob{s-\frac{1}{2}}$ while $F_t$ is bounded in $\ps[2s-1]$. Therefore, the same reasoning as in the proof of Proposition~\ref{thm:localized-rp-estimates} leads to the conclusion $\wf{s}(u)\cap (\mathcal{R}^+_+\sqcup\mathcal{R}^-_-)=\varnothing$. Then since $s>-\frac{1}{2}$, the above-threshold estimate implies $\wf{}(u)\cap (\mathcal{R}^+_+\sqcup\mathcal{R}^-_-)=\varnothing$, so $u\in\sch$.
\end{proof}

In particular, this means that any element of the kernel of the Feynman or anti-Feynman realization of $P$ is also in the kernel of all the other realizations, since $\sch\subset\mathcal{X}^{\mathsf{s}}$ for any $\mathsf{s}$.

\begin{theorem}
\label{thm:invertibility}
Assume that there exists a function $t\in S^{1,1,0,1,1}_{\epsilon}(\M)$ which satisfies
$$
\Big(\rho_{I^{\pm}}\rho_{\scri^{\pm}} t\Big)|_{I^{\pm}\cup\scri^{\pm}} >0,
\hspace{50pt}
\Big(\rho_{I^0}^{-2}\g(\nabla t,\nabla t)\Big)|_{I^{\pm}\cup\scri^{\pm}} < 0.
$$
Then $(\M^{\circ},\g)$ is globally hyperbolic and $P:\mathcal{X}^{\mathsf{s}}\to\mathcal{Y}^{\mathsf{s}}$ is invertible for any $\mathsf{s}$ satisfying the retarded $(-)$, advanced $(+)$, or Feynman/anti-Feynman (either sign) conditions as in Proposition~\ref{thm:regularity}.
\end{theorem}

Note that the conditions make sense because $\rho_{I^{\pm}}\rho_{\scri^{\pm}} t$ and $\rho_{I^0}^{-2} \g(\nabla t,\nabla t)$ are in $\soe(\M)$ near $I^{\pm}\cup \scri^{\pm}$ (the latter because $\g(\nabla t,\nabla t) = \g^{-1}(dt,dt)$, where $dt\in S^{0,0,-1,0,0}_{\epsilon}(\M;\ttm)$ and $\g^{-1}\in\soe(\M;\mathrm{Sym}^2(\ttm))$), so they are continuous up to and including the boundary. See \cite[Theorem 5.3]{HV-semilinear} for a result of this type in the setting of Lorentzian scattering spaces; our proof involves a similar energy estimate.

Unfortunately, the hypotheses of the theorem are not satisfied by time functions similar to the global inertial time in Minkowski space, which is first-order at spacelike infinity and second-order at null infinity due to the square-root blowup in the definition of $\M$ but for which $\g(\nabla t,\nabla t)=-1$, which is one order smaller at $\scri$ than one would expect just from the symbolic order of $t$. However, in Appendix~\ref{sec:time-fcn} we show that an appropriate time function $t$ exists for any Minkowski-like spacetime as defined in Section~\ref{sec:examples} (which include small perturbations of Minkowski space), so the result applies to them.

\begin{proof}
We consider the case when the assumptions are imposed at $I^+ \cup \scri^+$; the other case is analogous. Without loss of generality, we can assume that $t$ is zeroth-order at $I^-$ and $\scri^-$, so $t\leqslant C\rho_{I^+}^{-1}\rho_{\scri^+}^{-1}$ globally for some $C>0$. The assumptions imply that there exists an open neighborhood $U'$ of $I^+\cup\scri^+$, which can be taken disjoint from a neighborhood of $I^-\cup\scri^-$ and connected in every connected component of $\M$, on which $t \geqslant c\rho_{I^+}^{-1} \rho_{\scri^+}^{-1}$ and $\g(\nabla t,\nabla t)\leqslant -c\rho_{I^0}^2$ for some $c>0$. In particular, we must have $t\to +\infty$ at $I^+\cup\scri^+$ but $t$ is finite and continuous at points of $(I^0)^{\circ}$. For some $T_0>0$ we have $\{t\geqslant T_0\} \subset  \{\rho_{I^+}\rho_{\scri^+} \leqslant \frac{C}{T_0}\} \subset U'$; we denote $U= \{t\geqslant T_0\}$. 

For any $T>T_0$ we write $S_T = \{t=T\}\subset U \subset U'$. Since $t$ is smooth in $\M^{\circ}$, each level set $S_T$ is closed in $\M^{\circ}$. Since $t<T_0$ near $I^-\cup\scri^-$ and $t\to +\infty$ at $I^+\cup\scri^+$, due to the non-trapping assumption every null geodesic must cross any $S_T$ and re-emerge from it. Since $\nabla t$ is nonzero and timelike in $U'$ and the spacetime is assumed to be time-orientable, for any timelike curve in $\M^{\circ}$ the function $t$ is either monotone increasing or monotone decreasing along the part of that curve in $U'$ (thanks to the fact that we choose $U'$ to be connected in any connected component of $\M^{\circ}$). Then no two points of $S_T$ can be joined by a timelike curve which does not exit $U$, whereas a timelike curve which exits $U$ cannot enter it again; thus, no two points of $S_T$ can be joined by any timelike curve in $\M^{\circ}$. Therefore by a characterization due to Geroch, each $S_T$ is a Cauchy surface for $(\M^{\circ},\g)$ \cite[Property 6]{Geroch} and the spacetime is globally hyperbolic \cite[Theorem 11]{Geroch}.
 
The proof of Proposition~\ref{thm:Fredholm} shows that the invertibility statement will follow if we show that $\ker P|_{\mathcal{X}^{\mathsf{s}}}=0$ for all $\mathsf{s}$ satisfying any of the conditions. For $\mathsf{s}$ satisfying the Feynman/anti-Feynman conditions, we saw in Proposition~\ref{thm:Feynman-kernel} that any elements of the kernel are in $\sch$; for the retarded/advanced versions, on the other hand, propagation of singularities starting from the a priori assumption $u\in\sob{s}$, which is above-threshold at both components of $\mathcal{R}$ over $I^-$ in the retarded case and $I^+$ in the advanced case, proves that $u$ is Schwartz away from $I^+$ or $I^-$ respectively. Therefore, it is enough to show that the only solution $u\in\sch'$ to $Pu=0$ which is Schwartz away from $I^{\pm}$ is $u=0$. We show this (for $u$ Schwartz away from $I^-$) using an energy estimate similar to those in \cite{Vasy-AdS, HV-semilinear}.

Direct calculation shows that for any real vector field $V$ on $\M^{\circ}$ we have
$$
(V^*P+PV)u
=
\partial_{\alpha}^*\left(\left[
g^{\alpha\mu}(\partial_{\mu}v^{\beta})
+g^{\beta\mu}(\partial_{\mu} v^{\alpha})
-\frac{1}{\sqrt{|\det\g|}}\partial_{\mu}
\Big(\sqrt{|\det\g|}v^{\mu}g^{\alpha\beta}\Big)
\right]
\partial_{\beta} u\right)
-m^2 u\operatorname{div} V
$$
in any local coordinates (summing over repeated indices), where $\partial_{\alpha}^*$ is the $L^2(\M,\g)$-adjoint of $\partial_{\alpha}$. Since $P$ is real, we can assume without loss of generality that the solution $u$ is real-valued. Then as long as $V$ and $u$ are regular enough to justify integration by parts, 
\begin{equation}
\label{eq:energy-integral}
\langle (V^*P+PV)u, u\rangle
=
\int_{\M} \Big( C_V(du,du) - m^2u^2\operatorname{div} V\Big)\ d\mathrm{vol}_{\g},
\end{equation}
where $C_V$ is the bilinear form on the fibers of $T^*\M^{\circ}$ whose matrix in any coordinates is given by the expression in brackets above. Taking $V=\chi(t)\cdot\nabla t$ for a real-valued function $\chi\in C^{\infty}(\R)$, the integrand becomes
$$
C_V(du,du) -m^2u^2\operatorname{div} V
=
2\chi'\cdot T_u(\nabla t,\nabla t)
+\chi
\cdot \Big( C_{\nabla t}
(du,du) + m^2u^2 \Box_{\g}t\Big),
$$
where $T_u=\g(\nabla u,\bullet)\g(\nabla u,\circ)-\frac{1}{2}\Big(\g(\nabla u,\nabla u)+m^2u^2\Big)\g(\bullet,\circ)$ is the Klein-Gordon stress-energy tensor of $u$ and $C_{\nabla{t}}$ is defined in the same way as $C_V$ with $\nabla t$ in place of $V$.

Working in local coordinates $(\rho_0,x_0,y_1,\ldots,y_{d-1})$ on $\U_0$, where $y$ extends a local coordinate chart on $Y$, let us define $\tilde{d}f = \left(\rho_0^2 x_0 \frac{\partial f}{\partial\rho_0},\rho_0 x_0^2\frac{\partial f}{\partial x_0}, \rho_0 x_0^2\frac{\partial f}{\partial y_1},\ldots,\rho_0 x_0^2\frac{\partial f}{\partial y_{d-1}}\right)$, and let $\tilde{\g}$ be the Lorentzian bilinear form defined by the matrix of $\g^{-1}$ with respect to the frame $\left(\frac{d\rho_0}{\rho_0^2x_0^2},\frac{dx_0}{\rho_0x_0^2}, \frac{dy_1}{\rho_0x_0^2},\ldots,\frac{dy_{d-1}}{\rho_0x_0^2}\right)$, so $\tilde{\g}$ is in $\soe(\U_0)$ and nondegenerate up to and including the boundary. Then we can write 
$$
2T_u(\nabla t,\nabla t)
=
2\tilde{\g}(\tilde{d}u,\tilde{d}t)^2 - \Big(\tilde{\g}(\tilde{d}u,\tilde{d}u)+m^2u^2\Big)\tilde{\g}(\tilde{d}t,\tilde{d}t).
$$
The fact that $\g(\nabla t,\nabla t)<-c\rho_{I^0}^2$ in $U$ translates to $\tilde{\g}(\rho_{I^0}^{-1}\tilde{d}t, \rho_{I^0}^{-1}\tilde{d}t)<-c$, so $\rho_{I^0}^{-1}\tilde{d}t\in\soe(\M;\ttm)$ is timelike for $\tilde{\g}$ at all points of $U$ up to and including the boundary. Then we have the lower bound (uniform version of the \textit{weak energy condition})
$$
2T_u(\nabla t,\nabla t) \geqslant c' \rho_{I^0}^2 (|\tilde{d}u|^2+m^2u^2),
$$
where $|\tilde{d}u|$ is just the Euclidean norm of $\tilde{d}u = \left(\rho_0^2 x_0 \frac{\partial u}{\partial\rho_0},\rho_0 x_0^2\frac{\partial u}{\partial x_0}, \rho_0 x_0^2\frac{\partial u}{\partial y_1},\ldots,\rho_0 x_0^2\frac{\partial u}{\partial y_{d-1}}\right)$. Here we can take $c'>0$ independent of the point in $\U_0$ because the positivity holds at all points up to and including the boundary and $\tilde{\g}$ (hence the coefficients of the bilinear form) is continuous.

On the other hand, using the fact that $\g$ is a de,sc-metric and the symbolic order of $t$, one can check that, working in the same coordinates, we can write $C_{\nabla t}(du,du) = \tilde{C}_{\nabla t}(\tilde{d}u,\tilde{d}u)$, where $\tilde{C}_{\nabla t}$ is a bilinear form with coefficients in $\rho_{I^0} S^{\mathsf{-1}}(\M)$. We also have $\Box_{\g}t\in \rho_{I^0} S^{\mathsf{-1}}(\M)$. Therefore, there exists some $C'>0$ such that in $U\cap \U_0$
$$
|C_{\nabla t}(du,du) - m^2u^2\Box_{\g}t |\leqslant C'\rho_{I^+}\rho_{\scri^+}\rho_{I^0}^2(|\tilde{d}u|^2+m^2u^2).$$

Then, combining the two estimates and assuming $\chi,\chi'\geqslant 0$, we get
$$
C_V(du,du)-m^2u^2\operatorname{div}V
\geqslant
(c'\chi'-C' \rho_{I^+}\rho_{\scri^+} \chi)\rho_{I^0}^2(|\tilde{d}u|^2+m^2u^2).$$
Similar estimates hold in coordinate neighborhoods in $\U_T$ (with $\rho_0,x_0$ replaced by $\rho_T,x_T$), in coordinate neighborhoods of the interiors of $I^+$ and $I^0$ (with $\tilde{d}u$ built out of sc-derivatives of $u$), and in compact sets of $\M^{\circ}$ (with $\tilde{d}u$ built out of the ordinary coordinate derivatives of $u$). Since $U$ can be covered by a finite number of such coordinate neighborhoods, the constants $c',C'>0$ can be taken to be independent of the point in $U$.

Recall that we are assuming that $u$ is Schwartz away from $I^-$, and in particular in $U$. Fix $T>T_0$ and $N>\frac{C'}{c'}\sup_U \Big(\rho_{I^+}\rho_{\scri^+}(t-T)\Big)$ and set $\chi(t)=0$ for $t\leqslant T$ while $\chi(t)=(t-T)^N e^{-\frac{1}{t-T}}$ for $t>T$. Then
$$
\chi'(t) = \left( \frac{N}{t-T} + \frac{1}{(t-T)^2}\right) \chi(t),
$$
and we get $c'\chi'-C'\rho_{I^+} \rho_{\scri^+} \chi>0$ for $t>T$. Thus, the integrand in Eq.~(\ref{eq:energy-integral}) is zero on $\{t\leqslant T\}$ and strictly positive on $\{t>T\}$ except at points where $u=0$ and $du=0$.

$\chi(t)$ and all its derivatives on $\M$ are polynomially bounded. This means that, since $u$ is Schwartz on the support of $\chi(t)$, the vector field $V$ is regular enough for all the preceding integrals to be well-defined and the integrations by parts to be justified. Then due to the positivity result just established, we conclude that $\langle (V^*P+PV)u,u\rangle >0$ unless $u=0$ identically on $\{t>T\}$; but on the other hand $Pu=0$ implies that $\langle (V^*P+PV)u,u\rangle =0$. Therefore, we indeed have $u=0$ identically on $\{t>T\}$, in particular in a neighborhood of $S_{T'}$ for any $T'>T$. The fact that $S_{T'}$ is a Cauchy surface then implies that $u=0$ everywhere on $\M$.

\end{proof}

To summarize, in this section we showed that $P$ is Fredholm, and under the assumption of existence of an appropriate time function actually invertible, as an operator between Hilbert spaces $\X^{\mathsf{s}}\to\Y^{\mathsf{s}}$ for any orders $\mathsf{s}$ which satisfy, independently in each connected component of $\Sigma$, either the inequalities of Proposition~\ref{thm:regularity} or the opposite inequalities. The inverses define distinguished retarded, advanced, Feynman, and anti-Feynman propagators for the Klein-Gordon operator which are canonically determined by the compactification of the spacetime.
\begin{itemize}
\item The retarded and advanced propagators are defined by choosing the conditions with $s_{I^+}<-\frac{1}{2}<s_{I^-}$ (retarded) or $s_{I^-}<-\frac{1}{2}<s_{I^+}$ (advanced) in every component of $\Sigma$. This corresponds to decay faster than $\rho_{I^{\mp}}^{d/2}$ near $I^-$ ($I^+$) for functions in the range of the retarded (advanced) propagator.
\item The Feynman and anti-Feynman propagators are defined by choosing the conditions with $s_{I^{\pm}}|_{\Sigma^{\pm}}<-\frac{1}{2}<s_{I^{\mp}}|_{\mp}$ (Feynman) or $s_{I^{\mp}}|_{\Sigma^{\pm}}<-\frac{1}{2}<s_{I^{\pm}}|_{\pm}$ (anti-Feynman).
\end{itemize}
We will denote the Feynman propagator $P_+^{-1}$ and the anti-Feynman propagator $P_-^{-1}$ below.

We note that if $P:\X^{\mathsf{s}}\to\Y^{\mathsf{s}}$ is invertible for all orders $\mathsf{s}$ satisfying one set of conditions, then the inverses agree for different choices of such $\mathsf{s}$, and in particular for any $f\in \sch$ there exists a unique solution $u\in\sch'$ with wavefront set contained only in the components of $\mathcal{R}$ where $s_{I^{\pm}}<-\frac{1}{2}$. Indeed, considering the Feynman case, by the propagation theorems any solution $u\in\sch'$ to $Pu=f$ for $f\in\sch$ with above-threshold regularity at $\mathcal{R}^+_-\cup\mathcal{R}^-_+$ (which is true for any $u\in\X^{\mathsf{s}}$ with $\mathsf{s}$ satisfying the Feynman conditions) in fact has wavefront set of any order only in $\mathcal{R}^+_+\cup\mathcal{R}^-_-$; then the fact that $\sch\subset \Y^{\mathsf{s}}$ for any $\mathsf{s}$ implies the existence of the desired solutions in $\X^{\mathsf{s}}\subset\sch'$ (since orders $\mathsf{s}$ satisfying any of the conditions do exist, as noted in Remark~\ref{rmk:existence}). On the other hand, such solutions automatically belong to \textit{every} $\X^{\mathsf{s}}$ with $\mathsf{s}$ satisfying the conditions, since they have arbitrarily high below-threshold regularity at $\mathcal{R}^+_+\cup\mathcal{R}^-_-$ and are Schwartz elsewhere; thus, since the operator restricted to any one of these spaces is invertible, the solution must be unique.

Specializing from our general assumptions to a more easily described class of spacetimes, we can therefore state the following precise version of the existence-and-uniqueness part of Theorem 1.1, whose applicability to small perturbations of Minkowski space follows from the discussion in Example~\ref{ex:perturbations}.

\begin{theorem}
Let $\g$ be an asymptotically Minkowski metric on $\R^{d+1}$ satisfying the non-trapping assumption on null geodesics, and let $P=\Box_{\g}+m^2$ for $m>0$. Then for any $f\in\sch$ there exists a unique solution $u\in\sch'$ to $Pu=f$ such that $\wf{}(u)\subset (I^+\cap \Sigma^+) \cup (I^-\cap\Sigma^-)$.
\label{thm:intro-precise}
\end{theorem}
\begin{proof}
The construction of a time function in Appendix~\ref{sec:time-fcn} shows that Theorem~\ref{thm:invertibility} applies to such spacetimes, so all four distinguished realizations of $P$ are invertible. Then as just discussed, restricting attention to the Feynman realization, for any $f\in\sch$ there exists a unique solution $u\in\sch'$ to $Pu=f$ such that $\wf{}(u)\subset \mathcal{R}^+_+ \cup \mathcal{R}^-_-$. By the propagation theorems, any solution with $\wf{}(u)\subset (I^+\cap \Sigma^+) \cup (I^-\cap\Sigma^-)$ must in fact have $\wf{}(u)\subset \mathcal{R}^+_+ \cup \mathcal{R}^-_-$, so this is equivalent to the theorem.
\end{proof}

\section{Limiting absorption principle}
\label{sec:LAP}
In this section, we show that the Feynman/anti-Feynman propagators defined in the previous section can be understood as limits
$$P_{\pm}^{-1}=\lim_{\varepsilon\to 0^+}(P\mp i\varepsilon)^{-1},$$
where the inverses on the right-hand side can be interpreted via the functional calculus for self-adjoint operators, since $P$ with domain $\sch$ is essentially self-adjoint on $L^2(\M,\g)$ \cite{JMS} (though we will consider them defined directly using the Fredholm framework). For Lorentzian scattering spaces, convergence in the weak operator topology was established in \cite{Vasy-SA} following \cite[Section 2.7]{Vasy-AH-KdS}; essentially the same proof applies in our setting, and we also show that the result can be upgraded to convergence in the strong operator topology by extension from a dense subspace.

Let $\mathsf{s_{\pm}}=(s_f,\mathsf{s}_{base})$ be variable orders which are constant near any connected component of the characteristic set and satisfy simultaneously the inequalities
\begin{equation}
\begin{cases}
-s_{\scri^+} +2s_{I^+} < s_f-1 < -s_{\scri^-} +2s_{I^-}, \\
s_{\scri^+} < s_f-1 < s_{\scri^-}, \\
s_{\scri^+} < 2s_{I^0} - s_f + 1 < s_{\scri^-}, \\
s_{\scri^+} < s_{I^0} + s_f -\frac{1}{2} < s_{\scri^-}
\end{cases}
\label{eq:thresholds-complex}
\end{equation}
on $\Sigma^{\pm}$ and the opposite inequalities on $\Sigma^{\mp}$. In other words, $\mathsf{s_+}$ satisfies the Feynman inequalities and $\mathsf{s_-}$ the anti-Feynman inequalities, except the threshold condition on regularity at $I^T$ (the first inequality in Eq.~(\ref{eq:thresholds})) is dropped. For such $\mathsf{s_{\pm}}$ and any $\varepsilon>0$, define
$$
\X^{\mathsf{s_{\pm}}}_{\varepsilon} = \{u\in\sob{s_{\pm}}\ |\ (P \mp i\varepsilon)u\in\Y^{\mathsf{s_{\pm}}}\},
\hspace{30pt}
\Y^{\mathsf{s_{\pm}}} = H^{s_f-1,\mathsf{s_{base}+1}}(\M)
$$
$\X^{\mathsf{s_{\pm}}}_{\varepsilon}$ is a Hilbert space with the squared norm $\|u\|_{\X^{\mathsf{s_{\pm}}}_{\varepsilon}}^2=\|u\|_{\mathsf{s_{\pm}}}^2+\|(P\mp i\varepsilon)u\|_{\Y^{\mathsf{s_{\pm}}}}^2$. We recall the analogue of Proposition~\ref{thm:Fredholm} for $P\mp i\varepsilon$.

\begin{theorem}[\cite{JMS}, Theorem C and Remark 3.6]
\label{thm:invertibility-complex}
$(P\mp i\varepsilon):\X^{\mathsf{s_{\pm}}}_{\varepsilon}\to\Y^{\mathsf{s_{\pm}}}$ is invertible, and the inverse maps $(P\mp i\varepsilon)^{-1}:\sch\to\sch$.
\end{theorem}
\begin{proof}
The argument of \cite{JMS} for asymptotically Minkowski metrics on $\R^{d+1}$ applies without modification in the more general setting.

The outline of the argument that $P\mp i\varepsilon$ is Fredholm is almost the same as for $P$, with two main differences. First, due to the imaginary part, propagation of singularities (including into radial points) is only allowed in one direction along the Hamilton flow for either of $P\mp i\varepsilon$; because of this, estimates for $P-i\varepsilon$ can only be combined into a global statement on the Feynman spaces and for $P+i\varepsilon$ on the anti-Feynman spaces. Second, the propagation requires fewer steps because $P\mp i\varepsilon$ is elliptic away from fiber infinity; in particular, the radial set $\mathcal{R}$ plays no role, so neither the threshold inequalities on $s_{I^{\pm}}$ nor the part of the non-trapping assumptions relating to finite frequencies over timelike or spacelike infinity are necessary for this result.

The absence of thresholds can be used to show that the kernel of $P\mp i\varepsilon$ is a subspace of $\sch$ without the special construction of Proposition~\ref{thm:Feynman-kernel}, and the fact that $P$ is $L^2(\M,\g)$-symmetric then implies that the kernel is $0$. Since $(P+i\varepsilon)^*=(P-i\varepsilon)$ with respect to the $L^2(\M,\g)$ pairing, a functional-analytic argument like in the proof of Proposition~\ref{thm:Fredholm} then yields the invertibility of $(P\mp i\varepsilon):\X^{\mathsf{s_{\pm}}}_{\varepsilon}\to\Y^{\mathsf{s_{\pm}}}$. The fact that $(P\mp i\varepsilon)^{-1}$ maps $\sch\to\sch$ follows from the fact that, with the threshold inequality removed, there are solutions to the system Eq.~(\ref{eq:thresholds-complex}) with all orders arbitrarily high.
\end{proof}

The limiting absorption principle then follows from invertibility of $P_{\pm}$ along with the fact that the estimates for propagation of singularities (including into radial points) for $P\mp i\varepsilon$, and hence the corresponding Fredholm estimates, are uniform in $\varepsilon\in (0,1)$.

\begin{lemma}
Let $T$ be a bounded set of strictly positive real numbers. Let $\mathsf{s_{\pm}}=(s_f,\mathsf{s_{base}})$ be a set of orders satisfying the Feynman ($\mathsf{s_+}$) or anti-Feynman ($\mathsf{s_-}$) conditions, including the threshold conditions at $I^T$ if $0\in \overline{T}$ but not necessarily otherwise. For $\tau\in T$, let $g_{\tau}$ be a bounded family in $\Y^{\mathsf{s_{\pm}}}$. If the family $v_{\tau} = (P \mp i\tau)^{-1} g_{\tau} \in \sob{s_{\pm}}$ is bounded in $\sob{s'}$ for some $\mathsf{s'}\geqslant\mathsf{s_{\pm}-1}$, then $v_{\tau}$ is also bounded in $\sob{s_{\pm}}$.
\label{thm:bound-uniform}
\end{lemma}

\begin{proof}
Assuming first that $\mathsf{s_{\pm}}$ satisfies the threshold conditions, by using the uniform statements of Theorem~\ref{thm:pos} and Theorem~\ref{thm:localized-rp-main} in the order dictated by Figure~\ref{fig:connectivity-schematic} we conclude that $\wf{s_{\pm}}(\{v_{\tau}\})\cap\Sigma = \varnothing$. Then there exists an open neighborhood $U$ of $\Sigma$ in $\partial\PP$ which is disjoint from $\wf{s_{\pm}}(\{v_{\tau}\})$. To complete the proof, we need to show that $\wf{s_{\pm}}(\{v_{\tau}\}) \cap (\partial \PP\backslash U)=\varnothing$, i.e. prove a uniform elliptic estimate.

Fix $Q\in\ps[0]$ elliptic on $\partial\PP\backslash U$ such that $\wfs(Q)\cap\Sigma=\varnothing$. For every $\sigma\in \overline{T}$, let $S_{\sigma}\in\Psi^{-2,\mathsf{0}}_{\mathrm{de,sc}}(\M)$ be a microlocal elliptic parametrix for $P\mp i\sigma$ on $\wfs(Q)$. Let $R_{\sigma}=I-S_{\sigma}(P \mp i\sigma)$, so $\wfs(R_{\sigma})\cap\wfs(Q)=\varnothing$. Then
$$
\|Q v_{\tau}\|_{\mathsf{s_{\pm}}}
=
\|Q (S_{\sigma} ((P \mp i\tau) \pm i(\tau-\sigma))+R_{\sigma})v_{\tau}\|_{\mathsf{s_{\pm}}}
\leqslant
\|Q S_{\sigma} g_{\tau}\|_{\mathsf{s_{\pm}}}
+|\tau-\sigma| \|Q S_{\sigma} v_{\tau}\|_{\mathsf{s_{\pm}}}
+\|QR_{\sigma} v_{\tau}\|_{\mathsf{s_{\pm}}}
\leqslant
$$
$$
\leqslant
\|Q S_{\sigma} g_{\tau}\|_{\mathsf{s_{\pm}}}
+|\tau-\sigma| \| S_{\sigma} Q v_{\tau}\|_{\mathsf{s_{\pm}}}
+|\tau-\sigma| \| [Q,S_{\sigma}] v_{\tau}\|_{\mathsf{s_{\pm}}}
+\|QR_{\sigma} v_{\tau}\|_{\mathsf{s_{\pm}}}
\leqslant
$$
$$
\leqslant
C_{\sigma}\Big(
\|g_{\tau}\|_{s_f-2,\mathsf{s_{base}}}
+|\tau-\sigma| \|Q v_{\tau}\|_{s_f-2,\mathsf{s_{base}}}
+ \| v_{\tau}\|_{s_f-3,\mathsf{s_{base}-1}}
\Big)
\leqslant
C_{\sigma}\Big(
\|g_{\tau}\|_{\Y^{\mathsf{s_{\pm}}}}
+|\tau-\sigma| \|Q v_{\tau}\|_{\mathsf{s_{\pm}}}
+ \|v_{\tau}\|_{\mathsf{s'}}
\Big).
$$
Then for fixed $\sigma$, for any $\tau\in [\min(0,\sigma-\frac{1}{2 C_{\sigma}}), \sigma+\frac{1}{2C_{\sigma}}]$ we have
$$
\|Q v_{\tau}\|_{\mathsf{s_{\pm}}}
\leqslant
2 C_{\sigma} \Big(
\|g_{\tau}\|_{\Y^{\mathsf{s_{\pm}}}} +\|v_{\tau}\|_{\mathsf{s'}} \Big).
$$
Then since $\tau$ takes values in a bounded set $T$ which we can cover by a finite number of neighborhoods of the form $[\min(0,\sigma-\frac{1}{2 C_{\sigma}}), \sigma+\frac{1}{2C_{\sigma}}]$ for $\sigma\in \overline{T}$, we conclude that the family $Qv_{\tau}$ for $\tau\in T$ is bounded in $\sob{s_{\pm}}$ and thus $\wf{s_{\pm}}(\{v_{\tau}\})\cap (\partial\PP\backslash U)=\varnothing$. Combined with the propagation results, we conclude that $v_{\tau}$ is bounded in $\sob{s_{\pm}}$.

If $0\notin\overline{T}$, then for the elliptic estimate one can instead work in the complement of a small neighborhood of $\Sigma\cap\Gamma_f$, which is the characteristic set of $P \mp i\sigma$ for any $\sigma>0$. The propagation estimates within $\Sigma\cap \Gamma_f$ do not require the threshold conditions to hold, so they are not necessary for the result in this case.
\end{proof}

As the proof shows, despite the fact that the $\varepsilon>0$ inverses exist even without the threshold conditions, which can be shown using propagation of singularities only in $\Sigma\cap\Gamma_f$, the thresholds are necessary for the uniform boundedness result for small $\varepsilon$ because the elliptic estimate away from a small neighborhood of $\Sigma\cap\Gamma_f$ is not uniform as $\varepsilon\to 0^+$. Thus it is necessary to use propagation estimates in the full characteristic set $\Sigma$ of $P$, away from which the elliptic estimate \textit{is} uniform.

\begin{prop}
Assume that the Feynman/anti-Feynman realization of $P$ is invertible (implied e.g. by the assumptions of Theorem~\ref{thm:invertibility}). For any set of orders $\mathsf{s_{\pm}}$ satisfying the Feynman ($\mathsf{s_+}$)/anti-Feynman ($\mathsf{s_-}$) conditions (including the threshold conditions at $I^T$), there exists $C>0$ such that all $\varepsilon\in (0,1)$ and $f\in \Y^{\mathsf{s_{\pm}}}$ satisfy
\begin{equation}
\|(P \mp i\varepsilon)^{-1}f\|_{\mathsf{s_{\pm}}} \leqslant C \|f\|_{\Y^{\mathsf{s_{\pm}}}}.
\end{equation}
\label{thm:invertibility-uniform}
\end{prop}

\begin{proof}
Since $(P\mp i\varepsilon):\X^{\mathsf{s_{\pm}}}_{\varepsilon}\to\Y^{\mathsf{s_{\pm}}}$ is bounded and invertible, the inverses are also bounded, which implies boundedness as maps $(P \mp i\varepsilon)^{-1}:\Y^{\mathsf{s_{\pm}}}\to\sob{s_{\pm}}$. Then by the uniform boundedness principle, it is enough to show that for every $f\in\Y^{\mathsf{s_{\pm}}}$, the family $(P \mp i\varepsilon)^{-1}f$ is bounded in $\sob{s_{\pm}}$. Let us assume this is not the case, so there exists $f\in\Y^{\mathsf{s_{\pm}}}$ and a sequence $\varepsilon_n \in (0,1)$ such that $u_n=(P \mp i\varepsilon_n)^{-1} f$ has $\|u_n\|_{\mathsf{s_{\pm}}}\to +\infty$.

Take any $\mathsf{s'}$ such that $\mathsf{s_{\pm}-1}<\mathsf{s'}<\mathsf{s_{\pm}}$. By Lemma~\ref{thm:bound-uniform}, $u_n$ must be unbounded in $\sob{s'}$ as well. Then define $v_n = \frac{u_n}{\|u_n\|_{\mathsf{s'}}}$ and $g_n=\frac{f}{\|u_n\|_{\mathsf{s'}}}$, so we have $v_n=(P \mp i\varepsilon_n)^{-1} g_n$ while $\|v_n\|_\mathsf{s'}=1$.

There exists a sequence $n_k$ such that $\lim_{k\to\infty}\|u_{n_k}\|_{\mathsf{s'}}=+\infty$ and therefore $\lim_{k\to\infty} g_{n_k}=0$ in $\Y^{\mathsf{s_{\pm}}}$. In particular, $g_{n_k}$ is bounded in $\Y^{\mathsf{s_{\pm}}}$, so by Lemma~\ref{thm:bound-uniform} $v_{n_k}$ is bounded in $\sob{s_{\pm}}$. Since $\sob{s_{\pm}}$ is compactly embedded in $\sob{s'}$, we can choose the sequence $n_k$ so that $\lim_{k\to\infty} v_{n_k}=v\in\sob{s_{\pm}}$ in the topology of $\sob{s'}$. Additionally, since we are considering $\varepsilon_n\in (0,1)$, we can choose the sequence so that $\lim_{k\to\infty}\varepsilon_{n_k}=\varepsilon\geqslant 0$. Then we calculate, taking limits in $\sch'$,
$$
Pv = \lim_{k\to\infty} Pv_{n_k} = \lim_{k\to\infty} (g_{n_k} \pm  i\varepsilon_{n_k} v_{n_k}) = \pm i\varepsilon v.
$$
Since $\ker (P \mp i\varepsilon)|_{\sob{s_{\pm}}} =0$ (for $\varepsilon=0$, this uses the invertibility assumption), this means $v=0$; but this contradicts the fact that $\|v_n\|_{\mathsf{s'}}=1$. This completes the proof.
\end{proof}

\begin{theorem}[Limiting absorption principle]
\label{thm:lim-abs}
Assume that the Feynman/anti-Feynman realization of $P$ is invertible (implied e.g. by the assumptions of Theorem~\ref{thm:invertibility}). Then for any set of orders $\mathsf{s_{\pm}}$ satisfying the Feynman ($\mathsf{s_+}$)/anti-Feynman ($\mathsf{s_-}$) conditions (including the threshold conditions at $I^T$), for any $f\in\Y^{\mathsf{s_{\pm}}}$ we have
$\lim_{\varepsilon\to 0^+} (P\mp i\varepsilon)^{-1}f = P_{\pm}^{-1}f$ in $\sob{s_{\pm}}$. In other words, $\lim_{\varepsilon\to 0^+}(P\mp i\varepsilon)^{-1}=P_{\pm}^{-1}$ in the strong operator topology on $\mathcal{L}(\Y^{\mathsf{s_{\pm}}},\sob{s_{\pm}})$.
\end{theorem}

\begin{proof}
Since the Feynman/anti-Feynman inequalities are an open condition, there exists some $\mathsf{s''}>\mathsf{s_{\pm}}$ satisfying them. For any $g\in\Y^{\mathsf{s''}}$, the family $w_{\varepsilon}=(P \mp i\varepsilon)^{-1}g$ is bounded in $\sob{s''}$, which is compactly embedded in $\sob{s_{\pm}}$, so any sequence $\varepsilon_n\in (0,1)$ has a subsequence $\varepsilon_{n_k}$ such that $w_{\varepsilon_{n_k}}$ converges in $\sob{s_{\pm}}$ to some $w$. If $\varepsilon_n\to 0$, then for any such subsequence we can write $Pw_{\varepsilon_{n_k}} = g \pm i\varepsilon_{n_k} w_{\varepsilon_{n_k}} \to g $ in $\sch'$ as $k\to\infty$, so we must have $w=P_{\pm}^{-1}g$. This implies that $\lim_{\varepsilon\to 0^+} w_{\varepsilon} = P_+^{-1}g$ in $\sob{s^{\pm}}$. Then for any $f\in\Y^{\mathsf{s_{\pm}}}$, $g\in\Y^{\mathsf{s''}}$ we can write
$$
\|(P \mp i\varepsilon)^{-1} f - P_{\pm}^{-1}f \|_{\mathsf{s_{\pm}}}
\leqslant
\|(P \mp i\varepsilon)^{-1} (f-g) \|_{\mathsf{s_{\pm}}}
+\|(P \mp i\varepsilon)^{-1} g - P_{\pm}^{-1}g \|_{\mathsf{s_{\pm}}}
+\|P_{\pm}^{-1}(f-g)\|_{\mathsf{s_{\pm}}}
\leqslant
$$
$$
\leqslant
\|(P \mp i\varepsilon)^{-1} g - P_{\pm}^{-1}g \|_{\mathsf{s_{\pm}}} + 2C\|f-g\|_{\Y^{\mathsf{s_{\pm}}}},
$$
where we used Proposition~\ref{thm:bound-uniform} in the last line. Since $g$ can be taken arbitrarily close to $f$ in $\Y^{\mathsf{s_{\pm}}}$ while the first term converges to zero for any fixed $g$, we finally conclude that $(P \mp i\varepsilon)^{-1}f \to P_{\pm}^{-1}f$ in $\sob{s_{\pm}}$.
\end{proof}

We note that this characterization of the Feynman propagator means that it is invariant, for example as a map $P_+^{-1}:C_c^{\infty}(\M^{\circ}) \to \mathcal{D}'(\M^{\circ})$, under any change of compactification which preserves the Schwartz space $\sch$ as a set (as well as all our conditions on the metric). To see this, let $(\M,\g)$ be as elsewhere in the paper and such that the Feynman and anti-Feynman realizations of $P$ are invertible, and let $\mathcal{N}$ be another compactification of $\M^{\circ}$ of the form described in Section~\ref{sec:topology} such that $(\mathcal{N},\g)$ also satisfies all of the assumptions in Section~\ref{sec:metric} and such that the spaces $\sch$ for $\M$ and $\mathcal{N}$ coincide as sets of functions on their identified interiors. First of all, by Proposition~\ref{thm:Feynman-kernel} the kernel of the Feynman/anti-Feynman realizations of $P$ on $\mathcal{N}$ are subspaces of $\sch$, while the invertibility of the corresponding realizations on $\mathcal{M}$ implies the absence of Schwartz solutions $u$ to $Pu=0$, so the Feynman/anti-Feynman realizations on $\mathcal{N}$ must be invertible as well. Next, $(P-i\varepsilon)^{-1}$ for any $\varepsilon>0$ is a bijection $\sch\to\sch$, for any $f\in\sch$ giving the unique solution $u_{\varepsilon}\in\sch$ to $(P-i\varepsilon)u_{\varepsilon} = f$, which due to this uniqueness does not depend on the choice of compactification. Finally, by Theorem~\ref{thm:lim-abs} $u_{\varepsilon}\to P_+^{-1}f$ as $\varepsilon\to 0^+$ in the topology of $\sch'$ on either $\M$ or $\mathcal{N}$, both of which imply convergence in the compactification-independent topology of $\mathcal{D}'(\M^{\circ})$; since the families $u_{\varepsilon}$ do not depend on the compactification, the limits must coincide. Thus, $P_+^{-1}f$ does not depend on the compactification used to define the propagator.

In other words, for asymptotically flat spacetimes which admit compactifications of the form we consider, while there are potentially many different ways to compactify the spacetime consistently with the asymptotically flat structure, our prescription yields a single Feynman propagator canonically defined by the metric $\g$ as long as there are no competing compactifications which differ drastically enough that their Schwartz spaces do not agree.

\appendix

\section{Timelike and spacelike infinity for Minkowski-like metrics}
\label{sec:Minkowski-like}
Our general assumptions on the metric do not prescribe its form at spacelike or timelike infinity, so we cannot precisely describe the Hamilton flow there and need to rely on our causal structure/non-trapping assumptions for the qualitative picture instead. In this appendix, we describe the analysis at spacelike and timelike infinity for the special case of Minkowski-like metrics, since they provide the motivation for our general assumptions. Recall that we define the class of Minkowski-like metrics by
	$$
	\g = -v\frac{d\rho^2}{\rho^4} + \frac{d\rho\os dv}{\rho^3} + \frac{\h}{\rho^2} \mod S^{-\epsilon}(\M;\mathrm{Sym}^2(\ttm)),
	$$
where $\h\in C^{\infty}(\tilde{\M};T^*\tilde{\M})$ induces a smooth Riemannian metric on $Y$.

In a full neighborhood $U_0$ of $I^0$ a choice of defining functions of $I^0$, $\scri$ respectively is $\rho_0 = \frac{2\rho}{u_2\rho-v}$, $x_0=\sqrt{\frac{u_2\rho-v}{2}}$. Similarly, in a full neighborhood $U_T$ of $I^T$ a choice of defining functions of $I^T$, $\scri$ respectively is $\rho_T= \frac{2\rho}{v-u_1\rho}$, $x_T=\sqrt{\frac{v-u_1\rho}{2}}$. The metric then satisfies
\begin{equation}
\g|_{\partial\M \cap U_0} = 2\frac{d\rho_0^2}{\rho_0^4 x_0^2} + 4\frac{d\rho_0 \os dx_0}{\rho_0^3 x_0^3} + \frac{\h}{\rho_0^2 x_0^4},
\hspace{30pt}
\g|_{\partial\M \cap U_T} = -2\frac{d\rho_T^2}{\rho_T^4 x_T^2} - 4\frac{d\rho_T \os dx_T}{\rho_T^3 x_T^3} + \frac{\h}{\rho_T^2 x_T^4} ,
\label{eq:Minkowski-like-desc}
\end{equation}
i.e. compared to our general assumptions, the metric decays to this form not only at null infinity but also nearby at timelike and spacelike infinity.

\subsection{Near the interior of timelike infinity}
	Near timelike infinity away from null infinity, where $v>0$, we can define a new boundary-defining function $\hat{\rho}=\frac{\rho}{\sqrt{v}}$. Then a calculation shows that as a (de,)sc-metric, $\g$ is asymptotically of product form:
	\begin{equation}
\g|_{(I^{\pm})^{\circ}}
=
-\frac{d\hat{\rho}^2}{\hat{\rho}^4} + \frac{\hat{\h}}{\hat{\rho}^2},
	\end{equation}
where $\hat{\h}$ is the smooth Riemannian metric induced on $(I^{\pm})^{\circ}$ by $\frac{dv^2+4v\h}{4v^2}$. (We know $\hat{\h}$ must be Riemannian because we assume $\g$ is Lorentzian as a de,sc-metric). The local analysis is then similar in many respects to that of the Helmholtz operator on asymptotically conic manifolds, or Riemannian scattering spaces, as studied by Melrose \cite{Melrose-AES}.

Consider local coordinates $(y_1,\ldots,y_d)$ on a region $U_{\pm}\subset (I^{\pm})^{\circ}$. Consider a product neighborhood $\U_{\pm} \simeq [0,\varepsilon)_{\hat{\rho}} \times (U_{\pm})_y$ with $\U_{\pm}\cap\partial\M = U_{\pm}$. Let $(\hat{\xi},\hat{\eta}_1,\ldots,\hat{\eta}_d)$ be the (de,)sc-dual variables to $(\hat{\rho},y_1,\ldots,y_d)$. The principal symbol of $P$ over $\U_{\pm}$ is then 
\begin{equation}
p=-\hat{\xi}^2+\|\hat{\eta}\|^2_{\hat{\h}}+m^2 +Q,
\label{eq:symbol-sc}
\end{equation}
where $Q$ is a quadratic form in $\hat{\xi},\hat{\eta}$ with coefficients in $S^{-\epsilon}(\M)$. We see that there is $C>0$ such that $p>C(1+\hat{\xi}^2+\hat{\eta}^2)$ on the set $\{\hat{\rho}=0,\ \hat{\xi}=0\} \subset \pi^{-1}\U_{\pm}\backslash \Gamma_f$. Therefore, for any choice of defining functions for rescaling, $\tilde{p}$ does not vanish in a neighborhood of this set or its closure in $\PP$. Then we define two regions of $\pi^{-1}\U_{\pm}$ by
$$
\V_{\pm}^+ = \overline{\{\pm\hat{\xi}>\varepsilon\}}\cap \pi^{-1}\U_{\pm},
\hspace{30pt}
\V_{\pm}^- 
= \overline{\{\mp\hat{\xi}<\varepsilon\}}\cap \pi^{-1}\U_{\pm}.
$$
For $\varepsilon>0$ small enough, over a neighborhood of $U_{\pm}$, $\tilde{\Sigma}\cap\pi^{-1}\U_{\pm}$ is contained in the union of these two regions. In each of these regions, we introduce fiber coordinates well-defined on the compactification:
\begin{equation}
\hat{\varrho} = \frac{1}{|\hat{\xi}|},
\hspace{30pt}
\hat{\theta} = \frac{\hat{\eta}}{\hat{\xi}}.
\end{equation}
$\hat{\varrho}$ is a defining function of fiber infinity. In terms of these coordinates, the symbol restricted to $I^{\pm}$ is
$$
p = \frac{1}{\hat{\varrho}^2}(\|\hat{\theta}\|_{\tilde{\h}}^2 + m^2\hat{\varrho}^2-1).
$$

Starting from Eqs.~(\ref{eq:Hp-coordinates-general-sc}), (\ref{eq:symbol-sc}) and calculating the derivatives, we find that over $\U_{\pm}$ at finite frequencies,
\begin{equation}
H_p = \hat{\rho} \Bigg[
-(2\hat{\xi}+L_1)\hat{\rho} \frac{\partial}{\partial\hat{\rho}}
+ \sum_{i,j=1}^d (2\hat{h}^{ij}\hat{\eta}_j + L_2^i) \frac{\partial}{\partial y_i}
- (2\|\hat{\eta}\|_{\hat{\h}}^2 + Q_1) \frac{\partial}{\partial \hat{\xi}}
- \sum_{i=1}^d \left(2\hat{\xi}\hat{\eta}_i + \sum_{j,k=1}^d\frac{\partial \hat{h}^{jk}}{\partial y_i}\hat{\eta}_j\hat{\eta}_k + Q_2^i\right) \frac{\partial}{\partial \hat{\eta}_i}
\Bigg],
\end{equation}
where $L_1,L_2^i$ are linear forms and $Q_1,Q_2^i$ quadratic forms in $(\hat{\xi},\hat{\eta})$ with coefficients in $S^{-\epsilon}(\M)$. In terms of the coordinates valid on the compactification,
\begin{equation}
\begin{split}
H_p|_{\V^{\pm}_{\pm}}
=
\hat{\varrho}^{-1}\hat{\rho} \Bigg[
-2\hat{\rho} \frac{\partial}{\partial\hat{\rho}}
+ 2\sum_{i,j=1}^d \hat{h}^{ij}\hat{\theta}_j \frac{\partial}{\partial y_i}
+ 2\|\hat{\theta}\|_{\hat{\h}}^2 \hat{\varrho}\frac{\partial}{\partial\hat{\varrho}}
+
 \sum_{i=1}^d \left(2(\|\hat{\theta}\|_{\hat{\h}}^2-1)\hat{\theta}_i - \sum_{j,k=1}^d\frac{\partial \hat{h}^{jk}}{\partial y_i}\hat{\theta}_j \hat{\theta}_k \right) \frac{\partial}{\partial \hat{\theta}_i}
\Bigg]
\\
\mod \hat{\varrho}^{-1} \hat{\rho} S^{-\varepsilon}\Vb(\PP).
\end{split}
\label{eq:Hp-timelike}
\end{equation}
On $\V^{\pm}_{\mp}$, it is given by the same expression with the overall sign changed.

Since $\hat{\h}$ is nondegenerate, for the $\frac{\partial}{\partial y_i}$ terms to all vanish we must have $\hat{\theta}=0$, so we see that the vanishing set of $\Hp$ in $\Sigma$ in any of these regions is $\{\hat{\rho}=0,\ \hat{\theta}=0,\ \hat{\varrho}=\frac{1}{m}\}$. Thus we find that there are no radial points at fiber infinity, while at finite frequency there is one radial point in each sheet of the characteristic set over each point of the interior of $I^{\pm}$, which we identify as the set $\mathcal{R}_{\pm}\backslash\scri^{\pm}$.

We now check that $\mathcal{R}$ is nondegenerate in the sense that all points of it satisfy the assumptions of Theorem~\ref{thm:localized-rp-main}. The coordinate description we just found shows that, away from the corner with $\scri$, $\mathcal{R}$ is defined by the vanishing of several symbols as required. To check this also at the corner, we consider local coordinates $(\rho_T,x_T,y_1,\ldots,y_{d-1})$ on $\U_T$ like we did in Section~\ref{sec:Hamiltonian}, with defining functions $\rho_T,x_T$ as introduced above Eq.~(\ref{eq:Minkowski-like-desc}). For any $U_{\pm}$ which overlaps $\U_T\cap I^{\pm}$ we take coordinates $y$ on the $U_{\pm}$ factors of $\U_{\pm}$ such that $y_1,\ldots,y_{d-1}$ agree with local coordinates on the $Y$ factors of $\U_T$ while $y_d=\sqrt{v}$. Then a calculation shows that over the dual variables we used above over $\U_{\pm}$ and in Section~\ref{sec:Hamiltonian} over $\U_T$ are related over $(I^{\pm})^{\circ}$ by
$$
\hat{\xi}=\frac{\zeta_T}{\sqrt{2}},
\hspace{30pt}
\hat{\eta}_d=2(\xi_T-\zeta_T),
\hspace{30pt}
\hat{\eta}_i = \frac{\eta_i}{\sqrt{2}x_T}
\text{ for }i=1,\ldots,d-1
$$
The variables valid on the compactification are therefore related by $\hat{\varrho}= \sqrt{2}\varrho$, $\hat{\theta}_d = 2\sqrt{2}\omega$, $\hat{\theta}_i = \frac{\theta_i}{x_T}$ for $i=1,\ldots,d-1$. Then near the corner, $\mathcal{R}$ is defined by $\{\rho_T=0,\ \theta=0,\ \omega=0,\ \varrho=\frac{1}{\sqrt{2}m}\}$, which agrees with our calculation at the corner in Section~\ref{sec:Hamiltonian} and shows that $\mathcal{R}$ is indeed defined by the vanishing of several symbols as required. The signs of the $\hat{\rho}\frac{\partial}{\partial\hat{\rho}}$ components identify the component of $\Sigma$ in $\V^{\pm}_+$ and $\V^{\pm}_-$ as $\Sigma^{\pm}$.

The required properties of the linearization at the corner follow from the previous eigenvalue calculations. Away from the corner, consider $\alpha\in\mathcal{R}^+_+\backslash\scri^+\subset I^+$. Using $\hat{\varrho}^{-1}\hat{\rho}$ to rescale the vector field, we get $\Hp(\alpha) = - 2 \hat{\rho}\frac{\partial}{\partial\hat{\rho}}$ as a b-vector, so $\lambda_{I^+}=- 2$. Meanwhile
$$
L(\Hp|_{I^+})(\alpha)
=
2\sum_{i,j=1}^d \hat{h}^{ij}\hat{\theta}_j \frac{\partial}{\partial y_i}
-2 \sum_{i=1}^d \hat{\theta}_i \frac{\partial}{\partial \hat{\theta}_i}
$$
The eigenvector-eigenvalue pairs are $\left(\frac{\partial}{\partial\hat{\theta}_i}-\sum_{j=1}^d \hat{h}^{ij}(\alpha)\frac{\partial}{\partial y_j},-2\right)$, $\left(\frac{\partial}{\partial\hat{\varrho}},0\right)$, $\left(\frac{\partial}{\partial y_i},0\right)$. The vector $\frac{\partial}{\partial\hat{\varrho}}$ is transverse to $\Sigma\cap I^+$, so $\alpha$ is indeed nondegenerate. The analysis at the other components of $\mathcal{R}$ is similar.

This establishes the existence of $\mathcal{R}$ as postulated in our general assumptions for any Minkowski-like metric. As we now show, in fact the assumption on the dynamics at finite frequency over $I^{\pm}$ is necessarily satisfied as well.

The fiber coordinate $\hat{\xi}$ has a meaning as the value of a (de,)sc-one-form on the (de,)sc-vector $\hat{\rho}^2\frac{\partial}{\partial\hat{\rho}}$ independently of the choice of local coordinates on $I^{\pm}$. Therefore, in either $\V^+_{\pm}$ or $\V^-_{\pm}$ the coordinate $\hat{\varrho}$ is also defined independently of coordinates on the boundary and consistently between any overlapping coordinate charts $U_{\pm}$. Meanwhile, from Eq.~(\ref{eq:Hp-timelike}) we can read off that $\hat{\varrho}$ is monotone increasing along $\Hp$ in $\V^{\pm}_{\pm}$ and monotone decreasing in $\V^{\pm}_{\mp}$ at a rate which is bounded from below away from any fixed neighborhoods of $\Gamma_f=\{\hat{\varrho}=0\}$ and $\mathcal{R}=\{\hat{\varrho}=\frac{1}{m}\}$ within $\Sigma$. This, combined with our knowledge of the flow over $I^{\pm}\cap\scri^{\pm}$, means that all finite-frequency bicharacteristics over $I^{\pm}$ in $\Sigma^{\pm}$ must limit to $\mathcal{R}_{\pm}^{\pm}$ in the forward direction and to $\Gamma_f$ in the backward direction, whereas in $\Sigma^{\mp}$ they must limit to $\Gamma_f$ in the forward direction and to $\mathcal{R}^{\mp}_{\pm}$ in the backward direction.

\subsection{Near the interior of spacelike infinity}
Similarly, near spacelike infinity away from null infinity, where $v<0$, we can define a new boundary-defining function $\check{\rho}=\frac{\rho}{\sqrt{-v}}$. Then a calculation shows
	\begin{equation}
\g|_{(I^0)^{\circ}}
=
\frac{d\check{\rho}^2}{\check{\rho}^4}
+ \frac{\check{\h}}{\check{\rho}^2},
	\end{equation}
where $\check{\h}$ is the smooth metric, this time nondegenerate Lorentzian, induced on $(I^0)^{\circ}$ by $-\frac{dv^2+4v\h}{4v^2}$.

Consider again local coordinates $(y_1,\ldots,y_d)$ on a region $U\subset (I^0)^{\circ}$. Consider a product neighborhood $\U\simeq [0,\varepsilon)_{\check{\rho}} \times U_y$ with $\U\cap\partial\M=U$. Let $(\check{\xi},\check{\eta_1},\ldots,\check{\eta_d})$ be the (de,)sc-dual variables to $(\check{\rho},y_1,\ldots,y_d)$. The principal symbol of $P$ over $\U$ is
\begin{equation}
p = \check{\xi}^2 + \sum_{i,j=1}^d \check{h}^{ij}\check{\eta}_i\check{\eta}_j +m^2 +Q,
\end{equation}
where $Q$ is a quadratic form in $\check{\xi},\check{\eta}$ with coefficients in $S^{-\epsilon}(\M)$. The Hamilton vector field is
\begin{equation}
\begin{split}
H_p = \check{\rho} \Bigg[
(2\check{\xi}+L_1)\check{\rho} \frac{\partial}{\partial\check{\rho}}
+ \sum_{i,j=1}^d (2\check{h}^{ij}\check{\eta}_j + L_2^i) \frac{\partial}{\partial y_i}
- \left(2\sum_{i,j=1}^d \check{h}^{ij}\check{\eta}_i \check{\eta}_j  + Q_1\right) \frac{\partial}{\partial \check{\xi}}
+
\\
+ \sum_{i=1}^d \left(2\check{\xi}\check{\eta}_i - \sum_{j,k=1}^d\frac{\partial \check{h}^{jk}}{\partial y_i}\check{\eta}_j\check{\eta}_k + Q_2^i\right) \frac{\partial}{\partial \check{\eta}_i}
\Bigg],
\end{split}
\end{equation}
where $L_1,L_2^i$ are linear forms and $Q_1,Q_2^i$ quadratic forms in $(\check{\xi},\check{\eta})$ with coefficients in $S^{-\epsilon}(\M)$. We see that there is $C>0$ such that $p>C(1+\check{\xi}^2+\check{\eta}^2)$ on the set $\{\check{\eta}=0\}$, so any point of the characteristic set has a neighborhood on whose part in $\PP^{\circ}$ at least one of the $\check{\eta}_i$ must be nonzero. Then, since $\check{\h}$ is nondegenerate, the $\frac{\partial}{\partial y_i}$ components of $H_p$ cannot all vanish simultaneously on $\Sigma$.

Thus we see that there are no radial points over the interior of spacelike infinity. Unlike the situation at timelike infinity, we do not expect to also automatically get the non-trapping property over spacelike infinity for Minkowski-like metrics: the bicharacteristic flow at finite frequency is related to timelike geodesics, which are, roughly speaking, asymptotically transverse to the boundary at $I^{\pm}$ but tangent at $I^0$, so the behavior at $I^0$ can be expected to be less universal.

\subsection{Asymptotically Minkowski metrics}
In this section, we record the calculations showing that the non-trapping conditions at spacelike infinity and at infinite frequency over timelike infinity, which are not automatic for general Minkowski-like metrics, are satisfied for asymptotically Minkowski metrics. Since terms in the rescaled symbol which vanish at the boundary give rise to terms in $\Hp$ which vanish likewise and which moreover do not affect relevant properties of the linearization at any radial sets over spacetime infinity, it suffices to do the calculations for the exact Minkowski metric. It also suffices to work in the scattering phase space over the radial compactification since the non-trapping conditions just specify that the bicharacteristics of interest limit to null infinity.

Let $\tilde{\M}$ be the radial compactification of $\R^{d+1}$. Let $(t,x_1,\ldots,x_d)$ be the global inertial coordinates. Let $(\omega,k_1,\ldots,k_d)$ be the canonical dual variables, which turn out to be valid fiber-linear coordinates on ${}^{\mathrm{sc}}T^*\tilde{\M}$ up to and including the boundary. Then the Minkowski Klein-Gordon operator $P=\partial_t^2 - \sum_{i=1}^d \partial_{x_i}^2 +m^2$ has principal symbol $p=-\omega^2 + |k|^2+m^2$ and Hamilton vector field
$$H_p = -2\omega\frac{\partial}{\partial t} + 2\sum_{i=1}^d k_i\frac{\partial }{\partial x_i}.$$
In a neighborhood of either component of $\Sigma$, a defining function of fiber infinity is $\varrho=\frac{1}{|\omega|}$ and valid coordinates on fiber infinity are $\theta_i = \frac{k_i}{\omega}$. The symbol is $p=\frac{1}{\varrho^2}(|\theta|^2+m^2\varrho^2-1)$.

At timelike infinity, a valid boundary-defining function is $\rho=\frac{1}{|t|}$ and valid coordinates on the boundary are $y_i = \frac{x_i}{t}$. Then in terms of these, considering the region $t,\omega>0$ (the others are similar),
$$
H_p = 2\varrho^{-1}\rho \left( \rho\frac{\partial}{\partial\rho}
+\sum_{i=1}^d (y_i+\theta_i)\frac{\partial}{\partial y_i}
\right).
$$
Thus the radial set is defined by $\theta_i=-y_i$ and $\varrho=\frac{1}{m}\sqrt{1-|y|^2}$, where $|y|<1$ since we are at timelike infinity. At infinite frequency, in $\Sigma$ we have $|\theta|=1$. Since the value of $\theta$ is constant along $H_p$, from the $\frac{\partial}{\partial y_i}$ terms we can see that the bicharacteristics are straight lines in the $y$ variables radiating from $-\theta$ and, if one starts from a point with $|y|<1$ and goes in either direction, reaching any given neighborhood of $|y|=1$ in finite parameter time, which means that any infinite-frequency bicharacteristic limits to $\scri$ in both directions.

Spacelike infinity is covered by neighborhoods in each of which a valid boundary-defining function is $\rho=\frac{1}{|x_i|}$ for some $i$ and valid coordinates on the boundary are $s=\frac{t}{x_i}$, $y_j=\frac{x_j}{x_i}$ for $j\neq i$. Consider the region with $i=d$, $x_d>0$, $\omega>0$ (the others are similar). Then
$$
H_p = 
2\varrho^{-1}\rho
\left(
-\theta_d \rho \frac{\partial}{\partial\rho}
- (1+s\theta_d) \frac{\partial}{\partial s} 
+  \sum_{i=1}^{d-1} (\theta_i-y_i\theta_d) \frac{\partial}{\partial y_i}
\right).
$$
Since $|\theta|\leqslant 1$ in $\Sigma$ and $|s|<1$ in the interior of spacelike infinity, we see that $s$ is monotone along the flow and, if one starts at a point with $|s|<1$ and goes in either direction, reaches any neighborhood of $s=1$ in one direction or $s=-1$ in the other in finite parameter time, which means that any bicharacteristic limits to $\scri^+$ in one direction and $\scri^-$ in the other.

\section{Time function for Minkowski-like metrics}
\label{sec:time-fcn}
In this appendix, we construct a symbolic time function on any Minkowski-like spacetime which is first-order and has uniformly timelike differential near timelike and null infinity, showing that Theorem~\ref{thm:invertibility} applies to such spacetimes provided they also satisfy the non-trapping assumptions.

Consider a constant $u>0$ and $\chi\in C^{\infty}(\R)$ such that $\chi(s)=0$ for $s\leqslant \frac{1}{3}u$, $\chi$ is monotone increasing on $[\frac{1}{3}u,\frac{2}{3}u]$, and $\chi(s)=1$ for $s\geqslant \frac{2}{3}u$. Define
\begin{equation}
t = \chi\left(\frac{v}{\rho}\right) \frac{\sqrt{v}}{\rho} + \left( 1-\chi\left(\frac{v}{\rho}\right)\right) \frac{1}{\sqrt{u\rho -v}}.
\label{eq:time-function}
\end{equation}

\begin{itemize}

\item In the region where $\frac{v}{\rho}>0$, we can define $\rho_T = \frac{\rho}{v}$ and $x_T = \sqrt{v}$, which are local defining functions of $I^T$ and $\scri$ respectively. In terms of these variables,
$$
t = \Big(
\chi(\rho_T^{-1})\rho_T^{-1} + \Big(1-\chi(\rho_T^{-1})\Big) \Big(u\rho_T-1\Big)^{-\frac{1}{2}}
\Big) x_T^{-1},
$$
which shows that $t$ is a classical symbol of order 1 at all of timelike infinity and at the part of null infinity where $\frac{v}{\rho}>0$.

\item In the region where $\frac{v}{\rho}<u$, we can define $\rho_0=\frac{\rho}{u\rho-v}$ and $x_0=\sqrt{u\rho-v}$, which are local defining functions of $I^0$ and $\scri$ respectively. In terms of these variables,
$$
t = \Big(
\chi(u-\rho_0^{-1}) (u\rho_0-1)^{\frac{1}{2}} \rho_0^{-1}
+ \Big( 1- \chi(u-\rho_0^{-1})\Big)
\Big) x_0^{-1},
$$
which, since the first term is supported away from $I^0$, shows that $t$ is a classical symbol of order 0 at all of spacelike infinity and of order 1 at the part of null infinity where $\frac{v}{\rho}<u$.
\end{itemize}

Thus we conclude that $t\in S^{1,1,0,1,1}_{\mathrm{cl}}(\M)$ globally, and directly from the expressions we see that $\rho_{I^{\pm}}\rho_{\scri^{\pm}}t>0$ at $I^{\pm}\cup\scri^{\pm}$. Next, we check that its differential is timelike at $I^{\pm}\cup \scri^{\pm}$ with the uniform bound required by Theorem~\ref{thm:invertibility}, i.e. we check the values of $\rho_{I^0}^{-2}\g^{-1}(d t,d t)$ on $I^{\pm}\cup\scri^{\pm}$.
\begin{itemize}
\item In the neighborhood of timelike infinity where $\frac{v}{\rho}>\frac{2}{3}u$, we have $\chi(v/\rho)=1$ and $t=\frac{\sqrt{v}}{\rho}=\hat{\rho}^{-1}$, where $\hat{\rho}$ appeared in Appendix~\ref{sec:Minkowski-like}. Then $dt = -\frac{d\hat{\rho}}{\hat{\rho^2}}$, and using the product form of $\g|_{(I^{\pm})^{\circ}}$ found in Appendix~\ref{sec:Minkowski-like} we see that $\g^{-1}(dt,dt)=-1$ at any point of $I^{\pm}$.

\item In the region where $\frac{v}{\rho}>0$,
$$
dt = 
\left(
\left( \frac{1}{\rho_T} 
- \frac{1}{\sqrt{u\rho_T-1}} 
\right)  \chi'
-\chi
-\frac{u \rho_T^2}{2(u\rho_T-1)^{\frac{3}{2}}} (1-\chi) 
\right) 
\frac{d\rho_T}{\rho_T^2 x_T}
- 
\left(\chi + 
\frac{\rho_T}{\sqrt{u\rho_T-1}} (1-\chi)\right) \frac{dx_T}{\rho_T x_T^2},
$$
where $\chi$, $\chi'$ are always evaluated at $\rho_T^{-1}$. At points of $\scri$ in this region, the dual metric has the matrix $\begin{pmatrix}
0 & -1 \\
-1 & 1
\end{pmatrix}$ with respect to the frame $(\frac{d\rho_T}{\rho_T^2 x_T}, \frac{d x_T}{\rho_T x_T^2})$ (suppressing the other coordinates). Then we calculate
$$
\g^{-1}(dt,dt) = \left(
2 \left( \frac{1}{\rho_T} 
- \frac{1}{\sqrt{u\rho_T-1}} 
\right)  \chi'
-\chi
-\frac{\rho_T}{(u\rho_T-1)^{\frac{3}{2}}} (1-\chi) 
\right) 
\left(\chi + 
\frac{\rho_T}{\sqrt{u\rho_T-1}} (1-\chi)\right).
$$
On the support of $\chi'$, we have $\frac{3}{2u}\leqslant \rho_T \leqslant \frac{3}{u}$, so $\frac{1}{\rho_T}-\frac{1}{\sqrt{u\rho_T-1}} \leqslant \frac{2}{3}u - \frac{1}{\sqrt{2}}$. Then we fix  $u\leqslant \frac{3}{2\sqrt{2}}$ so the $\chi'$ term above is non-positive. Then all terms in the first factor are non-positive and all in the second are non-negative, and we can estimate
$$
\g^{-1}(dt,dt)
\leqslant 
-\chi^2-\frac{1}{(u-\rho_T^{-1})^2}(1-\chi)^2
\leqslant
-\chi^2 -\frac{1}{u}(1-\chi)^2
\leqslant -\min\left(\frac{1}{4},\frac{1}{4u}\right).
$$
Thus, $\g^{-1}(dt,dt)$ has a negative upper bound on the part of $\scri$ in the region where $\frac{v}{\rho}>0$.

\item It remains to check the part of $\scri$ in the region where $\frac{v}{\rho}\leqslant 0$. In this region, $\chi(\frac{v}{\rho})=0$ identically, so $t=x_0^{-1}$ and $dt =-\rho_0 \cdot \frac{dx_0}{\rho_0 x_0^2}$. At points of $\scri$ in this region, the dual metric has the matrix $\begin{pmatrix}
0 & 1 \\
1 & -1
\end{pmatrix}$ with respect to the frame $(\frac{d\rho_0}{\rho_0^2 x_0}, \frac{d x_0}{\rho_0 x_0^2})$, so $\g^{-1}(dt,dt) = -\rho_0^2$.
\end{itemize}

Combining the results, we conclude that $t \in S_{\mathrm{cl}}^{1,1,0,1,1}(\M)$ with 
$$\Big(\rho_{I^{\pm}}\rho_{\scri^{\pm}} t\Big)|_{I^{\pm}\cup \scri^{\pm}}>0,
\hspace{45pt}
\Big(\rho_{I^0}^{-2}\g(\nabla t,\nabla t)\Big)|_{I^{\pm}\cup \scri^{\pm}}<0.$$
Thus, by Theorem~\ref{thm:invertibility}, the distinguished Fredholm realizations of $P$ are invertible on any Minkowski-like spacetime which satisfies the non-trapping assumptions.

\printbibliography
\end{document}